\documentclass[journal]{ieee/IEEEtran}

\def\ARXIVVERSION{1}

\usepackage[T1]{fontenc}
\usepackage[utf8]{inputenc}

\usepackage{amsmath,amssymb,amsfonts}
\usepackage{mathtools}
\usepackage{amsthm}

\theoremstyle{definition}
\newtheorem{theorem}{Theorem}
\newtheorem{assumption}{Assumption}
\newtheorem{definition}{Definition}
\newtheorem{lemma}{Lemma}
\newtheorem{proposition}{Proposition}
\newtheorem{corollary}{Corollary}
\newtheorem{example}{Example}
\newtheorem{remark}{Remark}

\usepackage{graphicx}
\usepackage[caption=false,font=footnotesize]{subfig}
\usepackage{placeins}
\usepackage{subfiles}
\usepackage{booktabs}
\usepackage{tabularx}
\usepackage{threeparttable}
\usepackage{multirow}
\usepackage{diagbox}
\usepackage{enumitem}
\usepackage{siunitx}
\usepackage{algorithm}
\usepackage{algpseudocode}
\usepackage{tikz}
\usetikzlibrary{positioning, arrows.meta, fit, backgrounds, calc}

\usepackage{cite}

\usepackage[hidelinks]{hyperref}

\usepackage[table,dvipsnames]{xcolor}

\usepackage{optidef}

\usepackage{acro}

\DeclareAcronym{ns}{
  short = NS,
  long  = network slice
}

\DeclareAcronym{nss}{
  short = NSS,
  long  = network slice subnet
}

\DeclareAcronym{e2e}{
  short = E2E,
  long  = end-to-end
}

\DeclareAcronym{nsr}{
  short = NSR,
  long  = network slice request
}

\DeclareAcronym{embb}{
  short = eMBB,
  long  = enhanced mobile broadband
}

\DeclareAcronym{urllc}{
  short = URLLC,
  long  = ultra-reliable and low-latency communications
}

\DeclareAcronym{mmtc}{
  short = mMTC,
  long  = massive machine-type communications
}

\DeclareAcronym{dnsr}{
  short = D-NSR,
  long  = decomposed network slice request
}

\DeclareAcronym{nsrdp}{
  short = NSR-DP,
  long  = network slice request decomposition problem
}

\DeclareAcronym{nsst}{
  short = NSST,
  long  = network slice subnet template
}

\DeclareAcronym{gst}{
  short = GST,
  long  = generic network slice template
}

\DeclareAcronym{nest}{
  short = NEST,
  long  = network slice type
}

\DeclareAcronym{bwk}{
  short = BwK,
  long  = bandits with knapsacks
}

\DeclareAcronym{lincbwk}{
  short = linCBwK,
  long  = linear contextual bandits with knapsacks
}

\DeclareAcronym{ckb}{
  short = CKB,
  long  = constrained kernel bandits
}

\DeclareAcronym{cckb}{
  short = CCKB,
  long  = contextual constrained kernel bandits
}

\DeclareAcronym{gp}{
  short        = GP,
  short-plural = s,
  long         = Gaussian process,
  long-plural  = es
}

\DeclareAcronym{ard}{
  short = ARD,
  long  = automatic relevance determination
}

\DeclareAcronym{nlpd}{
  short = NLPD,
  long  = negative log predictive density
}

\DeclareAcronym{rkhs}{
  short = RKHS,
  long  = reproducing kernel Hilbert space
}

\DeclareAcronym{cgpucb}{
  short = CGP-UCB,
  long  = contextual GP-UCB
}

\DeclareAcronym{bo}{
  short = BO,
  long  = Bayesian optimization
}

\DeclareAcronym{rade}{
  short = RADE,
  long  = real-time adaptive decomposition
}

\DeclareAcronym{rails}{
  short = RAILS,
  long  = risk-aware iterated local search
}

\DeclareAcronym{config}{
  short = CONFIG,
  long  = constrained efficient global optimization
}

\DeclareAcronym{slsqp}{
  short = SLSQP,
  long  = sequential least-squares quadratic programming
}

\DeclareAcronym{an}{
  short = AN,
  long  = access network
}

\DeclareAcronym{tn}{
  short = TN,
  long  = transport network
}

\DeclareAcronym{cn}{
  short = CN,
  long  = core network
}

\DeclareAcronym{rb}{
  short = RB,
  long  = resource block
}

\DeclareAcronym{ue}{
  short = UE,
  long  = user equipment
}

\DeclareAcronym{gnb}{
  short = gNB,
  long  = next-generation Node B
}

\DeclareAcronym{nf}{
  short = NF,
  long  = network function
}

\DeclareAcronym{upf}{
  short = UPF,
  long  = user plane function
}

\DeclareAcronym{drl}{
  short = DRL,
  long  = deep reinforcement learning
}

\DeclareAcronym{qos}{
  short = QoS,
  long  = quality of service
}

\DeclareAcronym{mab}{
  short = MAB,
  long  = multi-armed bandit
}

\DeclareAcronym{cdf}{
  short = CDF,
  long  = cumulative distribution function
}

\DeclareAcronym{pasta}{
  short = PASTA,
  long  = Poisson arrivals see time averages
}

\DeclareAcronym{fifo}{
  short = FIFO,
  long  = {first-in, first-out}
}

\DeclareAcronym{ucb}{
  short = UCB,
  long  = upper confidence bound
}

\DeclareAcronym{lcb}{
  short = LCB,
  long  = lower confidence bound
}

\DeclareAcronym{sla}{
  short = SLA,
  long  = service level agreement
}

\newif\ifarxivversion
\newif\ifjournalversion
\newif\ifsupplementversion

\ifdefined\ARXIVVERSION
  \arxivversiontrue
\fi
\ifdefined\JOURNALVERSION
  \journalversiontrue
\fi
\ifdefined\SUPPLEMENTVERSION
  \supplementversiontrue
\fi

\newcommand{\arxivonly}[1]{\ifarxivversion#1\fi}
\newcommand{\journalonly}[1]{\ifjournalversion#1\fi}
\newcommand{\supplementonly}[1]{\ifsupplementversion#1\fi}

\newcommand{\startsupplementresults}{%
  \setcounter{section}{0}%
  \setcounter{equation}{0}%
  \setcounter{figure}{0}%
  \setcounter{table}{0}%
  \setcounter{algorithm}{0}%
  \setcounter{theorem}{0}%
  \setcounter{assumption}{0}%
  \setcounter{definition}{0}%
  \setcounter{corollary}{0}%
  \setcounter{proposition}{0}%
  \setcounter{lemma}{0}%
  \setcounter{example}{0}%
  \setcounter{remark}{0}%
  \renewcommand{\thesection}{S\arabic{section}}%
  \renewcommand{\theequation}{S\arabic{equation}}%
  \renewcommand{\thefigure}{S\arabic{figure}}%
  \renewcommand{\thetable}{S\arabic{table}}%
  \renewcommand{\thealgorithm}{S\arabic{algorithm}}%
  \renewcommand{\thetheorem}{S\arabic{theorem}}%
  \renewcommand{\theassumption}{S\arabic{assumption}}%
  \renewcommand{\thedefinition}{S\arabic{definition}}%
  \renewcommand{\thecorollary}{S\arabic{corollary}}%
  \renewcommand{\theproposition}{S\arabic{proposition}}%
  \renewcommand{\thelemma}{S\arabic{lemma}}%
  \renewcommand{\theexample}{S\arabic{example}}%
  \renewcommand{\theremark}{S\arabic{remark}}%
  \renewcommand{\theHsection}{S\arabic{section}}%
  \renewcommand{\theHequation}{S\arabic{equation}}%
  \renewcommand{\theHfigure}{S\arabic{figure}}%
  \renewcommand{\theHtable}{S\arabic{table}}%
  \renewcommand{\theHalgorithm}{S\arabic{algorithm}}%
  \renewcommand{\theHtheorem}{S\arabic{theorem}}%
  \renewcommand{\theHassumption}{S\arabic{assumption}}%
  \renewcommand{\theHdefinition}{S\arabic{definition}}%
  \renewcommand{\theHcorollary}{S\arabic{corollary}}%
  \renewcommand{\theHproposition}{S\arabic{proposition}}%
  \renewcommand{\theHlemma}{S\arabic{lemma}}%
  \renewcommand{\theHexample}{S\arabic{example}}%
  \renewcommand{\theHremark}{S\arabic{remark}}%
}

\ifjournalversion
  \usepackage{xr-hyper}
\fi
\ifsupplementversion
  \usepackage{xr-hyper}
\fi

\newcommand{\suppsecref}[1]{%
  \ifarxivversion
    Appendix~\ref{#1}%
  \else
    \ifjournalversion
      Supplementary Material, Sec.~\ref{#1}%
    \else
      Sec.~\ref{#1}%
    \fi
  \fi
}

\newcommand{\shortsuppsecref}[1]{%
  \ifarxivversion
    Appendix~\ref{#1}%
  \else
    \ifjournalversion
      Suppl. Mater., Sec.~\ref{#1}%
    \else
      Sec.~\ref{#1}%
    \fi
  \fi
}

\title{Contextual Bandit-Based Decomposition of Network Slice Requirements under Cumulative Resource Budget Constraints}

\author{Masaki~Kobayashi, Akito~Suzuki, Ryoichi Kawahara, Masahiro~Kobayashi
\thanks{Masaki Kobayashi, Akito Suzuki, Masahiro Kobayashi are with Network Service Systems Laboratories, NTT, Inc., 3--9--11 Midori-cho, Musashino-shi, Tokyo 180--8585 Japan
(e-mail: kobayashi.m@ntt.com; akito.suzuki@ntt.com; masahiropk.kobayashi@ntt.com).}
\thanks{Ryoichi Kawahara is with the Faculty of Information Networking for Innovation and Design, Toyo University, Kita-ku, Tokyo 115--8650, Japan (e-mail: ryoichi.kawahara@iniad.org).}
\arxivonly{\thanks{This work has been submitted to the IEEE for possible publication. Copyright may be transferred without notice, after which this version may no longer be accessible.}}
}

\begin{document}
\maketitle

\begin{abstract}
\Ac{e2e} \acp{ns} are provisioned across multiple domains of the 5G network.
In hierarchical \ac{ns} management, a tenant submits a \ac{nsr}, which specifies \ac{e2e} \ac{sla} requirements.
Rather than managing these domains directly, an \ac{e2e} controller decomposes each \ac{nsr} into domain-level \ac{sla} requirements and delegates resource allocation to domain-specific controllers, which return feasibility and resource-consumption feedback.
A poor decomposition policy can therefore cause rejection of the current request by producing infeasible requirements or reduce future admission opportunities by concentrating resource consumption in bottleneck domains.
We call this decomposition-policy optimization problem the \ac{nsrdp}.
For practical operation, online approaches to \ac{nsrdp} have been proposed.
Such approaches must jointly meet two requirements: \textit{(R1)} control long-term resource budgets and \textit{(R2)} adapt each decomposition to the performance targets and guarantee levels specified in the arriving \ac{nsr}'s \ac{sla}.
To meet these requirements, we introduce \ac{cckb} as an online solution for \ac{nsrdp}.
To address \textit{(R1)}, \ac{cckb} raises penalties for using resources that become tight, thereby discouraging decompositions that consume bottleneck resources.
To address \textit{(R2)}, it uses \acp{gp} to predict, for the current \ac{nsr}, the reward and resource usage of candidate decompositions, allowing it to select a decomposition suited to the performance targets and guarantee levels.
We establish high-probability guarantees for the resulting formulation and show through extensive 5G simulations across topology, bottleneck, and traffic-mixture settings that CCKB outperforms the baselines in the large majority of conditions.
\end{abstract}


\acresetall

\documentclass[src/main.tex]{subfiles}

\begin{document}

\section{Introduction}
\label{sec:introduction}
\subsection{Background}
5G networks are expected to support service categories such as \ac{embb}, \ac{urllc}, and \ac{mmtc}~\cite{itu_r_m_2083}.
Network slicing has been proposed to meet the diverse \ac{e2e} requirements.
A 5G network comprises multiple domains: the \ac{an}, \ac{tn}, and \ac{cn}. 
In each domain, resources are allocated to support individual \acp{nss}.
A \ac{ns} is formed by interconnecting these \acp{nss}~\cite{santos_Hierarchical_architecture}.

\begin{figure}[t]
    \centering
    \includegraphics[width=\columnwidth]{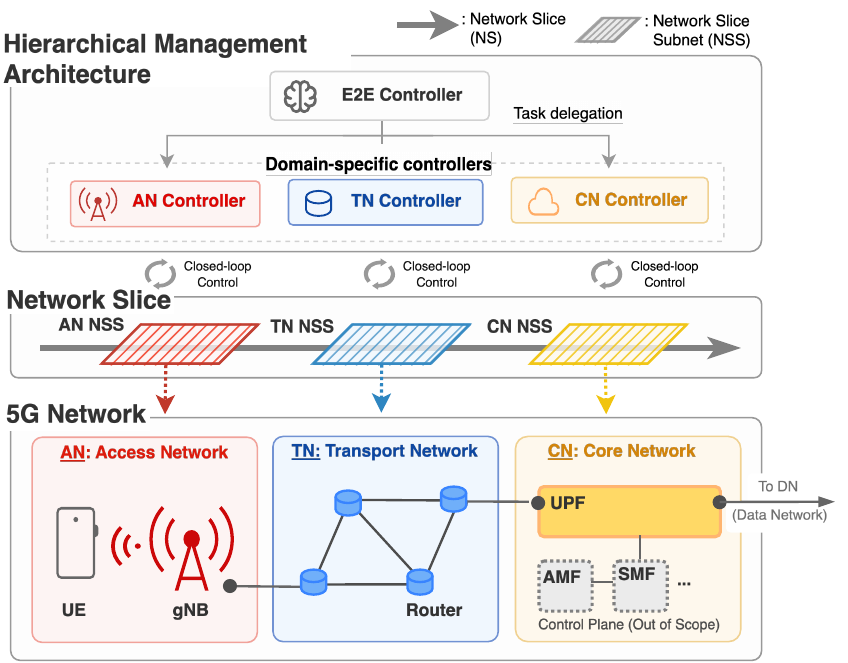}
    \caption{An overview of network slicing and its management architecture. A 5G network comprises multiple domains: the \ac{an}, \ac{tn}, and \ac{cn}. An \ac{ns} is formed by interconnecting the \acp{nss}. The hierarchical management architecture comprises the \ac{e2e} controller and domain-specific controllers.}
    \label{fig:overall_slicing}
\end{figure}

To manage \acp{ns}, a hierarchical management architecture has been specified~\cite{etsi_zsm_003}.
In this architecture, \textit{the \ac{e2e} controller} delegates \ac{ns} management tasks to subordinate \textit{domain-specific controllers},
which autonomously control their respective domains rather than being directly controlled by the \ac{e2e} controller.
Accordingly, management of each domain is entrusted to its domain-specific controller, whose internal control algorithm is treated as a black box by the \ac{e2e} controller.
\figurename~\ref{fig:overall_slicing} illustrates the concept of this \ac{ns} management architecture. 

Within this hierarchical architecture, an \ac{ns} is provisioned from a \ac{nsr} that specifies the tenant's \ac{sla} requirements.
As illustrated in \figurename~\ref{fig:ns_provisioning_procedure}, provisioning comprises three steps:
(i) \textit{\ac{nsr} decomposition},
(ii) \textit{feasibility check}, and
(iii) \textit{\ac{e2e} admission outcome}.
When admitted, the \ac{ns} is realized by consuming the resources, and the \ac{e2e} controller earns revenue from the tenant (Sec.~\ref{sec:system_model}).

\begin{figure*}
    \centering
    \includegraphics[width=0.92\linewidth]{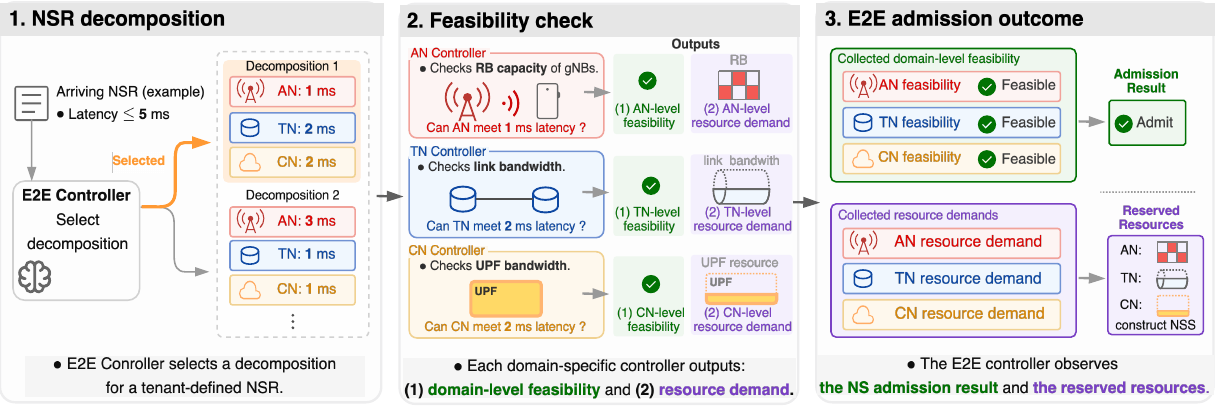}
    \caption{The \ac{ns} provisioning procedure.
    (i) In \textit{\ac{nsr} decomposition}, the \ac{e2e} controller decomposes the requirements specified in the \ac{nsr} and delegates the resulting domain-level \acp{nsr} to domain-specific controllers.
    (ii) In the \textit{feasibility check}, each domain-specific controller attempts resource allocation for its delegated \ac{nsr} and evaluates domain-level feasibility given the remaining resources and the demand required to satisfy the decomposed \ac{nsr}.
    (iii) In the \textit{\ac{e2e} admission outcome}, the \ac{e2e} controller determines whether the \ac{ns} is admitted and observes the corresponding reserved resource demands.}
    \label{fig:ns_provisioning_procedure}
\end{figure*}

\textbf{NSR-DP:}
Among these \ac{ns} provisioning steps, this work focuses on (i) \textit{\ac{nsr} decomposition}, in which the \ac{e2e} controller decomposes the requirements specified in the \ac{nsr} and delegates the resulting domain-level \acp{nsr} to domain-specific controllers.
The decomposition policy of the \ac{e2e} controller affects not only whether the current request can be admitted and generate revenue (reward),
but also how much 5G capacity is consumed (resource consumption).
For example, continuously imposing strict requirements (e.g., an overly short latency requirement) on the \ac{an} can deplete \ac{an} resources while leaving surplus resources in the \ac{tn} and \ac{cn}.
This cross-domain imbalance can leave one domain as a bottleneck, thereby reducing future admission opportunities and long-term cumulative revenue.
We call the problem of optimizing the decomposition policy under this trade-off the \textit{\acf{nsrdp}}.
Although \ac{ns} provisioning optimization has been widely studied~\cite{Afolabi2018SlicingSoftwarization,Su2019ResourceAllocationSlicingSurvey}, existing formulations assume a globally coordinated \ac{e2e} controller that directly manages resources across all 5G domains.
This assumption is misaligned with standardized architectures, where domain-specific controllers autonomously manage each domain~\cite{etsi_zsm_003,3gpp_ts_28_530}. Therefore, solving \ac{nsrdp} under this hierarchical architecture is a key problem.


\subsection{Research Objective and Challenges}
\textbf{Motivation for Online Learning:}
To solve \ac{nsrdp}, the \ac{e2e} controller must choose a decomposition before observing its realized provisioning outcome and thus needs predictive models of the reward and resource consumption of each \ac{nsr}--decomposition pair.
However, accurate model construction is challenging: (i) the \ac{e2e} controller cannot observe the internal resource-allocation logic of domain-specific controllers~\cite{etsi_zsm_002,etsi_zsm_003}, 
which makes accurate analytical modeling~\cite{sla_decomposition,sla_decomposition_multi_tier} difficult; 
and (ii) sufficiently broad logs of \ac{nsr}--decomposition outcomes are often unavailable because commercial 5G network-slicing deployments are still scaling~\cite{ericsson_diff_connectivity_2026}, 
which restricts offline supervised model learning~\cite{sla_decomp_machine_learning,sla_decomp_risk_model,sla_decomp_risk_model_nn}.

These challenges motivate learning from sequential feedback rather than relying only on fixed pre-collected data.
However, sequential learning also creates an exploration--exploitation trade-off: the controller must sometimes try uncertain decompositions to improve its models, rather than always selecting the decomposition that currently appears best.

\textbf{Research Objective:}
An exploration-aware online approach to \ac{nsrdp} must satisfy two operational requirements.
\textit{(R1) Long-term resource budget control:}
for each resource, the policy must control cumulative consumption so that long-term 5G network capacity constraints are respected.
Without this capability, the policy may consume a bottleneck resource too aggressively in early rounds, causing many later requests to be rejected even when substantial capacity remains in the other domains.
\textit{(R2) Request-conditioned decomposition:}
the decomposition must be selected according to the arriving \ac{nsr} because different \acp{nsr} induce different reward--resource trade-offs across domains.
For example, even within the \ac{urllc} class, requests may specify different latency targets or reliability guarantees and thus favor different domain-level allocations.
Without this capability, the controller may apply a decomposition suited to one \ac{sla} specification to a request with different requirements, leading to unnecessarily strict domain-level requirements or inefficient resource consumption.
Among exploration-aware online approaches, Odin~\cite{odin} uses \ac{bo} to optimize decompositions, but it satisfies neither (R1) nor (R2).
We therefore seek to learn an \ac{nsr} decomposition policy satisfying (R1) and (R2) from sequential feedback.

\textbf{Challenges:}
Our previous work~\cite{kobayashi2025icccn} addressed both requirements by formulating \ac{nsrdp} as a \ac{lincbwk} problem~\cite{lincbwk}.
However, it retained two key limitations:
\textit{(L1) finite/discrete decomposition search} and
\textit{(L2) linear realizability}.
Under (L1), the controller can select decompositions only from a predefined finite set, which can cause discretization-induced suboptimality.
Under (L2), expected reward and resource consumption are assumed to depend linearly on the \ac{nsr}--decomposition pair; nonlinear dependence can therefore cause inaccurate predictions and suboptimal decomposition.
Thus, we seek an online method that preserves (R1) and (R2) while relaxing (L1) and (L2).

As in our previous work~\cite{kobayashi2025icccn}, we analyze a stationary-response relaxation in which provisioning outcomes and resource consumption for each \ac{nsr}--decomposition pair are time-invariant and independent of residual resources, rather than the original state-dependent process (Sec.~\ref{sec:optimization_problem}).

\subsection{Our Proposal}
Our proposal has two parts.
(i) We introduce \acf{cckb} as an online solution for the relaxed \ac{nsrdp}.
It combines primal--dual cumulative resource control with \ac{gp}-based
request-conditioned decision making.
(ii) We introduce a conservative proxy formulation to address a learning
difficulty arising from resource-consumption feedback in \ac{nsrdp}.

\noindent\textbf{(i) \acs{cckb}-Based Online Solver:}
The solver comprises two components: \textit{(a)} a \ac{ckb}-style primal--dual mechanism~\cite{zhou2022kernelized_constraints} and \textit{(b)} contextual \ac{gp} models for reward and constraint prediction over \ac{nsr}--decomposition pairs, inspired by \ac{cgpucb}~\cite{krause2011cgpucb}.
The context and action are the arriving \ac{nsr} and its decomposition, respectively, and the long-term constraints represent resource budgets.
For \textit{(a)}, the primal--dual mechanism operates in a continuous decomposition space: \textit{the primal step} selects a decomposition,
while \textit{the dual step} updates budget-penalty weights from cumulative resource usage.
When a resource becomes tight, the corresponding dual penalty increases, discouraging decompositions that heavily consume that resource.
By optimizing the decomposition directly without finite-candidate enumeration, we relax \textit{(L1)} while supporting long-term budget control \textit{(R1)}.
For \textit{(b)}, the contextual \ac{gp} models condition reward and resource-consumption predictions on the current \ac{nsr}, thereby supporting contextual decomposition \textit{(R2)}.
Their nonlinear model class relaxes linear realizability \textit{(L2)}.

\noindent\textbf{(ii) Conservative Proxy Formulation:}
Directly learning resource-consumption models for the relaxed \ac{nsrdp} from provisioning feedback introduces an additional limitation:
\textit{(L3) zero-dominated resource-consumption observations}.
For example, if a resource is included in the request's \ac{ns} topology for only a small subset of requests and many of those requests are rejected, then most observations for that resource are zero.
Consequently, the model tends to make near-zero predictions, underestimating the resource consumption of a given decomposition, which can lead the controller to select resource-intensive decompositions and thereby degrade decision quality.
We address \textit{(L3)} by introducing a proxy target defined only on successful rounds and on resources included in the request's \ac{ns} topology, formulating the corresponding proxy problem, and instantiating \ac{cckb} for that proxy problem (Sec.~\ref{sec:proposed_method}).

\textbf{Performance Guarantees and Empirical Evaluation.}
We evaluate the proposed method theoretically under the relaxed \ac{nsrdp} and empirically under the original, state-dependent \ac{nsrdp}.
The theoretical analysis quantifies the optimality loss introduced by the
proxy formulation and establishes high-probability finite-time bounds on
regret and cumulative constraint violation.
The simulations assess whether the same proxy-based \ac{cckb} design remains effective in the original state-dependent setting (Sec.~\ref{sec:evaluation}).

The contributions of this paper are as follows:
\begin{itemize}
    \item \textbf{NSR-DP Formulation:}
We formulate the online \ac{nsrdp} for heterogeneous requests, where
each decomposition is conditioned on the arriving request and coupled across
a finite horizon through shared resource budgets.
    \item \textbf{CCKB-Based Online Solution:}
    We introduce \ac{cckb} as an online solution for the relaxed \ac{nsrdp},
    accommodating a continuous decision space and nonlinear reward and
    resource models.  We also construct a conservative resource proxy from successful-round
observations for resources in the request's \ac{ns} topology, thereby
mitigating the zero-dominated observation problem \textit{(L3)}.
    \item \textbf{Finite-Time Performance Analysis:}
    For the proxy problem, we establish finite-time regret and cumulative
    constraint-violation guarantees for \ac{cckb} and transfer these
    guarantees to the relaxed \ac{nsrdp}.  The transferred regret bound
    includes an additional term that quantifies the optimality gap induced
    by the proxy formulation.
    \item \textbf{Empirical Evaluation:}
    We evaluate our method in a 5G network simulator.
    Across 12 topology--bottleneck conditions, our method achieves higher mean Total Reward than our previous \ac{lincbwk}-based method~\cite{kobayashi2025icccn}; additional ablations assess its principal components.
\end{itemize}

This paper extends our conference version~\cite{kobayashi2025icccn} 
by generalizing the \ac{nsrdp} formulation, 
developing the \ac{cckb} algorithm to relax \textit{(L1)} and \textit{(L2)}, 
and introducing a proxy formulation that addresses \textit{(L3)} with corresponding theoretical guarantees.

\subsection{Paper Organization and Major Symbols}
\label{subsec:paper_organization}
We first review the related work (Sec.~\ref{sec:related_work}).
We then introduce the hierarchical \ac{ns} management modeling (Sec.~\ref{sec:system_model}) and formulate \ac{nsrdp} and its stationary-response relaxation (Sec.~\ref{sec:optimization_problem}).
Next, we present the proposed \ac{cckb} method, proxy formulation, and theoretical guarantees (Sec.~\ref{sec:proposed_method}), evaluate the resulting proxy-based method through 5G simulations (Sec.~\ref{sec:evaluation}), and discuss implementation-oriented considerations (Sec.~\ref{sec:discussion}).
Table~\ref{tab:symbols} summarizes the major symbols.

\begin{table*}[t]
\caption{Summary of Major Symbols Grouped by Their Roles}
\label{tab:symbols}
\centering
\footnotesize
\setlength{\tabcolsep}{4pt}
\renewcommand{\arraystretch}{1.08}
\newcommand{\catpad}{\rule[-\dimexpr119\dp\strutbox/100\relax]{0pt}{\dimexpr119\ht\strutbox/100+119\dp\strutbox/100\relax}}
\begin{tabularx}{\textwidth}{@{}>{\hspace*{0.8em}\arraybackslash}p{0.24\textwidth}X@{}}
\toprule
\rowcolor{black!10}
\multicolumn{2}{l}{\catpad\textbf{I. Core notation for NSR-DP and the proposed method}} \\
\rowcolor{black!4}
\multicolumn{2}{l}{\catpad\hspace{0.3em}\textbf{A. Online decompositions, outcomes, and budgets}} \\
\(s\in\mathcal{S}\), \(\mathbf{x}\in\mathcal{X}\) & Generic \ac{nsr} and decomposition, respectively. \\
\(t\in[T]\), \(T\in\mathbb{Z}_+\) & Round index and horizon length in the online formulation. \\
\(s_t,\ \mathbf{x}_t\) & Arriving \ac{nsr} and selected decomposition at round \(t\). \\
\(C_{d,j}^{\max}\) & Initial capacity and cumulative horizon budget of resource \((d,j)\). \\
\(Y_t\), \(\mathbf{U}_t\)
& Round-\(t\) \ac{e2e} admission indicator and resource-consumption vector,
respectively, with \(\mathbf{U}_t=\mathbf{0}\) when \(Y_t=0\). \\
\(W_t\), \(\kappa_{\mathrm{price}}(s)\), \(\bar p\) & Round reward \(W_t=\kappa_{\mathrm{price}}(s_t)Y_t\), where \(\kappa_{\mathrm{price}}(s)\) is the revenue for request \(s\) and is bounded by \(\bar p\). \\
\rowcolor{black!4}
\multicolumn{2}{l}{\catpad\hspace{0.3em}\textbf{B. Original and relaxed NSR-DP (Sec.~\ref{sec:optimization_problem})}} \\
\((Y,\mathbf{U})\sim\nu_{s,\mathbf{x}}\) & Generic provisioning outcome for a fixed request--decomposition pair under the stationary-response relaxation. \\
\(f(s,\mathbf{x})\), \(c_{d,j}(s,\mathbf{x})\) & Expected reward and expected resource consumption means under stationary-response relaxation. \\
\(h_{d,j}(s,\mathbf{x})\) & Normalized constraint of the relaxed NSR-DP, \(c_{d,j}(s,\mathbf{x})/C_{d,j}^{\max}-1/T\). \\
\(\pi\in\Pi,\ q(\cdot\mid s)\in\mathcal{Q}\) & \(\pi\): causal policy for the original online NSR-DP; \(q\): stationary policy for the relaxed NSR-DP. \\
\(\mathrm{OPT}^{\mathrm{rel}}\), \(\mathrm{Reg}^{\mathrm{rel}}(T)\), \(\mathrm{Vio}^{\mathrm{rel}}_{d,j}(T)\) & Optimal value, regret, and cumulative normalized violation of the relaxed NSR-DP. \\
\rowcolor{black!4}
\multicolumn{2}{l}{\catpad\hspace{0.3em}\textbf{C. Proxy formulation (Sec.~\ref{sec:proposed_method})}} \\
\(\Gamma_d(s)\), \(\mathbf{1}^{\Gamma}_{d,j}(s)\) & Path-induced resource-index set in domain \(d\) and its membership indicator, \(\mathbf{1}^{\Gamma}_{d,j}(s)=\mathbf{1}\{j\in\Gamma_d(s)\}\). \\
\(m_{d,j}^{\Gamma}(s,\mathbf{x})\), \(h^{\mathrm{prx}}_{d,j}(s,\mathbf{x})\) & Success-conditioned resource demand and its normalized proxy constraint, \(h^{\mathrm{prx}}_{d,j}=m_{d,j}^{\Gamma}/C_{d,j}^{\max}-1/T\). \\
\(\mathrm{OPT}^{\mathrm{prx}}\), \(\mathrm{Reg}^{\mathrm{prx}}(T)\), \(\mathrm{Vio}^{\mathrm{prx}}_{d,j}(T)\) & Optimal value, regret, and cumulative normalized violation of the proxy formulation. \\
\rowcolor{black!4}
\multicolumn{2}{l}{\catpad\hspace{0.3em}\textbf{D. CCKB and GP surrogates (Sec.~\ref{sec:proposed_method})}} \\
\(\boldsymbol{\phi}_t\), \(\rho\), \(V\) & Dual vector, dual cap, and dual-step-size parameter in the primal--dual updates. \\
\(\hat f_t,\hat h_{t,d,j}\), \(\bar f_t,\bar h_{t,d,j}\) & Raw exploration estimates and the clipped surrogates used by \ac{cckb}, respectively. \\
\(\mathcal{T}^{f}_{t-1}\), \(\mathcal{T}^{m_{d,j}^{\Gamma}}_{t-1}\) & Reward data (all rounds) and constraint data (successful rounds whose requested topology includes resource \((d,j)\)). \\
\(\mu_{t-1}^\bullet\), \(\sigma_{t-1}^\bullet\), \(\beta_t^f\), \(\beta_t^{m_{d,j}^{\Gamma}}\) & Posterior mean, posterior standard deviation, and exploration widths for reward and resource-consumption \acp{gp}. \\
\rowcolor{black!10}
\multicolumn{2}{l}{\catpad\textbf{II. Supporting notation for network and provisioning models}} \\
\rowcolor{black!4}
\multicolumn{2}{l}{\catpad\hspace{0.3em}\textbf{E. Basic sets, indices, and network resources (Sec.~\ref{sec:system_model})}} \\
\(\mathcal{D}\), \(d\in\mathcal{D}\) & Domain set and its index: \(\mathcal{D}=\{\mathrm{AN},\mathrm{TN},\mathrm{CN}\}\). \\
\(j\in[J_d]\), \(i\in[M]\) & Indices for resource in domain \(d\) and \ac{sla} component. \\
\(J_d\), \(J_{\mathrm{tot}}\), \((d,j)\) & Number of resources in domain \(d\), total resources \(\sum_{d\in\mathcal{D}}J_d\), and resource index pair. \\
\rowcolor{black!4}
\multicolumn{2}{l}{\catpad\hspace{0.3em}\textbf{F. NSR and decomposition model (Sec.~\ref{sec:system_model})}} \\
\(s=(A,\boldsymbol{\theta},\mathbf{R},\mathbf{g})\) & \ac{nsr} tuple: \ac{gnb} coverage \(A\), traffic/service profile \(\boldsymbol{\theta}\), \ac{sla} specification \((\mathbf{R},\mathbf{g})\). \\
\(\mathbf{R}=(R_i)_{i\in[M]}\), \(R_i=(H_i,T_i)\) & \ac{sla} components and each component's conditioning/target event pair. \\
\(\mathbf{g}=(g_i)_{i\in[M]}^\top\) & Required guarantee levels for \ac{sla} components. \\
\(R_{d,i}(s,\mathbf{x}_T)\), \(g_{d,i}(s,\mathbf{x}_g)\) & Domain-level \ac{sla} target-event component and guarantee-level allocation for component \(i\). \\
\bottomrule
\end{tabularx}
\end{table*}

\ifSubfilesClassLoaded{
  \bibliographystyle{ieee/IEEEtran}
  \bibliography{bib/references}
}{}

\end{document}

\documentclass[src/main.tex]{subfiles}

\begin{document}

\section{Related Work}
\label{sec:related_work}
This section first reviews studies on \ac{ns} provisioning optimization and
identifies the properties required of an online \ac{nsr} decomposition method
(Sec.~\ref{subsec:ns_provisioning_optimization}).  We then examine whether
existing online-learning formulations jointly provide these properties
(Sec.~\ref{subsec:online_learning}).

\subsection{\ac{ns} Provisioning Optimization}
\label{subsec:ns_provisioning_optimization}

Table~\ref{tab:related_work_comparison} positions representative prior
work families along four increasingly specific criteria:
\textit{(G1) compatibility with hierarchical \ac{ns} management},
\textit{(G2) exploration-aware online decision-making},
\textit{(G3) satisfaction of the online \ac{nsrdp} requirements (R1) and (R2)}, and
\textit{(G4) relaxation of the modeling restrictions (L1) and (L2)}.

\begin{table*}[t]
\centering
\begin{threeparttable}
\caption{Sequential positioning of representative prior work families under criteria (G1)--(G4).}
\label{tab:related_work_comparison}
\footnotesize
\setlength{\tabcolsep}{4pt}
\renewcommand{\arraystretch}{1.3}
\newcommand{\rwna}{\textemdash}
\begin{tabular}{@{}
  >{\raggedright\arraybackslash}p{0.08\textwidth}
  >{\raggedright\arraybackslash}p{0.16\textwidth}
  >{\raggedright\arraybackslash}p{0.17\textwidth}
  >{\raggedright\arraybackslash}p{0.16\textwidth}
  >{\raggedright\arraybackslash}p{0.17\textwidth}
  >{\raggedright\arraybackslash}p{0.16\textwidth}
@{}}
\toprule
\multicolumn{2}{c}{\textbf{Methods}}
&
\multicolumn{4}{c}{\textbf{Comparison criteria}}
\\
\cmidrule(lr){1-2}\cmidrule(lr){3-6}
\textbf{Ref.}
&
\textbf{Method family}
&
\textbf{\textit{(G1)} Hierarchical \ac{ns} management?}
&
\textbf{\textit{(G2)} Exploration-aware online?}
&
\textbf{\textit{(G3)} Satisfies \textit{(R1)/(R2)}?}
&
\textbf{\textit{(G4)} Relaxes \textit{(L1)/(L2)}?}
\\
\midrule
\rowcolor{black!4}
\multicolumn{6}{l}{\textbf{(1-i) Centralized \ac{e2e} provisioning}} \\

\cite{Liu2021OnSlicing,Xiao2019NFVdeep,reinforcement_partition_problem,Helmy2025SlicingAI,Zhang2021OnlineAdaptiveInterferenceAware,Tang2018QueueAwareReliableEmbedding,Vieira2024MobilityAwareSFCMigration,Camargo2023DynamicSlicingReconfiguration,Chen2025QoSAwareRoutingFlexible}
&
Centralized \ac{e2e} provisioning
&
No: directly controls domain-internal resources
&
\rwna
&
\rwna
&
\rwna
\\

\rowcolor{black!4}
\multicolumn{6}{l}{\textbf{(1-ii) Analytical and offline data-driven hierarchical decomposition}} \\

\cite{sla_decomposition,sla_decomposition_multi_tier}
&
Hierarchical analytical decomposition
&
Yes
&
No: uses explicit analytical models
&
\rwna
&
\rwna
\\

\cite{sla_decomp_machine_learning,sla_decomp_risk_model,sla_decomp_risk_model_nn}
&
Hierarchical data-driven decomposition
&
Yes
&
No: uses pre-collected logs
&
\rwna
&
\rwna
\\

\rowcolor{black!4}
\multicolumn{6}{l}{\textbf{(2-i) Heuristic online adaptation}} \\

\cite{risk_aware_online,risk_aware_multi_provider}
&
FIFO-based online risk-model adaptation
&
Yes
&
No: no exploration mechanism
&
\rwna
&
\rwna
\\

\rowcolor{black!4}
\multicolumn{6}{l}{\textbf{(2-ii) Exploration-aware online decomposition}} \\

\cite{odin}
&
\acs{config}-type~\cite{xu2023config}
&
Yes
&
Yes
&
No: misses \textit{(R1), (R2)}
&
\rwna
\\

\cite{kobayashi2025icccn}
&
\ac{lincbwk}~\cite{lincbwk} method
&
Yes
&
Yes
&
Yes
&
No: misses \textit{(L1), (L2)}
\\ \midrule

&
Ours
&
Yes
&
Yes
&
Yes
&
Yes
\\

\bottomrule
\end{tabular}
\vspace{2pt}
\begin{tablenotes}[flushleft]
\footnotesize
\item[] \textit{Note:} \textemdash{} indicates that a later criterion is not the main point of comparison once an earlier criterion is missed.
\end{tablenotes}
\end{threeparttable}
\end{table*}

\subsubsection{Controller-Architecture Assumptions: Centralized vs Hierarchical}
\label{subsubsec:rw_controller_architecture}
\ac{ns} provisioning optimization has been widely studied~\cite{Afolabi2018SlicingSoftwarization,Su2019ResourceAllocationSlicingSurvey}.
From the \ac{ns} management-architecture perspective, existing studies can be grouped into two lines.

\textbf{(i) Centralized \ac{e2e} provisioning.}
Representative examples include RL- and online-learning-based \ac{e2e} resource allocation and SFC deployment~\cite{Liu2021OnSlicing,Xiao2019NFVdeep,reinforcement_partition_problem,Helmy2025SlicingAI},
as well as optimization-based provisioning~\cite{Zhang2021OnlineAdaptiveInterferenceAware,Tang2018QueueAwareReliableEmbedding,Vieira2024MobilityAwareSFCMigration,Camargo2023DynamicSlicingReconfiguration,Chen2025QoSAwareRoutingFlexible}.
While important, these formulations assume centralized control over the underlying network resources and are therefore less aligned with standardized hierarchical management frameworks such as ETSI ZSM ISG and 3GPP management and orchestration,
where an \ac{e2e} controller delegates provisioning to domain-specific controllers rather than directly optimizing domain-internal resources~\cite{etsi_zsm_003,3gpp_ts_28_530}.
Thus, the centralized family does not satisfy \textit{(G1)}.

\textbf{(ii) Hierarchical \ac{nsr} decomposition.}
Here, the \ac{e2e} controller translates \ac{e2e} requirements into domain-level targets.
Representative examples include early SLA decomposition formulations~\cite{sla_decomposition,sla_decomposition_multi_tier}, offline machine-learning-based decomposition~\cite{sla_decomp_machine_learning,sla_decomp_risk_model,sla_decomp_risk_model_nn},
and adaptive extensions for dynamic or multi-provider settings~\cite{risk_aware_online,risk_aware_multi_provider}.
These studies are closer to our setting because they acknowledge delegated multi-domain operation,
but the black-box nature of domain-specific controllers and the limited availability of broad request--decomposition logs still make online learning particularly attractive in practice.
Accordingly, the analytical and offline data-driven families satisfy \textit{(G1)} but not \textit{(G2)}.

\subsubsection{Online NSR Decomposition}
Among methods compatible with the standardized hierarchical setting,
the relevant online studies fall into \textit{heuristic online adaptation} and \textit{exploration-aware online decomposition}.

\textbf{(i) Heuristic online adaptation.}
\Ac{rade}~\cite{risk_aware_online} and
\ac{rails}~\cite{risk_aware_multi_provider} adapt \ac{sla} decompositions online
by updating point-estimate risk models through online gradient descent using
recent provisioning feedback retained in a \ac{fifo} buffer.
Although these updates enable adaptation to changing conditions, their
action-selection rules optimize the resulting point estimates without using
predictive uncertainty and therefore do not explicitly manage the
exploration--exploitation trade-off.
We therefore categorize \ac{rade} and \ac{rails} as heuristic online-adaptation methods rather than exploration-aware online algorithms.
Accordingly, both methods satisfy \textit{(G1)} but not \textit{(G2)}.

\textbf{(ii) Exploration-aware online decomposition.}
In contrast, exploration-aware methods use predictive uncertainty in action selection to balance immediate provisioning performance against information acquisition.
Among such approaches, two method families are close.
We provide an overview here and compare their formulations with ours in detail in Sec.~\ref{subsec:requirements_prior_formulations}.

Odin~\cite{odin} is the closest exploration-aware method to our setting:
it tailors a \ac{bo} design inspired by \acf{config}~\cite{xu2023config}.
Accordingly, Odin satisfies \textit{(G1)} and \textit{(G2)} but not \textit{(G3)}:
its \ac{config}-type formulation neither models horizon-level resource budgets \textit{(R1)} nor conditions decomposition selection on the arriving \ac{nsr} \textit{(R2)}.

Our previous work~\cite{kobayashi2025icccn}, which formulates \ac{nsrdp} as a \ac{lincbwk} problem~\cite{lincbwk}, is also directly relevant.
Its formulation satisfies \textit{(G1)}--\textit{(G3)} through request-conditioned decisions under cumulative resource budgets, but not \textit{(G4)}:
it retains \textit{(L1)} by optimizing over a pre-enumerated finite candidate set and \textit{(L2)} by assuming linear reward/resource models.

\arxivonly{%
\underline{\textbf{Takeaway.}}
The preceding review identifies the target properties for an online
\ac{nsr} decomposition method: compatibility with hierarchical management and
exploration-aware learning, satisfaction of \textit{(R1)} and \textit{(R2)},
and relaxation of \textit{(L1)} and \textit{(L2)}.
}

\subsection{Online Learning}
\label{subsec:online_learning}
We next examine whether existing online-learning formulations satisfy
\textit{(R1)} and \textit{(R2)} while relaxing \textit{(L1)} and
\textit{(L2)}.  Table~\ref{tab:online_learning_comparison} compares
representative formulations.

\begin{table*}[t]
\centering
\begin{threeparttable}
\caption{Comparison of representative online-learning formulations under properties (R1), (R2), (L1), and (L2).}
\label{tab:online_learning_comparison}
\footnotesize
\setlength{\tabcolsep}{2.5pt}
\renewcommand{\arraystretch}{1.3}
\begin{tabular}{@{}
  >{\raggedright\arraybackslash}p{0.075\textwidth}
  >{\raggedright\arraybackslash}p{0.175\textwidth}
  >{\raggedright\arraybackslash}p{0.145\textwidth}
  >{\raggedright\arraybackslash}p{0.105\textwidth}
  >{\raggedright\arraybackslash}p{0.065\textwidth}
  >{\raggedright\arraybackslash}p{0.065\textwidth}
  >{\raggedright\arraybackslash}p{0.22\textwidth}
@{}}
\toprule
\multicolumn{2}{c}{\textbf{Methods}} &
\multicolumn{4}{c}{\textbf{Comparison properties}} &
\textbf{Context / Performance Target} \\
\cmidrule(lr){1-2}\cmidrule(lr){3-6}\cmidrule(lr){7-7}
\textbf{Ref.} &
\textbf{Method/formulation} &
\textbf{\textit{(R1)} Across-round resource control?} &
\textbf{\textit{(R2)} Request-conditioned?} &
\textbf{Relaxes \textit{(L1)}?} &
\textbf{Relaxes \textit{(L2)}?} &
\\
\midrule
\rowcolor{black!4}
\multicolumn{7}{l}{\textbf{(i) Continuous-action bandits with across-round constraints}} \\

\cite{shi2022continuum_constraints} &
GP-UCB with constraints &
Yes (cumulative) &
No &
Yes &
Yes &
\textemdash \\

\cite{zhou2022kernelized_constraints} &
CKB &
Yes (soft cumulative) &
No &
Yes &
Yes &
\textemdash \\

\rowcolor{black!4}
\multicolumn{7}{l}{\textbf{(ii) Contextual continuous-action methods}} \\

\cite{krause2011cgpucb,ayala2024rancb} &
CGP-UCB; RANCB &
No (none; step-wise) &
Yes &
Yes &
Yes &
\textemdash \\

\rowcolor{black!4}
\multicolumn{7}{l}{\textbf{(iii) Contextual bandits with across-round constraints}} \\

\cite{lincbwk} &
linCBwK &
Yes (hard budgets) &
Yes &
No &
No &
\textemdash \\

\cite{han2023squarecbwk} &
SquareCBwK &
Yes (hard budgets) &
Yes &
No &
Yes &
\textemdash \\

\cite{slivkins2024packing_covering} &
LagrangeCBwLC &
Yes (packing/covering) &
Yes &
No &
Yes &
\textemdash \\

\cite{guo2024loe2d,guo2025optimistic3} &
LOE2D; Optimistic\(^{3}\) &
Yes (cumulative) &
Yes &
No &
Yes &
\textemdash \\
\rowcolor{black!4}
\multicolumn{7}{l}{\textbf{(iv) Contextual primal--dual BO: shared update structure, different analysis target}}  \\

\cite{xu2023pdcbo} &
PDCBO &
Yes (time average) &
Yes &
Yes &
Yes &
Arbitrary time-varying contexts / context-wise optimum \\

&
CCKB (ours) &
Yes (soft cumulative) &
Yes &
Yes &
Yes &
Stochastic requests / distribution-optimized policy \\

\bottomrule
\end{tabular}
\vspace{2pt}
\begin{tablenotes}[flushleft]
\footnotesize
\item[] \textit{Note:} \textemdash{} indicates that this aspect is not highlighted in this comparison.
\end{tablenotes}
\end{threeparttable}
\end{table*}

\textbf{(i) Continuous-action bandits with across-round constraints.}
A Bayesian continuum-armed bandit with long-term constraints and \ac{ckb} both
optimize over continuous action domains under across-round
constraints~\cite{shi2022continuum_constraints,zhou2022kernelized_constraints}.
In particular, \ac{ckb} models nonlinear reward and constraint functions in
\acp{rkhs}.  These methods therefore satisfy \textit{(R1)} while relaxing
\textit{(L1)} and \textit{(L2)}, but are non-contextual and hence do not
satisfy \textit{(R2)}.

\textbf{(ii) Contextual continuous-action methods.}
\ac{cgpucb} places a \ac{gp} over the joint context--action space and supports
action and context sets that are not necessarily finite, whereas RANCB uses
nonlinear neural models for contextual continuous-action
selection~\cite{krause2011cgpucb,ayala2024rancb}.  Both methods satisfy
\textit{(R2)} while relaxing \textit{(L1)} and \textit{(L2)}.  However,
\ac{cgpucb} imposes no constraints, and RANCB imposes step-wise constraints
that must hold in every round rather than across the horizon.  These methods
therefore do not satisfy \textit{(R1)}.

\textbf{(iii) Contextual bandits with across-round constraints.}
The methods in this group make request-conditioned decisions under across-round
constraints and thus satisfy \textit{(R1)} and \textit{(R2)}.  \ac{lincbwk}
and SquareCBwK enforce hard budgets by stopping when a resource budget is
exhausted~\cite{badanidiyuru2014resourceful,lincbwk,han2023squarecbwk}.  In
contrast, the formulations of LagrangeCBwLC, LOE2D, and
Optimistic\(^{3}\) allow the interaction to continue throughout the horizon
and measure constraint violations cumulatively~\cite{slivkins2024packing_covering,guo2024loe2d,guo2025optimistic3}.
\ac{lincbwk} retains both \textit{(L1)} and \textit{(L2)} because it assumes
finite arms and linear reward and consumption models~\cite{lincbwk}.
SquareCBwK, LagrangeCBwLC, LOE2D, and Optimistic\(^{3}\) relax
\textit{(L2)} because they do not require linear reward and cost models, but
retain a finite action set and hence \textit{(L1)}.

\textbf{(iv) Contextual primal--dual \ac{bo}.}
At the algorithmic level, PDCBO and \ac{cckb} share the same core
primal--dual contextual BO structure: an optimistic Lagrangian primal step
followed by a projected dual update.  However, the two methods adopt different
performance benchmarks.  PDCBO allows arbitrary time-varying
contexts and compares each selected action with an optimal action for the
context observed in that round~\cite{xu2023pdcbo}.  In contrast, under i.i.d.
requests from an unknown distribution \(\mathbb P\), our analysis benchmarks
\ac{cckb} against a request-conditioned policy jointly optimized over
\(\mathbb P\) under a horizon-wide shared resource budget.  This
distribution-aware benchmark captures the opportunity cost of allocating
resources to the current request relative to future arrivals and therefore
matches the relaxed \ac{nsrdp}.
\journalonly{Accordingly, PDCBO's guarantees do not directly cover our analysis target.}

\arxivonly{%
\underline{\textbf{Takeaway.}}
Among the representative methods other than PDCBO,
none simultaneously satisfies \textit{(R1)} and \textit{(R2)} while relaxing both \textit{(L1)} and \textit{(L2)}.
Although PDCBO shares the same core primal--dual contextual BO structure as CCKB,
its theoretical guarantees do not directly apply to the analysis target considered in this paper.
}

\ifSubfilesClassLoaded{
  \bibliographystyle{ieee/IEEEtran}
  \bibliography{bib/references}
}{}

\end{document}

\documentclass[src/main.tex]{subfiles}

\begin{document}

\section{System Model}
\label{sec:system_model}
This section formalizes the hierarchical provisioning system used throughout the paper.
We consider a horizon of \(T\in\mathbb{Z}_+\) provisioning rounds, indexed by \(t\in[T]\).
\begin{samepage}
At a high level, for the \ac{e2e} controller, round \(t\) is characterized by
\[
\begin{aligned}
\text{\textit{observed request}:} \quad & s_t,\\
\text{\textit{selected decomposition}:} \quad & \mathbf{x}_t,\\
\text{\textit{observed feedback}:} \quad & (y_t,\mathbf{u}_t),
\end{aligned}
\]
where \(y_t\) is the realized \ac{e2e} admission outcome, and \(\mathbf{u}_t\) is the realized resource-consumption vector.
\end{samepage}

We first define the network and request model that determines the observed request \(s_t\) (Sec.~\ref{subsec:system_overview}), 
then represent \ac{nsr} decompositions as finite-dimensional vectors and construct the continuous decomposition space \(\mathcal{X}\) from which the controller selects \(\mathbf{x}_t\) (Sec.~\ref{subsec:admissible_decomposition}), 
and finally specify how the selected \(\mathbf{x}_t\) produces the observed admission outcome \(y_t\)
and resource-consumption vector \(\mathbf{u}_t\) (Sec.~\ref{subsec:per_round_provisioning}).

\textbf{Notation:}
We use \(\mathbb{Z}_{+}:=\{1,2,\ldots\}\), \(\mathbb{R}_{+}:=[0,\infty)\), and \([n]:=\{1,\dots,n\}\) for any \(n\in\mathbb{Z}_+\).
\(|\cdot|\) denotes cardinality (e.g., \(|\mathcal{A}|\)).
We denote finite-dimensional real vectors by bold symbols and their scalar components by the corresponding nonbold symbols.
For \(\mathbf{a}, \mathbf{b} \in \mathbb{R}^n\), \(\langle \mathbf{a}, \mathbf{b} \rangle := \mathbf{a}^{\top}\mathbf{b}\) denotes the Euclidean inner product.
For any interval \([a,b]\subset\mathbb{R}\) with \(a\le b\), define \(\operatorname{clip}_{[a,b]}(z):=\min\{b,\max\{a,z\}\}\).

\subsection{Network and NSR Model}
\label{subsec:network_nsr_model}
\label{subsec:system_overview}
We define the network resources (Sec.~\ref{subsubsec:network_model}) and request model (Sec.~\ref{subsubsec:nsr_model}) used in each provisioning round.

\subsubsection{Network Model}
\label{subsubsec:network_model}
Let the set of 5G domains be \(\mathcal{D}:=\{\text{AN},\text{TN},\text{CN}\}\), with \(d\in\mathcal{D}\) denoting a domain.
We represent the 5G network as a directed graph \(\mathcal{G}=(\mathcal{V},\mathcal{E})\), where \(\mathcal{V}=\mathcal{A}\cup\mathcal{R}\cup\mathcal{U}\) contains \acp{gnb}, routers, and \acp{upf}.
For each domain \(d\), let \(J_d\in\mathbb{Z}_+\) be the number of allocatable resources and let \(j\in[J_d]\) index them.
The total number of resources is \(J_{\mathrm{tot}}:=\sum_{d\in\mathcal{D}}J_d\), and resource \((d,j)\) has initial capacity \(C_{d,j}^{\max}\).
In our model, admitted reservations persist across rounds and reduce the remaining capacity; hence, \(C_{d,j}^{\max}\) is also the cumulative capacity budget available over the horizon.
Concretely, \ac{an} resources are schedulable radio units such as \ac{rb} budgets at \acp{gnb}, \ac{tn} resources are egress interfaces, and \ac{cn} resources are attachment interfaces at the \ac{upf} side.
This paper focuses on provisioning the \ac{an}, \ac{tn}, and \ac{cn}
resources that carry user data. Procedures that connect individual devices to
the network or exchange control messages are outside the model.

\subsubsection{\ac{nsr} Model}
\label{subsubsec:nsr_model}
An \ac{nsr} comprises a coverage set \(A\subseteq\mathcal{A}\), a traffic/service profile \(\boldsymbol{\theta}\), \ac{sla} requirement components \(\mathbf{R}:=(R_i)_{i\in[M]}\), and guarantee levels \(\mathbf{g}:=(g_i)_{i\in[M]}^\top\in(0,1]^M\).
\(M\in\mathbb{Z}_+\) is the number of \ac{sla} components and \(i\in[M]\) indexes them.

\textbf{(i) Coverage Set \(A\):}
The coverage set \(A\) specifies the \acp{gnb} that the requested \ac{ns} must cover.
For example, if \(A=\{a_1,a_2\}\), the \ac{ns} must cover \acp{gnb} \(a_1\) and \(a_2\).

\textbf{(ii) Traffic/Service Profile \(\boldsymbol{\theta}\):}
The traffic/service profile \(\boldsymbol{\theta}\) describes request-specific characteristics that influence performance, such as traffic volume.

\textbf{(iii) \ac{sla} Specification \((\mathbf{R},\mathbf{g})\):}
\(R_i=(H_i,T_i)\) specifies the events used to evaluate the \ac{sla}: \(H_i\) is the condition under which the guarantee is evaluated, and \(T_i\) is the performance event that must be achieved under that condition.
The value \(g_i\in(0,1]\) is the required guarantee level.
For a fixed request \(s\) and decomposition \(\mathbf{x}\) (defined in Sec.~\ref{subsec:admissible_decomposition}), let \(\Pr_{s,\mathbf{x}}\) denote the induced law.
Together, \(\mathbf{R}\) and \(\mathbf{g}\) specify the \ac{sla} requirements
imposed on an \ac{ns} provisioned for request \(s\) as
\begin{equation}
    \Pr_{s,\mathbf{x}}(T_i\mid H_i)\ge g_i,
    \quad \forall i\in[M].
    \label{eq:e2e_probabilistic_guarantee}
\end{equation}
In this definition, we assume \(\Pr_{s,\mathbf{x}}(H_i)>0\).

For illustration, let
\(\mathcal{N}\) denote the event of no packet drop over the \ac{e2e} path, and let
\(\Omega\) denote the sure event under \(\Pr_{s,\mathbf{x}}\).  Let \(D\) and \(\Theta\) denote the
\ac{e2e} delay and throughput random variables, respectively, with targets
\(\delta_i>0\) and \(\theta_i>0\).  Table~\ref{tab:e2e_sla_components}
gives three illustrative instantiations of \(H_i\) and \(T_i\).
\begin{table}[t]
\caption{Illustrative \ac{e2e} \ac{sla} requirement components.}
\label{tab:e2e_sla_components}
\centering
\small
\renewcommand{\arraystretch}{1.15} 
\setlength{\tabcolsep}{3pt}
\begin{tabularx}{\columnwidth}{@{}lcc>{\raggedright\arraybackslash}X@{}}
\toprule
\textbf{Component} & \(\boldsymbol{H_i}\) & \(\boldsymbol{T_i}\) & \textbf{Interpretation} \\
\midrule
Latency    & \(\mathcal{N}\) & \(\{D\le\delta_i\}\) & The probability that the \ac{e2e} delay is at most \(\delta_i\), conditioned on no packet drop, is at least \(g_i\). \\
Throughput & \(\Omega\)      & \(\{\Theta\ge\theta_i\}\) & The probability that the \ac{e2e} throughput is at least \(\theta_i\) is at least \(g_i\). \\
Non-drop   & \(\Omega\)      & \(\mathcal{N}\) & The probability of no packet drop over the \ac{e2e} path is at least \(g_i\). \\
\bottomrule
\end{tabularx}
\end{table}

\begin{example}
\label{ex:latency_requirement}
For the latency component in Table~\ref{tab:e2e_sla_components},
\eqref{eq:e2e_probabilistic_guarantee} becomes
\(
\Pr_{s,\mathbf{x}}(D\le\delta_i\mid\mathcal{N})\ge g_i.
\)
For example,
\(
\Pr_{s,\mathbf{x}}(D\le5\,\mathrm{ms}\mid\mathcal{N})\ge0.9999
\)
requires the \ac{e2e} delay to be at most \(5\) ms for at least \(99.99\%\) of successfully delivered packets.
\end{example}

Using this component representation, we define an \ac{nsr}.
\begin{definition}[Network slice request]
\label{def:network_slice_request}
An \ac{nsr} is the tuple
\[
s:=\big(A,\boldsymbol{\theta},\mathbf{R},\mathbf{g}\big).
\]
An \ac{ns} provisioned for \(s\) must satisfy \((\mathbf{R},\mathbf{g})\) over coverage set \(A\) under traffic/service profile \(\boldsymbol{\theta}\).
We denote the set of all \acp{nsr} by \(\mathcal{S}\).
The coverage set \(A\) and target events \(T_i\) correspond to
attributes specifiable in a GSMA NG.116 \ac{nest}~\cite{generic_slice_template}.\footnote{A NEST specifies per-use-case NS
requirements as target values, but not the condition under which each
target is evaluated, the probability with which it must hold, or the
traffic conditions that determine performance. The components $H_i$,
$g_i$, and $\theta$ supply these; $s$ therefore abstracts an NSR
rather than instantiating a NEST.}
\end{definition}


\subsection{Admissible NSR Decomposition}
\label{subsec:admissible_decomposition}
\label{subsec:e2e_nsr_decomp}
This subsection represents each \ac{nsr} decomposition by a finite-dimensional
parameter vector \(\mathbf{x}\) and constructs the continuous decomposition space
\(\mathcal{X}\) from which the \ac{e2e} controller selects \(\mathbf{x}\).
We first define how \(\mathbf{x}\) transforms an \ac{e2e} request
\(s\) into a delegated domain-level \ac{nsr} for each domain
(Sec.~\ref{subsubsec:domain_level_nsr_construction}).
We then state the admissibility conditions for preserving the original \ac{e2e}
\ac{sla} requirement when all delegated domain-level requirements are met
(Sec.~\ref{subsubsec:admissibility_conditions}).
Next, we give concrete constructions that satisfy these conditions and
parameterize the resulting admissible decompositions using nonnegative weight
vectors (Sec.~\ref{subsubsec:concrete_admissible_decomposition}).
Finally, we remove slack from the weight vectors, define exact splits, and
assemble \(\mathcal{X}\) as a product of simplices
(Sec.~\ref{subsubsec:exact_split_decomposition_space}).

\subsubsection{Domain-Level NSR Construction}
\label{subsubsec:domain_level_nsr_construction}
Given an \ac{e2e} \ac{nsr} \(s\) and decomposition
\(\mathbf{x}\), we define the domain-level \ac{nsr} \(s_d(s,\mathbf{x})\)
delegated to domain \(d\).  We first (i) define the components of
\(s_d(s,\mathbf{x})\).  We then (ii) state the operational meaning of
\(s_d(s,\mathbf{x})\): the probabilistic guarantee that the domain-specific
controller for \(d\) must satisfy.

\textbf{(i) Domain-level components:}
We partition the decomposition as
\(\mathbf{x}:=(\mathbf{x}_T,\mathbf{x}_g)\).  Selecting \(\mathbf{x}_T\)
determines the domain-level target events \(T_{d,i}(s,\mathbf{x}_T)\) from
\(T_i\), whereas selecting \(\mathbf{x}_g\) determines the domain-level
guarantee levels \(g_{d,i}(s,\mathbf{x}_g)\) from \(g_i\).  Concrete
parameterizations of \(\mathbf{x}_T\) and \(\mathbf{x}_g\) are given in
Sec.~\ref{subsubsec:concrete_admissible_decomposition}.
We assume that each \ac{e2e} conditioning event \(H_i\) is specified together
with domain-level evaluation conditions \(H_{d,i}(s)\), each evaluable within
domain \(d\), such that
\begin{equation}
H_i=\bigcap_{d\in\mathcal{D}}H_{d,i}(s).
\label{eq:domain_evaluation_condition_construction}
\end{equation}
These domain-level conditions are fixed before optimization.  Accordingly,
\(\mathbf{x}\) parameterizes only the decompositions of \(T_i\) and \(g_i\),
not the domain-level construction of \(H_i\).
For the latency requirement in Example~\ref{ex:latency_requirement},
\(H_i=\mathcal{N}\) and \(H_{d,i}(s)=\mathcal{N}_d\), where
\(\mathcal{N}_d\) is the event that the packet traverses domain \(d\) without
being dropped; hence, \(\mathcal{N}=\bigcap_{d\in\mathcal D}\mathcal{N}_d\).

Using the resulting domain-level quantities and \(\boldsymbol{\theta}\)
inherited from \(s\), we define the delegated domain-level \ac{nsr} as
\begin{equation}
s_d(s,\mathbf{x})
:=\big(\boldsymbol{\theta},\mathbf{R}_d(s,\mathbf{x}_T),
\mathbf{g}_d(s,\mathbf{x}_g)\big).
\label{eq:domain_level_nsr}
\end{equation}
Here,
\(R_{d,i}(s,\mathbf{x}_T):=\big(H_{d,i}(s),T_{d,i}(s,\mathbf{x}_T)\big)\),
\(\mathbf{R}_d(s,\mathbf{x}_T):=\big(R_{d,i}(s,\mathbf{x}_T)\big)_{i\in[M]}\),
and
\(\mathbf{g}_d(s,\mathbf{x}_g):=\big(g_{d,i}(s,\mathbf{x}_g)\big)_{i\in[M]}^\top\).

\textbf{(ii) Operational meaning:}
The following inequality is the domain-level counterpart of
\eqref{eq:e2e_probabilistic_guarantee} and defines when domain \(d\) satisfies
component \(i\) of its delegated \ac{nsr} \(s_d(s,\mathbf{x})\).  It represents hierarchical control
in the \ac{ns} management architecture: the \ac{e2e} controller delegates the
domain-level evaluation condition, target event, and guarantee level, and the
domain-specific controller provisions its resources to satisfy the resulting
probabilistic guarantee:
\begin{equation}
\Pr_{s,\mathbf{x}}\!\left(T_{d,i}(s,\mathbf{x}_T)\mid H_{d,i}(s)\right)
\ge g_{d,i}(s,\mathbf{x}_g),
\quad \forall d\in\mathcal D,\ i\in[M].
\label{eq:domain_probabilistic_guarantee}
\end{equation}
We assume \(\Pr_{s,\mathbf{x}}(H_{d,i}(s))>0\) whenever this conditional probability is used.

\subsubsection{Admissibility Conditions}
\label{subsubsec:admissibility_conditions}
A decomposition cannot be chosen arbitrarily because, even if every domain-specific controller can satisfy its delegated request, the original \ac{e2e} requirement need not be satisfied.
For example, for an \ac{e2e} delay target of \(10\) ms across the \ac{an}, \ac{tn}, and \ac{cn}, domain targets of \((2,3,5)\) ms preserve the target because their sum is \(10\) ms.
In contrast, targets of \((4,4,4)\) ms allow every domain target to be met while the \ac{e2e} delay reaches \(12\) ms.
Therefore, the \ac{e2e} controller cannot select domain-level targets independently.

Operationally, for each component \(i\), guarantee preservation requires that satisfaction of every delegated domain-level guarantee imply satisfaction of the original \ac{e2e} guarantee:
\begin{equation}
\begin{aligned}
&\forall d\in\mathcal{D},\quad
\Pr_{s,\mathbf{x}}\!\left(
T_{d,i}(s,\mathbf{x}_T)\,\middle|\,H_{d,i}(s)
\right)
\ge g_{d,i}(s,\mathbf{x}_g)
\\[-1mm]
&\hspace{24mm}\Rightarrow
\Pr_{s,\mathbf{x}}(T_i\mid H_i)\ge g_i.
\end{aligned}
\label{eq:operational_guarantee_preservation}
\end{equation}
We formalize this requirement as follows.
\begin{definition}[Admissible decomposition]
\label{def:admissible_decomp}
Fix an \ac{nsr} \(s\in\mathcal{S}\).
A decomposition \(\mathbf{x}\) is admissible for \(s\) if, for every \ac{sla} component \(i\in[M]\),
\begin{equation}
\begin{aligned}
&\textnormal{(A1)}\quad
H_i\cap\bigcap_{d\in\mathcal{D}}T_{d,i}(s,\mathbf{x}_T)\subseteq T_i,\\
&\textnormal{(A2)}\quad
g_i\le\prod_{d\in\mathcal{D}}g_{d,i}(s,\mathbf{x}_g).
\end{aligned}
\label{eq:admissibility_conditions}
\end{equation}
\end{definition}
Under cross-domain conditional-independence conditions that are reasonable
when domain-specific controllers operate on separate resources and no
common-cause failures or traffic fluctuations affect multiple domains,
\textnormal{(A1)} and \textnormal{(A2)} are sufficient for the
guarantee-preservation implication in
\eqref{eq:operational_guarantee_preservation}.  A proof and a detailed
discussion of the applicability of the independence conditions are provided
in \suppsecref{appendix:cond_independence_guarantee}.

\subsubsection{Concrete Construction of Admissible Decompositions}
\label{subsubsec:concrete_admissible_decomposition}
Definition~\ref{def:admissible_decomp} characterizes admissibility but does not
provide finite-dimensional constructions.  We therefore specify
\textit{(i) target decomposition}, in which \(\mathbf{x}_T\) determines the
domain-level target events \(T_{d,i}(s,\mathbf{x}_T)\), and
\textit{(ii) guarantee decomposition}, in which \(\mathbf{x}_g\) determines the
domain-level guarantee levels \(g_{d,i}(s,\mathbf{x}_g)\).  The feasible
parameter values are subsequently restricted to form the continuous decomposition space
\(\mathcal{X}\) used for optimization.
Table~\ref{tab:sla_decomposition_constructions} summarizes these constructions
for the illustrative \ac{sla} components introduced in
Table~\ref{tab:e2e_sla_components}.
\begin{table*}[t]
\caption{Illustrative target and guarantee decompositions.}
\label{tab:sla_decomposition_constructions}
\centering
\small
\renewcommand{\arraystretch}{1.15} 
\setlength{\tabcolsep}{3.5pt}
\begin{tabular}{@{}
>{\raggedright\arraybackslash}m{0.09\textwidth}
>{\centering\arraybackslash}m{0.12\textwidth}
>{\centering\arraybackslash}m{0.13\textwidth}
>{\centering\arraybackslash}m{0.11\textwidth}
>{\centering\arraybackslash}m{0.128\textwidth}
>{\centering\arraybackslash}m{0.085\textwidth}
>{\centering\arraybackslash}m{0.11\textwidth}
>{\centering\arraybackslash}m{0.128\textwidth}
@{}}
\toprule
& \multicolumn{4}{c}{\textbf{(i) Target decomposition: }\(T_i\xrightarrow{\mathbf{x}_T}\{T_{d,i}\}_{d\in\mathcal D}\)}
& \multicolumn{3}{c}{\textbf{(ii) Guarantee decomposition: }\(g_i\xrightarrow{\mathbf{x}_g}\{g_{d,i}\}_{d\in\mathcal D}\)} \\
\cmidrule(lr){2-5}\cmidrule(l){6-8}
\textbf{Component}
& \textbf{Rule}
& \textbf{E2E input} \newline \(T_i\)
& \textbf{Parameter} \newline \(\mathbf{x}_T\)
& \textbf{Domain output} \newline \(T_{d,i}\)
& \textbf{E2E input} \newline \(g_i\)
& \textbf{Parameter} \newline \(\mathbf{x}_g\)
& \textbf{Domain output} \newline \(g_{d,i}\) \\
\midrule
Latency
& (i-a) Allocation
& \(\{D\le\delta_i\}\)
& \(\substack{\mathbf w_i\ge\mathbf 0,\\ \sum_d w_{d,i}\le1}\)
& \(\{D_d\le w_{d,i}\delta_i\}\)
& \(g_i\)
& \(\substack{\boldsymbol\eta_i\ge\mathbf 0,\\ \sum_d\eta_{d,i}\le1}\)
& \(g_i^{\eta_{d,i}}\) \\
Throughput
& (i-b) Replication
& \(\{\Theta\ge\theta_i\}\)
& None
& \(\{\Theta_d\ge\theta_i\}\)
& \(g_i\)
& \(\substack{\boldsymbol\eta_i\ge\mathbf 0,\\ \sum_d\eta_{d,i}\le1}\)
& \(g_i^{\eta_{d,i}}\) \\
Non-drop
& (i-b) Replication
& Success event: \(\mathcal N\)
& None
& \(\mathcal N_d\)
& \(g_i\)
& \(\substack{\boldsymbol\eta_i\ge\mathbf 0,\\ \sum_d\eta_{d,i}\le1}\)
& \(g_i^{\eta_{d,i}}\) \\
\bottomrule
\end{tabular}
\end{table*}

\textbf{(i) Target Decomposition:}
Target decomposition constructs the domain-level target events
\(T_{d,i}(s,\mathbf{x}_T)\) from the \ac{e2e} target event \(T_i\) so that
\(H_i\) and all domain-level target events jointly imply \(T_i\), i.e.,
\(
H_i\cap\bigcap_{d\in\mathcal D}T_{d,i}(s,\mathbf{x}_T)\subseteq T_i,
\)
as required by \textnormal{(A1)}.  For the illustrative components in
Table~\ref{tab:sla_decomposition_constructions}, the construction depends on
how the \ac{e2e} target relates to the domain-level quantities.  (a) Target-value
allocation applies when a numerical \ac{e2e} bound is allocated among the
domains, as for latency. (b) Target-condition replication applies when no
numerical split is needed and the \ac{e2e} target follows from requiring every
domain to satisfy its corresponding local condition, as for throughput and
non-drop.

\textbf{(i-a) Target-value allocation:}
Target-value allocation constructs the domain-level target events \(T_{d,i}\)
by \textit{translating the \ac{e2e} bound in \(T_i\) into one bound for each domain.}
For example, consider the latency target \(T_i=\{D\le\delta_i\}\), where
\(
D=\sum_{d\in\mathcal D}D_d,
\)
and \(D_d\) is the latency contributed by domain \(d\).  Let
\(\delta_{d,i}\ge0\) denote the delay bound assigned to domain \(d\), and
choose these bounds such that
\(
\sum_{d\in\mathcal D}\delta_{d,i}\le\delta_i.
\)
For \(\delta_i=10\) ms, one choice is
\((\delta_{d,i})_{d\in\mathcal D}=(2,3,5)\) ms.  On \(H_i\), if every
domain-specific controller satisfies its delegated latency target
\(T_{d,i}=\{D_d\le\delta_{d,i}\}\), then
\begin{equation*}
D=\sum_dD_d
\le\sum_d\delta_{d,i}
\le\delta_i,
\end{equation*}
and hence
\(
H_i\cap\bigcap_{d\in\mathcal D}T_{d,i}\subseteq T_i,
\)
which verifies \textnormal{(A1)}.

\noindent\textbf{(Generic construction):}
The latency example above is an instance of the following construction for any
\ac{e2e} quantity that is the sum of its domain-level contributions.  We use
\(Q_i\), \(Q_{d,i}\), \(\tau_i\), and \(\tau_{d,i}\) to denote the generic
counterparts of \(D\), \(D_d\), \(\delta_i\), and \(\delta_{d,i}\), respectively,
with \(\tau_i>0\).  The construction is summarized as follows:
\begin{itemize}[leftmargin=*,labelsep=0.5em]
\item \textit{Input:} The \ac{e2e} target event is written as follows:
\begin{equation*}
T_i=\{Q_i\le\tau_i\},
\qquad
Q_i=\sum_{d\in\mathcal D}Q_{d,i}.
\end{equation*}
\item \textit{Allocation:} A nonnegative vector
\(\mathbf w_i=(w_{d,i})_{d\in\mathcal D}\), where \(w_{d,i}\) is the fraction
of the \ac{e2e} bound assigned to domain \(d\), satisfying
\(
\sum_{d\in\mathcal D}w_{d,i}\le1,
\)
and defining the domain-level bounds as
\(\tau_{d,i}:=w_{d,i}\tau_i\).
\item \textit{Output:} The domain-level target events
\begin{equation*}
T_{d,i}(s,\mathbf{x}_T):=\{Q_{d,i}\le \tau_{d,i}\}.
\end{equation*}
\end{itemize}
The latency targets
\((2,3,5)\) ms above correspond to \(\mathbf w_i=(0.2,0.3,0.5)\).
The output satisfies \(\bigcap_dT_{d,i}(s,\mathbf{x}_T)\subseteq T_i\), and
hence \textnormal{(A1)} holds.  If the required relation is instead
\(\sum_dQ_{d,i}\ge\tau_i\), the same construction applies by reversing the
inequalities in \(T_i\), \(T_{d,i}\), and the allocation-sum condition.

\textbf{(i-b) Target-condition replication:}
Unlike target-value allocation, target-condition replication does not divide a
numerical \ac{e2e} bound among the domains.  Instead, it defines \(T_{d,i}\) by
applying the condition in \(T_i\) to the corresponding domain.
Consider the \ac{e2e} throughput target
\(T_i=\{\Theta\ge\theta_i\}\) introduced in
Table~\ref{tab:e2e_sla_components}.  Let \(\Theta_d\) denote the throughput in
domain \(d\).  We model the \ac{e2e} throughput as the bottleneck across the
domains, i.e., \(\Theta=\min_d\Theta_d\).  Therefore,
\begin{equation*}
T_i=\{\Theta\ge\theta_i\}
=\bigcap_{d\in\mathcal D}\{\Theta_d\ge\theta_i\},
\end{equation*}
and setting \(T_{d,i}=\{\Theta_d\ge\theta_i\}\) consequently satisfies
\textnormal{(A1)}.  For example, if \(\theta_i=100\) Mbit/s, the same target
\(\Theta_d\ge100\) Mbit/s is imposed on every domain.

\noindent\textbf{(Generic construction):}
The throughput example above is an instance of the following construction.  We
use \(C_{d,i}(s)\) to denote a target event that can be evaluated within domain
\(d\); it is not a decomposition parameter.  In the throughput example,
\(C_{d,i}(s)=\{\Theta_d\ge\theta_i\}\), whereas, for the non-drop component,
\(C_{d,i}(s)=\mathcal N_d\).  The construction is summarized as follows:
\begin{itemize}[leftmargin=*,labelsep=0.5em]
\item \textit{Input:} The \ac{e2e} target event can be written as
\begin{equation*}
T_i=\bigcap_{d\in\mathcal D}C_{d,i}(s).
\end{equation*}
\item \textit{Allocation:} No numerical target value is allocated, and no
target-decomposition parameter is introduced.
\item \textit{Output:} The domain-level target events
\begin{equation*}
T_{d,i}(s,\mathbf{x}_T):=C_{d,i}(s), \qquad d\in\mathcal D.
\end{equation*}
\end{itemize}
The output satisfies
\(
H_i\cap\bigcap_{d\in\mathcal D}T_{d,i}(s,\mathbf{x}_T)
\subseteq\bigcap_{d\in\mathcal D}C_{d,i}(s)=T_i,
\)
which verifies \textnormal{(A1)}.

\textbf{(ii) Guarantee Decomposition:}
Guarantee decomposition constructs the delegated guarantee levels
\(g_{d,i}(s,\mathbf{x}_g)\) so that their product is at least the \ac{e2e}
level \(g_i\), as required by \textnormal{(A2)}.  We use the following
exponent-weight parameterization:
\begin{itemize}[leftmargin=*,labelsep=0.5em]
\item \textit{Input:} The \ac{e2e} guarantee level \(g_i\in(0,1]\).
\item \textit{Allocation:} Nonnegative exponent weights
\(\boldsymbol\eta_i=(\eta_{d,i})_{d\in\mathcal D}\) satisfying
\(\sum_d\eta_{d,i}\le1\).
\item \textit{Output:} The domain-level guarantee levels
\begin{equation*}
g_{d,i}(s,\mathbf{x}_g):=g_i^{\eta_{d,i}}.
\end{equation*}
\end{itemize}
Because \(g_i\in(0,1]\), the output satisfies
\(
g_i\le g_i^{\sum_d\eta_{d,i}}
=\prod_{d\in\mathcal D}g_{d,i}(s,\mathbf{x}_g).
\)
Thus, \textnormal{(A2)} holds. 
\begin{remark}
The latency, throughput, and non-drop components above are illustrative.  An additional \ac{sla} component can be incorporated by
defining domain-level target and guarantee constructions that satisfy
\textnormal{(A1)} and \textnormal{(A2)}, respectively.
\end{remark}

\subsubsection{Exact-Split Decomposition Space}
\label{subsubsec:exact_split_decomposition_space}
Sec.~\ref{subsubsec:concrete_admissible_decomposition} introduces nonnegative
domain-weight vectors: \(\mathbf w_i\) allocates a numerical target value,
whereas \(\boldsymbol\eta_i\) allocates a guarantee level.  We call such a
weight vector an \emph{exact split} if its entries sum to \(1\).  
We next explain \textit{why the decomposition space \(\mathcal X\) available to the
\ac{e2e} controller retains only exact splits.}

\textbf{(i) Slack Removal for Target Allocation:}
For an additive upper-bound target \(T_i=\{Q_i\le\tau_i\}\), admissibility
requires \(\sum_d w_{d,i}\le1\).  If \(\sum_d w_{d,i}<1\), part of the \ac{e2e} target
bound is not assigned to any domain.  Increasing one or more weights until they
sum to \(1\) relaxes the corresponding domain-level targets while preserving
\textnormal{(A1)}.  For example, the latency weights
\(\mathbf w_i=(0.2,0.3,0.4)\) are admissible but leave \(10\%\) of the \ac{e2e}
delay budget unused.  Replacing them with \((0.2,0.3,0.5)\) preserves
admissibility, uses the full budget, and makes no delegated target stricter.   

\textbf{(ii) Slack Removal for Guarantee Allocation:}
Guarantee admissibility requires \(\sum_d\eta_{d,i}\le1\).  If
\(\sum_d\eta_{d,i}<1\) and \(g_i<1\), then
\(\prod_d g_{d,i}>g_i\), so the delegated guarantee levels collectively exceed
the \ac{e2e} level required by \textnormal{(A2)}.  Because
\(g_i^{\eta_{d,i}}\) is nonincreasing in \(\eta_{d,i}\) for
\(g_i\in(0,1]\), increasing one or more exponent weights until they sum to
\(1\) makes no delegated guarantee stricter and preserves \textnormal{(A2)}.

\textbf{(iii) Product-of-Simplices Decomposition Space:}
The arguments in \textnormal{(i)} and \textnormal{(ii)} show that allocation
slack can be removed from both \(\mathbf w_i\) and \(\boldsymbol\eta_i\) without
making any delegated requirement stricter or violating \textnormal{(A1)} or
\textnormal{(A2)}.  We therefore restrict both weight vectors to exact splits.
By definition, each exact-split weight vector is nonnegative and sums to \(1\),
and hence belongs to the following \((|\mathcal D|-1)\)-dimensional simplex:
\begin{equation*}
\Delta^{|\mathcal D|-1}
:=\left\{\mathbf a\in\mathbb R_{\ge0}^{|\mathcal D|}
\;\middle|\;\sum_{d\in\mathcal D}a_d=1\right\}.
\end{equation*}
Fig.~\ref{fig:latency_weight_simplex} shows this simplex for a latency
allocation vector, together with one exact split and one non-admissible split.

\begin{figure}[t]
\centering
\begin{tikzpicture}[x=0.45cm,y=0.45cm,font=\footnotesize]
  \def\L{4.2}
  \def\H{3.637}
  \coordinate (AN) at (0,0);
  \coordinate (TN) at (\L,0);
  \coordinate (CN) at (\L/2,\H);
  \coordinate (Wgood) at (2.31,1.8185);
  \coordinate (Wbad) at (3.95,3.95);
  \fill[black!10] (AN) -- (TN) -- (CN) -- cycle;
  \draw[thick] (AN) -- (TN) -- (CN) -- cycle;
  \node[below, font=\normalsize] at (AN) {AN};
  \node[below, font=\normalsize] at (TN) {TN};
  \node[above, font=\normalsize] at (CN) {CN};
  \filldraw[black] (Wgood) circle (1.1pt);
  \node[anchor=west, align=left] (GoodLabel) at (2.95,2.45)
    {exact split: \(\mathbf{w}_i=(0.2,0.3,0.5)\)};
  \draw[->, thick] (GoodLabel.west) .. controls (2.55,2.25) and (2.35,2.00) .. (Wgood);
  \draw[thick] ($(Wbad)+(-0.10,-0.10)$) -- ($(Wbad)+(0.10,0.10)$);
  \draw[thick] ($(Wbad)+(-0.10,0.10)$) -- ($(Wbad)+(0.10,-0.10)$);
  \node[anchor=west, align=left] at (4.15,3.95)
    {non-admissible: \(\mathbf{w}_i=(0.4,0.4,0.4)\)\\\(\sum_d w_{d,i}=1.2\), so \((4,4,4)\) ms};
\end{tikzpicture}
\caption{Simplex for latency allocation vector with \(\delta_i=10\) ms. Each point in the simplex is a weight vector \(\mathbf{w}_i\in\Delta^2\), which yields domain-level delay targets \((w_{\mathrm{AN},i}\delta_i,w_{\mathrm{TN},i}\delta_i,w_{\mathrm{CN},i}\delta_i)\) whose sum is exactly \(\delta_i\).}
\label{fig:latency_weight_simplex}
\end{figure}
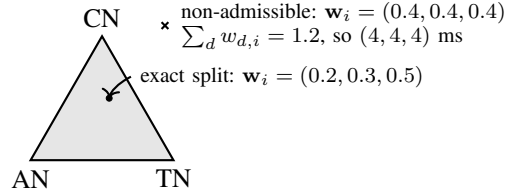

Let \(K_T,K_g\in\mathbb Z_{\ge0}\) denote the numbers of \ac{sla} components
whose target values and guarantee levels, respectively, are decomposed using
domain-weight vectors.  Accordingly, the target- and guarantee-decomposition
spaces are \(\mathcal X_T=\big(\Delta^{|\mathcal D|-1}\big)^{K_T}\) and
\(\mathcal X_g=\big(\Delta^{|\mathcal D|-1}\big)^{K_g}\), respectively.
Their Cartesian product gives the overall decomposition space
\begin{equation}
\mathcal X=\mathcal X_T\times\mathcal X_g
=\big(\Delta^{|\mathcal D|-1}\big)^{K_T+K_g}.
\label{eq:exact_split_decomposition_space}
\end{equation}

\subsection{Per-Round Hierarchical Provisioning Procedure}
\label{subsec:per_round_provisioning}
\label{subsec:ns_management_architecture}
We formalize the three-stage provisioning procedure illustrated in Fig.~\ref{fig:ns_provisioning_procedure}.
First, given \(s_t\) and the selected decomposition \(\mathbf{x}_t\), the
\ac{e2e} controller constructs and delegates the domain-level inputs
(Sec.~\ref{subsubsec:per_round_decomposition}).
Each domain-specific controller then evaluates provisioning feasibility and
reports its feasibility and resource-demand outputs
(Sec.~\ref{subsubsec:per_round_feasibility}).
Finally, the \ac{e2e} controller aggregates these outputs into the observed
admission outcome \(y_t\) and resource-consumption vector \(\mathbf{u}_t\)
(Sec.~\ref{subsubsec:per_round_outcome}).

\subsubsection{\ac{nsr} Decomposition and Domain-Level Delegation}
\label{subsubsec:per_round_decomposition}
Under the standardized hierarchical \ac{ns} management architecture, the \ac{e2e} controller delegates management tasks while domain-specific controllers autonomously manage their respective domains~\cite{etsi_zsm_003,3gpp_ts_28_530}.
In this paper, we model the information exchanged in this delegation by assuming that, at round \(t\), the \ac{e2e} controller provides each domain-specific controller \(d\in\mathcal{D}\) with two complementary inputs:
\begin{itemize}[leftmargin=*]
    \item \textbf{Path-Induced Resource Set \(\Gamma_d(s_t)\):}
    This set identifies \textit{which domain resources may be used to provision the
    \ac{nsr}.}
    \item \textbf{Delegated Request \(s_d(s_t,\mathbf{x}_t)\):}
    Defined in \eqref{eq:domain_level_nsr}, this request specifies the traffic
    profile and decomposed \ac{sla} requirements that the domain-specific
    controller uses to determine \textit{how much of each resource is required to provision the \ac{nsr}.}
\end{itemize}
Using these two inputs, each domain-specific controller determines the
allocation amounts and evaluates feasibility through its internal control
procedure, whose input--output behavior is specified in
Sec.~\ref{subsubsec:per_round_feasibility}.
Fig.~\ref{fig:path_mapping_and_decomposition} summarizes this delegation
interface and the corresponding role separation.

We now formally define \(\Gamma_d(s)\).
For each request \(s\) and covered \ac{gnb} \(a\in A\), let \(\Phi_{\mathrm{path}}(s,a)\) denote the directed downlink path selected by the \ac{e2e} controller using a fixed routing rule from the selected \ac{upf} \(u\in\mathcal{U}\) to \(a\), represented by the resources traversed along the path, including the downlink output resource of \(a\).
For example, \(\Phi_{\mathrm{path}}(s,a)\) may be obtained as a shortest admissible path from the selected \ac{upf} to \(a\).
The resulting resource set in domain \(d\) is
\begin{equation}
\Gamma_d(s)
:=
\left\{j\in[J_d]\mid \exists a\in A,\ (d,j)\in\Phi_{\mathrm{path}}(s,a)\right\}.
\label{eq:path_induced_resource_set}
\end{equation}
%
Because the routing rule is fixed, $\Gamma_d(s)$ is determined before
a decomposition is selected.\footnote{Path selection is therefore
outside the black box: the E2E controller computes
$\Phi_{\mathrm{path}}(s,a)$, and hence $\Gamma_d(s)$, itself. Only the
internal input--output mapping $\Phi_d$, which turns
$(\Gamma_d(s), s_d(s,\mathbf{x}))$ into a feasibility indicator and a
demand vector, is treated as a black box.}

\begin{figure}[t]
    \centering
    \includegraphics[width=0.8\columnwidth]{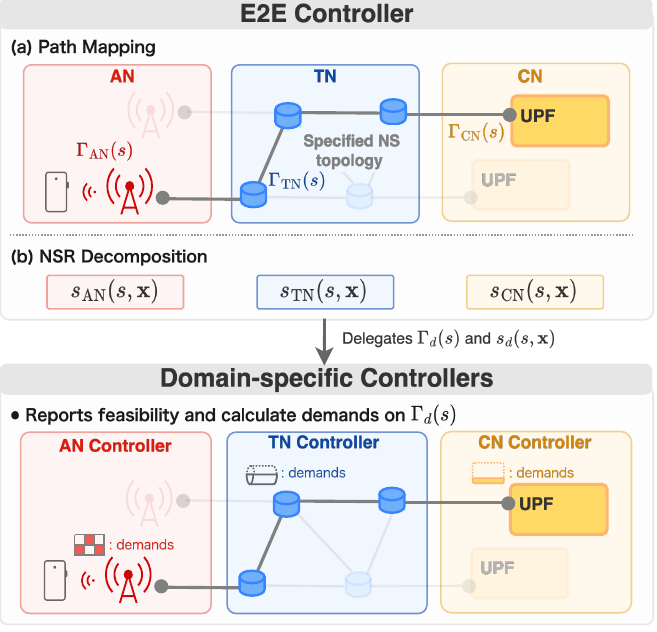}
    \caption{Role separation in hierarchical provisioning. The \ac{e2e} controller provides the path-induced resource set \(\Gamma_d(s_t)\) and delegated request \(s_d(s_t,\mathbf{x}_t)\); each domain-specific controller retains control of its internal resource allocation and reports feasibility and, when locally feasible, demand.}
    \label{fig:path_mapping_and_decomposition}
\end{figure}

\subsubsection{Domain-Level Feasibility Check}
\label{subsubsec:per_round_feasibility}
After receiving the two delegated inputs \(\Gamma_d(s_t)\) and \(s_d(s_t,\mathbf{x}_t)\), 
each domain-specific controller applies its internal control procedure \(\Phi_d\), and outputs the following two random domain-level variables:
\begin{itemize}[leftmargin=*]
    \item \textbf{Feasibility Indicator \(Y_{t,d} \in \{0,1\}\):}
    This binary indicator equals \(1\) when the controller can provision the
    delegated request \(s_d(s_t,\mathbf{x}_t)\) using the resources in
    \(\Gamma_d(s_t)\) under its current local conditions, and \(0\) otherwise.
    \item \textbf{Demand Vector
    \(\mathbf{B}_{t,d}\in\mathbb{R}_+^{J_d}\):}
    When \(Y_{t,d}=1\), \(B_{t,d,j}\) gives the amount that would
    be reserved from resource \((d,j)\) if the \ac{e2e} request were admitted.
    For \(j\notin\Gamma_d(s_t)\), \(B_{t,d,j}=0\).
\end{itemize}

We treat \(\Phi_d\) as a black box in the problem formulation.
\suppsecref{appendix:performance_pipeline_overview} gives its concrete
instantiation for the experiments.

\subsubsection{\ac{e2e} Admission and Consumption}
\label{subsubsec:per_round_outcome}
The \ac{e2e} controller aggregates the domain-level outputs \(Y_{t,d}\) and \(\mathbf{B}_{t,d}\) to determine two round-level quantities: the admission indicator \(Y_t\) and resource consumption \(\mathbf{U}_t\).

\noindent\textbf{(Admission Indicator \(Y_t\)):}
\(Y_t\) is defined as
\begin{equation}
Y_t:=\prod_{d\in\mathcal{D}}Y_{t,d}.
\label{eq:per_round_e2e_feasibility}
\end{equation}
Thus, \(Y_t=1\) if and only if every domain is locally feasible, and the request is admitted exactly in this case.

\noindent\textbf{(Realized Resource Consumption \(\mathbf{U}_t\)):}
The \ac{e2e} controller then constructs
\(\mathbf{U}_t\in\mathbb{R}_+^{J_{\mathrm{tot}}}\) as
\begin{equation}
U_{t,d,j}
:=
\begin{cases}
B_{t,d,j}, & Y_t=1 \ \land\ j\in\Gamma_d(s_t),\\
0, & Y_t=0 \ \lor\ j\notin\Gamma_d(s_t).
\end{cases}
\label{eq:realized_consumption_from_success_report}
\end{equation}
A rejected request therefore incurs no resource consumption, whereas on an
admitted round, only resources in the requested \ac{ns} topology incur
consumption, equal to their reported reservation demand \(B_{t,d,j}\).
In either case, \(0\le U_{t,d,j}\le C_{d,j}^{\max}\) almost surely.

At the end of the round, the \ac{e2e} controller observes realizations \(y_t\) and \(\mathbf{u}_t\) of the aggregate outcome \((Y_t,\mathbf{U}_t)\).


\end{document}

\documentclass[src/main.tex]{subfiles}

\begin{document}

\section{NSR-DP and Its Relaxation}
\label{sec:optimization_problem}
At each round \(t\), the provisioning procedure in Sec.~\ref{subsec:per_round_provisioning} is executed, and the \ac{e2e} controller's decision is the decomposition \(\mathbf{x}_t\in\mathcal{X}\).
We first formulate this sequential decomposition-selection problem as the \ac{nsrdp}, in which a causal policy selects \(\mathbf{x}_t\) from the current \ac{nsr} and past provisioning feedback (Sec.~\ref{subsec:original_nsrdp}).
Because its state-dependent provisioning responses make the problem difficult to analyze, we then introduce a stationary-response relaxation for tractable algorithm design (Sec.~\ref{subsec:assumption_relaxation}).
Finally, we define regret and cumulative normalized constraint violation as performance criteria for designing and evaluating online policies for the relaxed problem (Sec.~\ref{subsec:policy_evaluation_criteria}) and compare the relaxed \ac{nsrdp} formulation with prior formulations (Sec.~\ref{subsec:requirements_prior_formulations}).

\subsection{\ac{nsrdp}}
\label{subsec:original_nsrdp}
\label{subsec:protocol_and_formulation}
\label{subsec:iterational_resource_allocation}
\label{subsec:formulation}
We first define the round reward and causal policy class (Sec.~\ref{subsubsec:original_reward_policy}), and then formulate the \ac{nsrdp} (Sec.~\ref{subsubsec:original_formulation}).

\subsubsection{Reward and Causal Policy}
\label{subsubsec:original_reward_policy}
To quantify the value of successful provisioning, let \(\kappa_{\mathrm{price}}:\mathcal{S}\to\mathbb{R}_+\) denote the revenue earned by the network operator when request \(s\) is successfully provisioned, and assume \(0\le\kappa_{\mathrm{price}}(s)\le\bar p\) for some finite upper bound \(\bar p\).
The round reward is
\begin{equation}
W_t:=\kappa_{\mathrm{price}}(s_t)Y_t.
\label{eq:per_round_reward}
\end{equation}

Next, we define the history and policy class.
At the end of round \(t\), the \ac{e2e} controller observes realizations \(y_t\in\{0,1\}\) of \(Y_t\) and \(\mathbf{u}_t=(u_{t,d,j})_{d\in\mathcal{D},j\in[J_d]}\) of \(\mathbf{U}_t\).
The history is
\begin{equation}
\mathcal{H}_t
:=
\big((s_\tau,\mathbf{x}_\tau,y_\tau,\mathbf{u}_\tau)\big)_{\tau=1}^{t},
\qquad \mathcal{H}_0:=\emptyset.
\label{eq:nsrdp_history}
\end{equation}
Let \(\mathfrak{H}_t\) be the set of admissible histories through round \(t\), with \(\mathfrak{H}_0:=\{\emptyset\}\), and let \(\mathcal{P}(\mathcal{X})\) be the set of probability measures on \(\mathcal{X}\).
Let \(\Pi\) be the class of causal randomized policies \(\pi=(\pi_t)_{t=1}^T\) with
\(
\pi_t:\mathcal{S}\times\mathfrak{H}_{t-1}\to\mathcal{P}(\mathcal{X}).
\)

\subsubsection{Optimization Formulation}
\label{subsubsec:original_formulation}
The decision variable of the original \ac{nsrdp} is the policy \(\pi\in\Pi\).
For each \(\pi\in\Pi\), let \(\mathbb{E}_{\pi}\) denote expectation under the stochastic process induced by \(\pi\).
With the reward \(W_t\) and resource consumption \(U_{t,d,j}\) defined in \eqref{eq:per_round_reward} and \eqref{eq:realized_consumption_from_success_report}, respectively, the problem is
\begin{subequations}
\label{eq:nsrdp_exp}
\begin{align}
\mathrm{OPT}
:=
\max_{\pi\in\Pi}\quad
&\mathbb{E}_{\pi}\!\left[\sum_{t=1}^{T}W_t\right]
\label{eq:opt_target}\\
\text{s.t.}\quad
&\mathbb{E}_{\pi}\!\left[\sum_{t=1}^{T}U_{t,d,j}\right]
\le C_{d,j}^{\max},
\quad \forall \substack{d\in\mathcal{D}\\j\in[J_d]}.
\label{eq:opt_dj}
\end{align}
\end{subequations}
The objective \eqref{eq:opt_target} maximizes expected cumulative reward, and each constraint \eqref{eq:opt_dj} limits the expected cumulative consumption of resource \((d,j)\) by \(C_{d,j}^{\max}\).
\begin{remark}[Expected Constraints and Hard Capacity Enforcement]
For tractable online algorithm design and analysis, \eqref{eq:opt_dj} expresses
each resource budget over the horizon as a constraint on expected cumulative
consumption, following the standard long-term soft-constraint setting reviewed
in Sec.~\ref{subsec:online_learning}.
These per-resource constraints discourage the online policy from consuming
bottleneck resources too aggressively.
However, an expected constraint does not guarantee that \textit{realized cumulative
consumption remains within \(C_{d,j}^{\max}\) for every realization of requests
and provisioning outcomes.}
Because physical resource capacities cannot be exceeded in operation, the
experiments check the remaining capacities in every round, reject any request
that cannot be accommodated, and thereby evaluate the proposed method under
hard capacity enforcement.
Developing an alternative \ac{nsrdp} formulation that directly enforces
physical capacity constraints and deriving corresponding theoretical
guarantees remain future work.
\end{remark}

\subsection{Stationary-Response Relaxation}
\label{subsec:assumption_relaxation}
In the original \ac{nsrdp}, the conditional distribution of \((Y_t,\mathbf{U}_t)\) may depend on the history \(\mathcal{H}_{t-1}\), through past provisioning outcomes and the resulting remaining capacities.
To obtain a tractable model for algorithm design, 
we adopt a stationary-response relaxation that abstracts away this state dependence by assuming that \textit{the response distribution depends only on the current pair \((s_t,\mathbf{x}_t)\).}
We first formalize this relaxation through a stationary response law (Sec.~\ref{subsubsec:stationary_response_law}).
We then use this law to derive the relaxed \ac{nsrdp} (Sec.~\ref{subsubsec:stationary_policy_reduction}).

\subsubsection{Stationary Response Law}
\label{subsubsec:stationary_response_law}
We formalize the stationary-response relaxation through the following assumption.
\begin{assumption}[Stationary response law]
\label{ass:inter_round_stationarity}
There is a family of distributions \(\{\nu_{s,\mathbf{x}}\}_{(s,\mathbf{x})\in\mathcal{S}\times\mathcal{X}}\) on \(\{0,1\}\times\mathbb{R}_+^{J_{\mathrm{tot}}}\) such that, under any policy \(\pi\),
\begin{equation}
(Y_t,\mathbf{U}_t)
\mid
(\mathcal{H}_{t-1},s_t,\mathbf{x}_t)
\sim
\nu_{s_t,\mathbf{x}_t},
\quad \forall t\in[T].
\label{eq:stationary_response_law}
\end{equation}
\end{assumption}

That is, conditional on the current pair \((s_t,\mathbf{x}_t)\), the provisioning outcome \((Y_t,\mathbf{U}_t)\) is 
independent of the past history \(\mathcal{H}_{t-1}\), and the same response law applies across rounds.


\subsubsection{Derivation of the Relaxed \ac{nsrdp}}
\label{subsubsec:stationary_policy_reduction}
The reduction uses the following request-arrival assumption.
\begin{assumption}[Independent request arrivals]
\label{ass:independent_request_arrivals}
At the beginning of each round \(t\), request \(s_t\) is independently drawn from a common unknown distribution \(\mathbb{P}\) on \(\mathcal{S}\).
\end{assumption}

With Assumptions~\ref{ass:inter_round_stationarity} and~\ref{ass:independent_request_arrivals} imposed on the original \ac{nsrdp} in \eqref{eq:nsrdp_exp}, 
the conditional mean reward and resource consumption depend only on the current request--decomposition pair.
For each \((s,\mathbf{x})\), let \((Y,\mathbf{U})\sim\nu_{s,\mathbf{x}}\) and define these means as
\begin{equation}
\label{eq:mean_reward_consumption}
\begin{aligned}
f(s,\mathbf{x})
&:=
\mathbb{E}_{(Y,\mathbf{U})\sim\nu_{s,\mathbf{x}}}\!
\left[\kappa_{\mathrm{price}}(s)Y\right],\\
c_{d,j}(s,\mathbf{x})
&:=
\mathbb{E}_{(Y,\mathbf{U})\sim\nu_{s,\mathbf{x}}}\!
\left[U_{d,j}\right].
\end{aligned}
\end{equation}

For any \(\pi\in\Pi\), let \(\Pr_{\pi}\) denote the trajectory law induced by \(\pi\).
We represent its request-conditioned decomposition choices by
\(q_{\pi,t}(\cdot\mid s):=\Pr_{\pi}(\mathbf{x}_t\in\cdot\mid s_t=s)\) and their time average by
\(\bar q_{\pi}(\cdot\mid s):=T^{-1}\sum_{t=1}^Tq_{\pi,t}(\cdot\mid s)\).
Under Assumptions~\ref{ass:inter_round_stationarity} and~\ref{ass:independent_request_arrivals}, the objective and budget terms in \eqref{eq:nsrdp_exp} depend on \(\pi\) only through \(\bar q_\pi\).
In particular, the reward objective satisfies
\[
\begin{aligned}
\mathbb{E}_{\pi}\!\left[\sum_{t=1}^{T}W_t\right]
&=
\sum_{t=1}^{T}
\mathbb{E}_{s\sim\mathbb{P}}\!\left[
\mathbb{E}_{\mathbf{x}\sim q_{\pi,t}(\cdot\mid s)}[f(s,\mathbf{x})]
\right]
\\
&=
T\cdot\mathbb{E}_{s\sim\mathbb{P}}\!\left[
\mathbb{E}_{\mathbf{x}\sim\bar q_{\pi}(\cdot\mid s)}[f(s,\mathbf{x})]
\right].
\end{aligned}
\]
The budget terms in \eqref{eq:opt_dj} admit the same reduction with \(c_{d,j}\) in place of \(f\), yielding the average per-round constraints below.

\textbf{Relaxed Formulation and Exactness:}
Let \(\mathcal{Q}\) be the class of measurable stationary request-conditioned
policies \(q(\cdot\mid s)\in\mathcal{P}(\mathcal{X})\), which contains
\(\bar q_\pi\) for every \(\pi\in\Pi\).
Using the expected reward \(f(s,\mathbf{x})\) and expected consumption
\(c_{d,j}(s,\mathbf{x})\) of resource \((d,j)\) for each
request--decomposition pair, as defined in
\eqref{eq:mean_reward_consumption}, the resulting relaxed \ac{nsrdp} is
\begin{subequations}
\label{eq:relaxed_problem}
\begin{align}
\mathrm{OPT}^{\mathrm{rel}}
\!:=\!
\max_{q\in\mathcal{Q}}
&\mathbb{E}_{s\sim\mathbb{P}}\!\left[
\mathbb{E}_{\mathbf{x}\sim q(\cdot\mid s)}[f(s,\mathbf{x})]
\right]
\label{eq:relaxed_obj}\\
\text{s.t.}\quad\!\!\!
&\mathbb{E}_{s\sim\mathbb{P}}\!\left[
\mathbb{E}_{\mathbf{x}\sim q(\cdot\mid s)}
\left[c_{d,j}(s,\mathbf{x})\right]
\right]
\!\le\!\frac{C_{d,j}^{\max}}{T},
\quad\!\!\! \forall \substack{d\in\mathcal{D}\\j\in[J_d]}.
\label{eq:relaxed_con}
\end{align}
\end{subequations}
The objective \eqref{eq:relaxed_obj} maximizes expected per-round reward, while each constraint \eqref{eq:relaxed_con} limits expected per-round consumption to \(C_{d,j}^{\max}/T\).
The policy therefore determines both the decomposition used for each request
and the allocation of the shared resource budget across heterogeneous requests
over the decision horizon.
Conversely, any \(q\in\mathcal{Q}\) can be implemented by the causal policy \(\pi_t(\cdot\mid s_t,\mathcal{H}_{t-1})=q(\cdot\mid s_t)\).
Hence, under Assumptions~\ref{ass:inter_round_stationarity} and~\ref{ass:independent_request_arrivals}, the reduction is exact.


\subsection{Policy Evaluation Criteria}
\label{subsec:policy_evaluation_criteria}
\textbf{Difficulty of online learning for the relaxed \ac{nsrdp}:}
In the online setting, the primitives \(\mathbb{P}\) and \(\{\nu_{s,\mathbf{x}}\}_{(s,\mathbf{x})\in\mathcal{S}\times\mathcal{X}}\) are unknown, so the optimizer and optimal value of \eqref{eq:relaxed_problem} cannot be computed.
This information gap is the central challenge in solving the relaxed \ac{nsrdp} online.
With full information, the Lagrangian dual of \eqref{eq:relaxed_problem} decomposes across request types and can be solved by standard subgradient methods, with expectations evaluated exactly when tractable and approximated numerically otherwise.
To evaluate and design online policies under this information gap, 
we fix any \(\pi\in\Pi\) and consider both (i) its reward gap from the full-information benchmark and (ii) its cumulative budget excess.

\textbf{(i) Regret:}
To measure the reward gap, we compare the policy's cumulative mean reward with the relaxed optimum:
\begin{equation}
\mathrm{Reg}^{\mathrm{rel}}(T)
:=
T\,\mathrm{OPT}^{\mathrm{rel}}
-\sum_{t=1}^{T}f(s_t,\mathbf{x}_t).
\label{eq:policy_regret}
\end{equation}
Here, \(f(s,\mathbf{x})\), defined in
\eqref{eq:mean_reward_consumption}, is the expected reward for the
request--decomposition pair \((s,\mathbf{x})\).
Thus, \(\sum_{t=1}^{T}f(s_t,\mathbf{x}_t)\) is the cumulative mean reward
of the evaluated online policy, whereas
\(T\mathrm{OPT}^{\mathrm{rel}}\) is the cumulative expected reward
achieved by an offline optimal policy that knows the request distribution
\(\mathbb{P}\) and provisioning-outcome laws
\(\{\nu_{s,\mathbf{x}}\}\).
Accordingly, \(\mathrm{Reg}^{\mathrm{rel}}(T)\) measures \textit{the cumulative
mean-reward gap between the online policy, which operates without knowing
\(\mathbb{P}\) or \(\{\nu_{s,\mathbf{x}}\}\), and this offline optimum.}

\textbf{(ii) Violation:}
To measure how much cumulative consumption exceeds the capacity budget \(C_{d,j}^{\max}\), we
define the normalized per-round constraint function \(h_{d,j}\):
\begin{equation}
h_{d,j}(s,\mathbf{x})
:=
\frac{c_{d,j}(s,\mathbf{x})}{C_{d,j}^{\max}}-\frac{1}{T}.
\label{eq:method_exact_normalized_mean}
\end{equation}
We then define
\begin{equation}
\mathrm{Vio}^{\mathrm{rel}}_{d,j}(T)
:=
\left[
\sum_{t=1}^{T}h_{d,j}(s_t,\mathbf{x}_t)
\right]_+,
\quad \forall \substack{d\in\mathcal{D}\\j\in[J_d]}.
\label{eq:policy_relaxed_violation}
\end{equation}
Here, \(c_{d,j}(s,\mathbf{x})\), defined in
\eqref{eq:mean_reward_consumption}, is the expected consumption of resource
\((d,j)\) for the request--decomposition pair \((s,\mathbf{x})\).
Thus, \(\sum_{t=1}^{T}h_{d,j}(s_t,\mathbf{x}_t)\) is positive exactly when the evaluated online policy's
cumulative mean consumption of resource \((d,j)\) exceeds its capacity budget
\(C_{d,j}^{\max}\).
Accordingly, \(\mathrm{Vio}^{\mathrm{rel}}_{d,j}(T)\) measures the normalized
amount of this excess.

\subsection{Requirements and Applicability of Prior Methods}
\label{subsec:requirements_prior_formulations}
We first summarize the requirements
\textit{(R1)} and \textit{(R2)} encoded in the relaxed \ac{nsrdp} (Sec.~\ref{subsubsec:relaxed_nsrdp_requirements}).
We then compare the relaxed \ac{nsrdp} with the formulations underlying the
two exploration-aware online approaches \cite{odin, kobayashi2025icccn} identified in
Sec.~\ref{sec:related_work}.
We show that neither formulation directly targets the full relaxed
\ac{nsrdp} and explain how these mismatches can lead to suboptimal reward or
budget mismatch (Sec.~\ref{subsubsec:comparison_prior_formulations}).

\subsubsection{Requirements of the Relaxed \ac{nsrdp}}
\label{subsubsec:relaxed_nsrdp_requirements}
In the relaxed \ac{nsrdp}~\eqref{eq:relaxed_problem}, \textit{(R1)} is imposed 
by the resource constraints~\eqref{eq:relaxed_con}, which require every feasible \(q\in\mathcal{Q}\) to keep the expected per-round consumption at most \(C_{d,j}^{\max}/T\).
\textit{(R2)} is encoded by
optimizing over the policy class \(\mathcal{Q}\), whose elements map each
arriving request \(s\) to a distribution
\(q(\cdot\mid s)\in\mathcal{P}(\mathcal{X})\).

\subsubsection{Comparison with Prior Formulations}
\label{subsubsec:comparison_prior_formulations}
\mbox{{}}\newline
\textbf{Absence of \textit{(R1)} and \textit{(R2)} in the \acs{config} Formulation:}
Odin~\cite{odin} adopts a \ac{config}-type formulation~\cite{xu2023config} that does not account for these two requirements.
For \textit{(R1)}, \ac{config}, with constraint function \(g\), measures cumulative constraint violation as \(\sum_{t=1}^{T}[g(\mathbf{x}_t)]_+\)~\cite[Def.~2.8]{xu2023config}, whereas \eqref{eq:policy_relaxed_violation} uses the positive part of the sum.
Because \ac{config} applies the positive part separately in each round, a negative value of \(g(\mathbf{x}_t)\) in one round cannot offset a positive value in another; hence, it does not represent a resource budget shared across the \(T\) rounds.
For \textit{(R2)}, its decision \(\mathbf{x}_t\) does not depend on an arriving \ac{nsr} \(s_t\), unlike \(q(\cdot\mid s)\) in \eqref{eq:relaxed_problem}.

\textbf{Restricted Formulation in Our Previous Method:}
Our previous method~\cite{kobayashi2025icccn}, based on \ac{lincbwk}~\cite{lincbwk}, accounts for both \textit{(R1)} and \textit{(R2)}.
However, it restricts the decomposition space to a finite subset \(\bar{\mathcal{X}}\subset\mathcal{X}\) selected in advance, equivalently restricting \(q(\cdot\mid s)\) from \(\mathcal{P}(\mathcal{X})\) to \(\mathcal{P}(\bar{\mathcal{X}})\).
This restriction can exclude decompositions needed to attain the optimum of \eqref{eq:relaxed_problem}, giving rise to \textit{(L1) finite/discrete decomposition search}.
The previous method also assumes that the expected reward \(f(s,\mathbf{x})\) and resource consumption \(c_{d,j}(s,\mathbf{x})\) are linear in the request--decomposition pair \((s,\mathbf{x})\).
If these functions are nonlinear, this assumption cannot represent them accurately, giving rise to \textit{(L2) linear realizability}.

These limitations motivate an online solver that directly targets
\eqref{eq:relaxed_problem} without imposing \textit{(L1)} or \textit{(L2)}.

\end{document}

\documentclass[src/main.tex]{subfiles}

\begin{document}

\section{Proposed Method}
\label{sec:proposed_method}


\begin{figure*}[t]
\centering
\includegraphics[width=0.88\linewidth]{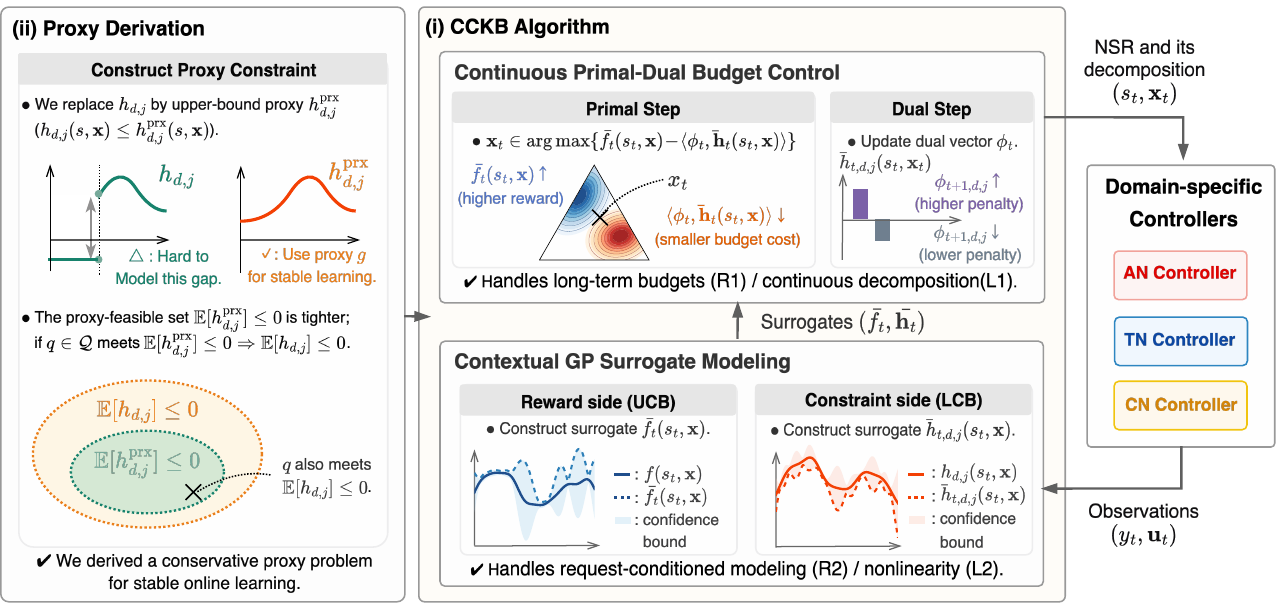}
\caption{Overview of the proposed method: (i) the \ac{cckb} algorithm; and (ii) the proxy derivation.}
\label{fig:method_overview_reduction}
\end{figure*}

\textbf{Method Overview:}
We first present \ac{cckb}, a general online learning algorithm
that can directly target the relaxed \ac{nsrdp} by treating requests as
contexts, decomposition decisions as actions, and resource budgets as
long-term constraints
(Sec.~\ref{subsec:primal_dual_algorithm}).
Direct application of \ac{cckb} to the relaxed \ac{nsrdp} requires
\ac{gp} regression for the original constraint functions \(h_{d,j}\).
However, such regression is hindered by
\textit{(L3) zero-dominated resource-consumption observations},
which can degrade the decision performance of \ac{cckb}.
To address \textit{(L3)}, we formulate a conservative proxy problem for the
relaxed \ac{nsrdp} (Sec.~\ref{subsec:reduction}).
We then instantiate \ac{cckb} for this proxy problem
(Sec.~\ref{subsec:gp_surrogates_and_kernel_design}).
Fig.~\ref{fig:method_overview_reduction} summarizes this overall procedure.

\textbf{Positioning and Contributions:}
\ac{cckb} combines the primal--dual resource-control mechanism of
\ac{ckb}~\cite{zhou2022kernelized_constraints} with contextual \ac{gp} models
over request--decomposition pairs, following the modeling principle of
\ac{cgpucb}~\cite{krause2011cgpucb}.  This combination shares the same core
primal--dual contextual \ac{bo} structure as PDCBO~\cite{xu2023pdcbo}.
Building on this structure, we formulate and analyze its application to the
relaxed \ac{nsrdp}, including a conservative proxy construction for
\textit{(L3)} and finite-time bounds on
\(\mathrm{Reg}^{\mathrm{rel}}(T)\) and every
\(\mathrm{Vio}^{\mathrm{rel}}_{d,j}(T)\).


\subsection{CCKB for the Relaxed \ac{nsrdp}}
\label{subsec:primal_dual_algorithm}
At round \(t\), after observing context \(s_t\) and selecting action
\(\mathbf{x}_t\), define \(z_t:=(s_t,\mathbf{x}_t)\).
For the relaxed \ac{nsrdp}, let \(f\) be the expected reward in
\eqref{eq:mean_reward_consumption}, and let
\(
\mathbf h:=(h_{d,j})_{d\in\mathcal D,\;j\in[J_d]}
\)
be the vector of normalized resource constraints defined in
\eqref{eq:method_exact_normalized_mean}.
Their function classes are
\(
\mathcal{F}_f:=\{f:\mathcal{S}\times\mathcal{X}\to[0,\bar p]\}
\)
and
\(
\mathcal{F}_{h_{d,j}}
:=\{h:\mathcal{S}\times\mathcal{X}\to[-1/T,\,1-1/T]\}.
\)
We first present the \ac{cckb} algorithm in this relaxed-\ac{nsrdp} notation
(Sec.~\ref{subsubsec:primal_dual_budget_control}) and then discuss the
difficulty that arises when \ac{cckb} is applied directly to the relaxed
\ac{nsrdp} (Sec.~\ref{subsubsec:original_constraint}).

\newcommand{\cckbphasebar}{%
  \makebox[0pt][r]{%
    \smash{\textcolor{black!45}{\rule[-0.25\baselineskip]{0.35pt}{1.25\baselineskip}}}%
    \hspace{0.65em}}}

\begin{algorithm}[t]
\caption{CCKB for the relaxed \ac{nsrdp}: contextual extension of the \ac{ckb} primal--dual framework~\cite{zhou2022kernelized_constraints}}
\label{alg:component1}
\begin{algorithmic}[1]
\Require Horizon \(T\), dual cap \(\rho\), dual step size \(V>0\), reward mappings \(\mathfrak{A}_{t}^{f}\), constraint mappings \(\mathfrak{A}_{t}^{h_{d,j}}\)\label{algline:component1:require}
\State \textbf{Initialize} \(\mathcal{H}_0\gets\emptyset\), \(\boldsymbol{\phi}_1\gets\mathbf{0}\in[0,\rho]^{J_{\mathrm{tot}}}\)
\For{\(t=1,\dots,T\)}\label{algline:component1:for}
    \State \(\triangleright\) \texttt{Construct surrogates.}
    \State \cckbphasebar \(\hat f_t\gets \mathfrak{A}_{t}^{f}(\mathcal{H}_{t-1})\)\label{algline:component1:surrogate-rhat}
    \State \cckbphasebar \(\hat{\mathbf h}_t\gets \bigl(\mathfrak{A}_{t}^{h_{d,j}}(\mathcal{H}_{t-1})\bigr)_{d\in\mathcal D,\;j\in[J_d]}\)\label{algline:component1:surrogate-ghat}
    \State \cckbphasebar \(\bar f_t(s,\mathbf{x})\gets\operatorname{clip}_{[0,\bar p]}(\hat f_t(s,\mathbf{x}))\) \label{algline:component1:clip-start}
    \State \cckbphasebar \(\bar{\mathbf h}_t(s,\mathbf{x})\gets\operatorname{clip}_{[-1/T,\,1-1/T]}(\hat{\mathbf h}_t(s,\mathbf{x}))\)\label{algline:component1:surrogate-end}
    \State \(\triangleright\) \texttt{Primal step.}
    \State \cckbphasebar \(\mathcal{A}_t(\mathbf{x}\mid s_t)\gets \bar f_t(s_t,\mathbf{x})-\langle\boldsymbol{\phi}_t,\bar{\mathbf h}_t(s_t,\mathbf{x})\rangle\)\label{algline:component1:primal-acq}
    \State \cckbphasebar Select \(\mathbf{x}_t \in \arg\max_{\mathbf{x}\in\mathcal{X}} \mathcal{A}_t(\mathbf{x}\mid s_t)\)\label{algline:component1:primal-end}
    \State \(\triangleright\) \texttt{Observe outcomes.}
    \State \cckbphasebar Observe feedback \((y_t,\mathbf{u}_t)\).
    \State \(\triangleright\) \texttt{Dual step.}
    \State \cckbphasebar \(\boldsymbol{\phi}_{t+1}\gets \operatorname{proj}_{[0,\rho]^{J_{\mathrm{tot}}}}\!\left(\boldsymbol{\phi}_t+\frac{1}{V}\bar{\mathbf h}_t(s_t,\mathbf{x}_t)\right)\)\label{algline:component1:dual}
    \State \cckbphasebar Update \(\mathcal H_t\).\label{algline:component1:dual-end}
\EndFor\label{algline:component1:endfor}
\end{algorithmic}
\end{algorithm}

\subsubsection{Primal--Dual CCKB Solver}
\label{subsubsec:primal_dual_budget_control}
Algorithm~\ref{alg:component1} summarizes \ac{cckb} for the relaxed \ac{nsrdp}.
The core elements of Algorithm~\ref{alg:component1} are (i) \textit{exploration strategies} and (ii) \textit{a dual vector}, which constitute the CKB-like primal--dual mechanism~\cite{zhou2022kernelized_constraints}.

\textbf{(i) Exploration Strategies:}
In lines~\ref{algline:component1:surrogate-rhat} and~\ref{algline:component1:surrogate-ghat}, the exploration mappings have the form
\(\mathfrak{A}_{t}^{f},\mathfrak{A}_{t}^{h_{d,j}}:
\mathcal{H}_{t-1}\rightarrow
\mathbb{R}^{\mathcal{S}\times\mathcal{X}}.
\)
Here, \(\mathbb R^{\mathcal S\times\mathcal X}\) denotes the set of all
real-valued functions on \(\mathcal S\times\mathcal X\).
Given the history \(\mathcal H_{t-1}\), \(\mathfrak A_t^f\) returns the raw
exploration estimate \(\hat f_t\) of the expected reward \(f\), whereas
\(\mathfrak A_t^{h_{d,j}}\) returns \(\hat h_{t,d,j}\) of the normalized mean
constraint \(h_{d,j}\).
Clipping these estimates to the ranges of their target functions yields
the reward and constraint surrogates
\(\bar f_t\in\mathcal F_f\) and
\(\bar h_{t,d,j}\in\mathcal F_{h_{d,j}}\).

We instantiate the reward mapping \(\mathfrak{A}_{t}^{f}\) to return an
\ac{ucb} of \(f\) and each constraint mapping
\(\mathfrak{A}_{t}^{h_{d,j}}\) to return a \ac{lcb} of \(h_{d,j}\).
The reward \ac{ucb} assigns larger values to decompositions
with a high predicted reward or high reward uncertainty, whereas the constraint
\ac{lcb} assigns smaller values to decompositions with low predicted resource
consumption or high consumption uncertainty.
Given the acquisition function \(\mathcal{A}_t\) in
line~\ref{algline:component1:primal-acq}, using these reward and constraint
mappings favors decompositions that may yield high rewards or consume few
resources.
For continuous \(\mathcal{X}\), the maximization in
line~\ref{algline:component1:primal-end} is solved numerically (e.g., by
multi-start \ac{slsqp}~\cite{kraft1988slsqp}).

To apply \ac{cckb} to a particular problem, its reward mapping
\(\mathfrak{A}_{t}^{f}\) and constraint mappings
\(\{\mathfrak{A}_{t}^{h_{d,j}}\}_{d\in\mathcal D,\,j\in[J_d]}\) must be
instantiated for that problem.
Section~\ref{subsubsec:original_constraint} describes how these mappings are
instantiated when \ac{cckb} is applied directly to the relaxed \ac{nsrdp},
whereas Sec.~\ref{subsec:gp_surrogates_and_kernel_design} describes how the
corresponding mappings are instantiated when \ac{cckb} is applied to the proxy
problem.

\textbf{(ii) Dual Vector:}
The dual vector
\(\boldsymbol{\phi}_t=(\phi_{t,d,j})_{d\in\mathcal D,\,j\in[J_d]}\)
stacks the resource-wise shadow prices in the primal--dual mechanism.
When resource \((d,j)\) is repeatedly tight, projected ascent increases
\(\phi_{t,d,j}\),
which amplifies the penalty term during action selection and discourages actions that rely heavily on that resource (line \ref{algline:component1:primal-acq}).
Hence, resource importance is reflected adaptively through the dual variables.

\subsubsection{Direct Application of \ac{cckb} to the Relaxed \ac{nsrdp}}
\label{subsubsec:original_constraint}
We first instantiate the reward and constraint mappings in
Algorithm~\ref{alg:component1} for the relaxed \ac{nsrdp}, and then explain
why this direct application encounters \textit{(L3)}.

\noindent\textbf{\underline{Direct Instantiation:}}
The mappings in Algorithm~\ref{alg:component1} are implemented as follows:
\begin{itemize}
    \item \textbf{Reward Mapping \(\mathfrak{A}_t^f\):}
    Use all past operator-reward observations
    \(\kappa_{\mathrm{price}}(s_\tau)y_\tau\) to fit a \ac{gp} regression
    model for \(f\), and return its raw \ac{ucb} estimate.
    \item \textbf{Constraint Mapping \(\mathfrak{A}_t^{h_{d,j}}\):}
    For each resource \((d,j)\), use all past normalized consumption
    observations \(u_{\tau,d,j}/C_{d,j}^{\max}-1/T\) to fit a \ac{gp}
    regression model for \(h_{d,j}\), and return its raw \ac{lcb} estimate.
\end{itemize}

However, this direct constraint \ac{gp} regression encounters
\textit{(L3) zero-dominated resource-consumption observations}.  By
\eqref{eq:realized_consumption_from_success_report}, its
regression response \(u_{t,d,j}/C_{d,j}^{\max}-1/T\) reflects the demand
\(B_{t,d,j}\) only when \(Y_t=1\) and \(j\in\Gamma_d(s_t)\), and otherwise
equals the repeated value \(-1/T\).  When such repeated values dominate the
training data, they obscure how \(B_{t,d,j}\) varies with the
request--decomposition pair \(z_t=(s_t,\mathbf{x}_t)\), making the constraint
\ac{gp} difficult to fit accurately.

\subsection{Conservative Proxy Problem}
\label{subsec:reduction}
Algorithm~\ref{alg:component1} relaxes \textit{(L1)} and \textit{(L2)} for the relaxed \ac{nsrdp} but encounters \textit{(L3)} when directly fitting a constraint \ac{gp} for \(h_{d,j}\).
To address this remaining limitation, we derive a proxy problem whose
constraint functions are more amenable to \ac{gp} fitting.
Specifically, we construct the proxy constraint \(h^{\mathrm{prx}}_{d,j}\) so that
it can be estimated using only resource-consumption observations
\(u_{\tau,d,j}\) from rounds satisfying
\(y_\tau=1\) and \(j\in\Gamma_d(s_\tau)\) (Sec.~\ref{subsubsec:reduction}).
We then formulate the proxy problem using \(h^{\mathrm{prx}}_{d,j}\), and show that any policy feasible for the proxy problem is also feasible
for the relaxed \ac{nsrdp}, which allows us to replace the original constraints with the proxy constraints while preserving feasibility
(Sec.~\ref{subsubsec:proxy_theoretical_implications}).

\subsubsection{Proxy Construction}
\label{subsubsec:reduction}
Starting from the selected observations \(u_{\tau,d,j}\), we
(i) define the conditional mean demand
\(m_{d,j}^{\Gamma}\) and use it to decompose the original mean consumption
\(c_{d,j}\), and then (ii) derive an upper bound on \(c_{d,j}\) and
substitute it into \(h_{d,j}\) to define the conservative proxy constraint
\(h^{\mathrm{prx}}_{d,j}\).

\textbf{(i) Mean Consumption Decomposition:}
We introduce the following quantities and use them to decompose the mean
consumption \(c_{d,j}\) in \eqref{eq:mean_reward_consumption}:
\begin{itemize}
    \item \textbf{Provisioning Success:}
    \(p(s,\mathbf{x}):=\Pr(Y=1\mid s,\mathbf{x})\) is the probability that
    provisioning succeeds for \((s,\mathbf{x})\).
    \item \textbf{Topology Membership:}
    For each \ac{nsr} \(s\), the topology-membership indicator records whether
    resource \((d,j)\) belongs to the requested \ac{ns} topology and is defined as
    \begin{equation}
    \mathbf{1}^{\Gamma}_{d,j}(s)
    :=
    \mathbf{1}\!\left\{j\in\Gamma_d(s)\right\}.
    \label{eq:method_path_indicator}
    \end{equation}
    \item \textbf{Conditional Mean Demand:}
    Define the mean demand represented by the selected observations as
    \begin{equation}
    m_{d,j}^{\Gamma}(s,\mathbf{x})
    :=
    \begin{cases}
    \mathbb{E}\!\left[
    U_{d,j}
    \mid
    Y=1,s,\mathbf{x}
    \right],
    & \!\! \substack{j\in\Gamma_d(s),\\
    p(s,\mathbf{x})>0},\\
    0,
    & \!\! \text{otherwise}.
    \end{cases}
    \label{eq:method_path_gated_mean}
    \end{equation}
\end{itemize}
By \eqref{eq:realized_consumption_from_success_report}, \(U_{d,j}=0\)
whenever \(Y=0\).  We distinguish the cases \(p(s,\mathbf{x})=0\) and
\(p(s,\mathbf{x})>0\).  For an on-path input with \(p(s,\mathbf{x})=0\),
\(Y=0\) almost surely.  Because the success-conditioned expectation
\(\mathbb{E}[U_{d,j}\mid Y=1,s,\mathbf{x}]\) is then not uniquely defined,
\eqref{eq:method_path_gated_mean} assigns its zero extension.  Consequently,
\[
c_{d,j}(s,\mathbf{x})
=
p(s,\mathbf{x})m_{d,j}^{\Gamma}(s,\mathbf{x})
=
0.
\]
For inputs with \(p(s,\mathbf{x})>0\), conditioning on \(Y\) and substituting
the definitions above gives
\begin{equation}
\begin{aligned}
c_{d,j}(s,\mathbf{x})
&=
\mathbb{E}\!\left[U_{d,j}\mid s,\mathbf{x}\right]\\
&=
p(s,\mathbf{x})
\mathbb{E}\!\left[U_{d,j}\mid Y=1,s,\mathbf{x}\right]\\
&=
p(s,\mathbf{x})m_{d,j}^{\Gamma}(s,\mathbf{x}).
\end{aligned}
\label{eq:method_consumption_factorization}
\end{equation}
The equality also holds when \(p(s,\mathbf{x})=0\), as shown above.

\textbf{(ii) Proxy Constraint Derivation:}
Equation~\eqref{eq:method_consumption_factorization} expresses the mean
consumption as the product of the provisioning-success probability and the
zero-extended conditional mean demand.  Because
\(p(s,\mathbf{x})\le 1\), it gives the upper bound
\begin{equation}
c_{d,j}(s,\mathbf{x})
\le
m_{d,j}^{\Gamma}(s,\mathbf{x}).
\label{eq:method_success_upper_bound}
\end{equation}
Replacing \(c_{d,j}\) in the definition of \(h_{d,j}\) in
\eqref{eq:method_exact_normalized_mean} with the upper bound
\eqref{eq:method_success_upper_bound}, we define the proxy
constraint\footnote{Further justification of this design choice and a discussion
of alternative modeling methods are provided in
Sec.~\ref{subsec:discussion_why_proxy}.}
\begin{equation}
h^{\mathrm{prx}}_{d,j}(s,\mathbf{x})
:=
\frac{m_{d,j}^{\Gamma}(s,\mathbf{x})}{C_{d,j}^{\max}}-\frac{1}{T}.
\label{eq:method_mean_functions}
\end{equation}
We write
\(
\mathbf h^{\mathrm{prx}}
:=(h^{\mathrm{prx}}_{d,j})_{d\in\mathcal D,\;j\in[J_d]}
\)
for the stacked proxy constraint vector.
Since \(0\le U_{d,j}\le C_{d,j}^{\max}\) almost surely, the proxy constraint satisfies
\(
-\frac{1}{T}\le h^{\mathrm{prx}}_{d,j}(s,\mathbf{x})\le 1-\frac{1}{T}.
\)

\subsubsection{Proxy Problem and Feasibility Transfer}
\label{subsubsec:proxy_theoretical_implications}
(i) We first formulate the proxy problem induced by \(\mathbf h^{\mathrm{prx}}\)
and define its performance metrics.
(ii) We then show that feasibility for the
proxy problem implies feasibility for the relaxed \ac{nsrdp}.

\textbf{(i) Proxy Problem and Performance Metrics:}
Using \(h^{\mathrm{prx}}_{d,j}\), we define the following proxy optimization problem \(\mathrm{OPT}^{\mathrm{prx}}\):
\begin{subequations}
\label{eq:proxy_relaxed_problem}
\begin{align}
\mathrm{OPT}^{\mathrm{prx}}
:=
\max_{q\in\mathcal{Q}}\quad
&\mathbb{E}_{s\sim\mathbb{P}}\!\left[
\mathbb{E}_{\mathbf{x}\sim q(\cdot\mid s)}\!\left[f(s,\mathbf{x})\right]
\right] \label{eq:proxy_relaxed_obj} \\
\text{s.t.}\quad
&\mathbb{E}_{s\sim\mathbb{P}}\!\left[
\mathbb{E}_{\mathbf{x}\sim q(\cdot\mid s)}\!\left[h^{\mathrm{prx}}_{d,j}(s,\mathbf{x})\right]
\right]
\le 0. \label{eq:proxy_relaxed_con}
\end{align}
\end{subequations}
Here, \(f(s,\mathbf{x})\), defined in
\eqref{eq:mean_reward_consumption}, is the expected reward for the
request--decomposition pair \((s,\mathbf{x})\).
The proxy objective \eqref{eq:proxy_relaxed_obj} is identical to the relaxed
objective \eqref{eq:relaxed_obj}, whereas the proxy resource constraint
\eqref{eq:proxy_relaxed_con} replaces the relaxed constraint
\eqref{eq:relaxed_con}.
In the remainder, we use \(\mathrm{OPT}^{\mathrm{prx}}\) as the benchmark.
For a policy \(\pi\), define the proxy counterparts of \eqref{eq:policy_regret} and \eqref{eq:policy_relaxed_violation} by
\begin{equation}
\mathrm{Reg}^{\mathrm{prx}}(T)
:=
T\cdot \mathrm{OPT}^{\mathrm{prx}}
-
\sum_{t=1}^T f(s_t,\mathbf{x}_t),
\label{eq:policy_proxy_regret}
\end{equation}
\begin{equation}
\mathrm{Vio}^{\mathrm{prx}}_{d,j}(T)
:=
\left[
\sum_{t=1}^{T} h^{\mathrm{prx}}_{d,j}(s_t,\mathbf{x}_t)
\right]_+.
\label{eq:policy_proxy_violation}
\end{equation}

\textbf{(ii) Feasibility Transfer to the Relaxed \ac{nsrdp}:}
Normalizing \eqref{eq:method_success_upper_bound} by \(C_{d,j}^{\max}\) and
subtracting \(1/T\) shows that
\begin{equation}
h_{d,j}(s,\mathbf{x})
\le
h^{\mathrm{prx}}_{d,j}(s,\mathbf{x}).
\label{eq:pointwise_domination_gorig_g}
\end{equation}
Taking expectations in \eqref{eq:pointwise_domination_gorig_g} gives the
following result.
\begin{lemma}[Feasibility transfer under conservative proxying]
\label{lem:proxy_feasible_implies_relaxed_feasible}
For \(q\in\mathcal Q\), if \(q\) is feasible for \eqref{eq:proxy_relaxed_problem}, then \(q\) is feasible for \eqref{eq:relaxed_problem}.
\end{lemma}
If Lemma~\ref{lem:proxy_feasible_implies_relaxed_feasible}
did not hold, a policy feasible for the proxy problem could violate the
relaxed \ac{nsrdp} constraints, so solving the proxy problem would not justify
resource-budget feasibility for the relaxed \ac{nsrdp}.  The lemma therefore
allows us to \textit{optimize and analyze the proxy problem while preserving
feasibility for the relaxed \ac{nsrdp}.}
Corollary~\ref{cor:concrete_finite_time_proxy_performance}(ii) quantifies the resulting optimality
loss.

\subsection{CCKB Instantiation for the Proxy Problem}
\label{subsec:gp_surrogates_and_kernel_design}
We now instantiate Algorithm~\ref{alg:component1} for the proxy
problem~\eqref{eq:proxy_relaxed_problem}.

\noindent\textbf{\underline{Proxy Instantiation:}}
The proxy problem retains the reward target \(f\) but replaces each
constraint target \(h_{d,j}\) with \(h^{\mathrm{prx}}_{d,j}\), defined in
\eqref{eq:method_mean_functions}.
Accordingly, the reward mapping
\(\mathfrak{A}_t^f\) remains unchanged, whereas each constraint mapping
\(\mathfrak{A}_t^{h_{d,j}}\) is replaced by
\(\mathfrak{A}_t^{h^{\mathrm{prx}}_{d,j}}\).  Other components of Algorithm~\ref{alg:component1}
remain unchanged.  The mappings are implemented as follows:
\begin{itemize}
    \item \textbf{Reward Mapping \(\mathfrak{A}_t^f\) (Unchanged):}
    Use the same reward mapping as in the direct instantiation
    (Sec.~\ref{subsubsec:original_constraint}).
    \item \textbf{Constraint Mapping \(\mathfrak{A}_t^{h^{\mathrm{prx}}_{d,j}}\):}
    Retain only past rounds satisfying
    \(y_\tau=1 \land \mathbf{1}^{\Gamma}_{d,j}(s_\tau)=1\), fit a \ac{gp}
    for \(m_{d,j}^{\Gamma}\) to the resource-consumption observations
    \(u_{\tau,d,j}\), and convert its \ac{lcb} into a raw exploration
    estimate of \(h^{\mathrm{prx}}_{d,j}\).
\end{itemize}

\begin{table}[t]
    \caption{Comparison of the direct and proxy \ac{cckb} instantiations.}
    \label{tab:direct_proxy_instantiation}
    \centering
    \footnotesize
    \setlength{\tabcolsep}{3pt}
\begin{tabularx}{\columnwidth}{@{}XXXX@{}}
\toprule
\multicolumn{2}{c}{\textbf{Surrogate Design}} &
\multicolumn{2}{c}{\textbf{\ac{cckb} Instantiation}} \\
\cmidrule(r){1-2}\cmidrule(l){3-4}
\textbf{Surrogate Type} & \textbf{Specification}
& \textbf{Direct} & \textbf{Proxy} \\
\midrule
\multirow{2}{*}{Reward}
 & Target  & \(f\) & \(f\) (unchanged) \\
 & Mapping & \(\mathfrak{A}_t^f\)
           & \(\mathfrak{A}_t^f\) (unchanged) \\
\midrule
\multirow{2}{*}{Constraint}
 & Target  & \(h_{d,j}\) & \(h^{\mathrm{prx}}_{d,j}\) \\
 & Mapping & \(\mathfrak{A}_t^{h_{d,j}}\)
           & \(\mathfrak{A}_t^{h^{\mathrm{prx}}_{d,j}}\) \\
\bottomrule
\end{tabularx}
\end{table}

Table~\ref{tab:direct_proxy_instantiation} summarizes the correspondence between the two instantiations.
The following subsubsections provide concrete \ac{gp}-based implementations of the exploration mappings \(\mathfrak{A}_t^f\) and \(\{\mathfrak{A}_t^{h^{\mathrm{prx}}_{d,j}}\}_{d\in\mathcal D,\,j\in[J_d]}\).
We first specify the common \ac{gp} posterior notation and feature-map design (Sec.~\ref{subsubsec:gp_modeling_general}),
and then instantiate the reward mapping \(\mathfrak{A}_t^f\) (Sec.~\ref{subsubsec:gp_reward_model}) and the proxy constraint mappings \(\{\mathfrak{A}_t^{h^{\mathrm{prx}}_{d,j}}\}_{d\in\mathcal D,\,j\in[J_d]}\) (Sec.~\ref{subsubsec:gp_constraint_model}).

\subsubsection{Common GP Settings}
\label{subsubsec:gp_modeling_general}
To fix the notation,
(i) we first specify common \ac{gp} posterior notation for a scalar function \(\bullet\),
which is followed by (ii) our additional design of feature maps.
For each modeled function \(\bullet\) (i.e., \(\bullet=f\) for reward and
\(\bullet=m_{d,j}^{\Gamma}\) for each constraint \((d,j)\)),
we define a dataset
\(
\mathcal{T}^{\bullet}_{t-1}
\subset
(\mathcal{S}\times\mathcal{X})\times\mathbb{R}
\),
with sample count
\(N_{t-1}^{\bullet}:=\left|\mathcal{T}_{t-1}^{\bullet}\right|\).

\textbf{(i) \ac{gp} Posterior Summaries:}
For each modeled function \(\bullet\), we use a zero-mean \ac{gp} surrogate
with kernel \(k_{\bullet}\) and regularization parameter
\(\eta_{\bullet}>0\).  Given \(\mathcal{T}^{\bullet}_{t-1}\), standard
\ac{gp} regression with the regularized Gram matrix
\(\mathbf{K}^{\bullet}_{t-1}+\eta_{\bullet}\mathbf{I}\) yields the posterior
mean \(\mu^{\bullet}_{t-1}(z)\) and standard deviation
\(\sigma^{\bullet}_{t-1}(z)\); see~\cite{rasmussen2006gpml}.  Here,
\(\mathbf{K}^{\bullet}_{t-1}\) is the kernel Gram matrix over the training
inputs.  When \(\mathcal{T}^{\bullet}_{t-1}=\emptyset\), these quantities
reduce to the prior values \(0\) and \(\sqrt{k_{\bullet}(z,z)}\), respectively.

\textbf{(ii) Feature Map:}
To allow different inductive biases across the modeled functions, we use
function-specific feature maps on the joint input \(z\):
\(\varphi_{\bullet}:\mathcal{S}\times\mathcal{X}\rightarrow\mathcal{Z}_{\bullet}.\)
We then define each kernel through \(\varphi_{\bullet}\) as
\(k_{\bullet}(z,z'):=\kappa_{\bullet}\!\big(\varphi_{\bullet}(z),
\varphi_{\bullet}(z')\big)\), where
\(z,z'\in\mathcal{S}\times\mathcal{X}\) and
\(\kappa_{\bullet}\) is a positive semidefinite kernel.
This is equivalent to defining a kernel on \(z\), but makes the representation explicit.

\subsubsection{Reward Mapping}
\label{subsubsec:gp_reward_model}
Reward learning is fully observed, and the reward-training dataset up to round \(t-1\) is
\[
\mathcal{T}^{f}_{t-1}
:=
\{(z_\tau,\kappa_{\mathrm{price}}(s_\tau)y_\tau)\}_{\tau=1}^{t-1}.
\]

Using these posterior summaries with \(\bullet=f\), we define the raw reward
\ac{ucb} estimate as
\begin{equation}
\hat f_t(z)
:=
\mu^{f}_{t-1}(z)+\beta^{f}_{t}(\alpha_f)\sigma^{f}_{t-1}(z).
\label{eq:reward_ucb}
\end{equation}
Since \(f(z)\in[0,\bar p]\), the reward surrogate used in Algorithm~\ref{alg:component1} is
\(
\bar f_t(z)
:=
\operatorname{clip}_{[0,\bar p]}\!\bigl(\hat f_t(z)\bigr).
\)
The exploration mapping \(\mathfrak{A}_{t}^{f}\) in Algorithm~\ref{alg:component1} is instantiated as
\(
\mathfrak{A}_{t}^{f}(\mathcal{H}_{t-1})
:=
\hat f_t.
\)
Here \(\alpha_f\in(0,1)\) controls the reward confidence level, and the
corresponding exploration width \(\beta_t^f(\alpha_f)\) is specified in
Sec.~\ref{subsubsec:concrete_proxy_finite_time_bounds}.

\subsubsection{Constraint Mapping}
\label{subsubsec:gp_constraint_model}
For each resource \((d,j)\), the constraint \ac{gp} models
\(m_{d,j}^{\Gamma}\) in \eqref{eq:method_path_gated_mean}, i.e., the
success-conditioned demand on resources included in the request's \ac{ns}
topology.
Selective constraint update keeps only rounds with \(y_t=1\) and \(\mathbf{1}^{\Gamma}_{d,j}(s_t)=1\). Although \(u_{t,d,j}\) is observed at every round, only those rounds are retained in \(\mathcal{T}^{m_{d,j}^{\Gamma}}_{t-1}\):
\[
\mathcal{T}^{m_{d,j}^{\Gamma}}_{t-1}
:=
\big\{(z_\tau,u_{\tau,d,j})\ \big|\ \tau\in[t-1],\,y_\tau=1,\mathbf{1}^{\Gamma}_{d,j}(s_\tau)=1\big\}.
\]
This filter addresses \textit{(L3)} by removing zeros caused by rejected requests and by resources outside the request's \ac{ns} topology before fitting the \ac{gp}.

For \(z=(s,\mathbf{x})\), we define the raw demand \ac{lcb} estimate from
the posterior summaries for \(m_{d,j}^{\Gamma}\):
\begin{equation}
\widehat m_{t,d,j}^{\Gamma}(z)
\!:=\!
\begin{cases}
\begin{aligned}
&\mu^{m_{d,j}^{\Gamma}}_{t-1}(z)\!-\!\beta_{t}^{m_{d,j}^{\Gamma}}\!\left(\!\frac{\alpha_h}{J_{\mathrm{tot}}}\!\right)
\sigma^{m_{d,j}^{\Gamma}}_{t-1}(z),
\end{aligned}
& \!\!\!\!\mathbf{1}^{\Gamma}_{d,j}(s)\!=\!1,\\
0,
& \!\!\!\!\mathbf{1}^{\Gamma}_{d,j}(s)\!=\!0.
\end{cases}
\label{eq:constraint_gp_posterior}
\end{equation}
The posterior quantities in the first branch are defined and evaluated only
at on-path inputs.
Here \(\alpha_h\in(0,1)\) is the total failure probability allocated to
simultaneous confidence bounds for all \(J_{\mathrm{tot}}\) unknown demand
functions.  For each resource \((d,j)\), we allocate failure probability
\(\alpha_h/J_{\mathrm{tot}}\) to the confidence bound for
\(m_{d,j}^{\Gamma}\), and
\(\beta_{t}^{m_{d,j}^{\Gamma}}(\alpha_h/J_{\mathrm{tot}})\in\mathbb{R}_{>0}\)
is the corresponding exploration width.
We then derive a raw exploration estimate of the proxy constraint as:
\begin{equation}
\widehat h^{\mathrm{prx}}_{t,d,j}(z)
:=
\frac{\widehat m_{t,d,j}^{\Gamma}(z)}{C_{d,j}^{\max}}
-\frac{1}{T}.
\label{eq:constraint_surrogate}
\end{equation}
Since \(h^{\mathrm{prx}}_{d,j}(z)\in[-1/T,\,1-1/T]\), the proxy constraint surrogate is
\(
\bar h^{\mathrm{prx}}_{t,d,j}(z)
:=
\operatorname{clip}_{[-1/T,\,1-1/T]}\!\bigl(\hat h^{\mathrm{prx}}_{t,d,j}(z)\bigr),
\)
and
\(
\mathfrak{A}_{t}^{h^{\mathrm{prx}}_{d,j}}(\mathcal{H}_{t-1})
:=
\hat h^{\mathrm{prx}}_{t,d,j}.
\)
Accordingly, the proxy instantiation substitutes
\(\mathfrak A_t^{h_{d,j}}\leftarrow
\mathfrak A_t^{h^{\mathrm{prx}}_{d,j}}\) and
\(\bar h_{t,d,j}\leftarrow\bar h^{\mathrm{prx}}_{t,d,j}\) in
Algorithm~\ref{alg:component1}.


\subsection{Theoretical Guarantees}
\label{subsec:success_only_feedback}
\label{subsec:theory}
Theorem~\ref{thm:general_cckb_regret_violation}
bounds the regret \(\mathrm{Reg}^{\mathrm{rel}}(T)\) and cumulative constraint
violation \(\mathrm{Vio}^{\mathrm{rel}}_{d,j}(T)\) for \ac{cckb} applied
directly to the relaxed \ac{nsrdp}.
For proxy-based \ac{cckb},
Corollary~\ref{cor:concrete_finite_time_proxy_performance} uses \ac{gp}
surrogate-error bounds to bound \(\mathrm{Reg}^{\mathrm{prx}}(T)\), transfers
this bound to \(\mathrm{Reg}^{\mathrm{rel}}(T)\) by accounting for the proxy
optimality gap, and bounds \(\mathrm{Vio}^{\mathrm{rel}}_{d,j}(T)\).

\textbf{Scope of Theoretical Results:}
Algorithm~\ref{alg:component1} remains executable when the analytical
assumptions fail, but its guarantees may not hold.  These guarantees concern
the relaxed \ac{nsrdp}; effectiveness for the original
\ac{nsrdp}~\eqref{eq:nsrdp_exp} is evaluated empirically in
Sec.~\ref{sec:evaluation} with residual resources explicitly managed.

\subsubsection{Regularity and Optimization Conditions}
\label{subsubsec:cckb_analytical_foundations}
We adopt the standard bounded-\ac{rkhs} conditions~\cite{zhou2022kernelized_constraints}.
\begin{assumption}[RKHS and kernel conditions]
\label{ass:cckb_kernel_regular}
The means \(f\) and \(h_{d,j}\) lie in the \acp{rkhs} of \(k_f\) and
\(k_{h_{d,j}}\), with \(\|f\|_{\mathcal H_{k_f}}\le B^f\) and finite
\(\|h_{d,j}\|_{\mathcal H_{k_{h_{d,j}}}}\).
On inputs with \(\mathbf1_{d,j}^{\Gamma}(s)=1\),
\(m_{d,j}^{\Gamma}\in\mathcal H_{k_{m_{d,j}^{\Gamma}}}\) with norm at most
\(B^{m_{d,j}^{\Gamma}}\).  All kernels satisfy \(k_\bullet(z,z)\le1\)
and are continuous in their action arguments.
\end{assumption}

\begin{assumption}[Measurability in the request]
\label{ass:cckb_request_measurability}
For every fixed \(\mathbf x\in\mathcal X\), the function
\(f(\cdot,\mathbf x)\) is measurable.  For every resource \((d,j)\),
\(h_{d,j}(\cdot,\mathbf x)\), the zero extension of
\(m_{d,j}^{\Gamma}(\cdot,\mathbf x)\), and the indicator
\(\mathbf1_{d,j}^{\Gamma}(\cdot)\) are measurable.
\end{assumption}

For \(q\in\mathcal Q\), define
\(v_q^f(s):=\mathbb E_{\mathbf x\sim q(\cdot\mid s)}[f(s,\mathbf x)]\)
and
\(\mathbf v_q^h(s):=\mathbb E_{\mathbf x\sim q(\cdot\mid s)}[\mathbf h(s,\mathbf x)]\),
whose \((d,j)\) component is \(v_q^{h_{d,j}}(s)\).

\begin{assumption}[Slater condition]
\label{ass:cckb_slater}
For each fixed horizon \(T\), there exist a policy
\(q^\circ\in\mathcal Q\) and a margin
\(\xi_T>0\) such that
\[
\mathbb E_{s\sim\mathbb P}[\mathbf v_{q^\circ}^h(s)]
\le-\xi_T\mathbf1.
\]
\end{assumption}

Provided that the relaxed \ac{nsrdp} and the proxy problem each admit an
optimal policy and satisfy Assumption~\ref{ass:cckb_slater} for their
respective constraints, Corollary~\ref{cor:cckb_slater_consequence}
guarantees the existence of optimal dual vectors with finite \(\ell_1\)
norms.  We denote corresponding norm bounds by
\(\Lambda_T^{\mathrm{rel}}\) and \(\Lambda_T^{\mathrm{prx}}\), respectively.
The corollary and its application to the proxy problem are stated in
\suppsecref{appendix:slater_and_generality}; the proof is in
\suppsecref{appendix:cckb_slater_consequence}.

\arxivonly{\subsubsection{Finite-Time Performance Guarantees for CCKB}}
\journalonly{\subsubsection{Generic Finite-Time Guarantee}}
\label{subsubsec:general_cckb_guarantee}
  \ifarxivversion
    We establish finite-time regret and cumulative constraint-violation
guarantees for \ac{cckb} applied directly to the relaxed \ac{nsrdp}.  We
first specify the surrogate conditions and request-concentration events used
in Theorem~\ref{thm:general_cckb_regret_violation}, together with their
failure-probability budgets, and then state the resulting finite-time bounds
and their overall confidence level.
We subsequently instantiate the same guarantee for the proxy problem.

\textbf{(i) Surrogate Conditions:}
Algorithm~\ref{alg:component1} selects actions and updates the dual vector
using the surrogates \(\bar f_t\) and \(\bar h_{t,d,j}\) in place of \(f\) and \(h_{d,j}\).
Following the GP-UCB analysis of the
original \ac{ckb}~\cite{zhou2022kernelized_constraints}, we characterize their
accuracy through (C1) one-sided bounds used for action selection and
(C2) cumulative-error bounds along the selected actions.

\noindent\textbf{(One-Sided-Bound Condition):}
\textit{The reward surrogate must not underestimate the reward, whereas each
constraint surrogate must not overestimate its target constraint.}
Specifically, for every \(t\in[T]\) and
\(z\in\mathcal S\times\mathcal X\),
\begin{equation}
\text{(C1):}\quad
f(z)\le\bar f_t(z),
\ \
\bar{\mathbf h}_t(z)\le\mathbf h(z).
\label{eq:cckb_surrogate_direction}
\end{equation}
Here, the vector inequality is interpreted componentwise.

\noindent\textbf{(Cumulative-Error Condition):}
\textit{The estimation errors along the selected inputs must not accumulate too
quickly.}  For \(z_t=(s_t,\mathbf x_t)\), let
\(\mathcal W_f(T)\ge0\) and
\(
\boldsymbol{\mathcal W}_h(T)\in\mathbb R_+^{J_{\mathrm{tot}}}
\)
denote deterministic upper bounds on the cumulative reward and
constraint estimation errors, respectively, and require the following, with
the vector inequality interpreted componentwise:
\begin{equation}
\text{(C2):}\quad
\begin{aligned}
&\sum_{t=1}^T\bigl(\bar f_t(z_t)-f(z_t)\bigr)
\le\mathcal W_f(T),\\
&\sum_{t=1}^T\bigl(\mathbf h(z_t)-\bar{\mathbf h}_t(z_t)\bigr)
\le\boldsymbol{\mathcal W}_h(T).
\end{aligned}
\label{eq:cckb_cumulative_surrogate_error}
\end{equation}

The surrogate functions \(\bar f_t\) and \(\bar{\mathbf h}_t\) are constructed
from the random history \(\mathcal H_{t-1}\), whereas the selected input
\(z_t=(s_t,\mathbf x_t)\) depends on the realized request \(s_t\) and the
selected action.  Consequently, (C1) and (C2) contain random quantities and
need not hold for every realization.  As in the original
\ac{ckb}~\cite{zhou2022kernelized_constraints}, let
\(\mathcal E_{\mathrm{sur}}^{\mathrm{reg}}\) denote the event on which (C1)
and the first inequality in (C2) hold, and let
\(\mathcal E_{\mathrm{sur}}^{\mathrm{vio}}\) denote the event on which (C1)
and both inequalities in (C2) hold.  Thus,
\(\mathcal E_{\mathrm{sur}}^{\mathrm{vio}}
\subseteq\mathcal E_{\mathrm{sur}}^{\mathrm{reg}}\).
We use \(\alpha_{\mathrm{sur}}\) as a common failure-probability budget.  The
regret bound requires
\(\Pr(\mathcal E_{\mathrm{sur}}^{\mathrm{reg}})
\ge 1-\alpha_{\mathrm{sur}}\), whereas the constraint-violation bound requires
\(\Pr(\mathcal E_{\mathrm{sur}}^{\mathrm{vio}})
\ge 1-\alpha_{\mathrm{sur}}\).

\textbf{(ii) Request-Sequence Concentration:}
Along the realized request sequence, the aggregate request-wise reward
values \(v_{q^\star}^f(s_t)\) and constraint values
\(\mathbf v_{q^\star}^h(s_t)\) may deviate from their expectations
\(\mathbb E_{s\sim\mathbb P}[v_{q^\star}^f(s)]\) and
\(\mathbb E_{s\sim\mathbb P}[\mathbf v_{q^\star}^h(s)]\), respectively.
The parameters \(\alpha_{\mathrm{ctx}}^f\) and
\(\alpha_{\mathrm{ctx}}^h\) are the failure-probability budgets for the
reward and constraint concentration inequalities in
Propositions~\ref{prop:context_concentration_fixed_q}
and~\ref{prop:context_concentration_fixed_q_constraint}, respectively
(\suppsecref{appendix:context_concentration_proof}).

\textbf{(iii) Performance Guarantee:}
We establish finite-time regret and cumulative constraint-violation bounds
for \ac{cckb} applied directly to the relaxed \ac{nsrdp}, with the regret
bound expressed in terms of \(\mathcal W_f(T)\) and the violation bound also
using \(\boldsymbol{\mathcal W}_h(T)\).
Each bound is established on the intersection of its required surrogate and
request-concentration events.  A union bound gives the overall probability
\(1\!-\!\alpha_{\mathrm{ctx}}^f\!-\!\alpha_{\mathrm{ctx}}^h
\!-\!\alpha_{\mathrm{sur}}\).
  \fi

  \ifjournalversion
    \input{src/sections/variants/sec_proposed_method_generic_guarantee_journal}%
  \fi

\begin{theorem}[Finite-time regret and constraint-violation bounds]
\label{thm:general_cckb_regret_violation}
Suppose that the relaxed \ac{nsrdp} admits an optimal policy and that
Assumptions~\ref{ass:independent_request_arrivals},
\ref{ass:cckb_kernel_regular}, and~\ref{ass:cckb_request_measurability} hold.
For failure probabilities
\(\alpha_{\mathrm{ctx}}^f,\alpha_{\mathrm{ctx}}^h,
\alpha_{\mathrm{sur}}\in(0,1)\) satisfying
\(\alpha_{\mathrm{ctx}}^f+\alpha_{\mathrm{ctx}}^h
+\alpha_{\mathrm{sur}}\!<\!1\), choose \(\rho\!>\!0\), and set
\(V\!:=\!\sqrt{J_{\mathrm{tot}}T}/\rho\) in
Algorithm~\ref{alg:component1}.

\emph{(i) Regret:}
\arxivonly{Suppose that
\(\Pr(\mathcal E_{\mathrm{sur}}^{\mathrm{reg}})
\ge1-\alpha_{\mathrm{sur}}\).}
\journalonly{%
Suppose that the one-sided bounds in \eqref{eq:cckb_surrogate_direction}
and the reward cumulative-error bound, i.e., the first inequality in
\eqref{eq:cckb_cumulative_surrogate_error}, hold jointly with probability at
least \(1-\alpha_{\mathrm{sur}}\).}
Then, with probability at least
\(1-\alpha_{\mathrm{ctx}}^f-\alpha_{\mathrm{ctx}}^h
-\alpha_{\mathrm{sur}}\), there exists a finite upper bound
\(B_{\mathrm{reg}}(T)\) such that
\begin{align}
\mathrm{Reg}^{\mathrm{rel}}(T)
&\le B_{\mathrm{reg}}(T)\notag\\
&\!=\!\mathcal O\!\left(
\begin{aligned}[c]
&\bar p\sqrt{T\log\tfrac{1}{\alpha_{\mathrm{ctx}}^f}}
+\rho J_{\mathrm{tot}}\sqrt{T\log\tfrac{1}{\alpha_{\mathrm{ctx}}^h}}\\
&{}+\mathcal W_f(T)+\rho\sqrt{J_{\mathrm{tot}}T}
\end{aligned}
\right).
\label{eq:general_cckb_regret_bound}
\end{align}

\emph{(ii) Constraint violation:}
Suppose, in addition, that Assumption~\ref{ass:cckb_slater} holds.  For the
fixed horizon \(T\), choose \(\rho\ge2\Lambda_T^{\mathrm{rel}}\), and suppose
\arxivonly{that
\(\Pr(\mathcal E_{\mathrm{sur}}^{\mathrm{vio}})
\ge1-\alpha_{\mathrm{sur}}\).}
\journalonly{that \(\Pr(\mathcal E_{\mathrm{sur}}^{\mathrm{vio}})\ge1-\alpha_{\mathrm{sur}}\), where
\(\mathcal E_{\mathrm{sur}}^{\mathrm{vio}}\) includes both cumulative-error bounds in
\eqref{eq:cckb_cumulative_surrogate_error}.}
Then, with probability at least
\(1-\alpha_{\mathrm{ctx}}^f-\alpha_{\mathrm{ctx}}^h
-\alpha_{\mathrm{sur}}\), there exists a finite upper bound
\(B_{\mathrm{vio}}(T)\) such that the following inequality holds for every
\((d, j)\):
\begin{align}
\mathrm{Vio}^{\mathrm{rel}}_{d,j}(T)
&\le B_{\mathrm{vio}}(T)\notag\\
&\!=\!\mathcal O\!\left(
\begin{aligned}[c]
&\frac{\mathcal W_f(T)}{\rho}
+J_{\mathrm{tot}}\sqrt{T\log\tfrac{1}{\alpha_{\mathrm{ctx}}^h}}\\
&{}+\!\left\|\boldsymbol{\mathcal W}_h(T)\right\|_1
+\sqrt{J_{\mathrm{tot}}T}
\end{aligned}
\right).
\label{eq:general_cckb_violation_bound}
\end{align}
\end{theorem}

The proof is in \suppsecref{appendix:general_cckb_theorem_proof}.
The \(\log(1/\alpha_{\mathrm{ctx}}^{f})\) and
\(\log(1/\alpha_{\mathrm{ctx}}^{h})\) terms in
\eqref{eq:general_cckb_regret_bound} and
\eqref{eq:general_cckb_violation_bound} arise from concentration over the realized request
sequence and have no counterpart in CKB.

\subsubsection{GP-Based Proxy Guarantee}
\label{subsubsec:concrete_proxy_finite_time_bounds}
Theorem~\ref{thm:general_cckb_regret_violation} leaves the cumulative
surrogate-error bounds \(\mathcal W_f(T)\) and
\(\boldsymbol{\mathcal W}_h(T)\) abstract.
We instantiates these quantities for proxy-based \ac{cckb}
using \ac{gp} confidence widths and maximum information gains.
We first specify the required \ac{gp} assumptions and exploration widths and
then derive explicit bounds on \(\mathcal W_f(T)\) and
\(\boldsymbol{\mathcal W}_h(T)\).
Substituting these bounds into
Theorem~\ref{thm:general_cckb_regret_violation} yields
Corollary~\ref{cor:concrete_finite_time_proxy_performance}.

\begin{assumption}[Reward observation noise]
\label{ass:reward_gp_regular}
The residual \(W_t-f(z_t)\) is conditionally \(\sigma^f\)-sub-Gaussian given
\((\mathcal{H}_{t-1},z_t)\).
\end{assumption}

\begin{assumption}[Proxy-demand observation noise]
\label{ass:success_only_gp_regular}
On rounds with \(Y_t=1\) and \(\mathbf{1}^{\Gamma}_{d,j}(s_t)=1\), the
residual \(U_{t,d,j}-m_{d,j}^{\Gamma}(z_t)\) is conditionally
\(\sigma^{m_{d,j}^{\Gamma}}\)-sub-Gaussian given
\((\mathcal{H}_{t-1},z_t,Y_t=1,\mathbf{1}^{\Gamma}_{d,j}(s_t)=1)\).
\end{assumption}

\begin{definition}[Maximum information gain]
For each modeled function \(\bullet\), define
\begin{equation}
\gamma_n^{\bullet}
:=
\max_{\mathbf Z_n\in\operatorname{dom}(\bullet)^n}
\frac{1}{2}\log\det\!\left(
\mathbf I+\eta_{\bullet}^{-1}\mathbf K_{\mathbf Z_n}^{\bullet}
\right),
\label{eq:maximum_information_gain}
\end{equation}
where \(\mathbf Z_n=(z_1,\ldots,z_n)\) and
\(\mathbf K_{\mathbf Z_n}^{\bullet}:=
[k_{\bullet}(z_a,z_b)]_{a,b=1}^n\).
For the reward model, \(\operatorname{dom}(f)=\mathcal S\times\mathcal X\).
For resource \((d,j)\),
\(
\operatorname{dom}(m_{d,j}^{\Gamma})
=
\{(s,\mathbf x)\in\mathcal S\times\mathcal X
\mid \mathbf1_{d,j}^{\Gamma}(s)=1\}.
\)
\end{definition}

\textbf{Exploration Widths:}
Under Assumptions~\ref{ass:cckb_kernel_regular},
\ref{ass:reward_gp_regular}, and~\ref{ass:success_only_gp_regular}, we follow
the GP-UCB instantiation of
\ac{ckb}~\cite[Corollary~1]{zhou2022kernelized_constraints} to set the
exploration widths used in the reward \ac{ucb}~\eqref{eq:reward_ucb} and the
proxy-demand \ac{lcb}~\eqref{eq:constraint_gp_posterior}.  With normalized
noise scales
\(
\widetilde\sigma^f:=\sigma^f/\sqrt{\eta_f}
\)
and
\(
\widetilde\sigma_{d,j}:=
\sigma^{m_{d,j}^{\Gamma}}/\sqrt{\eta_{m_{d,j}^{\Gamma}}}
\), these widths are
\begin{align}
\beta_t^f(\alpha_f)
&\!:=\!B^f\!+\!\widetilde\sigma^f
\sqrt{2\left(\gamma_{t-1}^f\!+\!1\!+\!\log\frac{1}{\alpha_f}\right)},
\label{eq:concrete_reward_exploration_width}\\
\beta_t^{m_{d,j}^{\Gamma}}
\!\left(\frac{\alpha_h}{J_{\mathrm{tot}}}\right)
&\!:=\!B^{m_{d,j}^{\Gamma}}\!+\!\widetilde\sigma_{d,j}\sqrt{2\left(
\gamma_{N_{t-1}^{m_{d,j}^{\Gamma}}}^{m_{d,j}^{\Gamma}}\!+\!1
\!+\!\log\frac{J_{\mathrm{tot}}}{\alpha_h}
\right)}.
\label{eq:concrete_proxy_demand_exploration_width}
\end{align}

\begin{assumption}[Existence of a valuable decomposition for every request]
\label{ass:valuable_decomposition}
For the fixed horizon \(T\), there exist constants \(M_T>0\) and \(a_T>0\)
such that, for \(\mathbb P\)-almost every request \(s\),
\begin{equation}
\max_{\mathbf x\in\mathcal X}
\left\{
f(s,\mathbf x)
-M_T
\sum_{d\in\mathcal D}
\sum_{j\in[J_d]}
\frac{m_{d,j}^{\Gamma}(s,\mathbf x)}{C_{d,j}^{\max}}
\right\}
\ge a_T.
\label{eq:valuable_decomposition_condition}
\end{equation}
In words, \textit{almost every request has at least one decomposition whose expected
reward exceeds \(M_T\) times its total normalized conditional resource demand
by at least \(a_T\).}
\end{assumption}

\begin{remark}[Interpretation of Assumption~\ref{ass:valuable_decomposition}]
The condition excludes request populations containing a non-negligible subset
for which every decomposition is unlikely to succeed or consumes too many
resources relative to its revenue.  Removing this condition remains future work.
\end{remark}

Define the nonnegative per-round proxy gap by
\begin{equation}
\Delta_{\mathrm{prx}}(T)
:=\mathrm{OPT}^{\mathrm{rel}}-\mathrm{OPT}^{\mathrm{prx}}.
\label{eq:proxy_optimality_gap_definition}
\end{equation}
The notation \(\widetilde{\mathcal O}\) suppresses logarithmic factors in
\(T\), inverse failure probabilities, and \(J_{\mathrm{tot}}\), while treating
the kernels \(k_\bullet\), the \ac{rkhs}-norm bounds \(B^f\) and
\(B^{m_{d,j}^{\Gamma}}\), the noise scales \(\widetilde\sigma^f\) and
\(\widetilde\sigma_{d,j}\), and the regularization parameters
\(\eta_\bullet\) as fixed.

\begin{corollary}[Concrete finite-time bounds for proxy-based CCKB]
\label{cor:concrete_finite_time_proxy_performance}
Fix \(\alpha_f,\alpha_h,\alpha_{\mathrm{ctx}}^f,
\alpha_{\mathrm{ctx}}^h\in(0,1)\) such that
\(\alpha_{\mathrm{ctx}}^f+\alpha_{\mathrm{ctx}}^h
+\alpha_f+\alpha_h<1\).
Suppose that Assumptions~\ref{ass:inter_round_stationarity},
\ref{ass:independent_request_arrivals},~\ref{ass:cckb_kernel_regular},
\ref{ass:cckb_request_measurability},~\ref{ass:reward_gp_regular},
and~\ref{ass:success_only_gp_regular} hold and that the proxy problem
admits an optimal policy.  Choose \(\rho>0\), set
\(V:=\sqrt{J_{\mathrm{tot}}T}/\rho\), and use the exploration widths in
\eqref{eq:concrete_reward_exploration_width} and
\eqref{eq:concrete_proxy_demand_exploration_width}.

\emph{(i) Regret:}
With probability at least
\(1-\alpha_{\mathrm{ctx}}^f-\alpha_{\mathrm{ctx}}^h
-\alpha_f-\alpha_h\),
\begin{align}
\mathrm{Reg}^{\mathrm{prx}}(T)\le B_{\mathrm{prx}}(T)
=\widetilde{\mathcal O}\!\left(
\bar p\sqrt T+\rho J_{\mathrm{tot}}\sqrt T+\gamma_T^f\sqrt T
\right).
\label{eq:concrete_gp_proxy_regret_bound}
\end{align}

\emph{(ii) Relaxed \ac{nsrdp} regret:}
If the relaxed \ac{nsrdp} also admits an optimal policy, then, on the same
event,
\begin{align}
\mathrm{Reg}^{\mathrm{rel}}(T)
&=\mathrm{Reg}^{\mathrm{prx}}(T)+T\Delta_{\mathrm{prx}}(T)\notag\\
&\le B_{\mathrm{prx}}(T)+T\Delta_{\mathrm{prx}}(T).
\label{eq:concrete_gp_relaxed_regret_transfer}
\end{align}
If, additionally, Assumption~\ref{ass:cckb_slater} holds for the proxy
constraints and Assumption~\ref{ass:valuable_decomposition} holds with
\(M_T\ge\Lambda_T^{\mathrm{rel}}\), then
\begin{equation}
\mathrm{Reg}^{\mathrm{rel}}(T)
\le B_{\mathrm{prx}}(T)
+\Lambda_T^{\mathrm{prx}}\left(\frac{2\bar p}{a_T}-1\right).
\label{eq:concrete_gp_relaxed_regret_bound}
\end{equation}

\emph{(iii) Constraint violation:}
Suppose also that Assumption~\ref{ass:valuable_decomposition} holds with
\(M_T\ge\rho\) and Assumption~\ref{ass:cckb_slater} holds for the proxy
constraints.  Choose \(\rho\ge2\Lambda_T^{\mathrm{prx}}\), and fix
\(\alpha_w\in(0,1)\) such that
\(\alpha_{\mathrm{ctx}}^f+\alpha_{\mathrm{ctx}}^h
+\alpha_f+\alpha_h+\alpha_w<1\).  Then, with probability at least
\(1-\alpha_{\mathrm{ctx}}^f-\alpha_{\mathrm{ctx}}^h
-\alpha_f-\alpha_h-\alpha_w\), the following bound holds simultaneously
for every resource \((d,j)\):
\begin{equation}
\mathrm{Vio}^{\mathrm{rel}}_{d,j}(T)
\le\widetilde{\mathcal O}\!\left(\sqrt T\left[
\begin{aligned}[c]
&\frac{\gamma_T^f}{\rho}+J_{\mathrm{tot}}
+\frac{J_{\mathrm{tot}}\gamma_T^f}{a_T}\\
&{}+\frac{\bar p}{a_T}\sum_{d'\in\mathcal D}\sum_{j'\in[J_{d'}]}
\frac{\gamma_T^{m_{d',j'}^{\Gamma}}}{C_{d',j'}^{\max}}
\end{aligned}\right]\right).
\label{eq:concrete_gp_relaxed_violation_bound}
\end{equation}
\end{corollary}

For fixed \(\rho\) and failure probabilities, the proxy-regret bound in
\eqref{eq:concrete_gp_proxy_regret_bound} is sublinear whenever
\(\gamma_T^f\sqrt T=o(T)\).  This condition alone does not ensure that the
relaxed-\ac{nsrdp} regret in
\eqref{eq:concrete_gp_relaxed_regret_transfer} is sublinear, because of the
separate term \(T\Delta_{\mathrm{prx}}(T)\).  Under the additional conditions
in part~\emph{(ii)}, this term is bounded by
\(\Lambda_T^{\mathrm{prx}}(2\bar p/a_T-1)\), which does not grow with \(T\)
explicitly; hence, if \(\Lambda_T^{\mathrm{prx}}\) and \(1/a_T\) remain
bounded in \(T\), proxy conservatism contributes only an \(\mathcal O(1)\)
cumulative-regret term and \eqref{eq:concrete_gp_relaxed_regret_bound} is
sublinear.
The explicit bounds and proof are in
\suppsecref{appendix:explicit_proxy_bounds}
(Corollary~\ref{cor:explicit_finite_time_proxy_performance}) and
\suppsecref{appendix:concrete_finite_time_guarantee_proof}.

\ifSubfilesClassLoaded{
  \bibliographystyle{ieee/IEEEtran}
  \bibliography{bib/references}
}{}

\end{document}

\documentclass[src/main.tex]{subfiles}

\begin{document}

\section{Evaluation}
\label{sec:evaluation}

\begin{table*}[t]
\caption{Evaluation objectives of the experiments.}
\label{tab:evaluation_overview}
\centering
\small
\renewcommand{\arraystretch}{1.1} 
\setlength{\tabcolsep}{4pt}
\begin{tabularx}{\textwidth}{@{}>{\raggedright\arraybackslash}p{0.11\textwidth}X@{}}
\toprule
\textbf{Section} & \textbf{Evaluation Objective} \\
\midrule
\rowcolor{black!8}
\multicolumn{2}{@{}l}{\rule[-0.35ex]{0pt}{2.5ex}\textbf{Effectiveness and Robustness}} \\
Sec.~\ref{subsec:result}
& \textbf{Overall Effectiveness Across Topologies and Domain Bottlenecks:} Evaluates whether the complete method improves cumulative reward and resource allocation relative to the baselines across network topologies and domain bottlenecks. \\
Sec.~\ref{subsec:traffic_mixture_tree}
& \textbf{Robustness to Heterogeneous Requests:} Evaluates whether the proposed method remains effective relative to the baselines when \ac{urllc} and \ac{embb} requests coexist and their traffic composition changes. \\
\addlinespace[2pt]
\rowcolor{black!8}
\multicolumn{2}{@{}l}{\rule[-0.35ex]{0pt}{2.5ex}\textbf{Ablation Studies}} \\
Sec.~\ref{subsec:ablation_relaxation_effect}
& \textbf{Effects of Relaxing \textit{(L1)} and \textit{(L2)}:} Isolates the effects of continuous decomposition search for relaxing \textit{(L1) finite/discrete decomposition search} and \ac{gp} surrogate modeling for relaxing \textit{(L2) linear realizability}, relative to our previous method~\cite{kobayashi2025icccn}. \\
Sec.~\ref{subsec:ablation_proxy_filters}
& \textbf{Effect of Addressing \textit{(L3)}:} Evaluates whether learning \ac{cckb}'s proxy constraint target only from successful rounds and resources included in the request's \ac{ns} topology mitigates \textit{(L3) zero-dominated resource-consumption observations}. \\
Sec.~\ref{subsec:ablation_context_dual}
& \textbf{Effects of Mechanisms for \textit{(R1)} and \textit{(R2)}:} Isolates the contributions of dual-based long-horizon resource-budget control for \textit{(R1)} and request-conditioned surrogate modeling for \textit{(R2)}. \\
\addlinespace[2pt]
\rowcolor{black!8}
\multicolumn{2}{@{}l}{\rule[-0.35ex]{0pt}{2.5ex}\textbf{Scalability and Sensitivity Analyses}} \\
Sec.~\ref{subsec:runtime_scalability}
& \textbf{Computational Scalability:} Evaluates per-round runtime and total runtime as the numbers of \acp{gnb} increase. \\
\shortsuppsecref{appendix:hyperparameter_robustness}
& \textbf{Hyperparameter Robustness:} Evaluates the sensitivity of total reward to the dual-cap parameter, which limits the dual penalty weights, and the \ac{gp} exploration-scale parameter, which controls the confidence widths. \\
\bottomrule
\end{tabularx}
\end{table*}

We conducted experiments in a 5G simulator.
We first describe the experimental 5G network and \ac{nsr} settings (Sec.~\ref{sec:experimental_setup}),
then summarize the comparison methods (Sec.~\ref{subsec:comparison_methods})
and define the evaluation metrics (Sec.~\ref{subsec:evaluation_metrics}).
We then present experiments and their results.
Table~\ref{tab:evaluation_overview} summarizes the experiments and their corresponding objectives.

\subsection{Experimental Setup}
\label{sec:experimental_setup}
\subsubsection{5G Network}
\label{subsubsec:network_setup}

\begin{figure*}[t]
\centering
\includegraphics[width=0.78\linewidth]{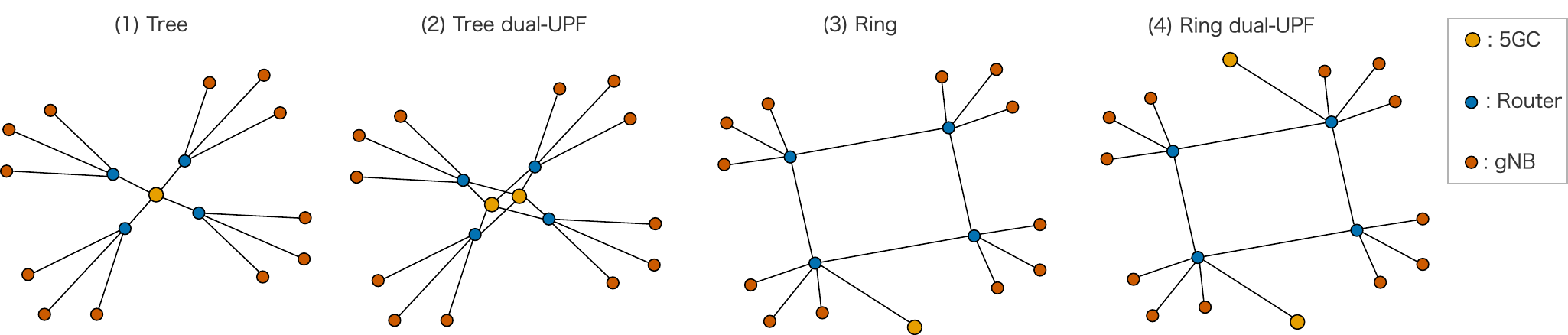}
\caption{Four topology configurations used in evaluation (\(|\mathcal{A}|=12\)): (1) Tree, (2) Tree dual-\ac{upf}, (3) Ring, and (4) Ring dual-\ac{upf}. }
\label{fig:topology_four_variants_gnb12}
\end{figure*}

We use four topologies with 12 \acp{gnb} (Fig.~\ref{fig:topology_four_variants_gnb12}): Tree, Tree dual-\ac{upf}, Ring, and Ring dual-\ac{upf}.

\textbf{General Settings:}
We model a 5G system as a finite-buffer queueing network in which each traversed interface is represented by an M/M/1/\(K\) queue.
The experiments generate only downlink user-plane traffic.
Across all four configurations, the maximum \ac{upf}-to-\ac{gnb} path length is at most \(100\)~km, which satisfies the \(20\) to \(300\)~km assumption range in \cite{itu_t_workshop_transport_distance_han_li}.
For propagation-delay calculations, we assume a propagation speed of \(1.96\times10^{8}\)~m/s.
Each access-link abstraction is configured with maximum capacity
\(C_{\mathrm{AN}}^{\max}=17.83\)~Gbit/s, 128 buffer slots, and zero
propagation distance.
These AN settings are shared across all four topologies.
The radio profile comprises one \ac{ue} per \ac{gnb}, i.e., 12 \acp{ue} in
total. We use \(C_{\mathrm{AN}}^{\max}\) as the nominal service-capacity
ceiling of each downlink queue. Its derivation from the custom wideband radio
profile, relation to the standardized NR bandwidth configurations, and mapping
to the queue service rate are given in
\suppsecref{appendix:access_link_service_capacity}.
Across transport and core-facing links, each directed link is parameterized by distance, bandwidth, and buffer depth, so queueing and propagation are modeled.
Following these settings, we set four topologies as follows:

\textbf{(1) Tree, (2) Tree dual-\ac{upf}:}
Both tree configurations share the same TN/AN structure: four TN routers, each serving three \acp{gnb}, i.e., twelve TN-to-AN links, each configured at \(70\)~km, \(20\)~Gbit/s, and 256 buffer slots.
In the single-\ac{upf} tree, one \ac{upf} connects radially to all four TN routers through four CN-to-TN links, each set to \(30\)~km, \(40\)~Gbit/s, and 512 slots.
In the dual-\ac{upf} tree, two \acp{upf} dual-home every TN router (eight CN-to-TN links in total), with each CN-to-TN link set to \(30\)~km, \(20\)~Gbit/s, and 512 slots.
Hence, total CN-side ingress capacity is \(160\)~Gbit/s in both tree configurations.

\textbf{(3) Ring, (4) Ring dual-\ac{upf}:}
Both configurations use the same access fan-out structure as the tree configurations (four TN routers, each serving three \acp{gnb}), but configure the twelve TN-to-AN links at \(10\)~km, \(20\)~Gbit/s, and 256 slots and connect the TN routers in a four-hop ring (\(30\)~km, \(50\)~Gbit/s, and 256 slots per TN hop).
In the single-\ac{upf} ring, one \ac{upf} attaches to one TN router through one CN-to-TN link (\(30\)~km, \(160\)~Gbit/s, 512 slots).
In the dual-\ac{upf} ring, two \acp{upf} attach to two opposite TN routers through two CN-to-TN links, each configured at \(30\)~km, \(80\)~Gbit/s, and 512 slots.
Thus, total CN-side ingress capacity is again \(160\)~Gbit/s in both configurations.

\subsubsection{Hierarchical Network Slice Management}
\label{subsubsec:hierarchical_ns_management_setup}
We instantiate the hierarchical \ac{ns} management framework as follows.

\textbf{(i) Path Assignment:}
Within the \ac{e2e} controller, the fixed path-assignment rule determines
\(\Gamma_d(s_t)\), the path-induced resource set in domain \(d\), as defined in
Sec.~\ref{subsubsec:per_round_decomposition}.
In the experiments, for each active \ac{gnb} \(a\), the controller selects
uniformly at random one of the minimum-distance paths from any \ac{upf} in
\(\mathcal{U}\) to \(a\).
The selected path defines \(\Phi_{\mathrm{path}}(s_t,a)\).
Controller ownership is fixed by interface type: \ac{cn} manages only core-facing attachment links (from the \ac{upf} to router ingress), \ac{tn} manages intra-transport links and transport-to-access links, and \ac{an} manages radio-side resources at \acp{gnb}.

\textbf{(ii) Domain-Specific Controller Procedure:}
Within each domain-specific controller, the internal procedure \(\Phi_d\)
computes resource demands and evaluates the local feasibility of the delegated
request, as defined in Sec.~\ref{subsubsec:per_round_feasibility}.
In the experiments, we instantiate \(\Phi_d\) using an M/M/1/\(K\) model over
the resources in \(\Gamma_d(s_t)\).
The controller first evaluates \ac{sla} feasibility at the physical upper ratio.
An infeasible upper ratio produces a local rejection; otherwise, the controller
computes the capacity-bounded estimate
\(\widehat\beta^{\mathrm{req}}_{t,d,j}\) of the minimum required ratio.
The complete experimental instantiation of \(\Phi_d\), including the random
allocation overhead and final check against remaining capacity, is given in
\suppsecref{appendix:performance_pipeline_overview}.

\subsubsection{Network Slice Request}
\label{subsubsec:nsr_setup}
We instantiate the \ac{nsr} tuple  \(s\!=\!(A,\boldsymbol{\theta},\mathbf{R},\mathbf{g})\) with its distribution \(\mathbb{P}^{\mathrm{NSR}}_{\mathrm{mix}}\) as follows.

\textbf{(i) \ac{nsr} Instantiation:}
We consider
two request classes, \ac{urllc} and \ac{embb}, which impose contrasting
latency, reliability, and throughput requirements on user-plane traffic.
We exclude \ac{mmtc} because evaluating its support for a very large number of
connected devices~\cite{itu_r_m_2410} would require modeling device
connections, which is outside our scope. For both classes, we set
\(\mathbf{R}\!=\!(R_1,R_2,R_3)\) and
\(\mathbf{g}\!=\!(g_1,g_2,g_3)^\top\).
We map \ac{urllc} and \ac{embb} to 5QI 86 and 5QI 6,
respectively~\cite{3gpp_ts_23_501}.
The complete mapping from each experimental parameter to an element of
\((A,\boldsymbol{\theta},\mathbf{R},\mathbf{g})\), together with the
class-conditional values, is given in
\journalonly{Supplementary }Table~\ref{tab:slice_parameter_basis}.
The decomposition rules are detailed in Sec.~\ref{subsubsec:exp_decomposition_rules}.
The \ac{urllc} packet-size setting reported in
\journalonly{Supplementary }Table~\ref{tab:slice_parameter_basis} follows 3GPP
TR~38.824~\cite{3gpp_tr_38_824}, whereas the \ac{embb} delay and throughput
settings follow 3GPP TS~28.530~\cite{3gpp_ts_28_530} and ITU-R
M.2410~\cite{itu_r_m_2410}.

\textbf{(ii) \ac{nsr} Distribution:}
Let \(\mathbb{P}^{\mathrm{NSR}}_{\mathrm{\acs{urllc}}}\) and \(\mathbb{P}^{\mathrm{NSR}}_{\mathrm{\acs{embb}}}\) denote the class-conditional \ac{nsr} distributions.
For each request, the parameters listed as intervals in
\journalonly{Supplementary }Table~\ref{tab:slice_parameter_basis} are sampled independently and uniformly
from those intervals, whereas the parameters listed as single values are fixed
within each class.
We use
\(\mathbb{P}^{\mathrm{NSR}}_{\mathrm{mix}}
=
\alpha\,\mathbb{P}^{\mathrm{NSR}}_{\mathrm{\acs{embb}}}
+
(1-\alpha)\,\mathbb{P}^{\mathrm{NSR}}_{\mathrm{\acs{urllc}}}\),
where \(\alpha\in[0,1]\) is the \ac{embb} mixing ratio.

The successful-provisioning revenue function is fixed class-wise as \(\kappa_{\mathrm{price}}(s)\!=\!1\) for \ac{urllc} and \(\kappa_{\mathrm{price}}(s)\!=\!50\) for \ac{embb}, reflecting the larger resource footprint of \ac{embb}.

\subsubsection{Instantiation of the Decomposition Space}
\label{subsubsec:exp_decomposition_rules}
Following the target and guarantee constructions in
Sec.~\ref{subsubsec:concrete_admissible_decomposition}, we use the exact
splits defined in Sec.~\ref{subsubsec:exact_split_decomposition_space}.
The components \(R_1\), \(R_2\), and \(R_3\) represent the latency,
throughput, and non-drop requirements, respectively. Only \(R_1\) uses
target-value allocation; \(R_2\) and \(R_3\) use target-condition replication
and therefore introduce no target-decomposition parameters. Thus, \(K_T=1\).
We apply guarantee decomposition to all three components, so
\(K_g=3\). With \(\mathcal D=\{\mathrm{AN},\mathrm{TN},\mathrm{CN}\}\),
\eqref{eq:exact_split_decomposition_space} therefore gives the experimental
decomposition space
\begin{equation}
\mathcal X_T=\Delta^2,\qquad
\mathcal X_g=(\Delta^2)^3,\qquad
\mathcal X=(\Delta^2)^4.
\label{eq:exp_exact_split_space}
\end{equation}


\subsubsection{Runtime}
All experiments were run in a Docker container on an Intel Xeon Platinum 8468 server (2 sockets \(\times\) 48 physical cores, 192 logical CPUs) with 256~GB RAM.
The kernel version was Linux 5.15.0.
The Python runtime was Python 3.12.12, and the \ac{bo}/\ac{gp} stack comprised PyTorch 2.10.0, BoTorch 0.16.1, and GPyTorch 1.15.1~\cite{pytorch_project,botorch_project,gpytorch_project}.

\subsection{Comparison Methods}
\label{subsec:comparison_methods}
We prepare three baselines:
\ac{lincbwk}~\cite{kobayashi2025icccn},
\ac{config}~\cite{xu2023config}, and Random.
They differ in \textit{(i) Surrogate Model:} how they model reward and resource consumption,
\textit{(ii) Search Space:} which decomposition space they search, and
\textit{(iii) Selection Algorithm:} how they select decompositions, as summarized in Table~\ref{tab:comparison_kernel_search_space}.
We present the settings shared across methods (Sec.~\ref{subsubsec:comparison_common_settings}), 
formalize the surrogate-model and decomposition search-space options (Sec.~\ref{subsubsec:comparison_design_components}), 
and provide the detailed method-specific instantiations (Sec.~\ref{subsubsec:comparison_methods_detail}).

\begin{table}[t]
\caption{Overview of Compared Methods}
\label{tab:comparison_kernel_search_space}
\centering
\footnotesize
\setlength{\tabcolsep}{3pt}
\renewcommand{\arraystretch}{1.1}
\begin{tabularx}{\columnwidth}{@{}l>{\raggedright\arraybackslash}X>{\raggedright\arraybackslash}p{0.2\columnwidth}>{\raggedright\arraybackslash}p{0.3\columnwidth}@{}}
\toprule
\begin{tabular}[c]{@{}l@{}}\textbf{Method}\end{tabular}
&
\begin{tabular}[c]{@{}c@{\;}l@{}}
\multirow{1}{*}{\textbf{(i)}} & \textbf{Surrogate} \\
& \textbf{Model}
\end{tabular}
&
\begin{tabular}[c]{@{}c@{\;}l@{}}
\multirow{1}{*}{\textbf{(ii)}} & \textbf{Search} \\
& \textbf{Space}
\end{tabular}
&
\begin{tabular}[c]{@{}c@{\;}l@{}}
\multirow{1}{*}{\textbf{(iii)}} & \textbf{Selection} \\
& \textbf{Algorithm}
\end{tabular}
\\
\midrule
CCKB (Ours) & Nonlinear \ac{gp} & Continuous & Primal--dual (Alg.~\ref{alg:component1}) \\
\ac{lincbwk}~\cite{kobayashi2025icccn} & Linear \ac{gp} & Discrete & Primal--dual (Alg.~\ref{alg:component1}) \\
\ac{config}~\cite{xu2023config} & Nonlinear \ac{gp} & Continuous & \ac{config} \\
Random & None & Continuous & Uniform random \\
\bottomrule
\end{tabularx}
\end{table}


\subsubsection{Common Settings}
\label{subsubsec:comparison_common_settings}
To ensure a fair policy comparison, we use common settings for surrogate
inputs and observations, hyperparameters, and the learning protocol where
applicable.

\textbf{(i) Surrogate Inputs and Observations:}
All surrogate-based methods use the same input, output processing,
and observation-noise models, as detailed in
\suppsecref{appendix:comparison_shared_components}.
These experimental fitting choices are distinct from the conditional
sub-Gaussian assumptions in Sec.~\ref{subsec:success_only_feedback}.

\textbf{(ii) Hyperparameters:}
For implementation, we replace the confidence widths
in \eqref{eq:concrete_reward_exploration_width} and
\eqref{eq:concrete_proxy_demand_exploration_width} with the common
tunable schedule \(\beta_t=c_{\beta}\sqrt{\log(t+2)}\), where
\(c_{\beta}>0\).  This schedule grows as
\(\Theta(\sqrt{\log t})\).  Although this simplification enables a common
implementation across the \ac{gp} models, it does not satisfy the
exploration-width conditions in
\eqref{eq:concrete_reward_exploration_width} and
\eqref{eq:concrete_proxy_demand_exploration_width}, and thus the corresponding
theoretical guarantee cannot be invoked for the experimental implementation.  Similar
practical modifications have been adopted in prior empirical
work~\cite{pmlr-v37-kandasamy15}.
Unless otherwise specified in each experimental setup, we set
\(c_{\beta}=0.1\).
For \ac{cckb} and \ac{lincbwk}, which use primal--dual updates, we set the dual-cap parameter to \(\rho=1.0\).

\textbf{(iii) Learning Protocol:}
Unless otherwise specified, each run consists of \(T=400\) rounds.
For \ac{gp}-based policies, each run begins with 50 rounds of uniformly random exploration.
Thereafter, the surrogate models are refitted every five rounds.
For each \ac{gp}, parameter learning is attempted up to 20 times.
A fit is considered unsuccessful if none of these attempts produces a numerically valid fit, either because parameter learning does not converge or because the covariance matrix cannot be stably factorized.
If the reward \ac{gp} or any constraint \ac{gp} cannot be fitted successfully at a round when surrogate refitting is scheduled, the event is counted once as a surrogate-refitting failure.
For each failed model, the \ac{e2e} controller reuses its most recent valid reward or constraint estimate, when available.
If any failed surrogate has no valid previous fit, the controller performs random exploration.
We report the frequency of such failures using the \textit{Surrogate-Refitting Failure Rate} defined in Sec.~\ref{subsec:evaluation_metrics}.

\subsubsection{Surrogate Models and Decomposition Search Spaces}
\label{subsubsec:comparison_design_components}
Of the three design dimensions summarized in
Table~\ref{tab:comparison_kernel_search_space}, this subsection formalizes the
surrogate-model and decomposition-search-space options used by the compared
methods.  

\textbf{(i) Surrogate-Model Options:}
  \ifarxivversion
    The two choices for the surrogate-model dimension are a nonlinear \ac{gp},
and a linear \ac{gp}. Following the notation in
Sec.~\ref{subsec:gp_surrogates_and_kernel_design}, we instantiate their kernels
on the feature vector \(\tilde{\mathbf{u}}(z)\) as follows.

\noindent
\textbf{(nonlinear):}
The Mat\'ern-\(5/2\) kernel with \ac{ard} is defined by
\begin{align}
k^{\bullet,\mathrm{Mat}}(z,z^{\prime})
&:=
\sigma_{k,\bullet}^{2}
\left(1+\sqrt{5}\,r_{\bullet}+\frac{5}{3}r_{\bullet}^{2}\right)e^{-\sqrt{5}r_{\bullet}},
\label{eq:comparison_kernel_matern_symbol}
\\
&\text{where}\quad
r_{\bullet}^{2}
:=
\sum_{j=1}^{D}
\frac{
\left(
\tilde u_{j}(z)-\tilde u_{j}(z^{\prime})
\right)^2
}{\ell_{\bullet,j}^{2}}.
\label{eq:comparison_kernel_matern_distance}
\end{align}

\noindent
\textbf{(linear):}
The linear kernel is defined by
\begin{equation}
k^{\bullet,\mathrm{Lin}}(z,z^{\prime})
:=
\sigma_{k,\bullet}^{2}\,\tilde{\mathbf{u}}(z)^{\top}\tilde{\mathbf{u}}(z^{\prime}).
\label{eq:comparison_kernel_linear_symbol}
\end{equation}
Here, \(\sigma_{k,\bullet}^{2}\) denotes the kernel output scale for target \(\bullet\).
The Mat\'ern-\(5/2\) form provides a standard nonlinear \ac{gp} prior with moderate smoothness on \(\tilde{\mathbf{u}}(z)\)~\cite{matern1986spatial}, and \ac{ard} assigns dimension-specific length-scales \(\{\ell_{\bullet,j}\}_{j=1}^{D}\) to capture anisotropy across feature dimensions.
  \fi

  \ifjournalversion
    \input{src/sections/variants/sec_evaluation_kernel_options_journal}%
  \fi

\textbf{(ii) Search-space Options:}
Recall from \eqref{eq:exp_exact_split_space} that the experimental
decomposition space is \(\mathcal{X}=(\Delta^2)^4\).
For numerical stability, all methods restrict \(\mathcal{X}\) by imposing the
coordinate-wise lower bound \(w_{\mathrm{lb}}=0.05\) on every simplex field.
Continuous-search methods use this restricted region directly, whereas
discrete-search methods use its grid-discretized subset.

\noindent
\textbf{(continuous):}
For a single simplex field, define the lower-bounded continuous region as
\(\mathcal{W}^{\mathrm{cont}}(w_{\mathrm{lb}})
:=\{\mathbf{w}\in\Delta^2 \mid w_d\ge w_{\mathrm{lb}},\ \forall d\in\mathcal{D}\}\).
The corresponding search space is
\begin{equation}
\mathcal{X}_{\mathrm{cont}}(w_{\mathrm{lb}})
:=
\left(\mathcal{W}^{\mathrm{cont}}(w_{\mathrm{lb}})\right)^4.
\label{eq:comparison_search_space_continuous}
\end{equation}

\noindent
\textbf{(discrete):}
For a single simplex field, let
\(\mathcal{W}^{\mathrm{disc}}_{\Delta_{\mathrm{grid}}}(w_{\mathrm{lb}})
\subseteq\mathcal{W}^{\mathrm{cont}}(w_{\mathrm{lb}})\) denote the
grid-discretized subset of the continuous region.
Here, \(\Delta_{\mathrm{grid}}>0\) is the simplex grid step size, i.e., \(\mathcal{W}^{\mathrm{disc}}_{\Delta_{\mathrm{grid}}}(w_{\mathrm{lb}})\) keeps feasible points whose coordinates lie on the \(\Delta_{\mathrm{grid}}\)-spaced lattice, so smaller \(\Delta_{\mathrm{grid}}\) yields a finer discretization and a larger candidate set.
The corresponding four-field search space is
\begin{equation}
\mathcal{X}_{\Delta_{\mathrm{grid}}}(w_{\mathrm{lb}})
:=
\left(\mathcal{W}^{\mathrm{disc}}_{\Delta_{\mathrm{grid}}}(w_{\mathrm{lb}})\right)^4.
\label{eq:comparison_search_space_discrete}
\end{equation}
Thus, \(\mathcal{X}_{\Delta_{\mathrm{grid}}}(w_{\mathrm{lb}})\) is the
grid-discretized subset of \(\mathcal{X}_{\mathrm{cont}}(w_{\mathrm{lb}})\).
Fig.~\ref{fig:comparison_search_space_simplex} visualizes this construction
for a single simplex field.

\begin{figure}[t]
\centering
\resizebox{\columnwidth}{!}{%
\begin{tikzpicture}[font=\scriptsize]
  \def\L{4.0}
  \def\H{3.464}
  \def\B{5}

  \begin{scope}[xshift=0cm]
    \node[font=\normalsize, align=center] at (2,-1.00) {(1) Continuous Region $\mathcal{W}^{\mathrm{cont}}(w_{\mathrm{lb}})$};

    \begin{scope}[shift={(2,0)},scale=0.75,shift={(-2,0)}]

    \fill[black!15] (0.3,0.1732) -- (3.7,0.1732) -- (2.0,3.1176) -- cycle;
    \draw[thick] (0,0) -- (\L,0) -- (\L/2,\H) -- cycle;
    \draw[densely dashed] (0.3,0.1732) -- (3.7,0.1732) -- (2.0,3.1176) -- cycle;

    \node[below, font=\normalsize] at (0,0) {AN};
    \node[below, font=\normalsize] at (\L,0) {TN};
    \node[above, font=\normalsize] at (\L/2,\H) {CN};
    \node[align=center, scale=1.6, transform shape] at (2.0,1.15) {$w_d\ge w_{\mathrm{lb}}$};
    \coordinate (tncnOuterMid) at ($(\L,0)!0.50!(\L/2,\H)$);
    \coordinate (tncnInnerMid) at ($(3.7,0.1732)!0.50!(2.0,3.1176)$);
    \coordinate (tncnInsetExt) at ($(tncnOuterMid)!1.5!(tncnInnerMid)$);
    \draw[thick] (tncnOuterMid) -- (tncnInsetExt);
    \node[right, scale=1.6, transform shape] (wlblabel) at (3.62,2.35) {$w_{\mathrm{lb}}$};
    \draw[thick,->] (wlblabel.west) .. controls (3.35,2.25) and (3.10,1.98) .. ($(tncnOuterMid)!0.55!(tncnInsetExt)$);
    \end{scope}
  \end{scope}

  \begin{scope}[xshift=5.6cm]
    \node[font=\normalsize, align=center] at (2,-1.00) {(2) Discrete Region $\mathcal{W}^{\mathrm{disc}}_{0.2}(w_{\mathrm{lb}})$};

    \begin{scope}[shift={(2,0)},scale=0.75,shift={(-2,0)}]

    \draw[thick] (0,0) -- (\L,0) -- (\L/2,\H) -- cycle;
    \node[below, font=\normalsize] at (0,0) {AN};
    \node[below, font=\normalsize] at (\L,0) {TN};
    \node[above, font=\normalsize] at (\L/2,\H) {CN};
    \foreach \i in {1,2,3,4} {
      \pgfmathsetmacro{\xbase}{\L*\i/\B}
      \pgfmathsetmacro{\tickval}{\i/\B}
      \draw[line width=0.6pt] (\xbase,0) -- (\xbase,-0.10);

      \pgfmathsetmacro{\xleft}{(\L/2)*\i/\B}
      \pgfmathsetmacro{\yleft}{\H*\i/\B}
      \draw[line width=0.6pt] (\xleft,\yleft) -- ++(-0.08,0.05);

      \pgfmathsetmacro{\xright}{\L-(\L/2)*\i/\B}
      \draw[line width=0.6pt] (\xright,\yleft) -- ++(0.08,0.05);
    }

    \foreach \a in {0,1,2,3,4,5} {
      \foreach \b in {0,1,2,3,4,5} {
        \pgfmathtruncatemacro{\c}{\B-\a-\b}
        \ifnum\c>-1
          \pgfmathsetmacro{\x}{\L*\b/\B + (\L/2)*\c/\B}
          \pgfmathsetmacro{\y}{\H*\c/\B}
          \fill[black!30] (\x,\y) circle (0.045);
        \fi
      }
    }

    \foreach \a/\b/\c in {1/1/3,1/2/2,1/3/1,2/1/2,2/2/1,3/1/1} {
      \pgfmathsetmacro{\x}{\L*\b/\B + (\L/2)*\c/\B}
      \pgfmathsetmacro{\y}{\H*\c/\B}
      \filldraw[black] (\x,\y) circle (0.08);
    }
    \pgfmathsetmacro{\xex}{\L*1/\B + (\L/2)*3/\B}
    \pgfmathsetmacro{\yex}{\H*3/\B}
    \node[anchor=west, align=left, scale=1.6, transform shape] (discExampleLabel) at (2.95,2.62)
      {$(w_{\mathrm{AN}},w_{\mathrm{TN}},w_{\mathrm{CN}})$\\$=(0.2,0.2,0.6)$};
    \draw[thick,->] (\xex,\yex) .. controls (2.22,2.24) and (2.58,2.50) .. (discExampleLabel.west);
    \end{scope}
  \end{scope}
\end{tikzpicture}}
\caption{Continuous and discrete candidate regions for a single simplex field. (1) Continuous: \(\mathcal{W}^{\mathrm{cont}}(w_{\mathrm{lb}})=\{\mathbf{w}\in\Delta^2 \mid w_d\ge w_{\mathrm{lb}},\ \forall d\}\). (2) Discrete: \(\mathcal{W}^{\mathrm{disc}}_{\Delta_{\mathrm{grid}}}(w_{\mathrm{lb}})\) for \(\Delta_{\mathrm{grid}}=0.2\); filled black points are retained candidates.}
\label{fig:comparison_search_space_simplex}
\end{figure}
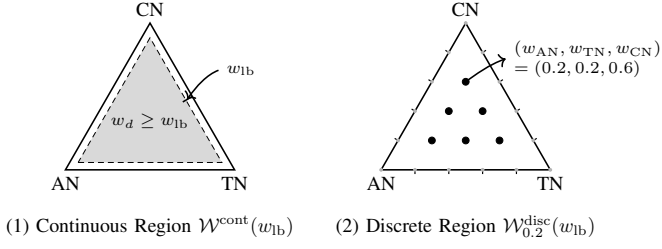

\subsubsection{Method-Specific Instantiations}
\label{subsubsec:comparison_methods_detail}
We now specify the implementation of each method.

\textbf{(i) CCKB (Ours):}
For the reward target \(f\) and every constraint target
\(m_{d,j}^{\Gamma}\), \ac{cckb} uses the nonlinear Mat\'ern-\(5/2\)
kernel with \ac{ard}
\arxivonly{defined in \eqref{eq:comparison_kernel_matern_symbol}--\eqref{eq:comparison_kernel_matern_distance}}
\journalonly{specified in \suppsecref{appendix:comparison_kernel_definitions}},
i.e., \(k^{f}=k^{f,\mathrm{Mat}}\) and
\(k^{m_{d,j}^{\Gamma}}=k^{m_{d,j}^{\Gamma},\mathrm{Mat}}\) for all \((d,j)\).
The primal step optimizes $\mathcal{A}_t(\mathbf{x}\mid s_t)$ over \(\mathcal{X}_{\mathrm{cont}}(w_{\mathrm{lb}})\).

\textbf{(ii) \ac{lincbwk}~\cite{kobayashi2025icccn}:}
For the reward target \(f\) and every constraint target
\(m_{d,j}^{\Gamma}\), \ac{lincbwk} uses the linear-kernel \ac{gp}
\arxivonly{defined in \eqref{eq:comparison_kernel_linear_symbol}}
\journalonly{specified in \suppsecref{appendix:comparison_kernel_definitions}},
i.e., \(k^{f}=k^{f,\mathrm{Lin}}\) and
\(k^{m_{d,j}^{\Gamma}}=k^{m_{d,j}^{\Gamma},\mathrm{Lin}}\) for all \((d,j)\).
With the linear kernel, each \ac{gp} posterior mean is linear in
\arxivonly{the common feature vector \(\tilde{\mathbf{u}}(z)\)}
\journalonly{its standardized input}, as in Bayesian linear
regression~\cite{rasmussen2006gpml}, thereby recovering
\textit{(L2) linear realizability} assumed in our previous
method~\cite{kobayashi2025icccn}.
Its primal step evaluates \(\mathcal{A}_t(\mathbf{x}\mid s_t)\) for every
\(\mathbf{x}\in\mathcal{X}_{\Delta_{\mathrm{grid}}}(w_{\mathrm{lb}})\) and
selects a maximizer.
By the four-field construction above, the number of candidates is
\(
\left|\mathcal{X}_{\Delta_{\mathrm{grid}}}(w_{\mathrm{lb}})\right|
=
\left|\mathcal{W}^{\mathrm{disc}}_{\Delta_{\mathrm{grid}}}(w_{\mathrm{lb}})\right|^4.
\)
We use \(\Delta_{\mathrm{grid}}=0.2\) to keep exhaustive evaluation tractable,
yielding \(1296\) candidate decompositions.

\textbf{(iii) \ac{config}~\cite{xu2023config}:}
Odin~\cite{odin} minimizes a sum of unknown domain-level cost functions
subject to an \ac{e2e} \ac{sla} constraint, whereas our formulation maximizes
provisioning reward subject to cumulative resource-budget constraints.
A direct implementation of Odin would therefore require changing our objective
and constraint model.
Instead, we use \ac{config}, the non-contextual constrained-\ac{bo} method on which
Odin's design is based, as the decomposition-selection algorithm under our
formulation.
\ac{config} uses the same \ac{gp} kernels as \ac{cckb}.
To retain \ac{config}'s non-contextual structure, we fix the \ac{nsr} component of
every \ac{gp} input to the first-round request, i.e.,
\(z_t^{\mathrm{cfg}}=(s_1,\mathbf{x}_t)\), rather than
\((s_t,\mathbf{x}_t)\).
Consequently, \ac{config} selects \(\mathbf{x}_t\) without conditioning on the
current request \(s_t\).
Finally, \ac{config} searches over \(\mathcal{X}_{\mathrm{cont}}(w_{\mathrm{lb}})\).

\textbf{(iv) Random:}
Each round samples a feasible decomposition from
\(\mathcal{X}_{\mathrm{cont}}(w_{\mathrm{lb}})\).

\subsection{Evaluation Metrics}
\label{subsec:evaluation_metrics}
We use \textit{(i) Total Reward} as the primary performance metric.
We additionally use \textit{(ii) Resource Usage Ratio} to examine how each
policy distributes resource consumption across domains.
For surrogate-based policies, we report \textit{(iii) Surrogate-Refitting
Failure Rate} to assess the numerical reliability of methods.
Under mixed traffic, we report both Total Reward and Successful-Slice
Count, together with a per-class breakdown; Sec.~\ref{subsec:traffic_mixture_tree}
explains why both aggregate metrics are needed.

For each policy under each experimental setting, we perform multiple
\(T\)-round runs with different random seeds.
We report Total Reward and Resource Usage Ratio as the mean and standard
deviation across runs.
In contrast, each reported Surrogate-Refitting Failure Rate pools the failed
and scheduled refitting operations over all runs covered by that result.

\textbf{(i) Total Reward:}
We define the total reward as
\[
R:=\sum_{t=1}^{T}\kappa_{\mathrm{price}}(s_t)\,y_t.
\]

\textbf{(ii) Resource Usage Ratio:}
Let \(C_{T,d,j}^{\mathrm{avail}}\) denote the remaining capacity
of resource \((d,j)\) at the end of one run.
The usage ratio is defined as
\[
u_d:=\frac{1}{J_d}\sum_{j=1}^{J_d}
\frac{C_{d,j}^{\max}-C_{T,d,j}^{\mathrm{avail}}}{C_{d,j}^{\max}}.
\]

\textbf{(iii) Surrogate-Refitting Failure Rate:}
We first define the
number of surrogate-refitting failures in one run as
\[
C_{\mathrm{fit}}
:=
\sum_{t=1}^{T}
\mathbf{1}\!\left[
\text{surrogate refitting fails at round }t
\right].
\]
A scheduled refitting operation is counted once as a failure if the reward
\ac{gp} or any constraint \ac{gp} exhausts its fitting attempts without a valid
fit.
This metric is not applicable to Random, which does not fit surrogate models.

\subsection{Overall Effectiveness Across Topologies and Domain Bottlenecks}
\label{subsec:result}
\subsubsection{Setup}
\label{subsubsec:topology_bottleneck_setup}

We consider only the \ac{urllc} class, corresponding to \(\alpha=0\) in the mixed-\ac{nsr} model \(\mathbb{P}^{\mathrm{NSR}}_{\mathrm{mix}}\).

We evaluate all four topology configurations in Fig.~\ref{fig:topology_four_variants_gnb12} under three domain-bottleneck profiles, each of which reduces the resource availability of the AN, TN, or CN.
These combinations yield 12 experimental conditions, each of which is evaluated over 10 runs with different random seeds.

The bottleneck profiles are implemented using the resource-availability vector
\((a_{\mathrm{AN}},a_{\mathrm{TN}},a_{\mathrm{CN}})\), where the capacity of resource \(j\in[J_d]\) in each domain
\(d\in\mathcal{D}\) is scaled to \(a_d C_{d,j}^{\max}\).
We use \((0.2,1.0,1.0)\), \((1.0,0.2,1.0)\), and \((1.0,1.0,0.2)\) for the AN, TN, and CN bottleneck profiles, respectively.

\subsubsection{Result}
\label{subsubsec:topology_bottleneck_result}

The comparison yields the three observations.

\begin{figure*}[t]
\centering
\includegraphics[width=0.83\textwidth]{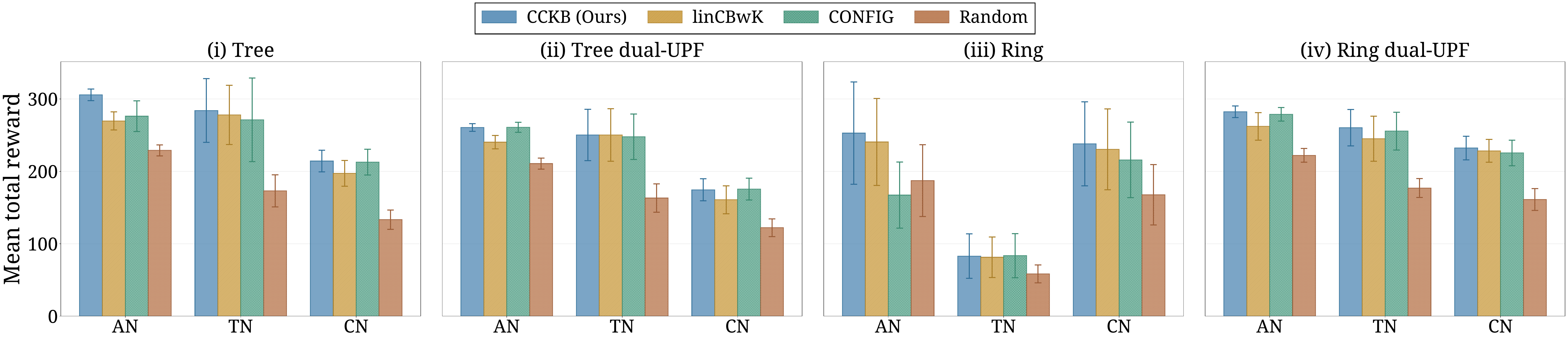}
\caption{Mean Total Reward under the three domain-bottleneck profiles for
\ac{cckb}, \ac{lincbwk}, \ac{config}, and Random:
(i) Tree, (ii) Tree dual-\ac{upf}, (iii) Ring, and
(iv) Ring dual-\ac{upf}. Each subplot corresponds to one topology,
and each AN, TN, or CN label denotes the bottleneck domain \(d\), for which
\(a_d=0.2\), while \(a_{d'}=1.0\) for every other domain \(d'\ne d\).
Bars and whiskers show the mean and standard
deviation over 10 runs, respectively. Higher is better.}
\label{fig:topology_bottleneck_reward_by_condition}
\end{figure*}

\begin{figure*}[t]
\centering
\includegraphics[width=0.86\textwidth]{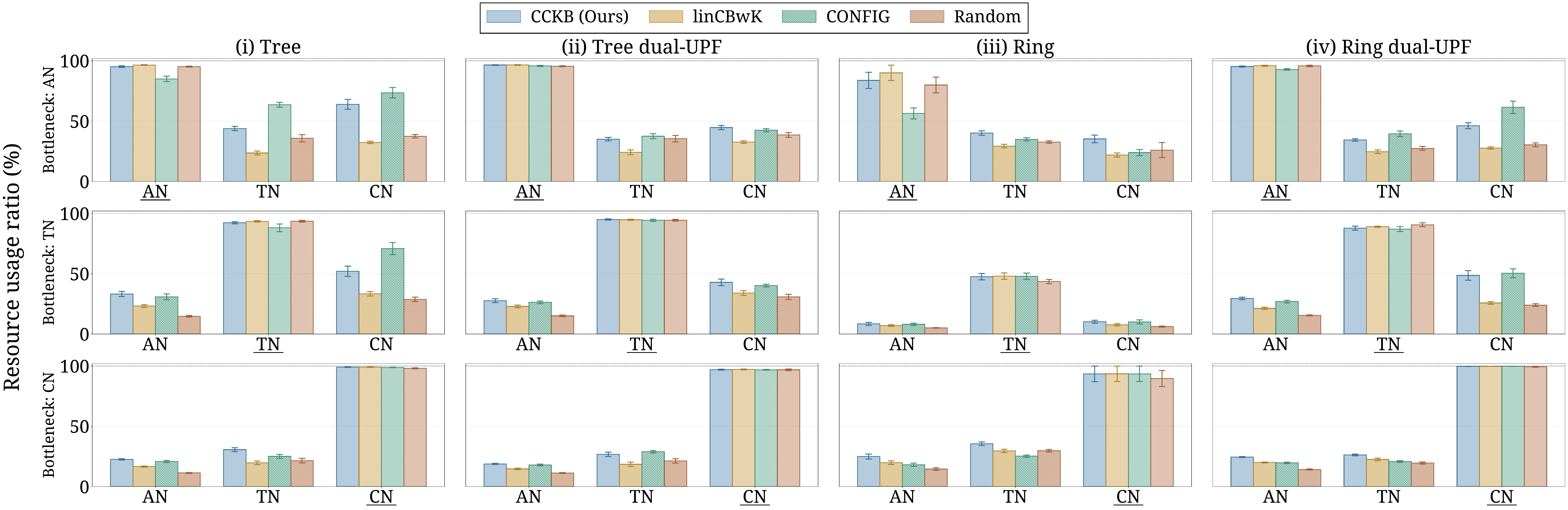}
\caption{Resource Usage Ratio at the end of each run for CCKB, linCBwK, \ac{config}, and Random. Bars show the mean, and whiskers show the standard deviation. Underlined domain labels (e.g., \(\underline{\mathrm{AN}}\)) denote the bottleneck; lower values on that domain and higher values on non-bottleneck domains are preferable.}
\label{fig:topology_bottleneck_urllc_main_usage_bar_whisker_grid}
\end{figure*}

\textbf{Total Reward:}
Across the 12 topology--bottleneck conditions, CCKB achieves higher mean Total Reward than linCBwK in every condition (Fig.~\ref{fig:topology_bottleneck_reward_by_condition}).
This consistent improvement supports the practical effectiveness of CCKB's two extensions over linCBwK that directly relax \textit{(L1) finite/discrete decomposition search} and \textit{(L2) linear realizability}.
Compared with \ac{config}, CCKB achieves higher mean Total Reward in 9 of the 12 conditions.
The exceptions are Tree dual-\ac{upf} with an AN bottleneck, Tree dual-\ac{upf} with a CN bottleneck, and Ring with a TN bottleneck, for which the reward gaps are visually small in Fig.~\ref{fig:topology_bottleneck_reward_by_condition}.
This comparison indicates that jointly addressing the two NSR-DP requirements, \textit{(R1) long-horizon multi-resource budget control} and \textit{(R2) request-conditioned decomposition}, is practically important.

\textbf{Resource Usage Ratio:}
Fig.~\ref{fig:topology_bottleneck_urllc_main_usage_bar_whisker_grid} provides a complementary view of resource allocation.
The bottleneck domain exhibits the highest usage ratio and is typically close to saturation across methods and conditions.
Overall, \ac{cckb}, \ac{lincbwk}, and \ac{config} tend to show higher usage ratios in non-bottleneck domains than Random.

For \ac{config}, however, higher non-bottleneck usage does not always translate into higher mean Total Reward.
For example, under the Tree topology with a TN bottleneck, \ac{config} places heavier load on non-bottleneck domains than \ac{cckb} but obtains lower mean Total Reward.
Our failure-case inspection shows that \ac{config} sometimes assigns overly strict domain latency targets \(\delta_{t,d}\) to selected non-bottleneck domains.
When \(\delta_{t,d}\) is smaller than the propagation delay required in the corresponding domain, slice construction fails.
A plausible explanation for selecting such targets is degraded reward-model fidelity under \ac{config}'s context-agnostic \ac{gp} evaluation, which may overestimate the benefit of strict decompositions and misjudge their provisioning success.

\textbf{Surrogate-Refitting Failure Rate:}
Table~\ref{tab:topology_bottleneck_refit_failure_rate_by_method} summarizes the surrogate-refitting failure rate.
CCKB and \ac{config} complete all 8{,}400 scheduled refitting operations without
failure, whereas linCBwK fails in 407 of 8{,}400 operations (4.85\%); the metric
is not applicable to Random because it does not fit surrogate models.
This failure rate directly indicates lower numerical fitting stability for
linCBwK in this setting and motivates the follow-up \ac{nlpd} analysis of
surrogate fidelity in Sec.~\ref{subsubsec:topology_bottleneck_followup}.

\begin{table}[t]
\caption{Surrogate-Refitting Failure Rate by method, aggregated over all topology--bottleneck conditions.}
\label{tab:topology_bottleneck_refit_failure_rate_by_method}
\centering
\small
\setlength{\tabcolsep}{5pt}
\resizebox{\columnwidth}{!}{%
\begin{tabular}{
    c
    c
    c
    c
}
\toprule
\multicolumn{1}{c}{\textbf{CCKB}} & \multicolumn{1}{c}{\textbf{linCBwK}} & \multicolumn{1}{c}{\textbf{\ac{config}}} & \multicolumn{1}{c}{\textbf{Random}} \\
\midrule
0/8{,}400 (0.00\%) & 407/8{,}400 (4.85\%) & 0/8{,}400 (0.00\%) & N/A \\
\bottomrule
\end{tabular}
}
\end{table}

\subsubsection{Follow-up Analysis}
\label{subsubsec:topology_bottleneck_followup}
To further examine the surrogate fidelity suggested by the resource-usage and
refitting-stability results, we use the Tree topology with an AN bottleneck as a
focused diagnostic condition and compare the constraint-\ac{gp} \ac{nlpd}
trajectories.
Here, \ac{nlpd} is used as a surrogate-fit metric; lower values indicate better fit.
Fig.~\ref{fig:topology_bottleneck_constraint_gp_metric_history_tree_an} shows a consistent ordering: \ac{cckb} keeps the lowest \ac{nlpd}, \ac{config} is lower than \ac{lincbwk} but still above \ac{cckb}, and \ac{lincbwk} remains highest with the broadest dispersion.
Thus, the contextual nonlinear constraint-\ac{gp} in \ac{cckb} provides the best predictive fit among the compared surrogates in this condition, while linCBwK's poorer fit is consistent with the higher surrogate-refitting failure rate in Table~\ref{tab:topology_bottleneck_refit_failure_rate_by_method}.

\begin{figure}[t]
\centering
\includegraphics[width=0.85\linewidth]{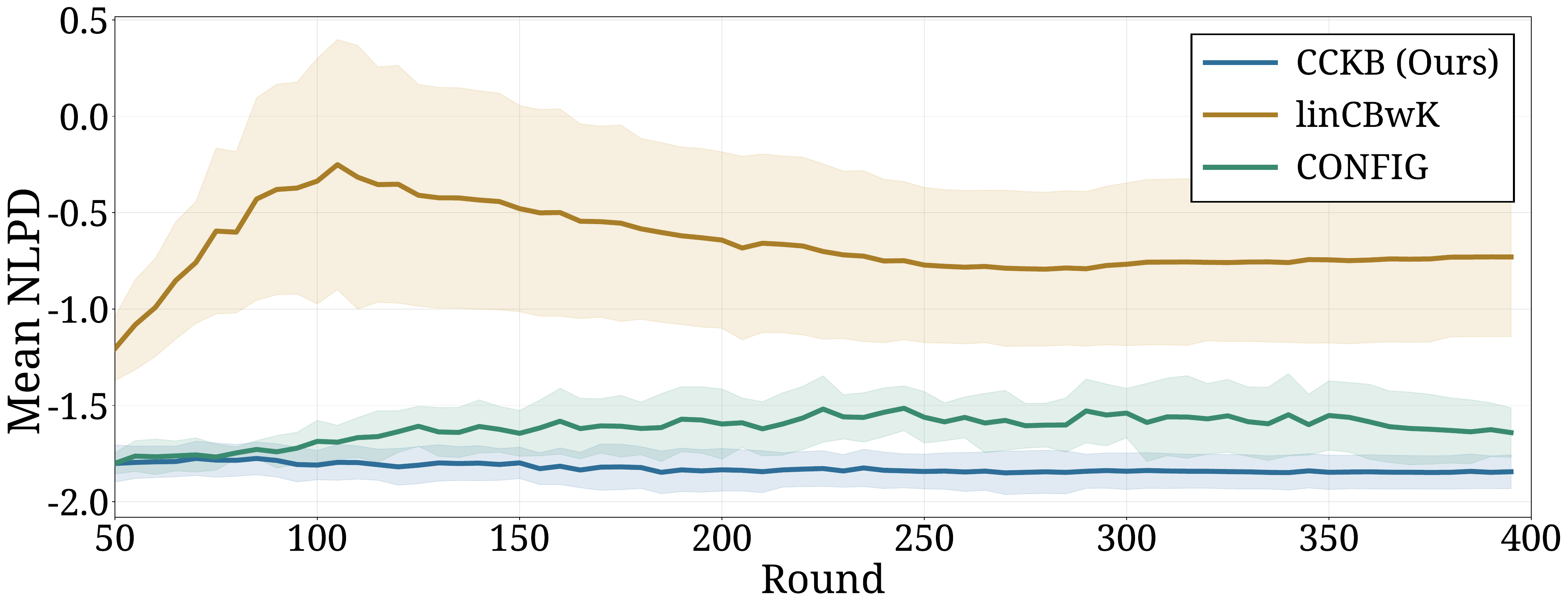}
\caption{Constraint-\ac{gp} \ac{nlpd} trajectories under the Tree topology with an AN bottleneck. Lines show the means over 10 runs, and shaded bands show cross-run variability. Lower values indicate better predictive quality.}
\label{fig:topology_bottleneck_constraint_gp_metric_history_tree_an}
\end{figure}

\subsection{Robustness to Heterogeneous Requests}
\label{subsec:traffic_mixture_tree}

\subsubsection{Setup}
\label{subsubsec:traffic_mixture_tree_setup}
We consider the Tree and Ring topologies under all three bottleneck profiles,
yielding six topology--bottleneck conditions.
For each condition, we vary the \ac{embb} mixing ratio over
\(
\alpha \in \{0,0.02,0.04,0.06,0.08\}.
\)
Each condition is evaluated over 10 runs with
different random seeds.

\textbf{Mixture Range:}
We determine the sweep range from the aggregate mean offered traffic rate of
each request class before provisioning and resource-allocation decisions,
rather than from the request counts alone.
For each class, this rate is the product of the number of covered \acp{gnb},
the packet-arrival rate per \ac{gnb}, and the mean packet size.
Let \(L_{\mathrm{urllc}}\) and \(L_{\mathrm{embb}}\) denote the corresponding
class-specific rates.
Under the settings in \journalonly{Supplementary }Table~\ref{tab:slice_parameter_basis}, a \ac{urllc}
request covers three \acp{gnb} and generates
\(L_{\mathrm{urllc}}=3\times20{,}000\times1{,}600
=96\)~Mbit/s, whereas an \ac{embb} request covers nine \acp{gnb} and
generates
\(L_{\mathrm{embb}}=9\times15{,}000\times9{,}600
=1{,}296\)~Mbit/s.
Thus, one \ac{embb} request generates \(13.5\) times as much offered traffic
as one \ac{urllc} request.
Consequently, the request mixing ratio alone does not reflect how strongly the
introduced \ac{embb} requests affect the aggregate offered traffic.
The two classes contribute equally to the mean offered traffic at
\(\alpha^\star=L_{\mathrm{urllc}}/
(L_{\mathrm{embb}}+L_{\mathrm{urllc}})\approx 0.069\).
Because \(\alpha=0.08\) is the first tested ratio above this balanced point,
the sweep captures the intended transition from \ac{urllc}-only traffic to a
mixed workload in which the two classes make comparable contributions to the
aggregate mean offered traffic.

\textbf{Evaluation Metric:}
Under mixed traffic, we report two metrics.  Total Reward is the optimization
objective itself.  However, \ac{embb} requests are few in number, and each is
priced 50 times higher than a \ac{urllc} request.  Total Reward is therefore
largely determined by the small number of admitted \ac{embb} slices.  This
number varies considerably across seeds, and the differences between methods
are comparable to or smaller than that variation.  Total Reward therefore has
limited resolution for comparing methods.  We use Successful-Slice Count as
the primary metric because it is less sensitive to this variation, and report
Total Reward and the per-class breakdown in
\suppsecref{appendix:mixed_traffic_reward_class_results}.

\textbf{Class-Specific Exploration Estimates:}
To avoid pooling observations from request classes with substantially different
traffic loads and \ac{sla} parameters, we maintain separate reward and constraint
exploration estimates for \ac{urllc} and \ac{embb} requests.
Let \(c_t\in\mathcal{C}:=\{\mathrm{\acs{urllc}},\mathrm{\acs{embb}}\}\)
denote the class of the round-\(t\) \ac{nsr}.
For each class \(k\in\mathcal{C}\), let \(\hat f_{t,k}\) and
\(\{\hat h_{t,d,j,k}\}_{d,j}\) denote its raw reward and constraint exploration estimates.
The active estimates at round \(t\) are constructed as
\begin{align}
\hat f_t(s_t,\mathbf{x})
&:=
\sum_{k\in\mathcal{C}}
\mathbf{1}\{c_t=k\}\,\hat f_{t,k}(s_t,\mathbf{x}),
\label{eq:traffic_mixture_reward_surrogate}\\
\hat h_{t,d,j}(s_t,\mathbf{x})
&:=
\sum_{k\in\mathcal{C}}
\mathbf{1}\{c_t=k\}\,\hat h_{t,d,j,k}(s_t,\mathbf{x}).
\label{eq:traffic_mixture_constraint_surrogate}
\end{align}
By \eqref{eq:traffic_mixture_reward_surrogate} and
\eqref{eq:traffic_mixture_constraint_surrogate}, only the estimates matching
the arriving request class are evaluated and updated, while the primal--dual
control loop remains unified.\footnote{The class-specific construction does not
share observations across request classes. Consequently, each model is
trained on fewer observations than a shared cross-class model; developing
a model that exploits similarities between the two classes is left for future
work.}
We apply the same class-specific split to \ac{cckb}, \ac{lincbwk}, and \ac{config},
so their performance differences are not attributable to class separation.

\subsubsection{Result}
\label{subsubsec:traffic_mixture_tree_result}
\ac{cckb} attains the highest mean in five, six, five, four, and five of the six
conditions at \(\alpha=0\), \(0.02\), \(0.04\), \(0.06\), and \(0.08\),
respectively (Fig.~\ref{fig:mixture_success_slices_vs_embb_ratio_tree_ring}).
Equivalently, it ranks first in 25 of the 30 evaluated points: 14 of 15 in the
Tree topology and 11 of 15 in the Ring topology.
Excluding the \ac{urllc}-only baseline, \ac{cckb} remains best in 20 of the 24
heterogeneous-traffic points, showing that its mean-count advantage persists
throughout the sweep toward bandwidth-balanced traffic rather than being
confined to the single-class setting.
The five exceptions are concentrated in the \ac{tn} bottleneck (Tree at
\(\alpha=0.04\) and Ring at \(\alpha\in\{0,0.06,0.08\}\)) and Ring--\ac{cn} at
\(\alpha=0.06\).
Even at these points, the largest absolute deficit from \ac{cckb} to the best
competing mean is 3.4 successfully provisioned slices, and the largest relative
deficit is 3.2\% of the competing mean.  Thus, the few ranking reversals are
small compared with the overall admission levels.

\suppsecref{appendix:mixed_traffic_reward_class_results} reports Total Reward
and the per-class breakdown for the same runs.  \ac{cckb} attains higher Total
Reward than \ac{lincbwk} at 23 of the 24 mixed-traffic points (mean +19)
and than Random at all 24 (mean +77), consistent with the
Successful-Slice Count ranking.  Against \ac{config}, the comparison is mixed
(15 of 24, mean +6): averaged over the mixing ratios, \ac{cckb} admits more
\ac{urllc} slices in each topology--bottleneck condition, and eight of the nine
points where \ac{config} leads are driven by the \ac{embb} component rather
than by \ac{urllc} admissions.

\begin{figure}[t]
\centering
\subfloat[Tree topology.\label{fig:mixture_success_slices_vs_embb_ratio_tree}]{
\includegraphics[width=0.95\linewidth]{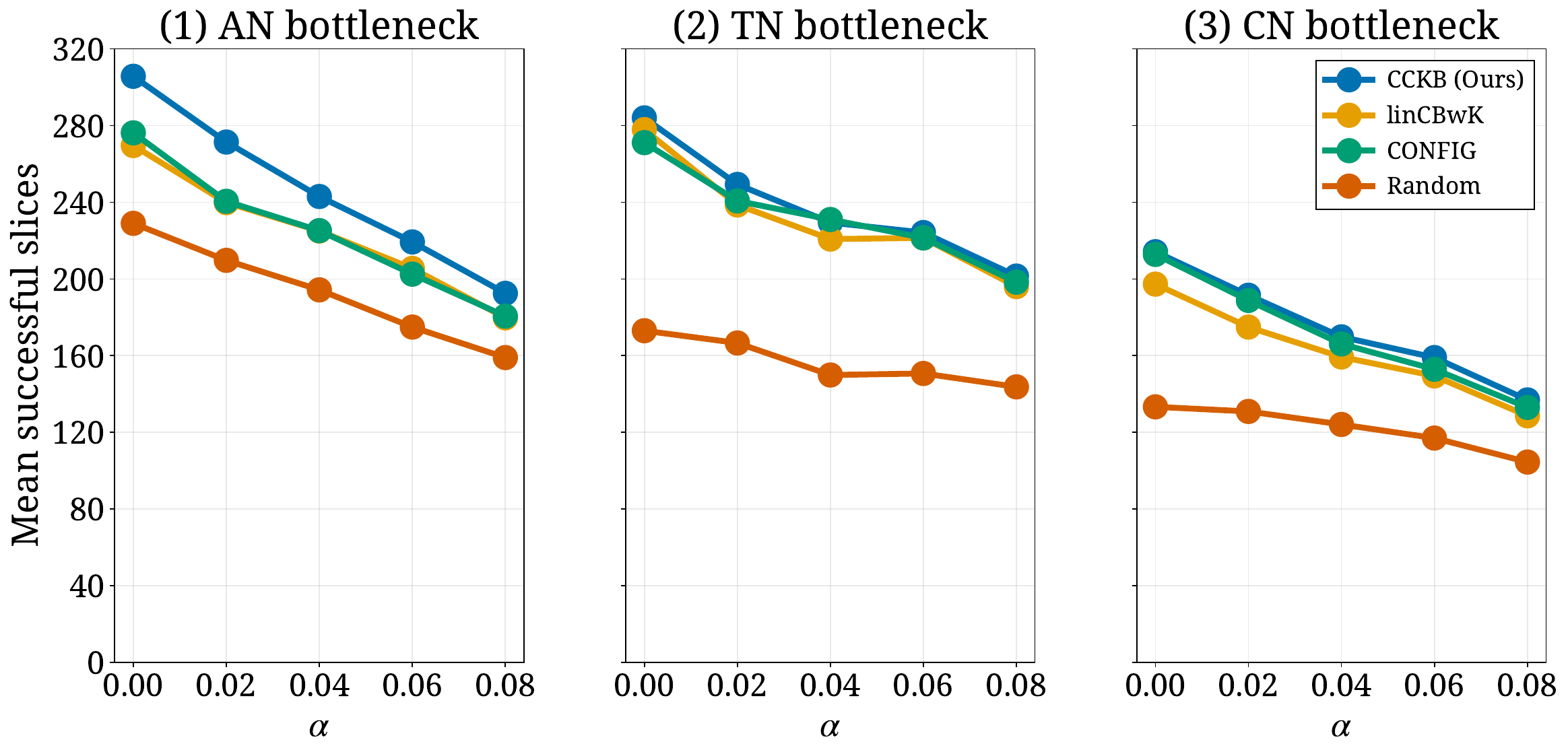}
}\\[0.6em]
\subfloat[Ring topology.\label{fig:mixture_success_slices_vs_embb_ratio_ring}]{
\includegraphics[width=0.95\linewidth]{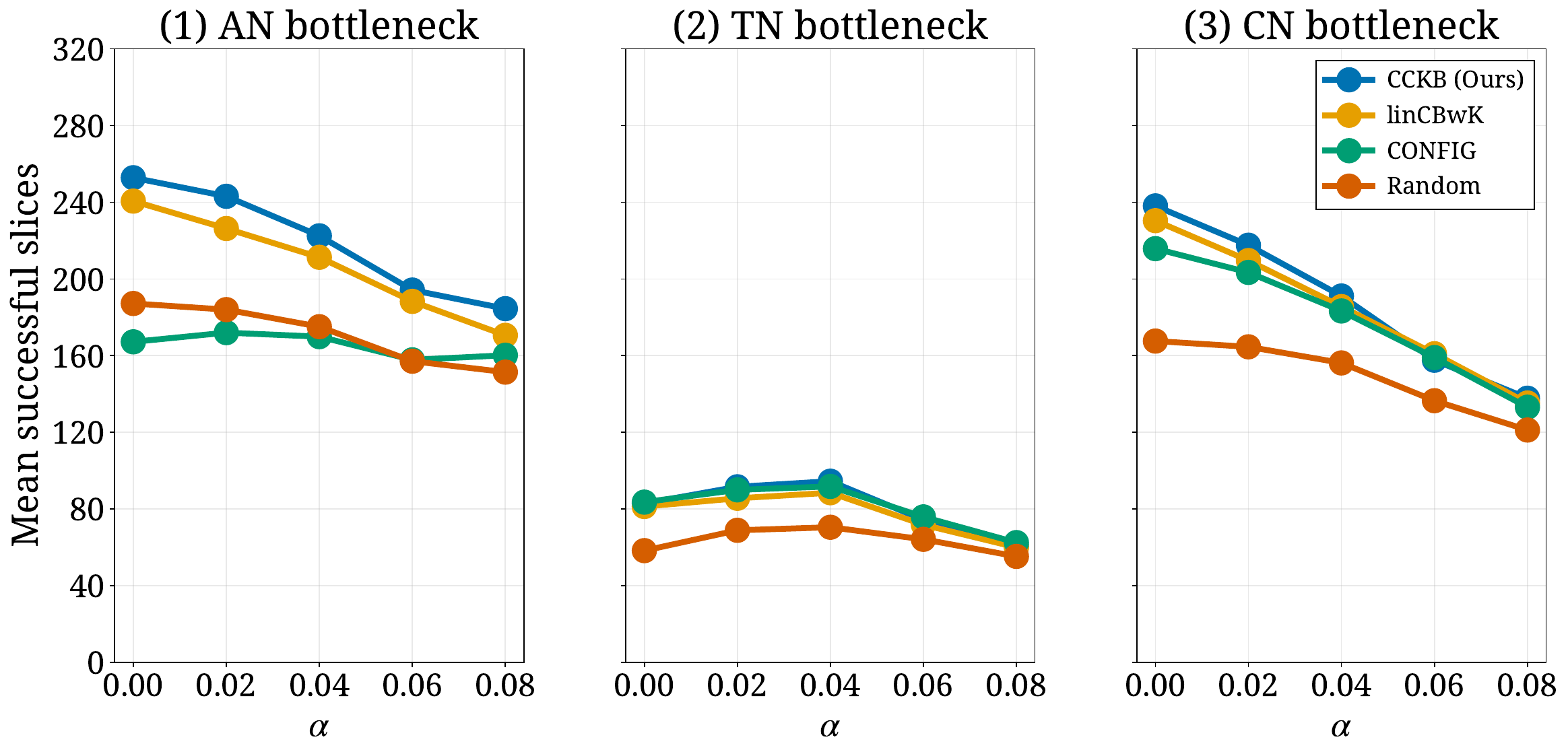}
}
\caption{Mean number of successfully provisioned \acp{ns} over 10 runs as the \ac{embb} mixing ratio varies over \(\alpha\in\{0,0.02,0.04,0.06,0.08\}\) under AN, TN, and CN bottleneck conditions in (a) Tree and (b) Ring topologies.}
\label{fig:mixture_success_slices_vs_embb_ratio_tree_ring}
\end{figure}

\subsection{Effects of Relaxing \textit{(L1)} and \textit{(L2)}}
\label{subsec:ablation_relaxation_effect}

\subsubsection{Setup}
\label{subsubsec:ablation_relaxation_setup}
The setup follows Sec.~\ref{subsubsec:topology_bottleneck_setup}, except that we
focus on the Tree topology and evaluate its three bottleneck conditions.

\textbf{Ablation Design:}
To isolate the effects of the two improvements addressing \textit{(L1)} and
\textit{(L2)}, we compare four variants obtained by independently combining
two surrogate-model options (nonlinear and linear \acp{gp}) and two
decomposition search-space options (continuous and discrete).
Table~\ref{tab:ablation_design_kernel_search} summarizes the resulting \(2\times2\) design.

\begin{table}[t]
\caption{Factorial Ablation Design ($2 \times 2$): Surrogate Models and Decomposition Search Spaces.}
\label{tab:ablation_design_kernel_search}
\centering
\small
\setlength{\tabcolsep}{3pt}
\renewcommand{\arraystretch}{1.1}
\begin{tabular}{@{}>{\raggedright\arraybackslash}p{0.25\columnwidth}>{\centering\arraybackslash}p{0.25\columnwidth}>{\centering\arraybackslash}p{0.25\columnwidth}@{}}
\toprule
\textbf{Method}
&
\begin{tabular}[c]{@{}c@{\;}l@{}}
\multirow{1}{*}{\textbf{(i)}} & \textbf{Surrogate} \\
& \textbf{Model}
\end{tabular}
&
\begin{tabular}[c]{@{}c@{\;}l@{}}
\multirow{1}{*}{\textbf{(ii)}} & \textbf{Search} \\
& \textbf{Space}
\end{tabular}
\\
\midrule
CCKB (Ours) & Nonlinear \ac{gp} & Continuous \\
Lin-Cont & Linear \ac{gp} & Continuous \\
Mat-Disc & Nonlinear \ac{gp} & Discrete \\
\ac{lincbwk} & Linear \ac{gp} & Discrete \\
\bottomrule
\end{tabular}
\end{table}

For both the reward and constraint surrogates, the nonlinear \ac{gp} option
uses \(k^{\bullet,\mathrm{Mat}}\), whereas the linear \ac{gp} option uses
\(k^{\bullet,\mathrm{Lin}}\).
\arxivonly{These kernels are defined in
\eqref{eq:comparison_kernel_matern_symbol}--\eqref{eq:comparison_kernel_matern_distance}
and \eqref{eq:comparison_kernel_linear_symbol}, respectively.}
\journalonly{These kernels are specified in
\suppsecref{appendix:comparison_kernel_definitions}.}
The continuous-search variants optimize over
\(\mathcal{X}_{\mathrm{cont}}(w_{\mathrm{lb}})\) defined in
\eqref{eq:comparison_search_space_continuous}, whereas the discrete-search
variants evaluate all candidates in
\(\mathcal{X}_{\Delta_{\mathrm{grid}}}(w_{\mathrm{lb}})\) defined in
\eqref{eq:comparison_search_space_discrete}.
All parameter settings are identical to those used in the preceding experiments.


\subsubsection{Result}
\label{subsubsec:ablation_relaxation_result}

The ablation yields two main observations.

\begin{figure}[t]
\centering
\includegraphics[width=0.55\linewidth]{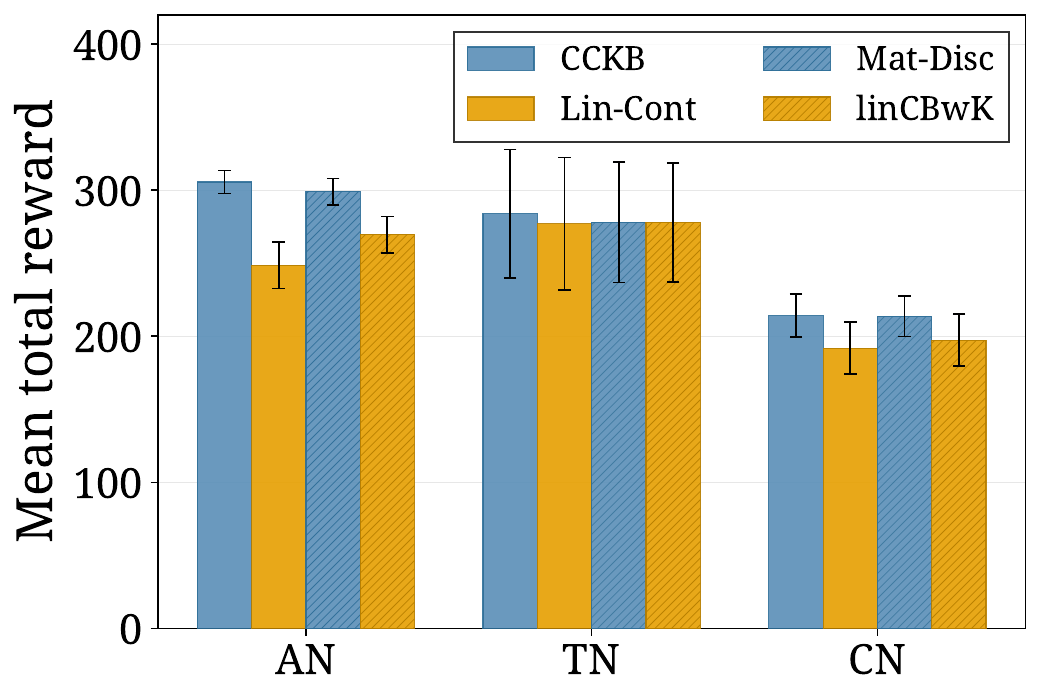}
\caption{Mean Total Reward by bottleneck for four ablation variants on the Tree topology (error bars indicate standard deviation).}
\label{fig:ablation_reward_mean_kernel_search}
\end{figure}

\begin{table}[t]
\caption{Surrogate-Refitting Failure Rate for the surrogate-model/search-space ablation on the Tree topology.}
\label{tab:ablation_refit_failure_rate_kernel_search}
\centering
\small
\setlength{\tabcolsep}{4pt}
\resizebox{\columnwidth}{!}{%
\begin{tabular}{c c c c}
\toprule
\textbf{CCKB} & \textbf{Lin-Cont} & \textbf{Mat-Disc} & \textbf{linCBwK} \\
\midrule
0/2{,}100 (0.00\%) & 71/2{,}100 (3.38\%) & 235/2{,}100 (11.19\%) & 86/2{,}100 (4.10\%) \\
\bottomrule
\end{tabular}
}
\end{table}

\textbf{Total Reward:}
CCKB achieves the highest mean Total Reward in all three bottleneck conditions
(Fig.~\ref{fig:ablation_reward_mean_kernel_search}).
The comparison with Lin-Cont indicates that, under the same continuous search
space, the nonlinear \ac{gp} instantiation is more effective than the linear
\ac{gp} instantiation.
The comparison with Mat-Disc further shows that, under the same nonlinear
surrogate model, continuous search is also beneficial.
By contrast, Lin-Cont does not consistently outperform linCBwK, suggesting that
enlarging the search space alone does not overcome the limitations of linear
reward and constraint surrogates.

\textbf{Surrogate-Refitting Failure Rate:}
The reward comparison alone does not reveal the numerical stability of the
surrogate models.
Table~\ref{tab:ablation_refit_failure_rate_kernel_search} shows that CCKB
completes all 2{,}100 scheduled surrogate-refitting operations without
failure, whereas the other three designs exhibit nonzero failure rates.
Thus, among the compared designs, CCKB achieves the highest Total Reward while
maintaining stable surrogate refitting.

\subsection{Effect of Addressing \textit{(L3)}}
\label{subsec:ablation_proxy_filters}

\subsubsection{Setup}
\label{subsubsec:ablation_proxy_filters_setup}
The setup follows Sec.~\ref{subsubsec:ablation_relaxation_setup}.

\textbf{Ablation Design:}
We compare three variants that differ only in the sample selection rule used for the constraint proxy learner:
\textit{(i): CCKB} uses only samples from rounds satisfying both \(y_\tau=1\)
and \(\mathbf{1}^{\Gamma}_{d,j}(s_\tau)=1\),
\textit{(ii): Path-Only} removes the condition \(y_\tau=1\) and uses samples
satisfying \(\mathbf{1}^{\Gamma}_{d,j}(s_\tau)=1\), and
\textit{(iii): No-Filter} uses all observations and directly learns the relaxed-\ac{nsrdp} constraint
\(h_{d,j}\).
The samples retained by the three variants are
\[
\begin{array}{@{}c@{\quad}l@{\quad}l@{}}
\textnormal{(i)}
& \textnormal{CCKB:}
& \big\{(z_\tau,u_{\tau,d,j})\ \big|\
\substack{\tau\in[t-1],\ y_\tau=1,\\
\mathbf{1}^{\Gamma}_{d,j}(s_\tau)=1}
\big\},\\
\textnormal{(ii)}
& \textnormal{Path-Only:}
& \big\{(z_\tau,u_{\tau,d,j})\ \big|\
\substack{\tau\in[t-1],\\ \mathbf{1}^{\Gamma}_{d,j}(s_\tau)=1}
\big\},\\
\textnormal{(iii)}
& \textnormal{No-Filter:}
& \big\{(z_\tau,u_{\tau,d,j})\ \big|\ \tau\in[t-1]\big\}.
\end{array}
\]

\subsubsection{Result}
\label{subsubsec:ablation_proxy_filters_result}

The filter ablation yields two observations.

\begin{figure*}[t]
\centering
\subfloat[Mean Total Reward by bottleneck.\label{fig:ablation_proxy_filters_reward_and_nlpd_reward}]{
\includegraphics[width=0.25\textwidth]{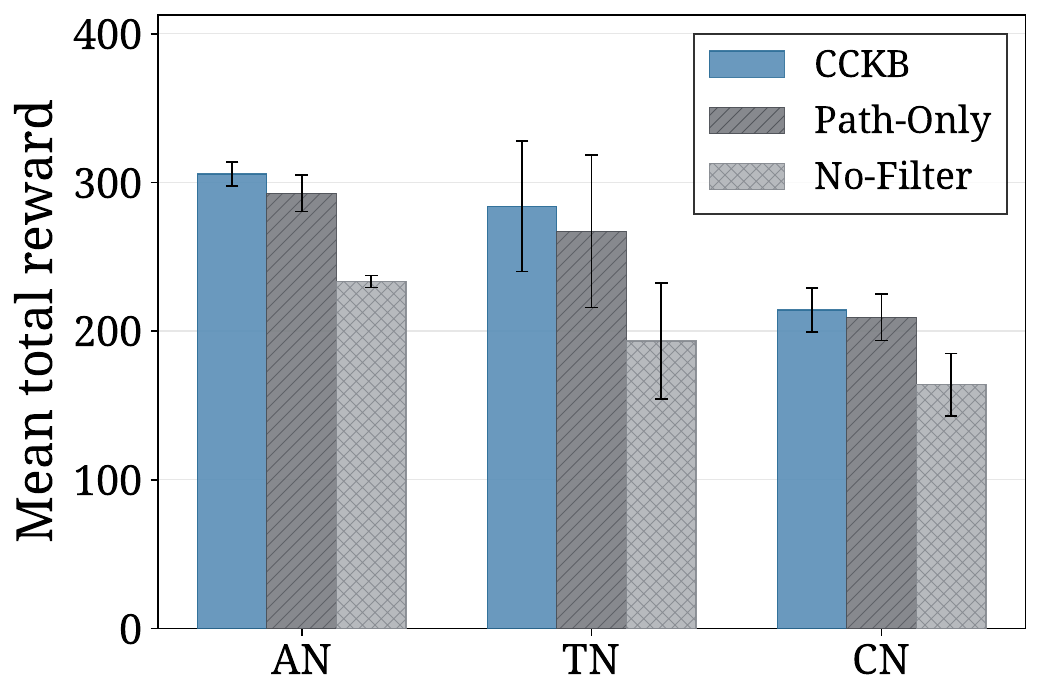}
}
\hfill
\subfloat[Constraint \ac{gp} \ac{nlpd} history by bottleneck.\label{fig:ablation_proxy_filters_reward_and_nlpd_nlpd}]{
\includegraphics[width=0.7\textwidth]{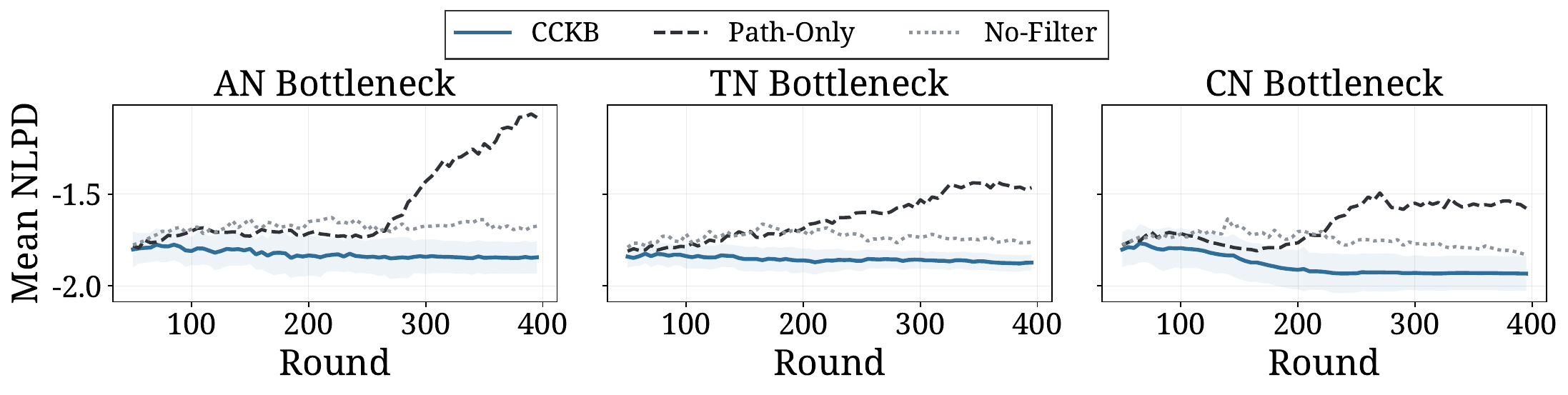}
}
\caption{Effects of the proxy-construction filters. (a) Mean Total Reward. (b) Constraint-\ac{gp} \ac{nlpd} trajectories; lower values indicate better predictive fit.}
\label{fig:ablation_proxy_filters_reward_and_nlpd}
\end{figure*}

\textbf{Total Reward:}
CCKB achieves the highest mean Total Reward in all three bottleneck conditions
(Fig.~\ref{fig:ablation_proxy_filters_reward_and_nlpd}(a)), supporting the
practical benefit of using both filters in the proxy construction.
All three variants complete all 2{,}100 scheduled
refitting operations without failure, so this reward comparison is not
confounded by surrogate-refitting instability.

\textbf{Constraint-Surrogate Fit:}
The constraint \ac{gp} in CCKB attains the lowest \ac{nlpd}
(Fig.~\ref{fig:ablation_proxy_filters_reward_and_nlpd}(b)), showing that the
combined filters also provide the best predictive fit among the three variants.
Path-Only is lower than No-Filter in the early rounds, but its \ac{nlpd} deteriorates once resource depletion starts (e.g., around round 260 in the AN-bottleneck case).

Together, the reward and \ac{nlpd} results support learning the proxy target
\(h^{\mathrm{prx}}_{d,j}\) with both sample-selection filters rather than directly learning
\(h_{d,j}\) from all observations, thereby supporting the practical effectiveness of the proposed treatment of \textit{(L3)}.

\subsection{Effects of Mechanisms for \textit{(R1)} and \textit{(R2)}}
\label{subsec:ablation_context_dual}

\subsubsection{Setup}
\label{subsubsec:ablation_context_dual_setup}
The setup follows Sec.~\ref{subsubsec:ablation_relaxation_setup}.

\textbf{Ablation Design:}
We compare CCKB with three variants: w/o Context removes context conditioning, w/o Dual removes dual-based budget control, and w/o Context \& Dual removes both components.
All four variants use the same kernel, continuous search space, training protocol, and hyperparameters.
Among the ablation variants, w/o Context retains the primal--dual mechanism
and therefore corresponds to a non-contextual, proxy-based \ac{ckb} instantiation.

\noindent
\textbf{(w/o Context):}
To ablate the request conditioning used to address \textit{(R2)}, the two
variants without context conditioning use the same fixed-request surrogate-input
construction as that used for \ac{config} in
Sec.~\ref{subsubsec:comparison_methods_detail}.
Specifically, we replace the current request \(s_t\) with the first-round request
\(s_1\), evaluating the surrogates at \((s_1,\mathbf{x})\), rather than
\((s_t,\mathbf{x})\).

\noindent
\textbf{(w/o Dual):}
To ablate the dual-based budget control used to address \textit{(R1)}, the two
variants without dual-based budget control skip the dual update in
Line~\ref{algline:component1:dual} of Algorithm~\ref{alg:component1}, so the
dual variables remain at zero.
Consequently, the dual penalty is inactive throughout the run.

\subsubsection{Result}
\label{subsubsec:ablation_context_dual_result}


\begin{figure}[t]
\centering
\includegraphics[width=0.6\linewidth]{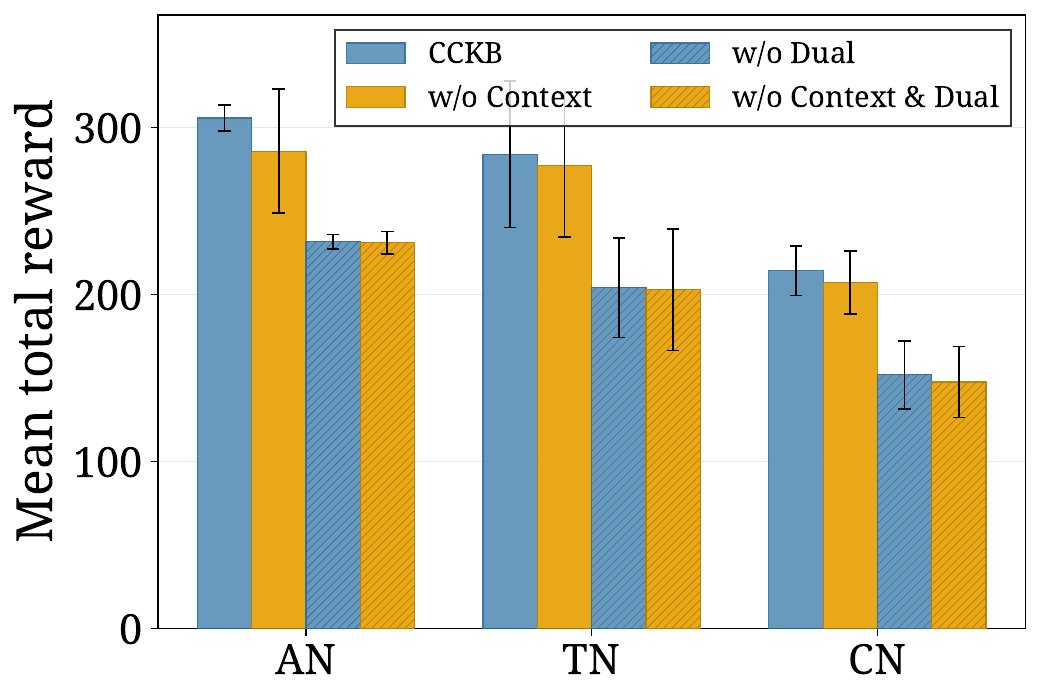}
\caption{Mean Total Reward by bottleneck for the four context/dual ablation variants on the Tree topology (error bars indicate standard deviation).}
\label{fig:ablation_context_dual_reward_mean}
\end{figure}

Fig.~\ref{fig:ablation_context_dual_reward_mean} shows that CCKB achieves the highest mean Total Reward in all three bottleneck conditions.
Relative to CCKB, w/o Dual reduces the mean reward by 74.2, 80.0, and 62.3
under the AN, TN, and CN bottlenecks, respectively.
By contrast, w/o Context yields smaller reductions of 19.9, 6.7, and 7.0,
respectively.
w/o Dual and w/o Context \& Dual form the lower-performing pair in every
condition, with only a small additional difference between them.

All four variants complete all 2{,}100 scheduled surrogate-refitting operations
without failure (0/2{,}100 per method); therefore, we omit a separate
failure-rate table for this ablation, and the reward differences are not
confounded by surrogate-refitting instability.
These descriptive results suggest that the dual-based mechanism for
\textit{(R1)} provides the dominant contribution, while the context-conditioning
mechanism for \textit{(R2)} provides a smaller complementary improvement.

\subsection{Computational Scalability}
\label{subsec:runtime_scalability}

\subsubsection{Setup}
\label{subsubsec:runtime_scalability_setup}
Unless stated below, the setup follows
Sec.~\ref{subsubsec:ablation_relaxation_setup}.
We use five runs with different random seeds, refit the \ac{gp} models every
round, and impose no domain bottleneck by setting the resource-availability
vector to
\((a_{\mathrm{AN}},a_{\mathrm{TN}},a_{\mathrm{CN}})=(1.0,1.0,1.0)\).

\textbf{Experiment Design:}
We use the Tree topology and increase the network size as
\(|\mathcal{A}|\in\{12,24,36,48\}\) and
\(|\mathcal{R}|\in\{4,8,12,16\}\).
We compare \ac{cckb}, \ac{lincbwk}, and Random under the same CPU allocation,
which averages 9 cores per run.
\ac{config} is omitted because its runtime is likewise dominated by the same \ac{gp}
computations as \ac{cckb}; including it would therefore not provide a distinct
scaling comparison.

\subsubsection{Result}
\label{subsubsec:runtime_scalability_result}
The evaluation yields two main observations.

\begin{figure}[t]
\centering
\subfloat[Per-round runtime at \(|\mathcal{A}|=12\).\label{fig:runtime_scalability_gnb12_step}]{
\includegraphics[width=0.95\linewidth]{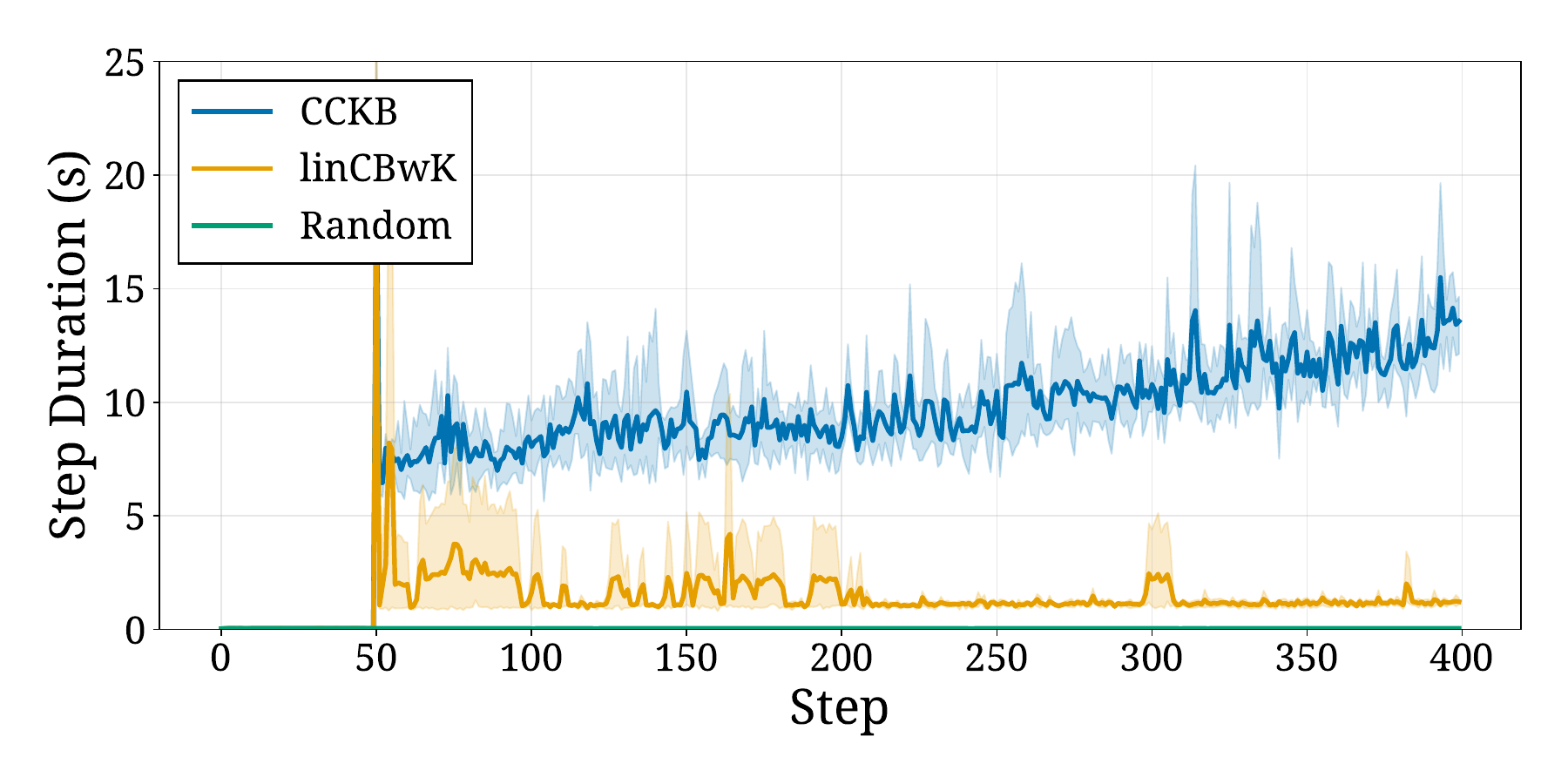}
}\\[0.4em]
\subfloat[Total runtime per \(T=400\)-round run as \(|\mathcal{A}|\) increases.\label{fig:runtime_scalability_total_vs_gnb}]{
\includegraphics[width=0.80\linewidth]{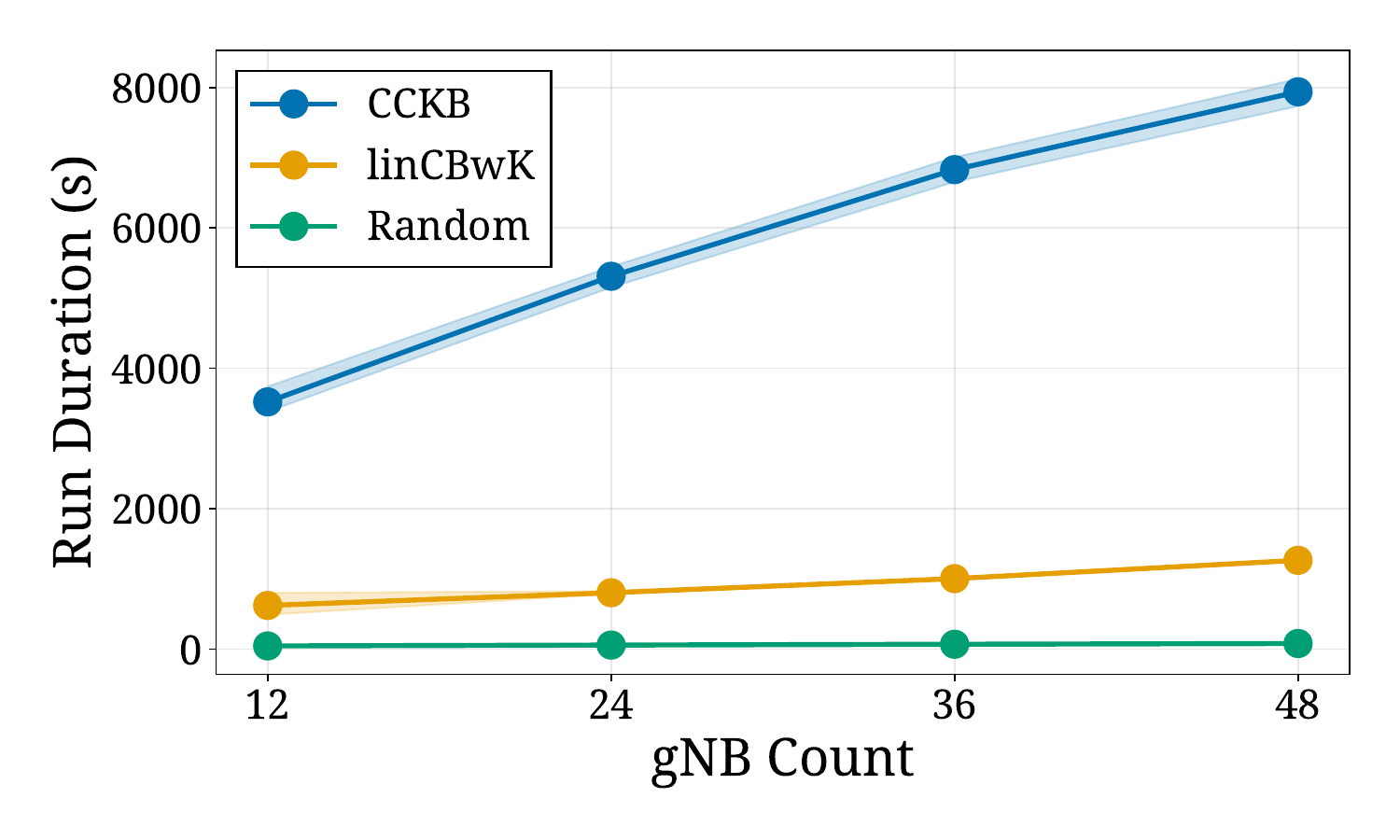}
}
\caption{Runtime scalability over five runs. Lines in (a) and markers in (b) show the means, and bands show the 10th--90th percentiles across runs.}
\end{figure}

\textbf{Per-Round Runtime:}
Fig.~\ref{fig:runtime_scalability_gnb12_step} shows that \ac{cckb} has the highest per-round runtime, followed by \ac{lincbwk}, while Random is much lower.
At \(|\mathcal{A}|=12\), the mean per-round runtime is \(8.64\)~s for \ac{cckb},
\(1.39\)~s for \ac{lincbwk}, and \(0.045\)~s for Random.
Although the measured operations differ, CCKB's computational overhead is
below the approximately 30~s reported for the instantiation of an \ac{ns} on an
experimental multi-slice platform~\cite{garcia2020experimenting}, suggesting
compatibility with orchestration workflows operating on a timescale of tens of
seconds.
CCKB's per-round runtime increases gradually over the horizon, whereas linCBwK
remains lower and Random remains nearly constant.
The gradual growth under CCKB is consistent with the selective updates in its
constraint-surrogate construction limiting the growth of effective training
samples.

\textbf{Scaling with Network Size:}
Fig.~\ref{fig:runtime_scalability_total_vs_gnb} shows monotonic runtime growth for all methods as \(|\mathcal{A}|\) increases.
For \ac{cckb}, the total run duration increases approximately linearly with \(|\mathcal{A}|\) over the tested range.
This trend is consistent with the complexity discussion in Sec.~\ref{subsec:discussion_time_complexity}, where the dominant cost scales linearly with the number of resource-side models.

\ifSubfilesClassLoaded{
  \bibliographystyle{ieee/IEEEtran}
  \bibliography{bib/references}
}{}

\end{document}

\documentclass[src/main.tex]{subfiles}

\begin{document}

\section{Discussion}
\label{sec:discussion}
\subsection{Time Complexity and Implementation Remarks}
\label{subsec:discussion_time_complexity}
We compare the computational complexity of the proposed method with that of our previous method~\cite{kobayashi2025icccn}.
We focus on the dominant cost of constructing the estimates \(\hat f_t\) and \(\hat h^{\mathrm{prx}}_{t,d,j}\), from which the clipped surrogates used by \ac{cckb} are obtained.

In our previous method~\cite{kobayashi2025icccn}, the corresponding reward and constraint estimates are computed by linear regression with a fixed \(p\)-dimensional parameter vector.
When the inverse covariance matrix is updated recursively via the Sherman--Morrison formula~\cite{sherman1950adjustment}, the per-model update cost is \(O(p^2)\).
Since the number of constraint models is on the order of the total number of resources \(J_{\mathrm{tot}}\), the per-round computational complexity is \(O(J_{\mathrm{tot}} p^2)\).
Therefore, over \(T\) rounds, the total complexity is \(O(J_{\mathrm{tot}} p^2 T)\).

By contrast, in the current naive exact-\ac{gp} implementation of the proposed method, each surrogate update requires inversion of the regularized Gram matrix \(\mathbf{K}^{\bullet}_{t-1}+\eta_{\bullet}\mathbf{I}\) specified in Sec.~\ref{subsubsec:gp_modeling_general}.
Thus, at round \(t\), the dominant per-model update cost is \(O(t^3)\) in the worst case.
Since the number of constraint models is again on the order of \(J_{\mathrm{tot}}\), the per-round complexity is upper-bounded by \(O(J_{\mathrm{tot}} t^3)\).
Summing this from \(t=1\) to \(T\), the total computational complexity becomes
\(
O(J_{\mathrm{tot}} T^4).
\)

Hence, compared with our previous method~\cite{kobayashi2025icccn}, 
the current exact-GP implementation significantly increases the computational cost from \(O(J_{\mathrm{tot}} p^2 T)\) to \(O(J_{\mathrm{tot}} T^4)\).
This suggests that, although the proposed method improves modeling flexibility, further acceleration techniques such as sparse/approximate \ac{gp} methods~\cite{quinonero-candela2005,titsias2009} will be important.

\subsection{Rationale for the Proxy Design}
\label{subsec:discussion_why_proxy}
The proxy constraint is learned from resource-consumption observations for
successful requests and resources included in the requested \ac{ns} topology,
thereby avoiding the many zero observations described in \textit{(L3)}.  It
replaces the unknown provisioning-success probability \(p(s,\mathbf{x})\) by
its upper bound of one.  An alternative is to estimate
\(p(s,\mathbf{x})\) and the conditional mean demand
\(m_{d,j}^{\Gamma}(s,\mathbf{x})\) separately and multiply the two estimates.
This alternative could make the constraint less conservative, but a confidence
bound for the resulting constraint would have to account for errors in both
estimates and in their product.  Moreover, estimating a probability from the
binary provisioning outcome generally requires numerical approximation rather
than the closed-form calculations used for \ac{gp} regression of the
demand~\cite{rasmussen2006gpml}.  We therefore use the current proxy because it
requires only the demand model and allows us to derive the confidence bounds
used in our theoretical analysis.  


\ifSubfilesClassLoaded{
  \bibliographystyle{ieee/IEEEtran}
  \bibliography{bib/references}
}{}

\end{document}

\documentclass[src/main.tex]{subfiles}

\begin{document}

\section{Conclusion}
This paper addressed online \ac{nsr} decomposition in hierarchical 5G management, where an \ac{e2e} controller must select request-conditioned continuous decompositions while respecting long-horizon multi-resource budgets.
We formulated the \ac{nsrdp} and introduced \ac{cckb} as an online solution for it.
\ac{cckb} combines primal--dual budget control with contextual \ac{gp} surrogates to learn nonlinear relationships between requests, decompositions, provisioning success, and resource demand without restricting the search to a finite set of decompositions.
To handle zero-dominated resource-consumption observations caused by rejected requests and resources outside the requested \ac{ns} topology, we derived a conservative proxy problem for the stationary-response relaxation of the \ac{nsrdp}.
We established high-probability finite-time regret and cumulative
constraint-violation bounds for the proxy problem and transferred both
guarantees to the relaxed \ac{nsrdp}.  The relaxed-problem regret bound
includes an explicit additive term that quantifies the optimality gap
induced by the proxy formulation.
In 5G network simulations, \ac{cckb} achieved higher mean Total Reward than \ac{lincbwk}~\cite{kobayashi2025icccn} in all 12 topology--bottleneck conditions and than \ac{config}~\cite{xu2023config} in 9 of the 12 conditions.
Extending these guarantees to state-dependent responses and improving scalability through approximate \ac{gp} updates remain important future directions.

\ifSubfilesClassLoaded{
  \bibliographystyle{ieee/IEEEtran}
  \bibliography{bib/references}
}{}

\end{document}

\bibliographystyle{ieee/IEEEtran}
\bibliography{bib/references}

\arxivonly{%
  \startsupplementresults
  \appendices
  \documentclass[src/main.tex]{subfiles}

\begin{document}

\section{Proof That Domain-Level Guarantees Imply the \ac{e2e} Guarantee}
\label{appendix:cond_independence_guarantee}
This appendix proves the implication in \eqref{eq:operational_guarantee_preservation}: for an admissible decomposition, satisfying every domain-level \ac{sla} guarantee implies satisfying the original \ac{e2e} guarantee.
The proof uses one assumption about target achievement after the evaluation conditions hold in every domain, together with the construction of the domain-level evaluation conditions in \eqref{eq:domain_evaluation_condition_construction}.
We first state the assumption and use the latency requirement to explain when the assumption is appropriate (Sec.~\ref{appendix:cond_ind_assumptions}).
We then prove \eqref{eq:operational_guarantee_preservation} (Sec.~\ref{appendix:cond_ind_guarantee_proof}).

\subsection{Assumption Used in the Proof}
\label{appendix:cond_ind_assumptions}
The proof requires the following assumption about the domain-level target events.

\begin{assumption}[Domain-level target events after all evaluation conditions hold]
\label{ass:cond_ind}
For each \(s \in \mathcal{S}\), \(i\in[M]\), and admissible decomposition \(\mathbf{x}\) for \(s\) such that
\(\Pr_{s,\mathbf{x}}(\bigcap_{d\in\mathcal{D}}H_{d,i}(s))>0\), the following conditions hold:
\begin{enumerate}[label=\textnormal{(\roman*)}]
    \item The events \(T_{d,i}(s,\mathbf{x}_T)\), \(d\in\mathcal{D}\), are mutually independent under
    \(\Pr_{s,\mathbf{x}}(\,\cdot\mid\bigcap_{d'\in\mathcal{D}}H_{d',i}(s))\).
    \item For every \(d\in\mathcal{D}\),
    \[
    \begin{aligned}
    &\Pr_{s,\mathbf{x}}\!\left(
    T_{d,i}(s,\mathbf{x}_T)
    \,\middle|\,
    \bigcap_{d'\in\mathcal{D}}H_{d',i}(s)
    \right)\\[-1mm]
    &\quad=
    \Pr_{s,\mathbf{x}}\!\left(
    T_{d,i}(s,\mathbf{x}_T)
    \mid H_{d,i}(s)
    \right).
    \end{aligned}
    \]
\end{enumerate}
\end{assumption}
Condition \textnormal{(i)} implies that the joint target-achievement probability,
conditional on all domain-level evaluation conditions, factorizes as
\[
\begin{aligned}
&
\Pr_{s,\mathbf{x}}\!\left(
\bigcap_{d\in\mathcal{D}}T_{d,i}(s,\mathbf{x}_T)
\,\middle|\,
\bigcap_{d\in\mathcal{D}}H_{d,i}(s)
\right)\\
&\quad=
\prod_{d\in\mathcal{D}}
\Pr_{s,\mathbf{x}}\!\left(
T_{d,i}(s,\mathbf{x}_T)
\,\middle|\,
\bigcap_{d'\in\mathcal{D}}H_{d',i}(s)
\right).
\end{aligned}
\]
Condition \textnormal{(ii)} states that whether the evaluation conditions also hold in the other domains does not change the target-achievement probability in domain \(d\).

\textbf{Latency Example and Scope:}
For the latency requirement, \(H_{d,i}(s)\) is the event that the packet is not dropped in domain \(d\), and \(T_{d,i}(s,\mathbf{x}_T)\) is the event that the delay in domain \(d\) does not exceed the limit allocated to domain \(d\).
Condition \textnormal{(i)} means that, among packets not dropped in any domain,
meeting or missing the allocated delay limit in one domain does not make the
limits in the other domains more or less likely to be met.
Condition \textnormal{(ii)} means that, among packets not dropped in domain \(d\),
selecting only those also not dropped in any other domain does not make the
delay limit in domain \(d\) more or less likely to be met.
Thus, packets that pass through all other domains are neither more nor less
likely to meet the delay limit in domain \(d\).
Conditions \textnormal{(i)} and \textnormal{(ii)} can be appropriate for the hierarchical \ac{ns} management architecture considered in this paper when each domain-specific controller operates independently, uses resources assigned only to the corresponding domain, and is not affected by a failure or traffic change that simultaneously affects several domains.

\subsection{Guarantee Preservation}
\label{appendix:cond_ind_guarantee_proof}
\begin{proof}
Fix \(s\in\mathcal{S}\), \(i\in[M]\), and an admissible decomposition \(\mathbf{x}\) for \(s\).
Equation~\eqref{eq:domain_evaluation_condition_construction} and \(\Pr_{s,\mathbf{x}}(H_i)>0\) ensure that \(\Pr_{s,\mathbf{x}}(\bigcap_{d\in\mathcal{D}}H_{d,i}(s))>0\).
Assumption~\ref{ass:cond_ind} gives
\begin{align}
&\Pr_{s,\mathbf{x}}\!\left(
\bigcap_{d\in\mathcal{D}}T_{d,i}(s,\mathbf{x}_T)
\,\middle|\,
\bigcap_{d\in\mathcal{D}}H_{d,i}(s)
\right)
\notag\\
&\overset{\textnormal{(i)}}{=}
\prod_{d\in\mathcal{D}}
\Pr_{s,\mathbf{x}}\!\left(
T_{d,i}(s,\mathbf{x}_T)
\,\middle|\,
\bigcap_{d'\in\mathcal{D}}H_{d',i}(s)
\right)
\notag \\
&\overset{\textnormal{(ii)}}{=}
\prod_{d\in\mathcal{D}}
\Pr_{s,\mathbf{x}}\!\left(
T_{d,i}(s,\mathbf{x}_T)\mid H_{d,i}(s)
\right).
\label{eq:cond_ind}
\end{align}
To prove \eqref{eq:operational_guarantee_preservation}, assume that all domain-level guarantees in \eqref{eq:domain_probabilistic_guarantee} hold.
\begin{align}
\Pr_{s,\mathbf{x}}(T_i \mid H_i)
&= \frac{\Pr_{s,\mathbf{x}}(T_i\cap H_i)}{\Pr_{s,\mathbf{x}}(H_i)} \notag \\
&\overset{\textnormal{(A1)}}{\ge} \frac{
\Pr_{s,\mathbf{x}}\!\left(
H_i \cap
\bigcap_{d\in\mathcal{D}} T_{d,i}(s,\mathbf{x}_T)
 \right)
}{
\Pr_{s,\mathbf{x}}(H_i)
} \notag \\
&=
\Pr_{s,\mathbf{x}}\!\left(
\bigcap_{d\in\mathcal{D}} T_{d,i}(s,\mathbf{x}_T)
\,\middle|\,
H_i
\right) \notag \\
&\overset{\eqref{eq:domain_evaluation_condition_construction}}{=}
\Pr_{s,\mathbf{x}}\!\left(
\bigcap_{d\in\mathcal{D}} T_{d,i}(s,\mathbf{x}_T)
\,\middle|\,
\bigcap_{d\in\mathcal{D}} H_{d,i}(s)
\right) \notag \\
&\overset{\eqref{eq:cond_ind}}{=} \prod_{d\in\mathcal{D}}
\Pr_{s,\mathbf{x}}\!\left(
T_{d,i}(s,\mathbf{x}_T) \mid H_{d,i}(s)
\right).
\label{eq:appendix_decomp_ind}
\end{align}
Using the domain-level guarantees in \eqref{eq:domain_probabilistic_guarantee} and condition \textnormal{(A2)} in Definition~\ref{def:admissible_decomp},
\begin{equation}
\Pr_{s,\mathbf{x}}(T_i \mid H_i)
\ge \prod_{d\in\mathcal{D}} g_{d,i}(s,\mathbf{x}_g)
\ge g_i.
\end{equation}
Hence, whenever all delegated inequalities in \eqref{eq:domain_probabilistic_guarantee} hold, the original \ac{e2e} requirement \(\Pr_{s,\mathbf{x}}(T_i \mid H_i)\ge g_i\) also holds.
Consequently, conditions \textnormal{(A1)} and \textnormal{(A2)} imply \eqref{eq:operational_guarantee_preservation} under Assumption~\ref{ass:cond_ind}.
\end{proof}


\end{document}

  \documentclass[src/main.tex]{subfiles}

\begin{document}

\section{Intermediate Results and Explicit Bounds}
\label{appendix:theory_intermediate_results}
This section records the intermediate results and explicit surrogate-error
scales underlying Corollary~\ref{cor:concrete_finite_time_proxy_performance}.
Table~\ref{tab:theory_analysis_stages} summarizes the proof stages.

\begin{table}[t]
    \caption{Problems and \ac{cckb} instantiations in the theoretical analysis.}
    \label{tab:theory_analysis_stages}
    \centering
    \footnotesize
    \setlength{\tabcolsep}{3pt}
    \renewcommand{\arraystretch}{1.2}
\begin{tabularx}{\columnwidth}{@{}
    >{\raggedright\arraybackslash}X
    >{\raggedright\arraybackslash}X
    >{\raggedright\arraybackslash}X@{}}
\toprule
\textbf{Analysis Stage} & \textbf{Benchmark Problem}
& \textbf{\ac{cckb} Instantiation} \\
\midrule
Generic finite-time guarantee
(Theorem~\ref{thm:general_cckb_regret_violation})
& Relaxed \ac{nsrdp}~\eqref{eq:relaxed_problem}
& Direct \\
Proxy finite-time guarantee
(Corollary~\ref{cor:proxy_regret_violation})
& Proxy problem~\eqref{eq:proxy_relaxed_problem}
& Proxy-based \\
Transferred finite-time guarantee
(Corollary~\ref{cor:relaxed_regret_violation})
& Relaxed \ac{nsrdp}~\eqref{eq:relaxed_problem}
& Proxy-based \\
Concrete guarantee with explicit surrogate-error bounds
(Corollary~\ref{cor:explicit_finite_time_proxy_performance})
& Proxy problem~\eqref{eq:proxy_relaxed_problem} and relaxed
\ac{nsrdp}~\eqref{eq:relaxed_problem}
& Proxy-based \\
\bottomrule
\end{tabularx}
\end{table}

\subsection{Slater Consequences and Generality}
\label{appendix:slater_and_generality}
The following corollary records the standard consequence of Slater's theorem
used in the finite-time analysis~\cite[Sec.~5.2.3]{boyd2004convex}.

\begin{corollary}[Consequence of Slater's condition]
\label{cor:cckb_slater_consequence}
Suppose that the relaxed \ac{nsrdp} in \eqref{eq:relaxed_problem} admits an
optimal policy.  Under Assumption~\ref{ass:cckb_slater}, its primal and dual
optimal values coincide.  Moreover, there exist an optimal dual vector
\(\boldsymbol\phi^\star\!\ge\!\mathbf0\) and, for the fixed horizon \(T\),
a finite bound \(\Lambda_T^{\mathrm{rel}}\) such that
\begin{equation}
\|\boldsymbol\phi^\star\|_1\le\Lambda_T^{\mathrm{rel}}<\infty.
\label{eq:cckb_rel_dual_norm_bound}
\end{equation}
\end{corollary}

The proof is provided in
\suppsecref{appendix:cckb_slater_consequence}.
The same conclusion applies to the proxy problem after replacing
\(\mathbf h\) with \(\mathbf h^{\mathrm{prx}}\), provided that the proxy
problem admits an optimal policy and Assumption~\ref{ass:cckb_slater} holds.
We denote the corresponding proxy dual-norm bound by \(\Lambda_T^{\mathrm{prx}}\).

\begin{remark}[Generality of the CCKB guarantee]
\label{rem:cckb_generality}
The finite-time guarantee in
Theorem~\ref{thm:general_cckb_regret_violation} is not limited to network
slicing and applies to other bounded contextual constrained decision-making
problems with the same structure.  Specifically, the theorem relies only on
the contextual objective and long-term-constraint structure in
\eqref{eq:relaxed_problem},
Assumptions~\ref{ass:independent_request_arrivals},
\ref{ass:cckb_kernel_regular}, and
\ref{ass:cckb_request_measurability}, and the
surrogate conditions in \eqref{eq:cckb_surrogate_direction} and
\eqref{eq:cckb_cumulative_surrogate_error}, rather than on
network-slicing-specific formulation details; the constraint-violation part
additionally uses Assumption~\ref{ass:cckb_slater}.
\end{remark}

\subsection{Proxy Instantiation}
\label{appendix:proxy_instantiation}
We now apply Theorem~\ref{thm:general_cckb_regret_violation} to the proxy
problem.  Let \(\mathcal E_{\mathrm{sur}}^{\mathrm{prx}}\) denote the event
on which (C1) and (C2) hold after replacing \(\mathbf h\) and
\(\bar{\mathbf h}_t\) with \(\mathbf h^{\mathrm{prx}}\) and
\(\bar{\mathbf h}^{\mathrm{prx}}_t\), respectively.  On this event, let
\(\mathcal W_f(T)\) and \(\boldsymbol{\mathcal W}_h(T)\) denote the
cumulative error bounds in \eqref{eq:cckb_cumulative_surrogate_error} after
these replacements.

\begin{corollary}[Finite-time bounds for proxy-based CCKB]
\label{cor:proxy_regret_violation}
Suppose that the proxy problem admits an optimal policy and that
Assumptions~\ref{ass:independent_request_arrivals},
\ref{ass:cckb_kernel_regular}, and~\ref{ass:cckb_request_measurability} hold.
Fix the failure probabilities as in
Theorem~\ref{thm:general_cckb_regret_violation}, choose \(\rho>0\), and set
\(V:=\sqrt{J_{\mathrm{tot}}T}/\rho\).

\emph{(i) Proxy regret:}
Suppose that the proxy versions of the one-sided bounds in
\eqref{eq:cckb_surrogate_direction} and the reward cumulative-error bound in
\eqref{eq:cckb_cumulative_surrogate_error} hold jointly with probability at
least \(1-\alpha_{\mathrm{sur}}\).  Then, with probability at least
\(1-\alpha_{\mathrm{ctx}}^f-\alpha_{\mathrm{ctx}}^h
-\alpha_{\mathrm{sur}}\),
\begin{equation}
\mathrm{Reg}^{\mathrm{prx}}(T)
\le B_{\mathrm{reg}}(T),
\label{eq:proxy_regret_bound_expanded_final}
\end{equation}
where \(B_{\mathrm{reg}}(T)\) is the bound in
\eqref{eq:general_cckb_regret_bound}, evaluated using the proxy reward
cumulative-error bound.

\emph{(ii) Proxy constraint violation:}
Suppose, in addition, that Assumption~\ref{ass:cckb_slater} holds for the
proxy problem.  For the fixed horizon \(T\), let \(\Lambda_T^{\mathrm{prx}}\) be a finite
bound on the norm of an optimal dual vector for that problem, choose
\(\rho\ge2\Lambda_T^{\mathrm{prx}}\), and suppose that
\(\Pr(\mathcal E_{\mathrm{sur}}^{\mathrm{prx}})
\ge1-\alpha_{\mathrm{sur}}\).  Then, with probability at least
\(1-\alpha_{\mathrm{ctx}}^f-\alpha_{\mathrm{ctx}}^h
-\alpha_{\mathrm{sur}}\),
\begin{equation}
\mathrm{Vio}^{\mathrm{prx}}_{d,j}(T)
\le B_{\mathrm{vio}}(T)
\label{eq:proxy_violation_bound_expanded_final}
\end{equation}
for every \(d\in\mathcal D\) and \(j\in[J_d]\), where
\(B_{\mathrm{vio}}(T)\) is the bound in
\eqref{eq:general_cckb_violation_bound}, evaluated using both proxy
cumulative-error bounds.
\end{corollary}

\subsection{Finite-Time Transfer to the Relaxed \ac{nsrdp}}
\label{subsubsec:proxy_transfer_guarantees}
The per-round optimality gap is defined in
\eqref{eq:proxy_optimality_gap_definition}.
Lemma~\ref{lem:proxy_feasible_implies_relaxed_feasible} shows that the proxy
feasible set is contained in the relaxed feasible set.  Because the two
problems have the same objective,
\(\Delta_{\mathrm{prx}}(T)\ge0\).
The following corollary uses \(\Delta_{\mathrm{prx}}(T)\) to relate the
relaxed and proxy regret and violation metrics for the same realized
request--action sequence.

\begin{corollary}[Finite-time transfer to the relaxed \ac{nsrdp}]
\label{cor:relaxed_regret_violation}
\emph{(i) Regret:}
On the event in part~\emph{(i)} of
Corollary~\ref{cor:proxy_regret_violation},
\begin{align}
\mathrm{Reg}^{\mathrm{rel}}(T)
&=\mathrm{Reg}^{\mathrm{prx}}(T)+T\Delta_{\mathrm{prx}}(T)
\notag\\
&\le B_{\mathrm{reg}}(T)+T\Delta_{\mathrm{prx}}(T).
\label{eq:relaxed_regret_transfer}
\end{align}
Here, \(\mathrm{Reg}^{\mathrm{prx}}(T)\) is the cumulative mean-reward
difference from the optimal proxy policy, whereas
\(T\Delta_{\mathrm{prx}}(T)\) is the optimal-value difference caused by
using the proxy problem.

\emph{(ii) Constraint violation:}
On the event in part~\emph{(ii)} of
Corollary~\ref{cor:proxy_regret_violation},
\begin{equation}
\mathrm{Vio}^{\mathrm{rel}}_{d,j}(T)
\le\mathrm{Vio}^{\mathrm{prx}}_{d,j}(T)
\le B_{\mathrm{vio}}(T),
\label{eq:relaxed_violation_transfer}
\end{equation}
for every \(d\in\mathcal D\) and \(j\in[J_d]\).
\end{corollary}

\begin{proof}
The regret identity follows directly from
\eqref{eq:policy_regret}, \eqref{eq:policy_proxy_regret}, and
\eqref{eq:proxy_optimality_gap_definition}.  For the violation bound,
\eqref{eq:pointwise_domination_gorig_g} gives
\[
\sum_{t=1}^T h_{d,j}(s_t,\mathbf x_t)
\le
\sum_{t=1}^T h^{\mathrm{prx}}_{d,j}(s_t,\mathbf x_t).
\]
Applying the nondecreasing function \([\,\cdot\,]_+\) and then
the corresponding part of
Corollary~\ref{cor:proxy_regret_violation} proves the claim.
\end{proof}

\subsection{Explicit Surrogate Errors and Performance Bounds}
\label{appendix:explicit_proxy_bounds}
Corollary~\ref{cor:explicit_finite_time_proxy_performance}, the final result
of this subsection, gives concrete finite-time bounds on the proxy regret,
its transfer to \(\mathrm{Reg}^{\mathrm{rel}}(T)\), and
\(\mathrm{Vio}^{\mathrm{rel}}_{d,j}(T)\) for proxy-based \ac{cckb} under
the \ac{gp} assumptions in Sec.~\ref{subsec:theory}.  To establish this result, the following
proposition
shows that the reward and proxy-constraint surrogates satisfy the conditions in
\eqref{eq:cckb_surrogate_direction} and
\eqref{eq:cckb_cumulative_surrogate_error} with high probability and gives
explicit finite-time bounds on the cumulative surrogate errors.

\begin{proposition}[Concrete bounds on cumulative surrogate errors]
\label{prop:proxy_surrogate_error_control}
Fix a horizon \(T\) and failure probabilities
\(\alpha_f,\alpha_h\in(0,1)\) whose sum is less than one.
Suppose that Assumptions~\ref{ass:inter_round_stationarity},
\ref{ass:cckb_kernel_regular},~\ref{ass:reward_gp_regular},
and~\ref{ass:success_only_gp_regular} hold.
Define
\begin{equation}
\label{eq:reward_surrogate_error_scale}
\Phi_f(T;\alpha_f)
:=\sqrt{T\gamma_T^f}\cdot\Bigl(B^f+\widetilde\sigma^f
\sqrt{\gamma_T^f+1+\log\tfrac{1}{\alpha_f}}\Bigr).
\end{equation}
If Algorithm~\ref{alg:component1} uses the exploration widths in
\eqref{eq:concrete_reward_exploration_width} and
\eqref{eq:concrete_proxy_demand_exploration_width}, then, with probability
at least \(1-\alpha_f-\alpha_h\), the proxy versions of the one-sided bounds
in \eqref{eq:cckb_surrogate_direction} hold and
\begin{equation*}
\mathcal W_f(T)
=\mathcal O\!\left(\Phi_f(T;\alpha_f)\right).
\end{equation*}

Suppose, in addition, that Assumption~\ref{ass:valuable_decomposition} holds
and \(M_T\ge\rho\).  Fix \(\alpha_w\in(0,1)\) such that
\(\alpha_f+\alpha_h+\alpha_w<1\).  For each resource \((d,j)\), define
\begin{equation}
\begin{aligned}
&\Psi_{d,j}(T;\alpha_f,\alpha_h,\alpha_w)\\
&:=\frac{1}{a_T}\Bigg[
\Phi_f(T;\alpha_f)\\
&\quad+
\frac{\bar p}{C_{d,j}^{\max}}
\left(
\sqrt{
\frac{T\gamma_T^{m_{d,j}^{\Gamma}}}
{\log\!\left(1+\eta_{m_{d,j}^{\Gamma}}^{-1}\right)}}
+\sqrt{T\log\frac{J_{\mathrm{tot}}}{\alpha_w}}
\right)\\
&\qquad\times
\left(
B^{m_{d,j}^{\Gamma}}
+\widetilde\sigma_{d,j}
\sqrt{\gamma_T^{m_{d,j}^{\Gamma}}+1
+\log\frac{J_{\mathrm{tot}}}{\alpha_h}}
\right)
\Bigg].
\end{aligned}
\label{eq:proxy_constraint_surrogate_error_scale}
\end{equation}
Let
\(
\boldsymbol{\Psi}_h(T;\alpha_f,\alpha_h,\alpha_w)
:=
(\Psi_{d,j}(T;\alpha_f,\alpha_h,\alpha_w))_
{d\in\mathcal D,\;j\in[J_d]}.
\)
Then, with probability at least \(1-\alpha_f-\alpha_h-\alpha_w\),
\(\mathcal E_{\mathrm{sur}}^{\mathrm{prx}}\) holds, the preceding bound on
\(\mathcal W_f(T)\) remains valid, and
\begin{equation*}
\bigl[\boldsymbol{\mathcal W}_h(T)\bigr]_{d,j}
=\mathcal O\!\left(
\Psi_{d,j}(T;\alpha_f,\alpha_h,\alpha_w)\right)
\end{equation*}
for every resource \((d,j)\).
\end{proposition}

The proof is provided in
\suppsecref{appendix:proxy_surrogate_control_proof}.

The following proposition bounds the optimality gap introduced by
conservative proxying.
\begin{proposition}[Proxy optimality gap]
\label{prop:proxy_optimality_gap}
Suppose that Assumptions~\ref{ass:cckb_kernel_regular},
\ref{ass:cckb_request_measurability}, and~\ref{ass:valuable_decomposition}
hold, the relaxed \ac{nsrdp} admits an optimal policy, the proxy problem
admits an optimal policy, and Assumption~\ref{ass:cckb_slater} holds for
\(\mathbf h^{\mathrm{prx}}\).  Under these conditions, the proxy problem and
the relaxed \ac{nsrdp} admit optimal dual vectors.  Let
\(\boldsymbol\phi_{\mathrm{prx}}^\star\) and
\(\boldsymbol\phi_{\mathrm{rel}}^\star\) denote such vectors, and, for the fixed
horizon \(T\), let finite bounds \(\Lambda_T^{\mathrm{prx}}\) and
\(\Lambda_T^{\mathrm{rel}}\) satisfy
\(
\|\boldsymbol\phi_{\mathrm{prx}}^\star\|_1\le\Lambda_T^{\mathrm{prx}},
\|\boldsymbol\phi_{\mathrm{rel}}^\star\|_1\le\Lambda_T^{\mathrm{rel}}.
\)
If \(\Lambda_T^{\mathrm{rel}}\le M_T\), then
\begin{equation}
0\le\Delta_{\mathrm{prx}}(T)
\le
\frac{1}{T}
\sum_{d\in\mathcal D}\sum_{j\in[J_d]}
\phi_{\mathrm{prx},d,j}^{\star}
\left(\frac{2\bar p}{a_T}-1\right).
\label{eq:proxy_optimality_gap}
\end{equation}
Consequently,
\begin{equation}
\Delta_{\mathrm{prx}}(T)
\le
\frac{\Lambda_T^{\mathrm{prx}}}{T}
\left(\frac{2\bar p}{a_T}-1\right).
\label{eq:proxy_optimality_gap_norm_bound}
\end{equation}
\end{proposition}

The proof is provided in
\suppsecref{appendix:proxy_optimality_gap_proof}.

\begin{corollary}[Concrete finite-time bounds for proxy-based CCKB]
\label{cor:explicit_finite_time_proxy_performance}
Fix \(\alpha_f,\alpha_h,\alpha_{\mathrm{ctx}}^f,
\alpha_{\mathrm{ctx}}^h\in(0,1)\) such that
\(\alpha_{\mathrm{ctx}}^f+\alpha_{\mathrm{ctx}}^h
+\alpha_f+\alpha_h<1\).
Suppose that Assumptions~\ref{ass:inter_round_stationarity},
\ref{ass:independent_request_arrivals},~\ref{ass:cckb_kernel_regular},
\ref{ass:cckb_request_measurability},~\ref{ass:reward_gp_regular},
and~\ref{ass:success_only_gp_regular} hold and that the proxy problem
admits an optimal policy.  Choose \(\rho>0\), set
\(V:=\sqrt{J_{\mathrm{tot}}T}/\rho\), and use the exploration widths in
\eqref{eq:concrete_reward_exploration_width} and
\eqref{eq:concrete_proxy_demand_exploration_width}.

\emph{(i) Proxy regret:}
With probability at least
\(1-\alpha_{\mathrm{ctx}}^f-\alpha_{\mathrm{ctx}}^h
-\alpha_f-\alpha_h\),
\begin{align}
\mathrm{Reg}^{\mathrm{prx}}(T)
&\le B_{\mathrm{prx}}(T)\notag\\
&:=\mathcal O\left(
\begin{aligned}[c]
&\bar p\sqrt{T\log\tfrac{1}{\alpha_{\mathrm{ctx}}^f}}
+\rho J_{\mathrm{tot}}\sqrt{T\log\tfrac{1}{\alpha_{\mathrm{ctx}}^h}}\\
&{}+\Phi_f(T;\alpha_f)+\rho\sqrt{J_{\mathrm{tot}}T}
\end{aligned}
\right).
\label{eq:explicit_gp_proxy_regret_bound}
\end{align}

\emph{(ii) Relaxed \ac{nsrdp} regret:}
If the relaxed \ac{nsrdp} also admits an optimal policy, then, on the same
event,
\begin{align}
\mathrm{Reg}^{\mathrm{rel}}(T)
&=\mathrm{Reg}^{\mathrm{prx}}(T)+T\Delta_{\mathrm{prx}}(T)\notag\\
&\le B_{\mathrm{prx}}(T)+T\Delta_{\mathrm{prx}}(T).
\label{eq:explicit_gp_relaxed_regret_transfer}
\end{align}
If, in addition, the conditions of
Proposition~\ref{prop:proxy_optimality_gap} hold, then
\begin{equation}
\mathrm{Reg}^{\mathrm{rel}}(T)
\le B_{\mathrm{prx}}(T)
+\Lambda_T^{\mathrm{prx}}\left(\frac{2\bar p}{a_T}-1\right).
\label{eq:explicit_gp_relaxed_regret_bound}
\end{equation}

\emph{(iii) Constraint violation:}
Suppose, in addition, that
Assumption~\ref{ass:valuable_decomposition} holds with \(M_T\ge\rho\), that
Assumption~\ref{ass:cckb_slater} holds for the proxy constraints, and that,
for the fixed horizon \(T\), let
\(\boldsymbol\phi_{\mathrm{prx}}^\star\) be an optimal proxy dual vector and
\(\Lambda_T^{\mathrm{prx}}\) a finite bound satisfying
\(\|\boldsymbol\phi_{\mathrm{prx}}^\star\|_1\le\Lambda_T^{\mathrm{prx}}<\infty\).
Choose \(\rho\ge2\Lambda_T^{\mathrm{prx}}\), and fix
\(\alpha_w\in(0,1)\) such that
\(\alpha_{\mathrm{ctx}}^f+\alpha_{\mathrm{ctx}}^h
+\alpha_f+\alpha_h+\alpha_w<1\).  Then, with probability at least
\(1-\alpha_{\mathrm{ctx}}^f-\alpha_{\mathrm{ctx}}^h
-\alpha_f-\alpha_h-\alpha_w\),
\begin{equation}
\mathrm{Vio}^{\mathrm{rel}}_{d,j}(T)
\le\mathcal O\left(
\begin{aligned}[c]
&\frac{\Phi_f(T;\alpha_f)}{\rho}
+J_{\mathrm{tot}}\sqrt{T\log\tfrac{1}{\alpha_{\mathrm{ctx}}^h}}\\
&{}+\left\|\boldsymbol{\Psi}_h
(T;\alpha_f,\alpha_h,\alpha_w)\right\|_1
+\sqrt{J_{\mathrm{tot}}T}
\end{aligned}
\right)
\label{eq:explicit_gp_relaxed_violation_bound}
\end{equation}
for every resource \((d,j)\).
\end{corollary}

For fixed \(\rho\), if \(\Phi_f(T;\alpha_f)=o(T)\), the proxy-regret bound in
\eqref{eq:explicit_gp_proxy_regret_bound} is sublinear in \(T\).  The relaxed
regret in \eqref{eq:explicit_gp_relaxed_regret_transfer} retains the separate
term \(T\Delta_{\mathrm{prx}}(T)\).

The proof is provided in
\suppsecref{appendix:concrete_finite_time_guarantee_proof}.

\end{document}

  \documentclass[src/main.tex]{subfiles}

\begin{document}

\section{Concentration Bounds for the Realized Request Sequence}
\label{appendix:context_concentration_proof}
This section states and proves two bounds that compare quantities evaluated on the realized request sequence with their expectations under the request distribution \(\mathbb P\).
Both proofs use standard concentration inequalities and follow the same three steps:
\textit{(i)} express the target difference as a centered sum,
\textit{(ii)} bound the range of each summand, and
\textit{(iii)} apply a concentration inequality.
For a fixed contextual policy, Proposition~\ref{prop:context_concentration_fixed_q} applies Hoeffding's inequality because, under Assumption~\ref{ass:independent_request_arrivals}, the reward terms form an i.i.d.\ sum.
Proposition~\ref{prop:context_concentration_fixed_q_constraint} instead applies the Azuma--Hoeffding inequality because the weights on the resource-consumption terms may depend on earlier rounds.

\begin{proposition}[Reward concentration over requests]
\label{prop:context_concentration_fixed_q}
For any \(q\in\mathcal Q\), define
\[
\Delta_{\mathrm{ctx}}^{f}(q,T)
:=
\sum_{t=1}^{T}\left(
\mathbb{E}_{s\sim\mathbb{P}}[v_q^f(s)]-v_q^f(s_t)
\right).
\]
Then, for any \(\delta\in(0,1)\), with probability at least \(1-\delta\),
\[
\left|\Delta_{\mathrm{ctx}}^{f}(q,T)\right|
\le
\bar p\sqrt{\frac{T}{2}\log\frac{2}{\delta}}.
\]
\end{proposition}

\begin{proof}
Fix any \(q\in\mathcal Q\), and let
\(
\mu_q^f:=\mathbb E_{s\sim\mathbb P}\!\left[v_q^f(s)\right].
\)
Define
\[
D_t^f:=\mu_q^f-v_q^f(s_t).
\]

\textit{(i) Centered i.i.d.\ sum.}
By Assumption~\ref{ass:independent_request_arrivals}, \(s_1,\dots,s_T\) are i.i.d.\ from \(\mathbb P\), so
\(D_1^f,\dots,D_T^f\) are i.i.d.\ and \(\mathbb E[D_t^f]=0\).
Moreover,
\[
\Delta_{\mathrm{ctx}}^f(q,T)=\sum_{t=1}^T D_t^f.
\]

\textit{(ii) Range bound.}
Since
\(
v_q^f(s)=\mathbb E_{\mathbf x\sim q(\cdot\mid s)}[f(s,\mathbf x)]
\)
and \(0\le f(s,\mathbf x)\le \bar p\), we have
\[
v_q^f(s_t)\in[0,\bar p]
\quad\Rightarrow\quad
D_t^f\in[\mu_q^f-\bar p,\mu_q^f].
\]
Hence each \(D_t^f\) has range width at most \(\bar p\).

\textit{(iii) Hoeffding.}
Applying Hoeffding's inequality~\cite{hoeffding1963probability}, for any \(\varepsilon>0\),
\[
\Pr\!\left(
\left|
\sum_{t=1}^{T}D_t^f
\right|
\ge \varepsilon
\right)
\le
2\exp\!\left(
-\frac{2\varepsilon^2}{T\bar p^2}
\right).
\]
Setting
\(
\varepsilon:=\bar p\sqrt{\frac{T}{2}\log\frac{2}{\delta}}
\)
yields
\[
\Pr\!\left(
\left|
\Delta_{\mathrm{ctx}}^f(q,T)
\right|
\le
\bar p\sqrt{\frac{T}{2}\log\frac{2}{\delta}}
\right)
\ge 1-\delta.
\]
\end{proof}

\begin{proposition}[Resource-consumption concentration over requests]
\label{prop:context_concentration_fixed_q_constraint}
For any \(q\in\mathcal Q\) and constant \(B\ge0\), let
\(\{\boldsymbol{\psi}_t\}_{t=1}^{T}\subseteq\mathbb R_+^{J_{\mathrm{tot}}}\)
be a sequence such that \(\boldsymbol\psi_t\) is measurable with respect to
the history available before \(s_t\) is observed and
\(\|\boldsymbol\psi_t\|_1\le B\) almost surely for every \(t\).  Define
\[
\Delta_{\mathrm{ctx}}^{h}(q,\boldsymbol{\psi},T)
:=
\sum_{t=1}^{T}\left\langle
\boldsymbol\psi_t,
\mathbf v_q^h(s_t)-\mathbb{E}_{s\sim\mathbb{P}}[\mathbf v_q^h(s)]
\right\rangle.
\]
Then, for any \(\delta\in(0,1)\), with probability at least \(1-\delta\),
\[
\left|\Delta_{\mathrm{ctx}}^{h}(q,\boldsymbol\psi,T)\right|
\le
B\sqrt{\frac{T}{2}\log\frac{2}{\delta}}.
\]
\end{proposition}

\begin{proof}
Fix any \(q\in\mathcal Q\), a constant \(B\ge0\), and a sequence
\(\{\boldsymbol\psi_t\}_{t=1}^T\subseteq\mathbb R_+^{J_{\mathrm{tot}}}\)
satisfying the conditions in Proposition~\ref{prop:context_concentration_fixed_q_constraint}.
Define
\(
\boldsymbol{\mu}_q^h:=\mathbb E_{s\sim\mathbb P}[\mathbf v_q^h(s)]\in\mathbb R^{J_{\mathrm{tot}}}.
\)
Let
\(
\mathcal F_t:=\sigma(\mathcal H_t)
\)
for \(t=0,\ldots,T\) be the natural filtration generated by the history in
\eqref{eq:nsrdp_history}.
Define
\[
D_t^h
:=
\left\langle
\boldsymbol{\psi}_t,
\mathbf v_q^h(s_t)-\boldsymbol{\mu}_q^h
\right\rangle.
\]

\textit{(i) Centered martingale-difference sum.}
By assumption, \(\boldsymbol\psi_t\) is \(\mathcal F_{t-1}\)-measurable.
By Assumption~\ref{ass:independent_request_arrivals} and the causal round protocol, \(s_t\sim\mathbb P\) is independent of \(\mathcal F_{t-1}\), so
\[
\mathbb E\!\left[
\mathbf v_q^h(s_t)\mid\mathcal F_{t-1}
\right]
=
\boldsymbol{\mu}_q^h.
\]
The \(\mathcal F_{t-1}\)-measurability of \(\boldsymbol\psi_t\) then gives
\[
\mathbb E\!\left[D_t^h\mid \mathcal F_{t-1}\right]
=
\left\langle
\boldsymbol{\psi}_t,
\boldsymbol{\mu}_q^h-\boldsymbol{\mu}_q^h
\right\rangle
=0.
\]
Thus, \(\{D_t^h\}_{t=1}^{T}\) is a martingale-difference sequence with
respect to \(\{\mathcal F_t\}_{t=0}^{T}\), and
\[
\Delta_{\mathrm{ctx}}^h(q,\boldsymbol\psi,T)
=
\sum_{t=1}^T D_t^h.
\]

\textit{(ii) Conditional range bound.}
For each coordinate \((d,j)\), by \(-\frac{1}{T}\le h_{d,j}(s,\mathbf x)\le 1-\frac{1}{T}\), we have
\[
v_q^{h_{d,j}}(s_t)
\in
\left[-\frac{1}{T},1-\frac{1}{T}\right].
\]
Because every coordinate of \(\boldsymbol\psi_t\) is nonnegative,
conditional on \(\mathcal F_{t-1}\),
\[
\begin{aligned}
D_t^h\in\biggl[
&-\frac{1}{T}\|\boldsymbol\psi_t\|_1
-\left\langle\boldsymbol\psi_t,\boldsymbol\mu_q^h\right\rangle,\\
&\left(1-\frac{1}{T}\right)\|\boldsymbol\psi_t\|_1
-\left\langle\boldsymbol\psi_t,\boldsymbol\mu_q^h\right\rangle
\biggr],
\end{aligned}
\]
so the conditional range width of \(D_t^h\) is at most
\(\|\boldsymbol\psi_t\|_1\le B\).

\textit{(iii) Azuma--Hoeffding.}
If \(B=0\), then \(\boldsymbol\psi_t=\mathbf0\) almost surely for every \(t\), so
\(\sum_{t=1}^T D_t^h=0\) and the bound is immediate.
Otherwise, applying the Azuma--Hoeffding
inequality for martingale differences with bounded conditional
ranges~\cite{azuma1967weighted,hoeffding1963probability}, for any
\(\varepsilon>0\),
\[
\Pr\!\left(
\left|
\sum_{t=1}^{T}D_t^h
\right|
\ge
\varepsilon
\right)
\le
2\exp\!\left(
-\frac{2\varepsilon^2}
{TB^2}
\right).
\]
Setting
\(
\varepsilon
:=
B\sqrt{\frac{T}{2}\log\frac{2}{\delta}}
\)
yields
\[
\Pr\!\left(
\left|
\Delta_{\mathrm{ctx}}^h(q,\boldsymbol\psi,T)
\right|
\le
B\sqrt{\frac{T}{2}\log\frac{2}{\delta}}
\right)
\ge
1-\delta.
\]
\end{proof}

\ifSubfilesClassLoaded{
  \bibliographystyle{ieee/IEEEtran}
  \bibliography{bib/references}
}{}

\end{document}

  \documentclass[src/main.tex]{subfiles}

\begin{document}

\section{Success-Weighted Posterior Width under Selective Constraint Updates}
\label{app:selective_width}
This section states and proves the success-weighted posterior-width result
used to bound the cumulative constraint-surrogate error.
We first state the result and then present the proof idea and formal proof.

\begin{lemma}[Success-weighted posterior-width bound under selective updates]
\label{lem:selective_width}
Fix a resource \((d,j)\).  Suppose that Assumptions~\ref{ass:inter_round_stationarity},
\ref{ass:cckb_kernel_regular}, and~\ref{ass:success_only_gp_regular}
hold.  Define
\[
\omega_{t,d,j}
:=
\begin{cases}
\sigma_{t-1}^{m_{d,j}^{\Gamma}}(z_t),
& \mathbf1_{d,j}^{\Gamma}(s_t)=1,\\
0,
& \mathbf1_{d,j}^{\Gamma}(s_t)=0.
\end{cases}
\]
Thus, the posterior standard deviation is evaluated only for inputs satisfying
\(\mathbf1_{d,j}^{\Gamma}(s_t)=1\).  Then, for every \(\delta\in(0,1)\), with
probability at least \(1-\delta\),
\[
\sum_{t=1}^{T}p(s_t,\mathbf x_t)\omega_{t,d,j}
\le
\sqrt{
\frac{2T\gamma_T^{m_{d,j}^{\Gamma}}}
{\log\!\left(1+\eta_{m_{d,j}^{\Gamma}}^{-1}\right)}
}
+\sqrt{2T\log\frac{1}{\delta}}.
\]
\end{lemma}

\noindent\textbf{Proof idea.}
A standard information-gain argument bounds the cumulative posterior width
over the rounds in which the constraint model is updated, namely, those
satisfying \(Y_t=1\) and \(\mathbf1_{d,j}^{\Gamma}(s_t)=1\).  The desired
bound weights the width on every on-path round by its provisioning-success
probability.  A martingale concentration argument transfers the standard
bound from the model-update rounds to this success-probability-weighted sum.

\begin{proof}[Proof of Lemma~\ref{lem:selective_width}]
For each round \(t\in[T]\), let \(X_t\) be the posterior standard deviation
counted only when successful provisioning produces an update of the constraint
model for resource \((d,j)\):
\[
X_t
:=
Y_t\omega_{t,d,j}.
\]

\textit{(i) Width on model-update rounds.}
By the definition of \(X_t\),
\[
\sum_{t=1}^{T}X_t
=
\sum_{\substack{t\in[T]\\
Y_t=1,\;\mathbf1_{d,j}^{\Gamma}(s_t)=1}}
\sigma_{t-1}^{m_{d,j}^{\Gamma}}(z_t).
\]
If \(N_T^{m_{d,j}^{\Gamma}}=0\), this sum is zero and the desired bound is
immediate.  Otherwise, enumerate the update rounds as
\(\tau_1<\cdots<\tau_{N_T^{m_{d,j}^{\Gamma}}}\).  For
\(q\in[N_T^{m_{d,j}^{\Gamma}}]\), let
\[
\mathbf Z_q
:=
\left(z_{\tau_1},\ldots,z_{\tau_q}\right),
\qquad
\mathbf Z_0:=\emptyset,
\]
and define
\[
v_q
:=
\left(
\sigma_{\tau_q-1}^{m_{d,j}^{\Gamma}}(z_{\tau_q})
\right)^2,
\qquad
q\in\left[N_T^{m_{d,j}^{\Gamma}}\right].
\]
Because the constraint posterior changes only in update rounds, \(v_q\) is
the posterior variance immediately before the \(q\)-th retained observation.
The Schur-complement determinant identity gives
\begin{align*}
&\det\!\left(
\mathbf I
+\eta_{m_{d,j}^{\Gamma}}^{-1}
\mathbf K_{\mathbf Z_q}^{m_{d,j}^{\Gamma}}
\right)\\
&\quad=
\det\!\left(
\mathbf I
+\eta_{m_{d,j}^{\Gamma}}^{-1}
\mathbf K_{\mathbf Z_{q-1}}^{m_{d,j}^{\Gamma}}
\right)
\left(1+\eta_{m_{d,j}^{\Gamma}}^{-1}v_q\right).
\end{align*}
Multiplying this identity over the update rounds, taking logarithms, and using
the definition of maximum information gain in
\eqref{eq:maximum_information_gain} give
\begin{align*}
\frac{1}{2}
\sum_{q=1}^{N_T^{m_{d,j}^{\Gamma}}}
\log\!\left(
1+\eta_{m_{d,j}^{\Gamma}}^{-1}v_q
\right)
&=
\frac{1}{2}\log\det\!\left(
\mathbf I
+\eta_{m_{d,j}^{\Gamma}}^{-1}
\mathbf K_{\mathbf Z_{N_T^{m_{d,j}^{\Gamma}}}}^{m_{d,j}^{\Gamma}}
\right)\\
&\le
\gamma_{N_T^{m_{d,j}^{\Gamma}}}^{m_{d,j}^{\Gamma}}.
\end{align*}
Assumption~\ref{ass:cckb_kernel_regular} and the fact that posterior
variance does not exceed prior variance imply \(0\le v_q\le1\).  By the
concavity of \(v\mapsto\log(1+\eta^{-1}v)\), for every \(\eta>0\) and
\(v\in[0,1]\),
\[
\log(1+\eta^{-1}v)
\ge
v\log(1+\eta^{-1}).
\]
Applying this inequality with \(\eta=\eta_{m_{d,j}^{\Gamma}}\), summing over
the update rounds, and combining it with the preceding information-gain bound
give
\begin{align*}
\log\!\left(1+\eta_{m_{d,j}^{\Gamma}}^{-1}\right)
\sum_{q=1}^{N_T^{m_{d,j}^{\Gamma}}}v_q
&\le
\sum_{q=1}^{N_T^{m_{d,j}^{\Gamma}}}
\log\!\left(1+\eta_{m_{d,j}^{\Gamma}}^{-1}v_q\right)\\
&\le
2\gamma_{N_T^{m_{d,j}^{\Gamma}}}^{m_{d,j}^{\Gamma}}.
\end{align*}
Therefore,
\begin{equation}
\sum_{q=1}^{N_T^{m_{d,j}^{\Gamma}}}v_q
\le
\frac{
2\gamma_{N_T^{m_{d,j}^{\Gamma}}}^{m_{d,j}^{\Gamma}}
}
{\log\!\left(1+\eta_{m_{d,j}^{\Gamma}}^{-1}\right)}.
\label{eq:update_round_variance_sum}
\end{equation}
To bound the cumulative posterior width over the model-update rounds, namely,
\(\sum_{t=1}^{T}X_t\), we apply the cumulative posterior-width argument
of~\cite[Lemma~4]{chowdhury2017kernelized} to the retained update sequence.
Specifically, we combine the variance-sum bound in
\eqref{eq:update_round_variance_sum} with the Cauchy--Schwarz inequality:
\begin{equation}
\begin{aligned}
\sum_{t=1}^{T}X_t
&=
\sum_{q=1}^{N_T^{m_{d,j}^{\Gamma}}}\sqrt{v_q}\\
&\overset{(a)}{\le}
\sqrt{
N_T^{m_{d,j}^{\Gamma}}
\sum_{q=1}^{N_T^{m_{d,j}^{\Gamma}}}v_q
}\\
&\overset{(b)}{\le}
\sqrt{
\frac{
2N_T^{m_{d,j}^{\Gamma}}
\gamma_{N_T^{m_{d,j}^{\Gamma}}}^{m_{d,j}^{\Gamma}}
}
{\log\!\left(1+\eta_{m_{d,j}^{\Gamma}}^{-1}\right)}
}.
\end{aligned}
\label{eq:width_transfer_replacement_a}
\end{equation}
Here, (a) follows from the Cauchy--Schwarz inequality, and (b) follows from
\eqref{eq:update_round_variance_sum}.

\textit{(ii) Transfer to the success-probability-weighted width.}
Using \(X_t=Y_t\omega_{t,d,j}\), we obtain
\[
\begin{aligned}
&\mathbb E[X_t\mid\mathcal H_{t-1},s_t,\mathbf x_t]\\
&\quad=
\mathbb E[Y_t\omega_{t,d,j}
\mid\mathcal H_{t-1},s_t,\mathbf x_t]\\
&\quad\overset{(a)}{=}
\omega_{t,d,j}
\Pr(Y_t=1\mid\mathcal H_{t-1},s_t,\mathbf x_t)\\
&\quad\overset{(b)}{=}
\omega_{t,d,j}p(s_t,\mathbf x_t).
\end{aligned}
\]
Here, (a) holds because \(\omega_{t,d,j}\) is determined by
\((\mathcal H_{t-1},s_t,\mathbf x_t)\) and \(Y_t\) is binary.
Equality (b) follows from the stationary response law in
Assumption~\ref{ass:inter_round_stationarity}; when
\(\mathbf1_{d,j}^{\Gamma}(s_t)=0\), both sides are zero by the definition of
\(\omega_{t,d,j}\).  Define
\[
D_t
:=
X_t-\mathbb E[X_t\mid\mathcal H_{t-1},s_t,\mathbf x_t].
\]
By definition,
\(\mathbb E[D_t\mid\mathcal H_{t-1},s_t,\mathbf x_t]=0\), which also gives
\(\mathbb E[D_t\mid\mathcal H_{t-1}]=0\).  Thus,
\(\{D_t\}_{t=1}^T\) is a martingale difference sequence.
By definition, \(\omega_{t,d,j}=0\) when
\(\mathbf1_{d,j}^{\Gamma}(s_t)=0\).  Otherwise, the fact that a posterior
variance does not exceed the corresponding prior variance and the kernel
normalization in Assumption~\ref{ass:cckb_kernel_regular} give
\[
0
\le
(\omega_{t,d,j})^2
\le
k_{m_{d,j}^{\Gamma}}(z_t,z_t)
\le 1.
\]
Since \(Y_t\in\{0,1\}\), this gives \(0\le X_t\le1\) and
\(|D_t|\le1\).
The Azuma--Hoeffding inequality~\cite{azuma1967weighted}
therefore implies that, with probability at least \(1-\delta\),
\begin{equation}
\sum_{t=1}^{T}D_t
\ge
-\sqrt{2T\log\frac{1}{\delta}}.
\label{eq:width_transfer_azuma_step}
\end{equation}

Summing the conditional-expectation equality over \(t\), using the definition
of \(D_t\), and applying
\eqref{eq:width_transfer_azuma_step} gives
\[
\begin{aligned}
\sum_{t=1}^{T}p(s_t,\mathbf x_t)\omega_{t,d,j}
&=
\sum_{t=1}^{T}
\mathbb E[X_t\mid\mathcal H_{t-1},s_t,\mathbf x_t]\\
&=
\sum_{t=1}^{T}X_t-\sum_{t=1}^{T}D_t\\
&\le
\sum_{t=1}^{T}X_t
+\sqrt{2T\log\frac{1}{\delta}}.
\end{aligned}
\]
Combining this inequality with \eqref{eq:width_transfer_replacement_a}
yields, with probability at least \(1-\delta\),
\[
\begin{aligned}
\sum_{t=1}^{T}p(s_t,\mathbf x_t)\omega_{t,d,j}
&\le
\sqrt{
\frac{
2N_T^{m_{d,j}^{\Gamma}}
\gamma_{N_T^{m_{d,j}^{\Gamma}}}^{m_{d,j}^{\Gamma}}
}
{\log\!\left(1+\eta_{m_{d,j}^{\Gamma}}^{-1}\right)}
}
+\sqrt{2T\log\frac{1}{\delta}}.
\end{aligned}
\]
Finally, \(N_T^{m_{d,j}^{\Gamma}}\le T\) and the maximum information gain is
nondecreasing in its sample count.  Substituting these two bounds gives the
claim in Lemma~\ref{lem:selective_width}.
\end{proof}

\end{document}

  \documentclass[src/main.tex]{subfiles}

\begin{document}

\section{Proofs of the CCKB and Proxy Guarantees}
\label{appendix:proxy_regret_proof}
This section establishes the theoretical guarantees underlying the general
\ac{cckb} analysis and its proxy instantiation.
We first prove Corollary~\ref{cor:cckb_slater_consequence}, which
establishes strong duality and bounds
\(\|\boldsymbol\phi^\star\|_1\)
(Sec.~\ref{appendix:cckb_slater_consequence}).
We then establish request-wise Lagrangian optimality
(Sec.~\ref{appendix:cckb_requestwise_lagrangian_optimality}).
Afterward, we prove Theorem~\ref{thm:general_cckb_regret_violation}, which gives
finite-time bounds on the relaxed regret \(\mathrm{Reg}^{\mathrm{rel}}(T)\) and
constraint violations \(\mathrm{Vio}^{\mathrm{rel}}_{d,j}(T)\) for general
\ac{cckb}
(Sec.~\ref{appendix:general_cckb_theorem_proof}).
Next, we prove Proposition~\ref{prop:proxy_surrogate_error_control}, which
bounds the cumulative proxy surrogate errors \(\mathcal W_f(T)\) and
\([\boldsymbol{\mathcal W}_h(T)]_{d,j}\)
(Sec.~\ref{appendix:proxy_surrogate_control_proof}).
We then prove Proposition~\ref{prop:proxy_optimality_gap}, which
bounds the optimality gap
\(\mathrm{OPT}^{\mathrm{rel}}-\mathrm{OPT}^{\mathrm{prx}}\)
(Sec.~\ref{appendix:proxy_optimality_gap_proof}).
Finally, we combine these two propositions to prove
Corollary~\ref{cor:explicit_finite_time_proxy_performance}
(Sec.~\ref{appendix:concrete_finite_time_guarantee_proof}).

\subsection{Proof of Corollary~\ref{cor:cckb_slater_consequence}}
\label{appendix:cckb_slater_consequence}

For the proof, write
\[
F(q):=\mathbb E_{s\sim\mathbb P}[v_q^f(s)],
\qquad
\mathbf G(q):=\mathbb E_{s\sim\mathbb P}[\mathbf v_q^h(s)].
\]
Using these quantities, define the Lagrangian and the dual function,
respectively, as
\[
L(q,\boldsymbol\phi)
:=F(q)-\langle\boldsymbol\phi,\mathbf G(q)\rangle,
\qquad
D(\boldsymbol\phi):=\sup_{q\in\mathcal Q}L(q,\boldsymbol\phi).
\]

\begin{proof}
\emph{1) Strong Duality:}
The relaxed benchmark \eqref{eq:relaxed_problem} depends on a policy \(q\)
only through the finite-dimensional vector \((F(q),\mathbf G(q))\).
By their definitions, \(F\) and \(\mathbf G\) are linear in \(q\).
Because \(\mathcal Q\) is closed under policy mixtures, its image under the
linear map \(q\mapsto(F(q),\mathbf G(q))\) is therefore
convex~\cite[Sec.~2.3.2]{boyd2004convex}.
The problem is therefore convex with finitely many inequality constraints.
Applying the standard Slater theorem~\cite[Sec.~5.2.3]{boyd2004convex} under
Assumption~\ref{ass:cckb_slater} gives strong duality and ensures that the
dual optimum is attained.  Let \(q^\star\) be an optimal policy, whose
existence is assumed in Corollary~\ref{cor:cckb_slater_consequence}, and let
\(\boldsymbol\phi^\star\) be an optimal dual vector.  Their optimality and
strong duality give
\[
F(q^\star)=\mathrm{OPT}^{\mathrm{rel}}=D(\boldsymbol\phi^\star).
\]

\emph{2) Finite Norm of an Optimal Dual Vector:}
Substituting the Slater policy \(q^\circ\) into the dual function and using
\(\mathbf G(q^\circ)\le-\xi_T\mathbf1\) give
\[
\begin{aligned}
\mathrm{OPT}^{\mathrm{rel}}
=D(\boldsymbol\phi^\star)
&\ge L(q^\circ,\boldsymbol\phi^\star)\\
&=F(q^\circ)
-\langle\boldsymbol\phi^\star,\mathbf G(q^\circ)\rangle\\
&\ge F(q^\circ)+\xi_T\|\boldsymbol\phi^\star\|_1.
\end{aligned}
\]
Rearranging and using
\(0\le F(q^\circ)\le\mathrm{OPT}^{\mathrm{rel}}\le\bar p\) give
\[
\|\boldsymbol\phi^\star\|_1
\le
\frac{\mathrm{OPT}^{\mathrm{rel}}-F(q^\circ)}{\xi_T}
\le
\frac{\bar p}{\xi_T}
<\infty.
\]
Hence, a finite bound \(\Lambda_T^{\mathrm{rel}}\) satisfying
\(\|\boldsymbol\phi^\star\|_1\le\Lambda_T^{\mathrm{rel}}\) exists.
\end{proof}

For any optimal policy \(q^\star\) and optimal dual vector
\(\boldsymbol\phi^\star\), primal feasibility and strong duality imply
complementary slackness:
\begin{equation}
\left\langle\boldsymbol\phi^\star,
\mathbb E_{s\sim\mathbb P}[\mathbf v_{q^\star}^h(s)]
\right\rangle
=0.
\label{eq:cckb_complementary_slackness}
\end{equation}

\subsection{Request-Wise Lagrangian Optimality}
\label{appendix:cckb_requestwise_lagrangian_optimality}

To derive the upper bound on \(\mathrm{Vio}^{\mathrm{rel}}_{d,j}(T)\), we
compare the benchmark policy \(q^\star\) with the selected action
\(\mathbf x_t\) for the same request in terms of
\(f(s,\mathbf x)-\langle\boldsymbol\phi^\star,\mathbf h(s,\mathbf x)\rangle\).
The following lemma shows that the corresponding conditional value under
\(q^\star\) is no smaller than the value of any action for almost every
request.

\begin{lemma}[Request-wise Lagrangian optimality]
\label{lem:cckb_requestwise_lagrangian_optimality}
Suppose that the relaxed \ac{nsrdp} admits an optimal policy, that
Assumptions~\ref{ass:cckb_kernel_regular}
and~\ref{ass:cckb_request_measurability} hold, and that
Assumption~\ref{ass:cckb_slater} holds.
Let \((q^\star,\boldsymbol\phi^\star)\) be a primal--dual optimal pair.
Then, for \(\mathbb P\)-almost every \(s\) and every
\(\mathbf x\in\mathcal X\),
\begin{equation}
\begin{aligned}
v_{q^\star}^f(s)
-\langle\boldsymbol\phi^\star,\mathbf v_{q^\star}^h(s)\rangle
&=\max_{\mathbf x'\in\mathcal X}
\left\{f(s,\mathbf x')
-\langle\boldsymbol\phi^\star,\mathbf h(s,\mathbf x')\rangle\right\}\\
&\ge f(s,\mathbf x)
-\langle\boldsymbol\phi^\star,\mathbf h(s,\mathbf x)\rangle.
\end{aligned}
\label{eq:cckb_contextwise_lagrangian_maximum}
\end{equation}
\end{lemma}

\begin{proof}
For fixed \(\boldsymbol\phi^\star\), define
\[
\ell_{\boldsymbol\phi^\star}(s,\mathbf x)
:=f(s,\mathbf x)
-\langle\boldsymbol\phi^\star,\mathbf h(s,\mathbf x)\rangle.
\]
Assumptions~\ref{ass:cckb_kernel_regular}
and~\ref{ass:cckb_request_measurability} make this function measurable in
\(s\) and continuous in \(\mathbf x\).  Because \(\mathcal X\) in
\eqref{eq:exact_split_decomposition_space} is compact, a maximizer exists for
every request and can be selected measurably
\cite[Thm.~18.19]{aliprantis2006infinite}.  The corresponding deterministic
policy belongs to \(\mathcal Q\).
For every \(q\in\mathcal Q\),
\[
\begin{aligned}
L(q,\boldsymbol\phi^\star)
&=\mathbb E_{s\sim\mathbb P}\!\left[
\mathbb E_{\mathbf x\sim q(\cdot\mid s)}
[\ell_{\boldsymbol\phi^\star}(s,\mathbf x)]
\right]\\
&\le
\mathbb E_{s\sim\mathbb P}\!\left[
\max_{\mathbf x\in\mathcal X}
\ell_{\boldsymbol\phi^\star}(s,\mathbf x)
\right],
\end{aligned}
\]
where the inequality holds because an average cannot exceed the maximum.
The measurable policy above attains the upper bound; hence,
\[
\sup_{q\in\mathcal Q}L(q,\boldsymbol\phi^\star)
=\mathbb E_{s\sim\mathbb P}
\left[\max_{\mathbf x\in\mathcal X}
\ell_{\boldsymbol\phi^\star}(s,\mathbf x)\right].
\]
By strong duality and complementary slackness in
\eqref{eq:cckb_complementary_slackness}, \(q^\star\) attains the supremum on
the left.  Therefore,
\[
\mathbb E_{s\sim\mathbb P}\!\left[
\max_{\mathbf x\in\mathcal X}
\ell_{\boldsymbol\phi^\star}(s,\mathbf x)
-\mathbb E_{\mathbf x\sim q^\star(\cdot\mid s)}
[\ell_{\boldsymbol\phi^\star}(s,\mathbf x)]
\right]
=0.
\]
The integrand is nonnegative and therefore equals zero for
\(\mathbb P\)-almost every \(s\).  Expanding
\(\ell_{\boldsymbol\phi^\star}\) gives
\eqref{eq:cckb_contextwise_lagrangian_maximum}.
\end{proof}

The same conclusion applies to the proxy problem under the corresponding
assumptions, after replacing \(\mathbf h\) with
\(\mathbf h^{\mathrm{prx}}\).  For a fixed request, the path indicators are
fixed and kernel regularity gives continuity in \(\mathbf x\), while
Assumption~\ref{ass:cckb_request_measurability} gives measurability in \(s\).
Thus, the same measurable-selection argument applies.

\subsection{Proof of Theorem~\ref{thm:general_cckb_regret_violation}}
\label{appendix:general_cckb_theorem_proof}
The proof follows the primal--dual argument used for general \ac{ckb}
guarantees~\cite[Theorem~3.3]{zhou2022kernelized_constraints}.
First, we construct the joint high-probability event used in the analysis
(Sec.~\ref{appendix:cckb_joint_event}).
Second, we establish the \ac{cckb} master inequality
(Sec.~\ref{appendix:cckb_master_inequality}).
Finally, we combine the master inequality with the primal--dual relations to
derive the regret and constraint-violation bounds
(Sec.~\ref{appendix:cckb_final_bounds}).

\subsubsection{Joint High-Probability Event}
\label{appendix:cckb_joint_event}
For \(q\in\mathcal Q\), a nonnegative predictable sequence
\(\boldsymbol\psi=\{\boldsymbol\psi_t\}_{t=1}^T\), and \(B\ge0\), define the
events
\begin{align}
\mathcal E_{\mathrm{ctx}}^f(q,\delta)
&:=\left\{|\Delta_{\mathrm{ctx}}^f(q,T)|
\le\bar p\sqrt{\frac T2\log\frac2\delta}\right\},
\label{eq:def_cckb_context_reward_event}\\
\mathcal E_{\mathrm{ctx}}^h(q,\boldsymbol\psi,B,\delta)
&:=\left\{|\Delta_{\mathrm{ctx}}^h(q,\boldsymbol\psi,T)|
\le B\sqrt{\frac T2\log\frac2\delta}\right\}.
\label{eq:def_cckb_context_constraint_event}
\end{align}
Let \(q^\star\) be an optimal policy.  For the regret bound, set
\[
\begin{aligned}
\mathcal E_{\mathrm{ctx}}^{\mathrm{reg}}
:={}&\mathcal E_{\mathrm{ctx}}^f
(q^\star,\alpha_{\mathrm{ctx}}^f)\\
&\cap\mathcal E_{\mathrm{ctx}}^h
\left(q^\star,\{\boldsymbol\phi_t\}_{t=1}^T,
\rho J_{\mathrm{tot}},\alpha_{\mathrm{ctx}}^h\right).
\end{aligned}
\]
Because \(q^\star\) is fixed,
Proposition~\ref{prop:context_concentration_fixed_q} applies to the first
event.  For the second event, Algorithm~\ref{alg:component1} determines
\(\boldsymbol\phi_t\) before observing \(s_t\) and ensures
\(\|\boldsymbol\phi_t\|_1\le\rho J_{\mathrm{tot}}\), so
Proposition~\ref{prop:context_concentration_fixed_q_constraint} applies.
A union bound gives
\[
\Pr(\mathcal E_{\mathrm{ctx}}^{\mathrm{reg}})
\ge1-\alpha_{\mathrm{ctx}}^f-\alpha_{\mathrm{ctx}}^h.
\]
Thus, the first two concentration events, together with the surrogate event
specified in part~\emph{(i)} of
Theorem~\ref{thm:general_cckb_regret_violation}, are sufficient for the
regret bound.

For the constraint-violation bound, let
\((q^\star,\boldsymbol\phi^\star)\) be the primal--dual pair in
Corollary~\ref{cor:cckb_slater_consequence}, let
\(\{\boldsymbol\phi^\star\}_{t=1}^T\) denote the constant sequence, and set
\[
\begin{aligned}
\mathcal E_{\mathrm{ctx}}^{\mathrm{vio}}
:={}&\mathcal E_{\mathrm{ctx}}^f
(q^\star,\alpha_{\mathrm{ctx}}^f)\\
&\cap\mathcal E_{\mathrm{ctx}}^h
\left(q^\star,\{\boldsymbol\phi_t\}_{t=1}^T,
\rho J_{\mathrm{tot}},\frac{\alpha_{\mathrm{ctx}}^h}{2}\right)\\
&\cap\mathcal E_{\mathrm{ctx}}^h
\left(q^\star,\{\boldsymbol\phi^\star\}_{t=1}^T,
\Lambda_T^{\mathrm{rel}},\frac{\alpha_{\mathrm{ctx}}^h}{2}\right).
\end{aligned}
\]
The constant sequence is predictable, and
\(\|\boldsymbol\phi^\star\|_1\le\Lambda_T^{\mathrm{rel}}\).  Applying the same two
concentration propositions and a union bound over all three events gives
\[
\Pr(\mathcal E_{\mathrm{ctx}}^{\mathrm{vio}})
\ge1-\alpha_{\mathrm{ctx}}^f-\alpha_{\mathrm{ctx}}^h.
\]
Intersecting this event with \(\mathcal E_{\mathrm{sur}}^{\mathrm{vio}}\) gives the
probability stated in part~\emph{(ii)} of the theorem.

\subsubsection{CCKB Master Inequality}
\label{appendix:cckb_master_inequality}
For brevity, write \(z_t=(s_t,\mathbf x_t)\), and define
\begin{equation}
\mathrm{Reg}^{\mathrm{rel}}_{\mathrm{seq}}(T)
:=\sum_{t=1}^T\bigl(v_{q^\star}^f(s_t)-f(z_t)\bigr).
\label{eq:def_cckb_sequential_regret}
\end{equation}

The following lemma is the vector-valued, contextual counterpart of the key
decomposition in the original \ac{ckb}
analysis~\cite[Theorem~3.3]{zhou2022kernelized_constraints}.  The scalar
constraint term is replaced by an inner product with the vector of contextual
constraints, and \(\Delta_{\mathrm{ctx}}^h\) accounts for the realized request
sequence.

\begin{lemma}[CCKB master inequality]
\label{lem:vector_ckb_master}
Set \(V=\sqrt{J_{\mathrm{tot}}T}/\rho\).  If the one-sided bounds
in \eqref{eq:cckb_surrogate_direction} and the reward cumulative-error bound
in \eqref{eq:cckb_cumulative_surrogate_error} hold, then the following
inequality holds for \(\boldsymbol\phi=\mathbf0\), with the term involving
\(\boldsymbol{\mathcal W}_h(T)\) equal to zero.  On
\(\mathcal E_{\mathrm{sur}}^{\mathrm{vio}}\), it holds for every
\(\boldsymbol\phi\in[0,\rho]^{J_{\mathrm{tot}}}\):
\begin{equation}
\begin{split}
&\mathrm{Reg}^{\mathrm{rel}}_{\mathrm{seq}}(T)
+\left\langle\boldsymbol\phi,
\sum_{t=1}^T\mathbf h(z_t)\right\rangle\\
&\le
\mathcal W_f(T)
+\left\langle\boldsymbol\phi,
\boldsymbol{\mathcal W}_h(T)\right\rangle
+\Delta_{\mathrm{ctx}}^h
\bigl(q^\star,\{\boldsymbol\phi_t\}_{t=1}^T,T\bigr)\\
&\quad+\frac{V\|\boldsymbol\phi\|_2^2}{2}
+\frac{\rho\sqrt{J_{\mathrm{tot}}T}}{2}.
\end{split}
\label{eq:general_cckb_master_bound}
\end{equation}
\end{lemma}

\begin{proof}
Fix \(\boldsymbol\phi\in[0,\rho]^{J_{\mathrm{tot}}}\).  We decompose the
left-hand side into four terms that can be controlled by the primal selection
rule, the surrogate-error bounds, the dual update, and the feasibility of
\(q^\star\), respectively:
\begin{equation}
\begin{split}
&\mathrm{Reg}^{\mathrm{rel}}_{\mathrm{seq}}(T)
+\left\langle\boldsymbol\phi,\sum_{t=1}^T\mathbf h(z_t)\right\rangle\\
&=T_1+T_2
+\sum_{t=1}^T\left\langle
\boldsymbol\phi-\boldsymbol\phi_t,\bar{\mathbf h}_t(z_t)\right\rangle
+\sum_{t=1}^T\left\langle
\boldsymbol\phi_t,\mathbf v_{q^\star}^h(s_t)\right\rangle,
\end{split}
\label{eq:general_cckb_master_decomposition}
\end{equation}
where
\[
\begin{split}
T_1:={}&\sum_{t=1}^T
\left(v_{q^\star}^f(s_t)
-\langle\boldsymbol\phi_t,\mathbf v_{q^\star}^h(s_t)\rangle\right)\\
&-\sum_{t=1}^T
\left(\bar f_t(z_t)
-\langle\boldsymbol\phi_t,\bar{\mathbf h}_t(z_t)\rangle\right),\\
T_2:={}&\sum_{t=1}^T\bigl(\bar f_t(z_t)-f(z_t)\bigr)
+\sum_{t=1}^T\left\langle\boldsymbol\phi,
\mathbf h(z_t)-\bar{\mathbf h}_t(z_t)\right\rangle.
\end{split}
\]
The surrogate directions in
\eqref{eq:cckb_surrogate_direction} and the primal maximization in
Algorithm~\ref{alg:component1} give \(T_1\le0\).  On
\(\mathcal E_{\mathrm{sur}}^{\mathrm{vio}}\), the cumulative bounds in
\eqref{eq:cckb_cumulative_surrogate_error} give
\(
T_2\le\mathcal W_f(T)
+\left\langle\boldsymbol\phi,
\boldsymbol{\mathcal W}_h(T)\right\rangle.
\)
For \(\boldsymbol\phi=\mathbf0\), the constraint-surrogate term in \(T_2\)
vanishes, so the same inequality follows from only the reward
cumulative-error bound.
The projected dual update implies
\[
\begin{split}
\left\langle\boldsymbol\phi-\boldsymbol\phi_t,
\bar{\mathbf h}_t(z_t)\right\rangle
\le{}&\frac V2\left(
\|\boldsymbol\phi_t-\boldsymbol\phi\|_2^2
-\|\boldsymbol\phi_{t+1}-\boldsymbol\phi\|_2^2\right)\\
&+\frac{\|\bar{\mathbf h}_t(z_t)\|_2^2}{2V}.
\end{split}
\]
Summation, \(\boldsymbol\phi_1=\mathbf0\), and
\(\|\bar{\mathbf h}_t(z_t)\|_2^2\le J_{\mathrm{tot}}\) bound the corresponding
sum by
\(V\|\boldsymbol\phi\|_2^2/2
+\rho\sqrt{J_{\mathrm{tot}}T}/2\).
For the fourth term, feasibility of \(q^\star\) gives
\(\mathbb E_{s\sim\mathbb P}[\mathbf v_{q^\star}^h(s)]\le\mathbf0\).
Together with \(\boldsymbol\phi_t\ge\mathbf0\), the definition of
\(\Delta_{\mathrm{ctx}}^h\) therefore gives
\[
\begin{aligned}
\sum_{t=1}^T\langle\boldsymbol\phi_t,
\mathbf v_{q^\star}^h(s_t)\rangle
&=\Delta_{\mathrm{ctx}}^h
\bigl(q^\star,\{\boldsymbol\phi_t\}_{t=1}^T,T\bigr)\\
&\quad+
\sum_{t=1}^T\left\langle\boldsymbol\phi_t,
\mathbb E_{s\sim\mathbb P}[\mathbf v_{q^\star}^h(s)]\right\rangle\\
&\le\Delta_{\mathrm{ctx}}^h
\bigl(q^\star,\{\boldsymbol\phi_t\}_{t=1}^T,T\bigr).
\end{aligned}
\]
Substitution into \eqref{eq:general_cckb_master_decomposition} proves the
claim.
\end{proof}

\subsubsection{Regret and Constraint-Violation Bounds}
\label{appendix:cckb_final_bounds}
\begin{proof}[Proof of Theorem~\ref{thm:general_cckb_regret_violation}]
\emph{1) Regret Bound:}
The population regret decomposes as
\begin{equation}
\mathrm{Reg}^{\mathrm{rel}}(T)
=\Delta_{\mathrm{ctx}}^f(q^\star,T)
+\mathrm{Reg}^{\mathrm{rel}}_{\mathrm{seq}}(T).
\label{eq:general_cckb_regret_split}
\end{equation}
On \(\mathcal E_{\mathrm{ctx}}^{\mathrm{reg}}\) and the surrogate
event specified in part~\emph{(i)} of the theorem, set
\(\boldsymbol\phi=\mathbf0\) in
Lemma~\ref{lem:vector_ckb_master}, and combine
\eqref{eq:general_cckb_regret_split} with the reward concentration bound in
\eqref{eq:def_cckb_context_reward_event}, instantiated with
\(q=q^\star\) and \(\delta=\alpha_{\mathrm{ctx}}^f\), and the constraint
concentration bound in \eqref{eq:def_cckb_context_constraint_event},
instantiated with \(q=q^\star\),
\(\boldsymbol\psi=\{\boldsymbol\phi_t\}_{t=1}^T\),
\(B=\rho J_{\mathrm{tot}}\), and
\(\delta=\alpha_{\mathrm{ctx}}^h\).  Because
\(\langle\mathbf0,\boldsymbol{\mathcal W}_h(T)\rangle=0\), this step does not
use the constraint cumulative-error bound.  This proves
\eqref{eq:general_cckb_regret_bound}.

\emph{2) Contextual Correction for the Violation Bound:}
For the violation bound, we reuse the coordinate-isolation argument of
Efroni et al.~\cite{efroni2020exploration}, as adopted in the scalar
\ac{ckb} analysis~\cite[Thm.~3.3]{zhou2022kernelized_constraints}.  Once
\eqref{eq:general_cckb_master_bound} is available, the coordinate choice and
the subsequent algebra are unchanged.  We only need to account for two
differences: the benchmark policy is defined under \(\mathbb P\), whereas
the sequential regret is evaluated on the realized requests, and the master
inequality \eqref{eq:general_cckb_master_bound} contains a vector-valued surrogate error.

We first establish the correction for the realized requests.  Applying
Lemma~\ref{lem:cckb_requestwise_lagrangian_optimality} to \(\mathbf x_t\),
summing
over \(t\), expanding \(\Delta_{\mathrm{ctx}}^h\), and then using
complementary slackness in \eqref{eq:cckb_complementary_slackness} give
\begin{equation}
\begin{aligned}
&\mathrm{Reg}^{\mathrm{rel}}_{\mathrm{seq}}(T)
+\left\langle\boldsymbol\phi^\star,
\sum_{t=1}^T\mathbf h(z_t)\right\rangle\\
&\quad\ge\sum_{t=1}^T
\left\langle\boldsymbol\phi^\star,
\mathbf v_{q^\star}^h(s_t)\right\rangle\\
&\quad\overset{(a)}{=}\Delta_{\mathrm{ctx}}^h
\bigl(q^\star,\{\boldsymbol\phi^\star\}_{t=1}^T,T\bigr)
+T\left\langle\boldsymbol\phi^\star,
\mathbb E_{s\sim\mathbb P}[\mathbf v_{q^\star}^h(s)]
\right\rangle\\
&\quad\overset{(b)}{=}\Delta_{\mathrm{ctx}}^h
\bigl(q^\star,\{\boldsymbol\phi^\star\}_{t=1}^T,T\bigr)\\
&\quad\overset{(c)}{\ge}-\Lambda_T^{\mathrm{rel}}\sqrt{\frac T2
\log\frac4{\alpha_{\mathrm{ctx}}^h}}.
\end{aligned}
\label{eq:cckb_contextual_closing_correction}
\end{equation}
Here, (a) follows from the definition of \(\Delta_{\mathrm{ctx}}^h\),
(b) from complementary slackness in
\eqref{eq:cckb_complementary_slackness}, and (c) from the
constraint-concentration event in
\eqref{eq:def_cckb_context_constraint_event}, instantiated with
\(B=\Lambda_T^{\mathrm{rel}}\) and \(\delta=\alpha_{\mathrm{ctx}}^h/2\).
\eqref{eq:cckb_contextual_closing_correction} is the additional step required
by the contextual setting.

\emph{3) Coordinate Isolation and Violation Bound:}
We now apply the standard coordinate-isolation step.  Fix \((d,j)\), let
\(\mathbf e_{d,j}\) be its standard basis vector, and set
\[
\boldsymbol\phi^{(d,j)}
:=\boldsymbol\phi^\star+\frac\rho2\mathbf e_{d,j}.
\]
Because \(\|\boldsymbol\phi^\star\|_1\le\Lambda_T^{\mathrm{rel}}\le\rho/2\), this vector lies
in \([0,\rho]^{J_{\mathrm{tot}}}\) and can be substituted into
\eqref{eq:general_cckb_master_bound}.  Its left-hand side decomposes as
\begin{equation}
\begin{aligned}
&\mathrm{Reg}^{\mathrm{rel}}_{\mathrm{seq}}(T)
+\left\langle\boldsymbol\phi^{(d,j)},
\sum_{t=1}^T\mathbf h(z_t)\right\rangle\\
&=
\mathrm{Reg}^{\mathrm{rel}}_{\mathrm{seq}}(T)
+\left\langle\boldsymbol\phi^\star,
\sum_{t=1}^T\mathbf h(z_t)\right\rangle
+\frac\rho2\sum_{t=1}^T h_{d,j}(z_t)\\
&\ge\frac\rho2\sum_{t=1}^T h_{d,j}(z_t)
-\Lambda_T^{\mathrm{rel}}\sqrt{\frac T2
\log\frac4{\alpha_{\mathrm{ctx}}^h}}.
\end{aligned}
\label{eq:cckb_coordinate_lhs_lower_bound}
\end{equation}
Here, the final inequality follows from
\eqref{eq:cckb_contextual_closing_correction}.

Moreover,
\(\|\boldsymbol\phi^{(d,j)}\|_2\le\rho\) and each of its coordinates is at
most \(\rho\).  On \(\mathcal E_{\mathrm{ctx}}^{\mathrm{vio}}\), the
constraint-concentration event for \(\{\boldsymbol\phi_t\}_{t=1}^T\) gives
\begin{equation}
\Delta_{\mathrm{ctx}}^h
\bigl(q^\star,\{\boldsymbol\phi_t\}_{t=1}^T,T\bigr)
\le\rho J_{\mathrm{tot}}\sqrt{\frac T2
\log\frac4{\alpha_{\mathrm{ctx}}^h}}.
\label{eq:cckb_algorithmic_dual_concentration_bound}
\end{equation}
Substituting \eqref{eq:cckb_algorithmic_dual_concentration_bound} into the
right-hand side of \eqref{eq:general_cckb_master_bound} and comparing it
with the left-hand-side lower bound in
\eqref{eq:cckb_coordinate_lhs_lower_bound} give
\[
\begin{aligned}
\frac\rho2\sum_{t=1}^T h_{d,j}(z_t)
&\le \mathcal W_f(T)
+\left\langle\boldsymbol\phi^{(d,j)},
\boldsymbol{\mathcal W}_h(T)\right\rangle\\
&\quad+\rho J_{\mathrm{tot}}\sqrt{\frac T2
\log\frac4{\alpha_{\mathrm{ctx}}^h}}\\
&\quad+\frac{V\|\boldsymbol\phi^{(d,j)}\|_2^2}{2}
+\frac{\rho\sqrt{J_{\mathrm{tot}}T}}{2}\\
&\quad+\Lambda_T^{\mathrm{rel}}\sqrt{\frac T2
\log\frac4{\alpha_{\mathrm{ctx}}^h}}.
\end{aligned}
\]
Taking the positive part, multiplying by \(2/\rho\), and using
\(V=\sqrt{J_{\mathrm{tot}}T}/\rho\) and \(\Lambda_T^{\mathrm{rel}}/\rho\le1/2\) yield
\[
\begin{aligned}
\left[\sum_{t=1}^T h_{d,j}(z_t)\right]_+
&\le
\underbrace{\frac{2\mathcal W_f(T)}{\rho}}
_{\mathcal O(\mathcal W_f(T)/\rho)}\\
&\quad+
\underbrace{\frac2\rho
\left\langle\boldsymbol\phi^{(d,j)},
\boldsymbol{\mathcal W}_h(T)\right\rangle}
_{\mathcal O(\|\boldsymbol{\mathcal W}_h(T)\|_1)}\\
&\quad+
\underbrace{2J_{\mathrm{tot}}\sqrt{\frac T2
\log\frac4{\alpha_{\mathrm{ctx}}^h}}}
_{\mathcal O(J_{\mathrm{tot}}\sqrt{T\log(1/\alpha_{\mathrm{ctx}}^h)})}\\
&\quad+
\underbrace{\frac{V\|\boldsymbol\phi^{(d,j)}\|_2^2}{\rho}
+\sqrt{J_{\mathrm{tot}}T}}
_{\mathcal O(\sqrt{J_{\mathrm{tot}}T})}\\
&\quad+
\underbrace{\frac{2\Lambda_T^{\mathrm{rel}}}{\rho}\sqrt{\frac T2
\log\frac4{\alpha_{\mathrm{ctx}}^h}}}
_{\mathcal O(\sqrt{T\log(1/\alpha_{\mathrm{ctx}}^h)})}.
\end{aligned}
\]
Since \(J_{\mathrm{tot}}\ge1\), collecting these terms gives
\[
\begin{split}
\left[\sum_{t=1}^T h_{d,j}(z_t)\right]_+
=\mathcal O\!\left(
\frac{\mathcal W_f(T)}{\rho}
+\left\|\boldsymbol{\mathcal W}_h(T)\right\|_1\right.\\
\left.{}+J_{\mathrm{tot}}
\sqrt{T\log\tfrac1{\alpha_{\mathrm{ctx}}^h}}
+\sqrt{J_{\mathrm{tot}}T}
\right),
\end{split}
\]
which proves \eqref{eq:general_cckb_violation_bound}.
\end{proof}

Assumption~\ref{ass:cckb_slater} guarantees that a finite \(\Lambda_T^{\mathrm{rel}}\)
exists for each fixed horizon, as required by the constraint-violation part
of Theorem~\ref{thm:general_cckb_regret_violation}.

\subsection{Proof of
Proposition~\ref{prop:proxy_surrogate_error_control}}
\label{appendix:proxy_surrogate_control_proof}

\begin{proof}[Proof of Proposition~\ref{prop:proxy_surrogate_error_control}]
\emph{1) Confidence Events:}
Applying the standard kernelized-bandit confidence
theorem~\cite{srinivas2012information,chowdhury2017kernelized} to the reward
observations gives an event \(\mathcal E_f\) with probability at least
\(1-\alpha_f\) on which, simultaneously for every \(t\in[T]\) and
\(z\in\mathcal S\times\mathcal X\),
\begin{equation}
|f(z)-\mu_{t-1}^f(z)|
\le\beta_t^f(\alpha_f)\sigma_{t-1}^f(z).
\label{eq:proxy_reward_confidence}
\end{equation}
Apply the same theorem to the retained subsequence of each constraint model,
allocate failure probability \(\alpha_h/J_{\mathrm{tot}}\) to each model,
and take a union bound.  This gives an event \(\mathcal E_h\) with
probability at least \(1-\alpha_h\) on which, simultaneously for every
\(t\in[T]\), resource \((d,j)\), and \(z=(s,\mathbf x)\) satisfying
\(\mathbf1_{d,j}^{\Gamma}(s)=1\),
\begin{equation}
\left|m_{d,j}^{\Gamma}(z)
-\mu_{t-1}^{m_{d,j}^{\Gamma}}(z)\right|
\le
\beta_t^{m_{d,j}^{\Gamma}}
\!\left(\frac{\alpha_h}{J_{\mathrm{tot}}}\right)
\sigma_{t-1}^{m_{d,j}^{\Gamma}}(z).
\label{eq:proxy_constraint_confidence}
\end{equation}
No confidence relation is needed off path:
the zero extensions in \eqref{eq:method_path_gated_mean} and
\eqref{eq:constraint_gp_posterior} make both the proxy constraint and its
surrogate equal \(-1/T\) there.

\emph{2) Pointwise Surrogate Bounds:}
On \(\mathcal E_f\), the definition of the reward \ac{ucb} in
\eqref{eq:reward_ucb} and the reward confidence relation in
\eqref{eq:proxy_reward_confidence} bound the gap between \(\hat f_t\) and
\(f\).  Since \(f(z)\in[0,\bar p]\), clipping
\(\hat f_t\) to this interval preserves its optimistic direction and cannot
increase the gap.  Therefore, the following bound holds uniformly in \(z\)
and, in particular, at \(z_t\):
\begin{equation}
0\le\bar f_t(z_t)-f(z_t)
\le2\beta_t^f(\alpha_f)\sigma_{t-1}^f(z_t).
\label{eq:proxy_reward_pointwise_error}
\end{equation}
On \(\mathcal E_h\), the demand \ac{lcb} in
\eqref{eq:constraint_gp_posterior} and the constraint confidence relation in
\eqref{eq:proxy_constraint_confidence} similarly bound the gap between
\(m_{d,j}^{\Gamma}\) and
\(\widehat m_{t,d,j}^{\Gamma}\).  Substitution into
\eqref{eq:constraint_surrogate}, followed by clipping to the known proxy
constraint range, gives, for every \((d,j)\),
\begin{equation}
0\le h^{\mathrm{prx}}_{d,j}(z_t)-\bar h^{\mathrm{prx}}_{t,d,j}(z_t)
\le
\frac{2\beta_t^{m_{d,j}^{\Gamma}}
(\alpha_h/J_{\mathrm{tot}})}{C_{d,j}^{\max}}
\omega_{t,d,j}.
\label{eq:proxy_constraint_pointwise_error}
\end{equation}

\emph{3) Cumulative Reward-Surrogate Error:}
The cumulative reward-surrogate error is bounded as
\begin{equation}
\begin{aligned}
\sum_{t=1}^T\bigl(\bar f_t(z_t)-f(z_t)\bigr)
&\overset{(a)}{\le}
2\sum_{t=1}^T
\beta_t^f(\alpha_f)\sigma_{t-1}^f(z_t)\\
&\overset{(b)}{\le}
2\beta_T^f(\alpha_f)
\sum_{t=1}^T\sigma_{t-1}^f(z_t)\\
&\overset{(c)}{=}
\mathcal O\!\left(
\beta_T^f(\alpha_f)\sqrt{T\gamma_T^f}
\right)\\
&\overset{(d)}{=}
\mathcal O\!\left(\Phi_f(T;\alpha_f)\right).
\end{aligned}
\label{eq:proxy_reward_cumulative_error}
\end{equation}
Here, (a) follows from \eqref{eq:proxy_reward_pointwise_error}, (b) from
the monotonicity of \(\beta_t^f(\alpha_f)\), (c) from the cumulative
posterior-width bound~\cite[Lemma~4]{chowdhury2017kernelized}, and (d) from
the definitions of \(\beta_T^f(\alpha_f)\) and \(\Phi_f(T;\alpha_f)\).
The event \(\mathcal E_f\cap\mathcal E_h\) has probability at least
\(1-\alpha_f-\alpha_h\), so the pointwise directions and
\eqref{eq:proxy_reward_cumulative_error} prove the first part of the
proposition without Assumption~\ref{ass:valuable_decomposition}.

\emph{4) Cumulative Constraint-Surrogate Error:}
For the remainder of the proof, additionally suppose that
Assumption~\ref{ass:valuable_decomposition} holds and \(M_T\ge\rho\).

\emph{Auxiliary reward-UCB floor:}
We first establish an auxiliary lower bound on the reward UCB \(\bar f_t\).
Define the shifted proxy quantities
\[
g^{\mathrm{prx}}_{d,j}(s,\mathbf x)
:=h^{\mathrm{prx}}_{d,j}(s,\mathbf x)+\frac1T
=\frac{m_{d,j}^{\Gamma}(s,\mathbf x)}{C_{d,j}^{\max}}
\]
and
\begin{equation}
\bar g^{\mathrm{prx}}_{t,d,j}(s,\mathbf x)
:=\bar h^{\mathrm{prx}}_{t,d,j}(s,\mathbf x)+\frac1T.
\label{eq:shifted_proxy_surrogate_definition}
\end{equation}
On \(\mathcal E_f\cap\mathcal E_h\), the pointwise bounds and clipping give
\[
0\le\bar g^{\mathrm{prx}}_{t,d,j}(s,\mathbf x)
\le g^{\mathrm{prx}}_{d,j}(s,\mathbf x)\le1.
\]
For the realized request \(s_t\),
Assumption~\ref{ass:valuable_decomposition} gives a pointwise comparator
\(\mathbf x_t^\circ\in\mathcal X\) such that
\begin{equation}
f(s_t,\mathbf x_t^\circ)
-M_T\sum_{d\in\mathcal D}\sum_{j\in[J_d]}
g^{\mathrm{prx}}_{d,j}(s_t,\mathbf x_t^\circ)
\ge a_T.
\label{eq:valuable_decomposition_pointwise_comparator}
\end{equation}
Using \(\mathbf x_t^\circ\) as a comparator yields
\begin{equation}
\begin{aligned}
\bar f_t(z_t)
&\overset{\mathrm{(a)}}{\ge}
\bar f_t(z_t)
-\left\langle\boldsymbol\phi_t,
\bar{\mathbf g}^{\mathrm{prx}}_t(z_t)\right\rangle\\
&\overset{\mathrm{(b)}}{\ge}
\bar f_t(s_t,\mathbf x_t^\circ)
-\left\langle\boldsymbol\phi_t,
\bar{\mathbf g}^{\mathrm{prx}}_t(s_t,\mathbf x_t^\circ)\right\rangle\\
&\overset{\mathrm{(c)}}{\ge}
f(s_t,\mathbf x_t^\circ)
-\rho\sum_{d\in\mathcal D}\sum_{j\in[J_d]}
g^{\mathrm{prx}}_{d,j}(s_t,\mathbf x_t^\circ)\\
&\overset{\mathrm{(d)}}{\ge}
f(s_t,\mathbf x_t^\circ)
-M_T\sum_{d\in\mathcal D}\sum_{j\in[J_d]}
g^{\mathrm{prx}}_{d,j}(s_t,\mathbf x_t^\circ)\\
&\overset{\mathrm{(e)}}{\ge}a_T.
\end{aligned}
\label{eq:proxy_reward_ucb_floor}
\end{equation}
Here, (a) follows from the nonnegativity of \(\boldsymbol\phi_t\) and
\(\bar{\mathbf g}^{\mathrm{prx}}_t\); (b) follows because, by
\eqref{eq:shifted_proxy_surrogate_definition}, replacing
\(\bar{\mathbf h}^{\mathrm{prx}}_t\) with \(\bar{\mathbf g}^{\mathrm{prx}}_t\)
subtracts the action-independent constant \(\|\boldsymbol\phi_t\|_1/T\) from
the acquisition and hence preserves its maximizer; (c) from reward optimism,
\(\boldsymbol\phi_t\in[0,\rho]^{J_{\mathrm{tot}}}\), and
\(0\le\bar{\mathbf g}^{\mathrm{prx}}_t\le\mathbf g^{\mathrm{prx}}\);
(d) from \(M_T\ge\rho\) and the nonnegativity of
\(\mathbf g^{\mathrm{prx}}\); and (e) from
\eqref{eq:valuable_decomposition_pointwise_comparator}.

\emph{Partition of rounds:}
Fix a resource \((d,j)\).  \eqref{eq:proxy_constraint_pointwise_error}
gives
\[
0\le h^{\mathrm{prx}}_{d,j}(z_t)
-\bar h^{\mathrm{prx}}_{t,d,j}(z_t)
\le\frac{2\beta_t^{m_{d,j}^{\Gamma}}(\alpha_h/J_{\mathrm{tot}})}
{C_{d,j}^{\max}}\omega_{t,d,j}.
\]
Split the rounds into
\[
\mathcal H:=\left\{t:p(s_t,\mathbf x_t)\ge\frac{a_T}{2\bar p}\right\}
\]
and
\[
\mathcal L:=\left\{t:p(s_t,\mathbf x_t)<\frac{a_T}{2\bar p}\right\}.
\]
Accordingly, the cumulative error decomposes into its contributions over
\(\mathcal H\) and \(\mathcal L\), which we bound separately.

\emph{High-success rounds (\(t\in\mathcal H\)):}
Apply Lemma~\ref{lem:selective_width} to every resource with failure
probability \(\alpha_w/J_{\mathrm{tot}}\).  A union bound gives an event
\(\mathcal E_w\) with probability at least \(1-\alpha_w\) on which all these
bounds hold.  On this event, the pointwise bound and the monotonicity of the
exploration parameters in
\eqref{eq:concrete_proxy_demand_exploration_width} give
\[
\begin{aligned}
&\sum_{t\in\mathcal H}
\bigl(h^{\mathrm{prx}}_{d,j}(z_t)
-\bar h^{\mathrm{prx}}_{t,d,j}(z_t)\bigr)\\
&\quad\le
\frac{2\beta_T^{m_{d,j}^{\Gamma}}
(\alpha_h/J_{\mathrm{tot}})}{C_{d,j}^{\max}}
\sum_{t\in\mathcal H}\omega_{t,d,j}\\
&\quad\le
\frac{4\bar p\beta_T^{m_{d,j}^{\Gamma}}
(\alpha_h/J_{\mathrm{tot}})}
{a_TC_{d,j}^{\max}}
\sum_{t\in\mathcal H}
p(s_t,\mathbf x_t)\omega_{t,d,j}\\
&\quad\le
\frac{4\bar p\beta_T^{m_{d,j}^{\Gamma}}
(\alpha_h/J_{\mathrm{tot}})}
{a_TC_{d,j}^{\max}}
\sum_{t=1}^T p(s_t,\mathbf x_t)\omega_{t,d,j}\\
&\quad=\mathcal O\!\Biggl(
\frac{\bar p}{a_TC_{d,j}^{\max}}
\left(
B^{m_{d,j}^{\Gamma}}
+\widetilde\sigma_{d,j}
\sqrt{\gamma_T^{m_{d,j}^{\Gamma}}+1
+\log\frac{J_{\mathrm{tot}}}{\alpha_h}}
\right)\\
&\hspace{18mm}\times\Biggl[
\sqrt{\frac{T\gamma_T^{m_{d,j}^{\Gamma}}}
{\log(1+\eta_{m_{d,j}^{\Gamma}}^{-1})}}+\sqrt{T\log\frac{J_{\mathrm{tot}}}{\alpha_w}}
\Biggr]\Biggr).
\end{aligned}
\]
The second inequality uses
\(1\le 2\bar p\,p(s_t,\mathbf x_t)/a_T\) for \(t\in\mathcal H\);
the third uses the nonnegativity of \(p(s_t,\mathbf x_t)\omega_{t,d,j}\);
and the final equality follows from
\eqref{eq:concrete_proxy_demand_exploration_width} and
Lemma~\ref{lem:selective_width}.

\emph{Low-success rounds (\(t\in\mathcal L\)):}
For every \(t\in\mathcal L\), the definition of \(\mathcal L\) gives
\[
\begin{aligned}
f(z_t)
&=\kappa_{\mathrm{price}}(s_t)p(s_t,\mathbf x_t)\\
&\le\bar p\,p(s_t,\mathbf x_t)<\frac{a_T}{2}.
\end{aligned}
\]
Combining this inequality with \eqref{eq:proxy_reward_ucb_floor} gives
\[
\bar f_t(z_t)-f(z_t)\ge\frac{a_T}{2}.
\]
Because the proxy target and surrogate both lie in
\([-1/T,1-1/T]\), their pointwise difference satisfies
\[
0\le h^{\mathrm{prx}}_{d,j}(z_t)
-\bar h^{\mathrm{prx}}_{t,d,j}(z_t)
\le1
\le\frac{2}{a_T}\bigl(\bar f_t(z_t)-f(z_t)\bigr).
\]
Therefore,
\[
\sum_{t\in\mathcal L}
\bigl(h^{\mathrm{prx}}_{d,j}(z_t)
-\bar h^{\mathrm{prx}}_{t,d,j}(z_t)\bigr)
\le\frac{2}{a_T}\mathcal W_f(T)
=\mathcal O\!\left(\frac{\Phi_f(T;\alpha_f)}{a_T}\right).
\]
Combining the two parts yields
\[
\bigl[\boldsymbol{\mathcal W}_h(T)\bigr]_{d,j}
=\mathcal O\!\left(
\Psi_{d,j}(T;\alpha_f,\alpha_h,\alpha_w)\right)
\]
simultaneously for all resources.

\emph{5) Joint Event:}
The joint event \(\mathcal E_f\cap\mathcal E_h\cap\mathcal E_w\) has
probability at least
\(1-\alpha_f-\alpha_h-\alpha_w\).  Consequently, under the proxy substitutions
\(h_{d,j}\leftarrow h^{\mathrm{prx}}_{d,j}\) and
\(\bar h_{t,d,j}\leftarrow\bar h^{\mathrm{prx}}_{t,d,j}\), the directional
conditions in \eqref{eq:cckb_surrogate_direction} and the cumulative-error
conditions in \eqref{eq:cckb_cumulative_surrogate_error} hold.  Thus,
\(\mathcal E_{\mathrm{sur}}^{\mathrm{prx}}\) holds on this joint event, which proves the
proposition.
\end{proof}

\subsection{Proof of Proposition~\ref{prop:proxy_optimality_gap}}
\label{appendix:proxy_optimality_gap_proof}
\begin{proof}[Proof of Proposition~\ref{prop:proxy_optimality_gap}]
\emph{1) Dual Existence and Norm Bounds:}
We first establish the existence of the optimal dual vectors
\(\boldsymbol\phi_{\mathrm{prx}}^\star\) and
\(\boldsymbol\phi_{\mathrm{rel}}^\star\), together with their finite norm
bounds \(\Lambda_T^{\mathrm{prx}}\) and \(\Lambda_T^{\mathrm{rel}}\).
Corollary~\ref{cor:cckb_slater_consequence} applies directly to the proxy
problem and yields an optimal dual vector
\(\boldsymbol\phi_{\mathrm{prx}}^\star\) with finite norm bound \(\Lambda_T^{\mathrm{prx}}\).
Let \(q^\circ\) and \(\xi_T\) denote the Slater policy and margin for the
proxy constraints, respectively.  Because
\(\mathbf h\le\mathbf h^{\mathrm{prx}}\),
\[
\begin{aligned}
\mathbb E_{s\sim\mathbb P}[\mathbf v_{q^\circ}^h(s)]
\le\mathbb E_{s\sim\mathbb P}
[\mathbf v_{q^\circ}^{h^{\mathrm{prx}}}(s)]\le-\xi_T\mathbf1.
\end{aligned}
\]
Hence, \(q^\circ\) is also a Slater policy for the relaxed constraints with
the same margin.  The corollary therefore also applies to the relaxed
\ac{nsrdp} and yields an optimal dual vector
\(\boldsymbol\phi_{\mathrm{rel}}^\star\) with finite norm bound
\(\Lambda_T^{\mathrm{rel}}\).

\emph{2) Support Restriction:}
We first show that an optimal relaxed policy assigns probability only to
decompositions satisfying
\(p(s,\mathbf x)\ge a_T/(2\bar p)\).
\eqref{eq:method_success_upper_bound} and the definitions of
\(h_{d,j}\) and \(h^{\mathrm{prx}}_{d,j}\) give
\[
0\le h_{d,j}(s,\mathbf x)+\frac1T
\le h^{\mathrm{prx}}_{d,j}(s,\mathbf x)+\frac1T.
\]
Moreover,
\(\|\boldsymbol\phi_{\mathrm{rel}}^\star\|_\infty
\le\|\boldsymbol\phi_{\mathrm{rel}}^\star\|_1
\le\Lambda_T^{\mathrm{rel}}\le M_T\).
Hence, Assumption~\ref{ass:valuable_decomposition} permits the same
shift-and-comparator argument used to establish
\eqref{eq:proxy_reward_ucb_floor}, with
\((\bar f_t,\bar{\mathbf h}^{\mathrm{prx}}_t,\boldsymbol\phi_t)\) replaced by
\((f,\mathbf h,\boldsymbol\phi_{\mathrm{rel}}^\star)\).  Thus, every
request-wise relaxed-Lagrangian maximizer \(\mathbf x\) satisfies
\(f(s,\mathbf x)\ge a_T\) for \(\mathbb P\)-almost every \(s\).
Since \(f(s,\mathbf x)\le\bar p\,p(s,\mathbf x)\), every such maximizer satisfies
\(p(s,\mathbf x)\ge a_T/\bar p\).
Lemma~\ref{lem:cckb_requestwise_lagrangian_optimality} states that an optimal
relaxed policy \(q_{\mathrm{rel}}^\star\) can put probability mass only on
request-wise Lagrangian maximizers.  Therefore,
\begin{equation}
p(s,\mathbf x)\ge\frac{a_T}{\bar p}
\ge\frac{a_T}{2\bar p}
\label{eq:relaxed_optimal_support_success}
\end{equation}
for \(\mathbb P\)-almost every \(s\) and
\(q_{\mathrm{rel}}^\star(\cdot\mid s)\)-almost every \(\mathbf x\).

\emph{3) Proxy Constraint Bound:}
We next bound the proxy-constraint values of the optimal relaxed policy.
Fix a resource \((d,j)\).  On the support of
\(q_{\mathrm{rel}}^\star\),
\eqref{eq:method_consumption_factorization} and
\eqref{eq:relaxed_optimal_support_success} give
\begin{equation}
\begin{aligned}
h^{\mathrm{prx}}_{d,j}(s,\mathbf x)
&=\frac{c_{d,j}(s,\mathbf x)}
{p(s,\mathbf x)C_{d,j}^{\max}}-\frac1T\\
&\overset{\mathrm{(a)}}{\le}\frac{2\bar p}{a_T}
\frac{c_{d,j}(s,\mathbf x)}{C_{d,j}^{\max}}-\frac1T\\
&=\frac{2\bar p}{a_T}
\left(\frac{c_{d,j}(s,\mathbf x)}{C_{d,j}^{\max}}-\frac1T\right)
+\frac1T\left(\frac{2\bar p}{a_T}-1\right)\\
&\overset{\mathrm{(b)}}{=}\frac{2\bar p}{a_T}h_{d,j}(s,\mathbf x)
+\frac1T\left(\frac{2\bar p}{a_T}-1\right).
\end{aligned}
\label{eq:proxy_gap_normalized_bound}
\end{equation}
Here, (a) follows from \eqref{eq:relaxed_optimal_support_success}, and (b)
follows from the definition of \(h_{d,j}\) in
\eqref{eq:method_exact_normalized_mean}.
Taking expectations under \(q_{\mathrm{rel}}^\star\) gives
\begin{equation}
\begin{aligned}
&\mathbb E_{s\sim\mathbb P}
\mathbb E_{\mathbf x\sim q_{\mathrm{rel}}^\star(\cdot\mid s)}
[h^{\mathrm{prx}}_{d,j}(s,\mathbf x)]\\
&\quad\overset{\mathrm{(a)}}{\le}
\frac{2\bar p}{a_T}
\mathbb E_{s\sim\mathbb P}
\mathbb E_{\mathbf x\sim q_{\mathrm{rel}}^\star(\cdot\mid s)}
[h_{d,j}(s,\mathbf x)]
+\frac1T\left(\frac{2\bar p}{a_T}-1\right)\\
&\quad\overset{\mathrm{(b)}}{\le}
\frac1T\left(\frac{2\bar p}{a_T}-1\right).
\end{aligned}
\label{eq:proxy_gap_policy_constraint_bound}
\end{equation}
Here, (a) follows from \eqref{eq:proxy_gap_normalized_bound}, and (b)
follows from the feasibility of \(q_{\mathrm{rel}}^\star\) for the relaxed
\ac{nsrdp}.

\emph{4) Optimality-Gap Bound:}
We finally convert the preceding proxy-constraint bound into an
optimality-gap bound through the proxy dual problem.
For \(q\in\mathcal Q\), define the proxy constraint vector by
\[
\bigl[\mathbf G^{\mathrm{prx}}(q)\bigr]_{d,j}
:=
\mathbb E_{s\sim\mathbb P}
\mathbb E_{\mathbf x\sim q(\cdot\mid s)}
[h^{\mathrm{prx}}_{d,j}(s,\mathbf x)].
\]
It follows that
\[
\begin{aligned}
\mathrm{OPT}^{\mathrm{prx}}
&\overset{\mathrm{(a)}}{=}
\sup_{q\in\mathcal Q}
\left\{F(q)-\langle\boldsymbol\phi_{\mathrm{prx}}^\star,
\mathbf G^{\mathrm{prx}}(q)\rangle\right\}\\
&\overset{\mathrm{(b)}}{\ge}
F(q_{\mathrm{rel}}^\star)
-\langle\boldsymbol\phi_{\mathrm{prx}}^\star,
\mathbf G^{\mathrm{prx}}(q_{\mathrm{rel}}^\star)\rangle\\
&\overset{\mathrm{(c)}}{=}
\mathrm{OPT}^{\mathrm{rel}}
-\langle\boldsymbol\phi_{\mathrm{prx}}^\star,
\mathbf G^{\mathrm{prx}}(q_{\mathrm{rel}}^\star)\rangle.
\end{aligned}
\]
Here, (a) follows from strong duality and the dual optimality of
\(\boldsymbol\phi_{\mathrm{prx}}^\star\), (b) follows because
\(q_{\mathrm{rel}}^\star\in\mathcal Q\), and (c) follows from the
optimality of \(q_{\mathrm{rel}}^\star\) for the relaxed \ac{nsrdp} and the
identical objectives of the relaxed and proxy problems.  Rearranging this
inequality and applying \eqref{eq:proxy_gap_policy_constraint_bound}
componentwise give
\[
\begin{aligned}
\mathrm{OPT}^{\mathrm{rel}}
&\le\mathrm{OPT}^{\mathrm{prx}}
+\langle\boldsymbol\phi_{\mathrm{prx}}^\star,
\mathbf G^{\mathrm{prx}}(q_{\mathrm{rel}}^\star)\rangle\\
&\overset{\mathrm{(d)}}{\le}\mathrm{OPT}^{\mathrm{prx}}
+\frac1T\sum_{d\in\mathcal D}\sum_{j\in[J_d]}
\phi_{\mathrm{prx},d,j}^\star
\left(\frac{2\bar p}{a_T}-1\right).
\end{aligned}
\]
Here, (d) also uses
\(\boldsymbol\phi_{\mathrm{prx}}^\star\ge\mathbf0\).
Subtracting \(\mathrm{OPT}^{\mathrm{prx}}\) proves the upper bound in
\eqref{eq:proxy_optimality_gap}.  The lower bound follows from
Lemma~\ref{lem:proxy_feasible_implies_relaxed_feasible}.  Finally,
\eqref{eq:proxy_optimality_gap_norm_bound} follows from
\(\|\boldsymbol\phi_{\mathrm{prx}}^\star\|_1\le\Lambda_T^{\mathrm{prx}}\) and
\(\boldsymbol\phi_{\mathrm{prx}}^\star\ge\mathbf0\).
\end{proof}

\subsection{Proof of
Corollary~\ref{cor:explicit_finite_time_proxy_performance}}
\label{appendix:concrete_finite_time_guarantee_proof}

\begin{proof}
\emph{1) Explicit Finite-Time Bounds:}

\emph{(i) Proxy Regret:}
The first part of
Proposition~\ref{prop:proxy_surrogate_error_control} gives the proxy
one-sided bounds and
\[
\mathcal W_f(T)
=\mathcal O\left(\Phi_f(T;\alpha_f)\right)
\]
with failure probability at most \(\alpha_f+\alpha_h\).
Substitution into part~\emph{(i)} of
Corollary~\ref{cor:proxy_regret_violation} proves
\eqref{eq:explicit_gp_proxy_regret_bound}.

\emph{(ii) Relaxed \ac{nsrdp} Regret:}
Part~\emph{(i)} of
Corollary~\ref{cor:relaxed_regret_violation} then gives
\eqref{eq:explicit_gp_relaxed_regret_transfer}.  Under its additional
conditions, Proposition~\ref{prop:proxy_optimality_gap} gives
\[
T\Delta_{\mathrm{prx}}(T)
\le
\Lambda_T^{\mathrm{prx}}
\left(\frac{2\bar p}{a_T}-1\right),
\]
which proves \eqref{eq:explicit_gp_relaxed_regret_bound}.

\emph{(iii) Constraint Violation:}
Under the additional conditions in part~\emph{(iii)}, the second part of
Proposition~\ref{prop:proxy_surrogate_error_control} also gives
\[
\left\|\boldsymbol{\mathcal W}_h(T)\right\|_1
=\mathcal O\left(
\left\|\boldsymbol{\Psi}_h
(T;\alpha_f,\alpha_h,\alpha_w)\right\|_1
\right)
\]
with joint failure probability at most
\(\alpha_f+\alpha_h+\alpha_w\).  Substitution into part~\emph{(ii)} of
Corollary~\ref{cor:proxy_regret_violation}, followed by part~\emph{(ii)} of
Corollary~\ref{cor:relaxed_regret_violation}, proves
\eqref{eq:explicit_gp_relaxed_violation_bound}.  Combining each surrogate
event with its two request-concentration events gives the stated
probabilities.

\emph{2) Simplified Rates:}
For the rates in Corollary~\ref{cor:concrete_finite_time_proxy_performance},
hold the kernel, norm, noise, and regularization parameters fixed.
For every nonzero modeled kernel, monotonicity gives
\(\gamma_T^\bullet\ge\gamma_1^\bullet>0\), so
\eqref{eq:reward_surrogate_error_scale} and
\eqref{eq:proxy_constraint_surrogate_error_scale} give
\begin{align*}
\Phi_f(T;\alpha_f)
&=\widetilde{\mathcal O}(\gamma_T^f\sqrt T),\\
\Psi_{d,j}(T;\alpha_f,\alpha_h,\alpha_w)
&=\widetilde{\mathcal O}\!\left(
\frac{\sqrt T}{a_T}\left[
\gamma_T^f+\frac{\bar p}{C_{d,j}^{\max}}\gamma_T^{m_{d,j}^{\Gamma}}
\right]\right).
\end{align*}
An identically zero kernel models only the zero function and has zero
posterior width, so its surrogate-error contribution vanishes directly.
Substituting the rate for \(\Phi_f\) into
\eqref{eq:explicit_gp_proxy_regret_bound} gives
\eqref{eq:concrete_gp_proxy_regret_bound}.
For constraint violation, summing the per-resource rates for \(\Psi_{d,j}\)
and using \(\sqrt{J_{\mathrm{tot}}}\le J_{\mathrm{tot}}\) in
\eqref{eq:explicit_gp_relaxed_violation_bound} gives
\eqref{eq:concrete_gp_relaxed_violation_bound}.
The regret-transfer identity and the proxy-optimality-gap term are unaffected
by this simplification, which gives
\eqref{eq:concrete_gp_relaxed_regret_transfer} and
\eqref{eq:concrete_gp_relaxed_regret_bound}.
\end{proof}

\end{document}

  \documentclass[src/main.tex]{subfiles}

\begin{document}

\section{Experimental Network Slice Request Instantiation}
\label{appendix:experimental_nsr_instantiation}
Table~\ref{tab:slice_parameter_basis} provides the complete parameter mapping and class-conditional
values for the experimental setup in Sec.~\ref{subsubsec:nsr_setup}.

\begin{table*}[t]
\caption{Explicit mapping from experimental parameters to \ac{nsr} tuple elements and decompositions}
\label{tab:slice_parameter_basis}
\centering
\small
\renewcommand{\arraystretch}{1.1} 
\setlength{\tabcolsep}{4pt} 
\begin{tabularx}{\textwidth}{@{}>{\raggedright\arraybackslash}m{0.036\textwidth}>{\raggedright\arraybackslash}m{0.29\textwidth}>{\raggedright\arraybackslash}m{0.156\textwidth}>{\raggedright\arraybackslash}m{0.156\textwidth}>{\raggedright\arraybackslash}m{0.3\textwidth}@{}}
\toprule
\textbf{NSR} & \textbf{Instantiation Rule} & \textbf{\ac{urllc} (5QI 86)} & \textbf{\ac{embb} (5QI 6)} & \textbf{Decomposition} \\
\midrule
\(A\) & Active-cell ratio \(\rho_A\) & \(\rho_A=0.25\) & \(\rho_A=0.75\) & No decomposition. \\
\midrule
\multirow{2}{*}{\(\boldsymbol{\theta}\)} &
  Packet-arrival rate per \acs{gnb} (packets/s) & 20,000 & 15,000 &
  \multirow{2}{*}{\parbox[c]{0.3\textwidth}{No decomposition.}} \\
 & Packet-size mean \(m_B^{(1)}\) (bits) & 1,600 & 9,600 & \\
\midrule
\multirow{6}{*}{\(\mathbf{R}\)} &
  \(R_1\): Latency component \newline \(\displaystyle R_1=(\mathcal{N}_t,\{D_t\le \delta_t\})\), where \(D_t\) is the delay random variable (s). &
  \(\delta_t\): target delay; \newline \(\delta_t\in[2,5]~\mathrm{ms}\) &
  \(\delta_t\): target delay; \newline \(\delta_t=0.3~\mathrm{s}\) &
  Choose \(w_{t,d}\) (\(\sum_d w_{t,d}=1\));\newline split delay budget as \(\delta_{t,d}=w_{t,d}\delta_t\). \\[3ex] 
 &
  \(R_2\): Throughput component \newline \(\displaystyle R_2=(\Omega,\{\Theta_t\ge \theta_t\})\), where \(\Theta_t\) is the throughput random variable (Mbit/s). &
  \(\theta_t\): target rate; \newline \(\theta_t\in[10,25]~\mathrm{Mbit/s}\) &
  \(\theta_t\): target rate; \newline \(\theta_t\in[100,120]~\mathrm{Mbit/s}\) &
  Fixed per-domain check \(\Theta_{t,d}\ge\theta_t\); \newline no event decomposition parameter. \\[3ex] 
 &
  \(R_3\): Non-drop component \newline \(\displaystyle R_3=(\Omega,\mathcal{N}_t)\) (fixed template) &
  No class-specific \newline parameter &
  No class-specific \newline parameter &
  No event decomposition. \\
\midrule
\multirow{5}{*}{\(\mathbf{g}\)} &
  \(g_1\): Guarantee for latency component \(R_1\) & \(g_1 = 99.99\%\) & \(g_1 = 98\%\) & Choose \(\eta_{t,d,1}\) (\(\sum_d\eta_{t,d,1}=1\)); \newline split as \(g_{t,d,1}=g_{t,1}^{\eta_{t,d,1}}\). \\[2ex]
 &
  \(g_2\): Guarantee for throughput component \(R_2\) & \(g_2\in[99,99.9]\%\) & \(g_2\in[95,98]\%\) & Choose \(\eta_{t,d,2}\) (\(\sum_d\eta_{t,d,2}=1\)); \newline split as \(g_{t,d,2}=g_{t,2}^{\eta_{t,d,2}}\). \\[2ex]
 &
  \(g_3\): Guarantee for non-drop component \(R_3\) & \(g_3 = 99.999\%\) & \(g_3 = 99.9999\%\) & Choose \(\eta_{t,d,3}\) (\(\sum_d\eta_{t,d,3}=1\)); \newline split as \(g_{t,d,3}=g_{t,3}^{\eta_{t,d,3}}\). \\
\bottomrule
\end{tabularx}
\end{table*}

\end{document}

  \documentclass[src/main.tex]{subfiles}

\begin{document}

\section{Implementation Details for Comparison Methods}
\label{appendix:comparison_shared_components}
For the joint request--decomposition input
\(z=(s_t,\mathbf{x})\in\mathcal{S}\times\mathcal{X}\), this section details
the construction of the standardized kernel input \(\tilde{\mathbf{u}}(z)\)
used by the \acp{gp} underlying the raw reward estimate \(\hat f_t\) and the
raw resource-wise constraint estimates \(\widehat h_{t,d,j}\)
(Sec.~\ref{appendix:comparison_feature_map}).
It also specifies how the constraint outputs are processed and which
observation-noise models are used to fit these \acp{gp}
(Sec.~\ref{appendix:comparison_noise_models}) in the comparison setup
summarized in Sec.~\ref{subsec:comparison_methods}.
\supplementonly{The exact kernels used by the comparison methods are defined in
Sec.~\ref{appendix:comparison_kernel_definitions}.}

\subsection{Feature Map for Decomposed SLA Requirements}
\label{appendix:comparison_feature_map}
The kernel input \(\tilde{\mathbf{u}}(z)\) is constructed in two stages:
\[
z
\overset{\text{(i)}}{\longmapsto}
\mathbf{u}^{\mathrm{raw}}(z)
\overset{\text{(ii)}}{\longmapsto}
\tilde{\mathbf{u}}(z).
\]
Stage~(i) transforms the decomposed \ac{sla} requirements so that larger
values consistently represent stricter requirements, despite differences in
their original units and directions, and concatenates the transformed values
into \(\mathbf{u}^{\mathrm{raw}}(z)\).
Stage~(ii) standardizes every dimension using
the current training inputs of each \ac{gp}, producing
\(\tilde{\mathbf{u}}(z)\) whose dimensions have comparable numerical scales
for kernel evaluation.

\textbf{(i) Fixed Input Transformation:}
In the experiment of Sec.~\ref{subsubsec:exp_decomposition_rules},
\(\mathbf{x}=(\mathbf{w},\boldsymbol{\eta}_1,\boldsymbol{\eta}_2,\boldsymbol{\eta}_3)\), where
\(\mathbf{w},\boldsymbol{\eta}_1,\boldsymbol{\eta}_2,\boldsymbol{\eta}_3
\in\Delta^{|\mathcal{D}|-1}\).
For the numerical values expressed in the units listed in
Table~\ref{tab:slice_parameter_basis}, we transform the decomposed requirements
so that a larger value represents a stricter domain-level \ac{sla} requirement:
\[
\begin{aligned}
\textnormal{(latency)}\quad
&-\log(\delta_t w_d),\\
\textnormal{(throughput)}\quad
&\log(\theta_t),\\
\textnormal{(guarantee)}\quad
&-\log(1-g_{t,i}^{\eta_{d,i}}),
\qquad i\in\{1,2,3\}.
\end{aligned}
\]
These transforms assign larger values to stricter requirements and are finite
over the parameter ranges in Table~\ref{tab:slice_parameter_basis}.

We concatenate the latency and throughput components over
\(d\in\mathcal{D}\), and the guarantee components over
\((d,i)\in\mathcal{D}\times\{1,2,3\}\), to form the unstandardized kernel
feature vector \(\mathbf{u}^{\mathrm{raw}}(z)\).
The coverage set \(A_t\) and traffic profile \(\boldsymbol{\theta}_t\) are not
included in this input transformation.
Their use in constraint-\ac{gp} training and prediction is described in
Sec.~\ref{appendix:comparison_noise_models}.

\textbf{(ii) Standardization Using Training Data:}
At every scheduled refit, each reward or constraint \ac{gp} computes the mean
and standard deviation of each dimension of
\(\mathbf{u}^{\mathrm{raw}}(z)\) over its current training inputs.
The \ac{gp} uses these statistics to standardize the raw input vector \(\mathbf{u}^{\mathrm{raw}}(z)\).
The resulting vector is denoted by \(\tilde{\mathbf{u}}(z)\).
The statistics remain fixed until the next scheduled refit and are applied to
both the training inputs and candidate inputs.
Dimensions that are constant in the current training inputs are mapped to zero.

  \ifsupplementversion
    \input{src/sections/variants/app_comparison_kernel_definitions_supplement}%
  \fi

\subsection{Output Processing and Observation-Noise Models for Surrogate Fitting}
\label{appendix:comparison_noise_models}
For \(z_t=(s_t,\mathbf{x}_t)\), both the reward and constraint \acp{gp} use
Gaussian likelihoods.
The reward observation is modeled as
\begin{equation}
\begin{aligned}
\textnormal{(reward)}\quad
&W_t
=
f(z_t)+\varepsilon_t^f.
\end{aligned}
\label{eq:comparison_reward_observation}
\end{equation}
For the constraint observations, define the aggregate mean offered traffic
rate of the round-\(t\) request as
\(L_t:=|A_t|\lambda_t m_{B,t}^{(1)}\), where \(|A_t|\) is the number of
covered \acp{gnb}, \(\lambda_t\) is the packet-arrival rate per \ac{gnb}, and
\(m_{B,t}^{(1)}\) is the mean packet size.
This quantity is used to express the observed resource consumption per unit of
offered traffic rate before fitting the constraint \acp{gp}.
Define the traffic-normalized observation and its corresponding mean as
\begin{equation}
V_{t,d,j}
:=
\frac{U_{t,d,j}}{L_t},
\qquad
\mu^V_{d,j}(z_t)
:=
\frac{m^\Gamma_{d,j}(z_t)}{L_t}.
\label{eq:comparison_traffic_normalized_observation}
\end{equation}
For the successful, on-path rounds included in the constraint training set,
the observation model is
\begin{equation}
\begin{aligned}
\textnormal{(constraint)}\quad
&V_{t,d,j}
=
\mu^V_{d,j}(z_t)+\varepsilon_{t,d,j}^{V}.
\end{aligned}
\label{eq:comparison_constraint_observation}
\end{equation}

\textbf{Noise Models and Fitting:}
The reward \ac{gp} uses one noise variance shared by all observations
(homoscedastic noise), whereas each constraint \ac{gp} uses a noise variance
estimated separately for each observation (heteroscedastic noise), as detailed
below.
Since Gaussian random variables are sub-Gaussian, the homoscedastic and
heteroscedastic Gaussian likelihood models are consistent with the noise
conditions in Assumptions~\ref{ass:reward_gp_regular}
and~\ref{ass:success_only_gp_regular}.

We next describe how the reward and constraint \acp{gp} model observation
noise and fit their parameters.
At each scheduled refit, the kernel and likelihood parameters are fitted by
maximizing the marginal likelihood.
In the descriptions below, \(\tau\in[t-1]\) indexes a training round for the
reward or constraint model used in round \(t\).

\subsubsection{Reward Estimate (Homoscedastic Noise)}
\label{appendix:comparison_reward_noise_model}
For the reward \ac{gp} in Sec.~\ref{subsubsec:gp_reward_model}, the fitting
model for \eqref{eq:comparison_reward_observation} is
\begin{equation}
\varepsilon_\tau^f\sim\mathcal{N}(0,v^f),
\end{equation}
where \(v^f\) is one noise variance fitted and shared by all reward
observations.

\subsubsection{Constraint Estimates (Traffic Normalization and Heteroscedastic Noise)}
\label{appendix:comparison_constraint_noise_model}
The constraint estimates combine two operations with distinct purposes.
\textit{(i)} Traffic normalization removes the increase in resource consumption caused by
a higher offered traffic rate, whereas \textit{(ii)} the heteroscedastic noise model
represents the residual variation that remains after normalization.
The two operations are described below.

\textbf{(i) Traffic Normalization and Output Rescaling:}
Resource demand depends not only on the decomposed \ac{sla} requirements but
also directly on the coverage set \(A_t\) and traffic profile
\(\boldsymbol{\theta}_t\).
For example, increasing the number of covered \acp{gnb} in \(A_t\) or the
packet-arrival rate specified by \(\boldsymbol{\theta}_t\) increases the
aggregate traffic that the slice must carry, even when the decomposed
requirements remain fixed.
The resulting increase in traffic raises the demand on the \ac{an}, \ac{tn},
and \ac{cn} resources that carry that traffic.

In our experimental implementation of the constraint \acp{gp}, we represent
this dominant traffic-volume effect through the aggregate offered traffic rate
\(L_t=|A_t|\lambda_t m_{B,t}^{(1)}\) and model the total demand as scaling
with \(L_t\).
Because this known multiplicative structure can be encoded directly, we do
not include \(A_t\) or \(\boldsymbol{\theta}_t\) in the kernel input.
This choice avoids requiring each resource-wise \ac{gp} to infer the same
traffic-volume scaling from a limited number of observations and instead
focuses kernel learning on the relationship between the decomposed
\ac{sla} requirements and the resource demand per unit of offered traffic.

The fitting and prediction procedures are summarized below.
\begin{description}[itemsep=0.4em]
\item[\textbf{(Fitting):}]
For resource \((d,j)\), fitting uses the traffic-normalized observations
\(V_{\tau,d,j}\).
Before fitting, observations above the empirical 99th percentile of the
training set are replaced by that percentile to prevent a small number of
unusually large values from dominating the fit.
After this preprocessing, the \ac{gp} is fitted, yielding a posterior mean and
standard deviation on the \(V_{t,d,j}\) scale.

\item[\textbf{(Prediction):}]
At prediction, the posterior mean and standard deviation on the
\(V_{t,d,j}\) scale are multiplied by \(L_t/C_{d,j}^{\max}\).
This rescaling converts the per-unit demand into the request's total resource
demand expressed as a fraction of the resource capacity.
The rescaled posterior quantities are then used to form a lower confidence
bound.
\end{description}

\textbf{(ii) Heteroscedastic Noise Estimation:}
Each constraint \ac{gp} models the residual in
\eqref{eq:comparison_constraint_observation} using a separate variance
\(v_\tau\) for each training round \(\tau\):
\begin{equation}
\varepsilon_{\tau,d,j}^{V}
\sim
\mathcal{N}(0,v_\tau).
\label{eq:comparison_constraint_heteroscedastic_likelihood}
\end{equation}
The observation-specific variances allow the amount of unexplained variation
in resource-consumption observations to differ across request--decomposition
pairs; the single variance \(v^f\) used by the reward \ac{gp} would not
represent these differences.
To estimate \(v_\tau\) in
\eqref{eq:comparison_constraint_heteroscedastic_likelihood}, we first fit a constraint \ac{gp} with one shared noise variance to
\(V_{\tau,d,j}\).
Using the mean \(\hat V^{(0)}_{d,j}(z_\tau)\) predicted by the initial \ac{gp}
at each training input, we compute
\begin{equation}
r_\tau
:=
V_{\tau,d,j}-\hat V^{(0)}_{d,j}(z_\tau),
\qquad
\xi_\tau^{v}
:=
\log\!\left(r_\tau^2+\epsilon_{\log}\right).
\end{equation}
Here, \(\epsilon_{\log}>0\) is a stabilizing offset that keeps the logarithmic
target finite when \(r_\tau=0\).
A second \ac{gp} with one shared noise variance is fitted to \(\xi_\tau^{v}\).
The mean \(\hat\xi_\tau^{v}\) predicted by the second \ac{gp} determines
\begin{equation}
v_\tau
=
\operatorname{clip}_{[v_{\min},\,v_{\max}]}
\!\left(\exp(\hat\xi_\tau^{v})\right).
\end{equation}
Here, \(0<v_{\min}<v_{\max}\) are the lower and upper clipping bounds on the
observation-specific noise variance, respectively.
In all experiments, we set
\(\epsilon_{\log}=10^{-8}\), \(v_{\min}=2\times10^{-3}\), and
\(v_{\max}=1\).
The constraint \ac{gp} is then refitted once using \(v_\tau\) as the fixed
noise variance for training round \(\tau\).


\end{document}

  \documentclass[src/main.tex]{subfiles}

\begin{document}

\section{The Performance Evaluation Model}
\label{appendix:performance_evaluation_model}
\label{appendix:performance_pipeline_overview}
This section specifies the experimental implementation of the
domain-specific controllers.
The experimental performance model evaluates downlink user-plane traffic from
the selected \acp{upf} to the active \acp{gnb}.
Accordingly, each selected route is oriented from CN through TN to AN and
includes the downlink radio-side output resource at its \ac{gnb} endpoint,
which abstracts the final transmission to the served terminal.
Given the path-induced resource set \(\Gamma_d(s_t)\) and the delegated request
\(s_d(s_t,\mathbf{x}_t)\), each controller works backward from the \ac{sla}
targets specified in \(s_d(s_t,\mathbf{x}_t)\) through an M/M/1/\(K\) queueing
model to determine the required capacity of each resource in
\(\Gamma_d(s_t)\).
As specified in Sec.~\ref{subsubsec:nsr_setup}, the experiments instantiate
\(M=3\), with \(R_{t,1}\), \(R_{t,2}\), and \(R_{t,3}\) representing conditional
latency, throughput, and non-drop, respectively.
The simulator performs this calculation from CN through TN to AN, passes the
output traffic rate of each resource to the next resource, and admits the
request only when every traversed resource is feasible.

In this section,
we first describe how the simulator passes traffic through the resources from
CN to AN and admits a request only if every required resource can support it
(Sec.~\ref{appendix:performance_sequential_composition}).
We then explain how a controller determines the capacity required at one
resource and checks it against the available capacity
(Sec.~\ref{appendix:performance_resource_level_evaluator}).
Finally, we specify the M/M/1/\(K\)-based conditional-latency, throughput, and
non-drop probabilities evaluated at each resource and their comparison with
the delegated probability targets
(Sec.~\ref{appendix:performance_per_resource_metrics}).
Algorithm~\ref{alg:appendix_phi_mapping} summarizes the complete per-request
procedure.
Table~\ref{tab:performance_model_symbol_map} summarizes the principal symbols
used throughout this section.

\begin{table}[t]
\caption{Principal symbols for the experimental controller implementation}
\label{tab:performance_model_symbol_map}
\centering
\small
\setlength{\tabcolsep}{3pt}
\begin{tabular}{@{} >{\raggedright\arraybackslash}p{0.3\columnwidth} >{\raggedright\arraybackslash}p{0.65\columnwidth} @{}}
    \toprule
    \multicolumn{2}{c}{\textbf{Network-level Composition and Admission}} \\
    \midrule
    \(\boldsymbol{\lambda}_t\) & Simulator state collecting the current resource-level input and output traffic rates. \\
    \(\lambda_{t,d,j}^{\mathrm{in}},\lambda_{t,d,j}^{\mathrm{out}}\) & Input and success-branch output packet-arrival rates at resource \((d,j)\). \\
    \(Y_{t,d,j}\), \(Y_{t,d}\), \(Y_t\) & Resource-level, domain-level, and \ac{e2e} feasibility indicators. \\
    \(B_{t,d,j}\) & Domain demand report containing \(\widetilde U_{t,d,j}\) for \(j\in\Gamma_d(s_t)\) and zero otherwise. \\
    \(C^{\mathrm{rem}}_{t,d,j}\) & Remaining capacity; \(\mathbf C_t^{\mathrm{rem}}:=(C^{\mathrm{rem}}_{t,d,j})_{d,j}\). \\
    \midrule
    \multicolumn{2}{c}{\textbf{Per-resource Feasibility and Allocation}} \\
    \midrule
    \(\delta_{t,\mathrm{TN},j}^{\mathrm{edge}}, g_{t,\mathrm{TN},i}^{\mathrm{edge}}\) & TN edge-level thresholds produced by the deterministic TN mapping (\(i\in\{1,2,3\}\)). \\
    \(\bar\alpha_{t,d,j}\) & Physically available ratio \(C^{\mathrm{rem}}_{t,d,j}/C_{d,j}^{\max}\). \\
    \(\widehat\beta^{\mathrm{req}}_{t,d,j}\) & Feasible-side bisection estimate of the minimum required ratio within the physical range. \\
    \(\widetilde\beta_{t,d,j}\) & Trial required ratio after the nonnegative random overhead. \\
    \(\widetilde U_{t,d,j}\) & Trial-demand value \(\widetilde\beta_{t,d,j}C_{d,j}^{\max}\) for an evaluated resource. \\
    \midrule
    \multicolumn{2}{c}{\textbf{Queueing and \ac{sla} Metrics}} \\
    \midrule
    \(\mu_{t,d,j}(\beta)\) & Service rate under the candidate capacity \(\beta C_{d,j}^{\max}\). \\
    \(\rho_{t,d,j}(\beta)\) & Traffic intensity for M/M/1/\(K_{d,j}\) at candidate ratio \(\beta\). \\
    \(p_{t,d,j,n}(\beta)\) & Stationary probability that queue length is \(n\). \\
    \(\Lambda^{\mathrm{out}}_{t,d,j}(\beta)\) & Accepted-packet rate at resource \((d,j)\) under candidate ratio \(\beta\). \\
    \(\Theta_{t,d,j,T_{\mathrm{win}}}\) & Windowed throughput random variable at resource \((d,j)\). \\
    \bottomrule
\end{tabular}
\end{table}

\begin{algorithm}[t]
\caption{Sequential Domain Composition and Atomic Admission for One Request}
\label{alg:appendix_phi_mapping}
\begin{algorithmic}[1]
\Require \(s_t\), \(\mathbf{x}_t\);
\(\{\Gamma_d(s_t),s_d(s_t,\mathbf{x}_t)\}_{d\in\mathcal D}\);
current remaining capacities \(\mathbf C_t^{\mathrm{rem}}\)
\State \(\boldsymbol{\lambda}_t\gets\Call{InitializeSources}{s_t}\)
\For{\(d=\mathrm{CN},\mathrm{TN},\mathrm{AN}\), in this order}
    \State
    \(\begin{aligned}
    &(Y_{t,d},\boldsymbol{\lambda}_t,\mathbf B_{t,d})\\[-1mm]
    &\quad\gets\Call{EvaluateDomain}{
    d,\Gamma_d(s_t),s_d(s_t,\mathbf{x}_t),\boldsymbol{\lambda}_t}
    \end{aligned}\)
    \If{\(Y_{t,d}=0\)}
        \State \Return \(0,\mathbf0,\mathbf C_t^{\mathrm{rem}}\)
    \EndIf
\EndFor
\State
\(\begin{aligned}
&(Y_t,\mathbf U_t,\mathbf C_{t+1}^{\mathrm{rem}})\\[-1mm]
&\quad\gets\textsc{AtomicAdmission}\bigl(
\{\mathbf B_{t,d}\}_{d\in\mathcal D},
\mathbf C_t^{\mathrm{rem}}\bigr)
\end{aligned}\)
\State \Return \(Y_t,\mathbf U_t,\mathbf C_{t+1}^{\mathrm{rem}}\)
\Statex
\Procedure{EvaluateDomain}{$d,\Gamma_d(s_t),s_{t,d},
\boldsymbol{\lambda}_t$}
\State \(Y_{t,d}\gets1\)
\State \(\mathbf B_{t,d}\gets\mathbf0\)
\For{each \(j\in\Gamma_d(s_t)\), in topological order}
    \State
    \(\begin{aligned}
    &(Y_{t,d,j},\widetilde U_{t,d,j},
    \lambda^{\mathrm{out}}_{t,d,j})\\[-1mm]
    &\quad\gets \Call{EvaluateResource}{
    s_{t,d},\lambda^{\mathrm{in}}_{t,d,j},
    C^{\mathrm{rem}}_{t,d,j}}
    \end{aligned}\)
    \State \(Y_{t,d}\gets Y_{t,d}Y_{t,d,j}\)
    \State \(B_{t,d,j}\gets\widetilde U_{t,d,j}\)
    \If{\(Y_{t,d,j}=0\)}
        \State \Return \(Y_{t,d}\), \(\boldsymbol{\lambda}_t\), and
        \(\mathbf B_{t,d}\)
    \EndIf
    \State
    \(\boldsymbol{\lambda}_t
    \gets\Call{PropagateRate}{
    d,j,\lambda^{\mathrm{out}}_{t,d,j},\boldsymbol{\lambda}_t}\)
\EndFor
\State \Return \(Y_{t,d}\), \(\boldsymbol{\lambda}_t\), and
\(\mathbf B_{t,d}\)
\EndProcedure
\Statex
\Procedure{AtomicAdmission}{$\{\mathbf B_{t,d}\}_{d\in\mathcal D},
\mathbf C_t^{\mathrm{rem}}$}
\State \(\mathbf U_t\gets(\mathbf B_{t,d})_{d\in\mathcal D}\)
\State \(\mathbf C_{t+1}^{\mathrm{rem}}
\gets\mathbf C_t^{\mathrm{rem}}-\mathbf U_t\)
\State \Return \(1,\mathbf U_t,\mathbf C_{t+1}^{\mathrm{rem}}\)
\EndProcedure
\end{algorithmic}
\end{algorithm}

\subsection{Network-level Composition and Atomic Admission}
\label{appendix:performance_sequential_composition}
At round \(t\), each \(\Phi_d\) receives the path-induced resource set
\(\Gamma_d(s_t)\) and delegated request \(s_d(s_t,\mathbf{x}_t)\).
Because the input traffic rate of each resource is determined by the output
traffic rates of its preceding resources, the simulator evaluates resources
sequentially from CN through TN to AN and propagates each output rate
downstream (Sec.~\ref{appendix:performance_traffic_propagation}).
This sequencing only carries the upstream traffic rates downstream; each
controller otherwise uses its delegated inputs and local operational state
independently.
Within each domain \(d\), the resource-level success indicators
\(Y_{t,d,j}\) are combined into \(Y_{t,d}\).
Evaluation stops at the first resource-level failure and returns \(Y_t=0\) and
\(\mathbf U_t=\mathbf0\) without evaluating the remaining resources.
If every traversed resource is feasible, the domain demand reports are
committed atomically as the realized resource-consumption vector \(\mathbf U_t\)
(Sec.~\ref{appendix:performance_atomic_admission}).

\subsubsection{Source Initialization and Traffic Propagation}
\label{appendix:performance_traffic_propagation}
To determine the traffic offered to each resource, the simulator first assigns
the request traffic to the selected CN roots
(\textsc{InitializeSources}).
It then passes the output traffic rate of each successfully evaluated resource
to the resources that follow it on the selected paths
(\textsc{PropagateRate}).
These two steps update \(\boldsymbol{\lambda}_t\) as the simulator proceeds from
CN through TN to AN.

\textbf{(i) \textsc{InitializeSources}.}
This procedure initializes \(\boldsymbol{\lambda}_t\) by assigning the CN-root
input rates separately for each selected \ac{upf}; all other resource-level
rates are populated by subsequent calls to \textsc{PropagateRate}.
Let \(\iota_t(a)\in\mathcal U\) be the \ac{upf} selected for active \ac{gnb}
\(a\in A_t\).
For each \(c\in\mathcal U\), define the active \acp{gnb} assigned to \(c\) as
\[
A_t(c):=
\{a\in A_t\mid \iota_t(a)=c\}.
\]
For each \(a\in A_t(c)\), let \(P_t(c,a)\) denote the selected directed
downlink path from \(c\) to \(a\).
For each CN root edge \(j\) leaving \(c\), define
\begin{equation}
A_t(c,j)
:=
\{a\in A_t(c)\mid (\mathrm{CN},j)\in P_t(c,a)\}.
\end{equation}
The offered source rate associated with \ac{upf} \(c\) is
\begin{equation}
\lambda^{\mathrm{src}}_{t,c}
:=
|A_t(c)|\lambda_t,
\end{equation}
where \(\lambda_t\) is the common offered packet-arrival rate per active
\ac{gnb}, specified by the traffic/service profile
\(\boldsymbol{\theta}_t\) of request \(s_t\)
(Definition~\ref{def:network_slice_request}).
For each CN root edge \(j\) leaving \(c\), the input rate is initialized as
\begin{equation}
\lambda^{\mathrm{in}}_{t,\mathrm{CN},j}
:=
|A_t(c,j)|\lambda_t.
\end{equation}
Thus, traffic is partitioned across the selected \acp{upf} and their root
edges rather than replicated.

\textbf{(ii) \textsc{PropagateRate}.}
\textsc{PropagateRate} distributes the output traffic rate
\(\lambda^{\mathrm{out}}_{t,d,j}\) of a successfully evaluated resource
\((d,j)\) among the resources \((d',j')\) that follow it on the selected paths
and records their resulting input rates
\(\lambda^{\mathrm{in}}_{t,d',j'}\) in \(\boldsymbol{\lambda}_t\).
In the four experimental topologies, the selected paths form a directed tree
rooted at each \ac{upf}.
Thus, if resource \((d,j)\) feeds resource \((d',j')\), every active \ac{gnb}
whose path contains \((d',j')\) also has a path containing \((d,j)\).
Because all active \acp{gnb} in a request have the same offered rate
\(\lambda_t\), the output of \((d,j)\) is divided in proportion to the number
of these \acp{gnb} served through each following resource.
To express this number, define
\begin{equation}
n_t(d,j)
:=
\left|
\left\{
a\in A_t\mid (d,j)\in P_t(\iota_t(a),a)
\right\}
\right|.
\end{equation}
The fraction of the output of \((d,j)\) directed to \((d',j')\) is therefore
\(n_t(d',j')/n_t(d,j)\).
The corresponding contribution and the total input rate of \((d',j')\) are
\begin{equation}
\begin{aligned}
\lambda^{\mathrm{contrib}}_{t,(d,j)\rightarrow(d',j')}
&:=
\lambda^{\mathrm{out}}_{t,d,j}
\frac{n_t(d',j')}{n_t(d,j)},\\
\lambda^{\mathrm{in}}_{t,d',j'}
&:=
\sum_{(d,j)\rightarrow(d',j')}
\lambda^{\mathrm{contrib}}_{t,(d,j)\rightarrow(d',j')}.
\end{aligned}
\end{equation}
Only resources traversed by the current request are processed, so the
denominator \(n_t(d,j)\) is positive whenever this propagation rule is used.

\subsubsection{Domain Demand Reporting and Atomic Admission}
\label{appendix:performance_atomic_admission}
The simulator does not reduce any resource capacity while evaluating the
individual resource requirements.
Instead, \textsc{EvaluateDomain} aggregates the resource-level indicators
\(\{Y_{t,d,j}\}_{j\in\Gamma_d(s_t)}\) into \(Y_{t,d}\) and, when all of them
equal \(1\), constructs the domain demand report \(\mathbf B_{t,d}\).
\textsc{EvaluateDomain} and the outer domain loop terminate at the first local
failure.
If all domain evaluations succeed, \textsc{AtomicAdmission} commits
\(\{\mathbf B_{t,d}\}_{d\in\mathcal D}\).
This prevents a request that fails in one domain from consuming capacity.

\textbf{(i) \textsc{EvaluateDomain}.}
Let \(Y_{t,d,j}\in\{0,1\}\) indicate whether resource \((d,j)\) can support its
trial demand \(\widetilde U_{t,d,j}\).
Conceptually, the domain-level indicator combines these resource-level
indicators:
\begin{equation}
Y_{t,d}
:=
\prod_{j\in\Gamma_d(s_t)}Y_{t,d,j}.
\end{equation}
The implementation evaluates this conjunction by short circuit: once some
\(Y_{t,d,j}=0\), it returns \(Y_{t,d}=0\) and skips the remaining resources.
Here, \(\widetilde U_{t,d,j}\) is the internal trial-demand value returned by
\textsc{EvaluateResource}, whereas \(B_{t,d,j}\) is the corresponding
domain-level report exposed by \textsc{EvaluateDomain}, consistent with the
controller interface in Sec.~\ref{subsubsec:per_round_feasibility}.
When \(Y_{t,d}=1\), every traversed resource has been evaluated, and the
implementation uses the zero-extension
\(\widetilde U_{t,d,j}:=0\) for \(j\notin\Gamma_d(s_t)\).
The complete domain demand report is then
\begin{equation}
B_{t,d,j}
:=
\widetilde U_{t,d,j},
\qquad j\in[J_d].
\end{equation}

\textbf{(ii) \textsc{AtomicAdmission}.}
The \ac{e2e} admission result is the conjunction of the domain-level
indicators:
\begin{equation}
Y_t
:=
\prod_{d\in\mathcal D}Y_{t,d}.
\end{equation}
The algorithm evaluates this conjunction by short circuit.
If \(Y_t=0\), it returns \(\mathbf U_t=\mathbf0\) and leaves
\(\mathbf C_t^{\mathrm{rem}}\) unchanged.
If \(Y_t=1\), all domain demand reports are complete, and
\textsc{AtomicAdmission} commits them jointly.
Let \(C^{\mathrm{rem}}_{t,d,j}\) denote the capacity of resource \((d,j)\)
available immediately before round \(t\), initialized as
\(C^{\mathrm{rem}}_{1,d,j}:=C_{d,j}^{\max}\), and write
\(\mathbf C_t^{\mathrm{rem}}:=(C^{\mathrm{rem}}_{t,d,j})_{d,j}\).
On the successful branch, the realized consumption and remaining capacity are
\begin{equation}
\begin{aligned}
U_{t,d,j}
&:=
B_{t,d,j},\\
C^{\mathrm{rem}}_{t+1,d,j}
&:=
C^{\mathrm{rem}}_{t,d,j}-U_{t,d,j}.
\end{aligned}
\end{equation}
Hence, if any traversed resource is infeasible, no capacity is consumed and all
remaining capacities are unchanged.

\subsection{Per-resource Required Capacity and Feasibility Check}
\label{appendix:performance_resource_level_evaluator}
For each resource \((d,j)\), the controller first tests whether the remaining
capacity can satisfy the delegated \ac{sla} checks under the current input
traffic and, if so, determines the minimum required capacity within that
physical range.
This calculation uses the delegated request
\(s_{t,d}:=s_d(s_t,\mathbf{x}_t)\), the input rate
\(\lambda^{\mathrm{in}}_{t,d,j}\), and the remaining capacity
\(C^{\mathrm{rem}}_{t,d,j}\).
Upon local success, it also determines the output rate
\(\lambda^{\mathrm{out}}_{t,d,j}\) passed to the resources evaluated next.
Algorithm~\ref{alg:appendix_phi_mapping} implements this calculation as
\textsc{EvaluateResource}.
The steps below assume \(\lambda^{\mathrm{in}}_{t,d,j}>0\)
(the implementation skips a zero-input resource and returns
\((Y_{t,d,j},\widetilde U_{t,d,j},\lambda^{\mathrm{out}}_{t,d,j})=(1,0,0)\)).

We first translate the delegated domain-level \ac{sla} requirements in
\(s_{t,d}\) into resource-level checks for the traversed resources
(Sec.~\ref{appendix:resource_level_check_construction}).
We then test the physical upper ratio and, when that upper ratio is feasible,
use bisection to estimate the minimum required ratio within the physical range
(Sec.~\ref{appendix:capacity_bounded_minimum_ratio}).
Finally, we add the experimental allocation variation, form the internal trial
demand \(\widetilde U_{t,d,j}\), check it against the remaining capacity, and
determine \(Y_{t,d,j}\) and, upon local success,
\(\lambda^{\mathrm{out}}_{t,d,j}\)
(Sec.~\ref{appendix:performance_stage2_output_construction}).

\subsubsection{Constructing Per-resource \ac{sla} Checks}
\label{appendix:resource_level_check_construction}
For every resource on the selected path, the evaluator derives resource-level
conditional-latency, throughput, and non-drop checks from the delegated
domain-level requirements.

\textbf{AN and CN (Direct Use):}
An experimental path contains one resource from each of the AN and CN domains.
The corresponding delegated domain-level checks are therefore applied directly
to those resources.

\textbf{TN (Fixed Uniform Split):}
A TN path may contain several TN resources, so the controller translates each
delegated TN-domain requirement into checks for the traversed resources.
This translation applies the same target- and guarantee-allocation construction
as the admissible decomposition in
Definition~\ref{def:admissible_decomp} and
Table~\ref{tab:sla_decomposition_constructions}.
Let \(H_{\max}^{\mathrm{TN}}\) be the maximum number of TN-managed resources on
any selected \ac{upf}--\ac{gnb} path in the environment.
Its values for the Tree, Tree dual-\ac{upf}, Ring, and Ring dual-\ac{upf}
configurations are \(1\), \(1\), \(3\), and \(2\), respectively.
For the current TN-domain latency threshold \(\delta_{t,\mathrm{TN}}\) and
probability targets \(g_{t,\mathrm{TN},i}\), the controller specializes both
the latency-target weight and each guarantee-allocation exponent uniformly to
\(1/H_{\max}^{\mathrm{TN}}\), yielding
\begin{equation}
\begin{aligned}
\delta_{t,\mathrm{TN},j}^{\mathrm{edge}}
&:=
\frac{\delta_{t,\mathrm{TN}}}{H_{\max}^{\mathrm{TN}}},
\\
g_{t,\mathrm{TN},i}^{\mathrm{edge}}
&:=
\left(g_{t,\mathrm{TN},i}\right)^{1/H_{\max}^{\mathrm{TN}}},
\qquad i\in\{1,2,3\}.
\end{aligned}
\end{equation}
On a selected path containing \(n_{\mathrm{hop}}\le H_{\max}^{\mathrm{TN}}\) TN resources,
and under the simulator's approximation that their success events are
independent, the two preservation relations are
\begin{equation}
\begin{aligned}
\sum_{j\text{ on path}}\delta_{t,\mathrm{TN},j}^{\mathrm{edge}}
&=
\frac{n_{\mathrm{hop}}}{H_{\max}^{\mathrm{TN}}}\delta_{t,\mathrm{TN}}
\le
\delta_{t,\mathrm{TN}},\\
\prod_{j\text{ on path}}g_{t,\mathrm{TN},i}^{\mathrm{edge}}
&=
\left(g_{t,\mathrm{TN},i}\right)^{n_{\mathrm{hop}}/H_{\max}^{\mathrm{TN}}}
\ge
g_{t,\mathrm{TN},i},
\quad i\in\{1,2,3\}.
\end{aligned}
\end{equation}
Thus, this split satisfies the same target-preservation and
guarantee-allocation relations as an admissible decomposition.

\subsubsection{Physical Feasibility and Minimum-Ratio Search}
\label{appendix:capacity_bounded_minimum_ratio}
To determine the resource allocation for the delegated request \(s_{t,d}\),
the controller first searches for the minimum allocation ratio, measured
relative to the initial capacity \(C_{d,j}^{\max}\), needed at resource
\((d,j)\).
It uses \(\beta\ge0\) as this normalized search variable, so the corresponding
candidate capacity is \(\beta C_{d,j}^{\max}\).
For the current \(s_{t,d}\) and input rate
\(\lambda_{t,d,j}^{\mathrm{in}}\), let
\(\mathrm{Feasible}_{t,d,j}(\beta)\in\{0,1\}\) indicate whether the
candidate capacity \(\beta C_{d,j}^{\max}\) satisfies the resource-local
\ac{sla} requirements.
Its formal definition is given in
Sec.~\ref{appendix:resource_local_feasibility_predicate}.
Under that definition, \(\mathrm{Feasible}_{t,d,j}(\beta)\) is nondecreasing
in \(\beta\).
Accordingly, the evaluation has two steps: (i) the controller first tests the
physical upper ratio and then, only when that ratio is feasible, (ii) searches for
the minimum required ratio within the physically available range.

\textbf{(i) Physical Upper-Ratio Check.}
The physical upper ratio before evaluating the current request is
\begin{equation}
\bar\alpha_{t,d,j}
:=
\frac{C^{\mathrm{rem}}_{t,d,j}}{C_{d,j}^{\max}}
\in[0,1].
\end{equation}
The stated range follows because the remaining capacity is initialized to
\(C_{d,j}^{\max}\) and only decreases upon admission.
The implementation sets
\(\mathrm{Feasible}_{t,d,j}(0)=0\) and first evaluates
\(\mathrm{Feasible}_{t,d,j}(\bar\alpha_{t,d,j})\).
If \(\mathrm{Feasible}_{t,d,j}(\bar\alpha_{t,d,j})=0\), monotonicity
implies that \textit{no ratio within the physically available range is
feasible}.
The procedure therefore returns
\((Y_{t,d,j},\widetilde U_{t,d,j},\lambda^{\mathrm{out}}_{t,d,j})
=(0,0,0)\); the short-circuit rules in
Algorithm~\ref{alg:appendix_phi_mapping} then terminate the request evaluation.

\textbf{(ii) Capacity-Bounded Minimum-Ratio Search.}
If \(\mathrm{Feasible}_{t,d,j}(\bar\alpha_{t,d,j})=1\), the controller
applies bisection over \([0,\bar\alpha_{t,d,j}]\).
The feasible side of the final interval is returned as
\(\widehat\beta^{\mathrm{req}}_{t,d,j}\), the numerical estimate of the
minimum required ratio within the physical range.

\subsubsection{Trial Demand, Local Feasibility, and Output Rate}
\label{appendix:performance_stage2_output_construction}
Using the physical upper ratio \(\bar\alpha_{t,d,j}\) and the result of the
preceding minimum-ratio search, the \textsc{EvaluateResource} procedure
(i) constructs the trial demand \(\widetilde U_{t,d,j}\), (ii) determines the
resource-level feasibility indicator \(Y_{t,d,j}\), and, only upon local
success, (iii) computes the output rate \(\lambda^{\mathrm{out}}_{t,d,j}\)
propagated to downstream resources.
These three steps are detailed below.

\textbf{(i) Trial Demand Construction.}
To represent variation in a domain-specific controller's allocation decisions
that is not captured by the analytical queueing calculation, the procedure adds
a small nonnegative random overhead to
\(\widehat\beta^{\mathrm{req}}_{t,d,j}\).
After
\(\mathrm{Feasible}_{t,d,j}(\bar\alpha_{t,d,j})=1\) has been established, the
procedure samples
\begin{equation}
\varepsilon_{t,d,j}
\sim
\operatorname{Uniform}(0,\varepsilon_{\max,d})
\end{equation}
independently across resources and rounds, with
\(\varepsilon_{\max,d}=0.05\) for every \(d\in\mathcal D\) in all experiments,
and forms
\begin{equation}
\begin{aligned}
\widetilde\beta_{t,d,j}
&:=
\widehat\beta^{\mathrm{req}}_{t,d,j}e^{\varepsilon_{t,d,j}},\\
\widetilde U_{t,d,j}
&:=
\widetilde\beta_{t,d,j}C_{d,j}^{\max}.
\end{aligned}
\end{equation}
The feasible-side estimate \(\widehat\beta^{\mathrm{req}}_{t,d,j}\) represents
the computed minimum required ratio, so the nonnegative multiplier models
random allocation overhead above that value.
This perturbation is the simulator-level implementation of the allocation
variation described above; it is distinct from the stationary-response law
\(\nu_{s,\mathbf{x}}\) for the
final provisioning outcome \((Y,\mathbf U)\).

\textbf{(ii) Local Feasibility Check.}
Since
\(\widetilde\beta_{t,d,j}
=\widehat\beta^{\mathrm{req}}_{t,d,j}e^{\varepsilon_{t,d,j}}\)
and \(\varepsilon_{t,d,j}\ge0\), we have
\(\widetilde\beta_{t,d,j}\ge
\widehat\beta^{\mathrm{req}}_{t,d,j}\).
The monotonicity of \(\mathrm{Feasible}_{t,d,j}(\beta)\) ensures that
\(\widetilde\beta_{t,d,j}\) also satisfies the resource-local \ac{sla}
requirements; local feasibility additionally requires
\(\widetilde\beta_{t,d,j}\le\bar\alpha_{t,d,j}\).
Within this branch, the resource-level feasibility indicator is therefore
\begin{equation}
Y_{t,d,j}
:=
\mathbf{1}\{\widetilde\beta_{t,d,j}\le\bar\alpha_{t,d,j}\}.
\end{equation}
If \(Y_{t,d,j}=0\), the request evaluation terminates before computing or
propagating an output rate; the fixed return signature uses zero for its unused
third component.

\textbf{(iii) Output-Rate Evaluation.}
When \(Y_{t,d,j}=1\), the procedure evaluates and returns the output rate at the
trial ratio:
\begin{equation}
\lambda^{\mathrm{out}}_{t,d,j}
:=
\Lambda^{\mathrm{out}}_{t,d,j}
\!\left(\widetilde\beta_{t,d,j}\right).
\end{equation}
Here, \(\Lambda^{\mathrm{out}}_{t,d,j}(\beta)\) denotes the accepted-packet
rate under candidate ratio \(\beta\), formally defined in
\eqref{eq:resource_accepted_packet_rate}.

\subsection{M/M/1/K-based SLA Metric Calculation}
\label{appendix:performance_per_resource_metrics}
This subsection first instantiates the M/M/1/\(K\) model for each resource from the
environment configuration, request-dependent traffic descriptors, current
input rate \(\lambda_{t,d,j}^{\mathrm{in}}\), and candidate allocation ratio
\(\beta\) (Sec.~\ref{appendix:mm1k_modeling}).
It then uses the instantiated model to determine
\(\mu_{t,d,j}(\beta)\), \(\rho_{t,d,j}(\beta)\), the queue-length
distribution, and the accepted-packet rate
\(\Lambda_{t,d,j}^{\mathrm{out}}(\beta)\)
(Sec.~\ref{appendix:per_resource_queueing_evaluation}).
Finally, these quantities determine the three \ac{sla}-success probabilities
(Sec.~\ref{appendix:per_resource_sla_probability_models}), which are
combined to define \(\mathrm{Feasible}_{t,d,j}(\beta)\)
(Sec.~\ref{appendix:resource_local_feasibility_predicate}).

\subsubsection{M/M/1/K Model Instantiation}
\label{appendix:mm1k_modeling}
This subsubsection instantiates the common M/M/1/\(K\) model used
for every traversed resource \((d,j)\).
The environment configuration assigns each resource its maximum capacity
\(C_{d,j}^{\max}\), queue depth \(K_{d,j}\), and propagation distance
\(\ell_{d,j}\), and fixes the propagation speed \(v_{\mathrm{prop}}\) and
throughput window \(T_{\mathrm{win}}\).
For request \(s_t\), the traffic/service profile
\(\boldsymbol{\theta}_t\) specifies the common per-\ac{gnb} source
packet-arrival rate \(\lambda_t\) and the mean packet size
\(m_{B,t}^{(1)}\) in bits.
At each resource, the experimental model treats packet arrivals as a Poisson
process with the current input rate \(\lambda_{t,d,j}^{\mathrm{in}}\) and
packet sizes as independent exponential random variables.
Consequently, the second packet-size moment is
\(m_{B,t}^{(2)}=2(m_{B,t}^{(1)})^2\).

\subsubsection{Queueing Quantities and Accepted Rate}
\label{appendix:per_resource_queueing_evaluation}
For each resource \((d,j)\), candidate ratio \(\beta>0\), and round-specific
input rate \(\lambda_{t,d,j}^{\mathrm{in}}\), the service rate
\(\mu_{t,d,j}(\beta)\) and traffic intensity
\(\rho_{t,d,j}(\beta)\) used by the capacity-bounded feasibility search are
\begin{equation}
\begin{aligned}
\mu_{t,d,j}(\beta)
&:=
\frac{\beta C_{d,j}^{\max}}{m_{B,t}^{(1)}},\\
\rho_{t,d,j}(\beta)
&:=
\frac{\lambda_{t,d,j}^{\mathrm{in}}}{\mu_{t,d,j}(\beta)}.
\end{aligned}
\end{equation}
\(Q_{d,j}\in\{0,\ldots,K_{d,j}\}\) denotes the steady-state number of packets
in the M/M/1/\(K_{d,j}\) system at resource \((d,j)\), including the packet in
service.
Here and below, \(\Pr_{\beta}\) denotes probability under the current request,
input rate, and candidate ratio \(\beta\).
For \(\beta>0\), let
\(p_{t,d,j,n}(\beta):=\Pr_{\beta}(Q_{d,j}=n)\).
Writing \(\rho:=\rho_{t,d,j}(\beta)\), it is given by
\begin{equation}
p_{t,d,j,n}(\beta)
=
\begin{cases}
\dfrac{1-\rho}{1-\rho^{K_{d,j}+1}}\,\rho^n,
& \rho \neq 1,\\[2mm]
\dfrac{1}{K_{d,j}+1}, & \rho = 1,
\end{cases}
\end{equation}
for \(n\in\{0,\ldots,K_{d,j}\}\).
For \(\beta>0\), the resulting accepted-packet rate is
\begin{equation}
\Lambda^{\mathrm{out}}_{t,d,j}(\beta)
:=
\lambda^{\mathrm{in}}_{t,d,j}
\left(1-p_{t,d,j,K_{d,j}}(\beta)\right).
\label{eq:resource_accepted_packet_rate}
\end{equation}

\subsubsection{Per-resource SLA Probability Models}
\label{appendix:per_resource_sla_probability_models}
The queueing quantities above determine the non-drop, conditional-latency,
and throughput probabilities used by the resource-local feasibility
predicate \(\mathrm{Feasible}_{t,d,j}\).

\textbf{(i) Non-drop Probability.}
Define the per-resource non-drop event by
\(\mathcal{N}_{d,j}:=\{Q_{d,j}<K_{d,j}\}\).
For candidate ratio \(\beta>0\), the non-drop probability is
\begin{equation}
\Pr_{\beta}(\mathcal{N}_{d,j})
=
1-p_{t,d,j,K_{d,j}}(\beta).
\end{equation}

\textbf{(ii) Conditional Latency Probability.}
By the \ac{pasta} property~\cite{wolff1982pasta}, an arriving packet observes
the stationary queue-length distribution \(p_{t,d,j,n}(\beta)\).
Hence, its queue-length distribution conditional on
\(\mathcal{N}_{d,j}\) is
\begin{equation}
\begin{aligned}
p^{\mathcal N}_{t,d,j,n}(\beta)
&:=
\Pr_{\beta}(Q_{d,j}=n\mid\mathcal{N}_{d,j})\\
&=
\frac{p_{t,d,j,n}(\beta)}
{1-p_{t,d,j,K_{d,j}}(\beta)}.
\end{aligned}
\end{equation}
This definition applies for \(n=0,\ldots,K_{d,j}-1\).
The Erlang mixture models the system delay
\(D^{\mathrm{sys}}_{d,j}\), which includes both waiting and service time.
Let \(\ell_{d,j}\) be the propagation distance of resource \((d,j)\), and define
\begin{equation}
\tau^{\mathrm{prop}}_{d,j}
:=
\frac{\ell_{d,j}}{v_{\mathrm{prop}}},
\qquad
D_{d,j}
:=
D^{\mathrm{sys}}_{d,j}+\tau^{\mathrm{prop}}_{d,j}.
\end{equation}
For a positive input rate and any latency threshold \(\delta\), the conditional
latency probability is
\begin{equation}
\begin{aligned}
&\Pr_{\beta}(D_{d,j}\le\delta\mid\mathcal{N}_{d,j})\\
&\quad =
\sum_{n=0}^{K_{d,j}-1}
p^{\mathcal N}_{t,d,j,n}(\beta)
F_{\mathrm{Erlang}(n+1,\mu_{t,d,j}(\beta))}
\!\left(\delta-\tau^{\mathrm{prop}}_{d,j}\right).
\end{aligned}
\end{equation}
Here, for \(k\in\mathbb{N}\), \(\mu>0\), and
\(Z\sim\mathrm{Erlang}(k,\mu)\), the CDF
\(F_{\mathrm{Erlang}(k,\mu)}\) is given by
\begin{equation}
\begin{aligned}
F_{\mathrm{Erlang}(k,\mu)}(x)
&=\Pr(Z\le x)\\
&=
\begin{cases}
0, & x<0,\\
1-e^{-\mu x}\displaystyle\sum_{m=0}^{k-1}\dfrac{(\mu x)^m}{m!},
& x\ge 0.
\end{cases}
\end{aligned}
\end{equation}

\textbf{(iii) Throughput Probability.}
Let \(T_{\mathrm{win}}>0\) denote the duration of the observation window used
for the throughput check.
In all experiments, we set \(T_{\mathrm{win}}=0.1~\mathrm{s}\).
The random variable \(\Theta_{t,d,j,T_{\mathrm{win}}}\) represents the average
accepted-packet bit rate over this window, obtained by dividing the total size
of the packets accepted during the window by \(T_{\mathrm{win}}\):
\begin{equation}
\Theta_{t,d,j,T_{\mathrm{win}}}
:=
\frac{1}{T_{\mathrm{win}}}
\sum_{m=1}^{N_{t,d,j}(T_{\mathrm{win}})} L^{\mathrm{pkt}}_{t,m},
\end{equation}
where \(N_{t,d,j}(\cdot)\) is the accepted-packet counting process and
\(\{L^{\mathrm{pkt}}_{t,m}\}_m\) are packet sizes, independent of one another
and of \(N_{t,d,j}(\cdot)\), with moments \(m_{B,t}^{(1)}\) and
\(m_{B,t}^{(2)}\).
To obtain a tractable throughput probability, the experimental model adopts
two approximations for a candidate ratio \(\beta>0\):
it models \(N_{t,d,j}(\cdot)\) as a Poisson process with rate
\(\Lambda^{\mathrm{out}}_{t,d,j}(\beta)\) and applies a normal approximation
to the resulting compound-Poisson sum
\(T_{\mathrm{win}}\Theta_{t,d,j,T_{\mathrm{win}}}\).
Under these approximations,
\begin{equation}
\Theta_{t,d,j,T_{\mathrm{win}}}
\approx
\mathcal{N}\!\left(
\Lambda^{\mathrm{out}}_{t,d,j}(\beta)m_{B,t}^{(1)},\;
\frac{\Lambda^{\mathrm{out}}_{t,d,j}(\beta)m_{B,t}^{(2)}}{T_{\mathrm{win}}}
\right).
\end{equation}
Hence, 
\begin{equation}
\Pr_{\beta}(\Theta_{t,d,j,T_{\mathrm{win}}}\ge\theta)
\approx
1-\Phi\!\left(
\frac{\theta-\Lambda^{\mathrm{out}}_{t,d,j}(\beta)m_{B,t}^{(1)}}
{\sqrt{\Lambda^{\mathrm{out}}_{t,d,j}(\beta)m_{B,t}^{(2)}
/T_{\mathrm{win}}}}
\right).
\end{equation}

The two approximations above are motivated as follows.
\noindent\textbf{(Poisson Approximation.)}
The model assumes that finite-buffer blocking is sufficiently rare for its
effect on packet-arrival timing at downstream resources to be negligible.
This low-blocking assumption is expressed as
\begin{equation}
\begin{aligned}
p_{t,d,j,K_{d,j}}(\beta)
&=
\Pr_{\beta}(Q_{d,j}=K_{d,j})\\
&=
1-\Pr_{\beta}(\mathcal{N}_{d,j})
\ll 1.
\end{aligned}
\end{equation}
Under this Poisson approximation, the throughput probability is evaluated
using only the accepted-packet rate
\(\Lambda^{\mathrm{out}}_{t,d,j}(\beta)\) in
\eqref{eq:resource_accepted_packet_rate}, without tracking the queue-state
dependence of accepted packet arrivals.
Because this rate incorporates the mean packet loss predicted by the
M/M/1/\(K_{d,j}\) model, the approximation retains its mean effect while
neglecting its effect on downstream packet-arrival timing.

\noindent\textbf{(Normal Approximation.)}
The model assumes that the expected number of accepted packets
\(\Lambda^{\mathrm{out}}_{t,d,j}(\beta)T_{\mathrm{win}}\) is sufficiently
large.
Because the independent packet sizes have a finite second moment, the central
limit theorem then supports the normal approximation to the compound-Poisson
sum \(T_{\mathrm{win}}\Theta_{t,d,j,T_{\mathrm{win}}}\).

\subsubsection{Resource-local Feasibility Predicate}
\label{appendix:resource_local_feasibility_predicate}
For the current delegated request \(s_{t,d}\) and a resource with positive
input rate \(\lambda_{t,d,j}^{\mathrm{in}}\), the controller evaluates each
candidate ratio \(\beta>0\) by combining conditional-latency, throughput, and
non-drop checks.
To express their domain-specific targets uniformly, define
\begin{equation}
\begin{aligned}
\bar{\delta}_{t,d,j}
&:=
\begin{cases}
\delta_{t,\mathrm{TN},j}^{\mathrm{edge}}, & d=\mathrm{TN},\\
\delta_{t,d}, & d\in\{\mathrm{AN},\mathrm{CN}\},
\end{cases}\\
\bar g_{t,d,j,i}
&:=
\begin{cases}
g_{t,\mathrm{TN},i}^{\mathrm{edge}}, & d=\mathrm{TN},\\
g_{t,d,i}, & d\in\{\mathrm{AN},\mathrm{CN}\},
\end{cases}
\end{aligned}
\end{equation}
where \(i=1,2,3\) correspond to conditional latency, throughput, and non-drop,
respectively.
The derivation of the TN-specific targets is given in
Sec.~\ref{appendix:resource_level_check_construction}.

Using these targets, \textsc{EvaluateResource} evaluates the predicate
\begin{equation}
\begin{aligned}
\mathrm{Feasible}_{t,d,j}(\beta)
&:=
\mathbf{1}\!\left\{
\mathcal{C}^{(D)}_{t,d,j}(\beta)\wedge
\mathcal{C}^{(\Theta)}_{t,d,j}(\beta)
\wedge \mathcal{C}^{(N)}_{t,d,j}(\beta)
\right\},
\end{aligned}
\end{equation}
where
\begin{equation}
\begin{aligned}
\mathcal{C}^{(D)}_{t,d,j}(\beta)
&:=
\Pr_{\beta}(D_{d,j}\le \bar{\delta}_{t,d,j}\mid \mathcal{N}_{d,j})
\ge \bar g_{t,d,j,1},
\\
\mathcal{C}^{(\Theta)}_{t,d,j}(\beta)
&:=
\Pr_{\beta}(\Theta_{t,d,j,T_{\mathrm{win}}}\ge \theta_t)
\ge \bar g_{t,d,j,2},
\\
\mathcal{C}^{(N)}_{t,d,j}(\beta)
&:=
\Pr_{\beta}(\mathcal{N}_{d,j})\ge \bar g_{t,d,j,3}.
\end{aligned}
\end{equation}

As shown below, \(\mathrm{Feasible}_{t,d,j}(\beta)\) is nondecreasing in
\(\beta\), which justifies applying bisection to compute
\(\widehat\beta^{\mathrm{req}}_{t,d,j}\), the numerical estimate of the
minimum feasible ratio (Sec.~\ref{appendix:capacity_bounded_minimum_ratio}).

\textbf{Monotonicity for Bisection.}
Increasing \(\beta\) increases the service rate
\(\mu_{t,d,j}(\beta)\), which does not decrease the conditional-latency or
non-drop probability.
It also does not decrease the accepted-packet rate
\(\Lambda^{\mathrm{out}}_{t,d,j}(\beta)\), and hence the throughput probability
under the normal approximation.
Thus, under the queueing and approximation models above, the probabilities
defining \(\mathcal{C}^{(D)}_{t,d,j}(\beta)\),
\(\mathcal{C}^{(\Theta)}_{t,d,j}(\beta)\), and
\(\mathcal{C}^{(N)}_{t,d,j}(\beta)\) are nondecreasing in \(\beta\).
Therefore, \(\mathrm{Feasible}_{t,d,j}(\beta)\) is also nondecreasing.

\subsection{NR-Parameterized Access-Link Capacity}
\label{appendix:access_link_service_capacity}
Each downlink is represented by a
finite-buffer queue with 128 packet slots and zero propagation distance. Its
capacity is calculated independently per \ac{gnb} from the following profile:
273 \acp{rb}, numerology \(\mu=3\) (120-kHz subcarrier spacing), 256-QAM
(\(Q_m=8\)), eight MIMO layers, maximum coding-rate factor
\(R_{\max}=948/1024\), downlink overhead \(\mathrm{OH}=0.18\), scaling
factor \(f=1\), and downlink duty factor \(r_{\mathrm{DL}}=1\).

We convert these parameters to the nominal downlink rate using the
TS~38.306 approximate maximum-data-rate expression~\cite{3gpp:38306}. For one
downlink, its bit/s form is
\begin{equation}
\label{eq:custom_access_capacity_reference}
C_{\mathrm{AN}}^{\max}
=
v_{\mathrm{layers}} Q_m f R_{\max}
\left(\frac{N_{\mathrm{RB}}\cdot 12}{T_s^{\mu}}\right)
(1-\mathrm{OH})r_{\mathrm{DL}},
\end{equation}
where \(T_s^{\mu}=10^{-3}/(14\cdot 2^{\mu})\) is the average OFDM symbol
duration under the normal cyclic prefix. Substitution gives
\begin{equation}
\label{eq:custom_access_capacity_value}
C_{\mathrm{AN}}^{\max}
\approx17.83\ \mathrm{Gbit/s},
\end{equation}
per \ac{gnb}.

The pair \(N_{\mathrm{RB}}=273\) and \(\mu=3\) is not a transmission-bandwidth
configuration specified in TS~38.104: at 120-kHz subcarrier spacing, its listed
maximum is 264 \acp{rb} for a 400-MHz channel~\cite{3gpp:38104}. We therefore
do not present this profile as a standards-compliant NR carrier. Instead, it is
an NR-parameterized custom wideband experimental profile.


For candidate allocation ratio \(\beta\), the access-queue service rate is
\begin{equation}
\label{eq:custom_access_queue_service_rate}
\mu_{t,\mathrm{AN},j}(\beta)
=
\frac{\beta C_{\mathrm{AN}}^{\max}}{m_{B,t}^{(1)}}.
\end{equation}
\eqref{eq:custom_access_queue_service_rate} is consistent with the
general per-resource definition in
Sec.~\ref{appendix:per_resource_queueing_evaluation}.


\end{document}

  \documentclass[src/main.tex]{subfiles}

\begin{document}

\section{Total Reward and Per-Class Results under Mixed Traffic}
\label{appendix:mixed_traffic_reward_class_results}

This section complements the Successful-Slice Count results in
Sec.~\ref{subsec:traffic_mixture_tree} with Total Reward and its class-specific
components from the same runs.  Let \(N_{r}^{\mathrm{U}}\) and
\(N_{r}^{\mathrm{B}}\) be the admitted \ac{urllc} and \ac{embb} counts in run
\(r\), respectively.  Under the class prices in
Sec.~\ref{subsubsec:nsr_setup}, Total Reward is
\(R_r=N_{r}^{\mathrm{U}}+50N_{r}^{\mathrm{B}}\).  Fig.~\ref{fig:mixed_traffic_reward_class}
shows the two class contributions.

\begin{figure*}[t]
\centering
\subfloat[Tree topology.\label{fig:mixed_traffic_reward_class_tree}]{%
\includegraphics[width=0.49\textwidth]{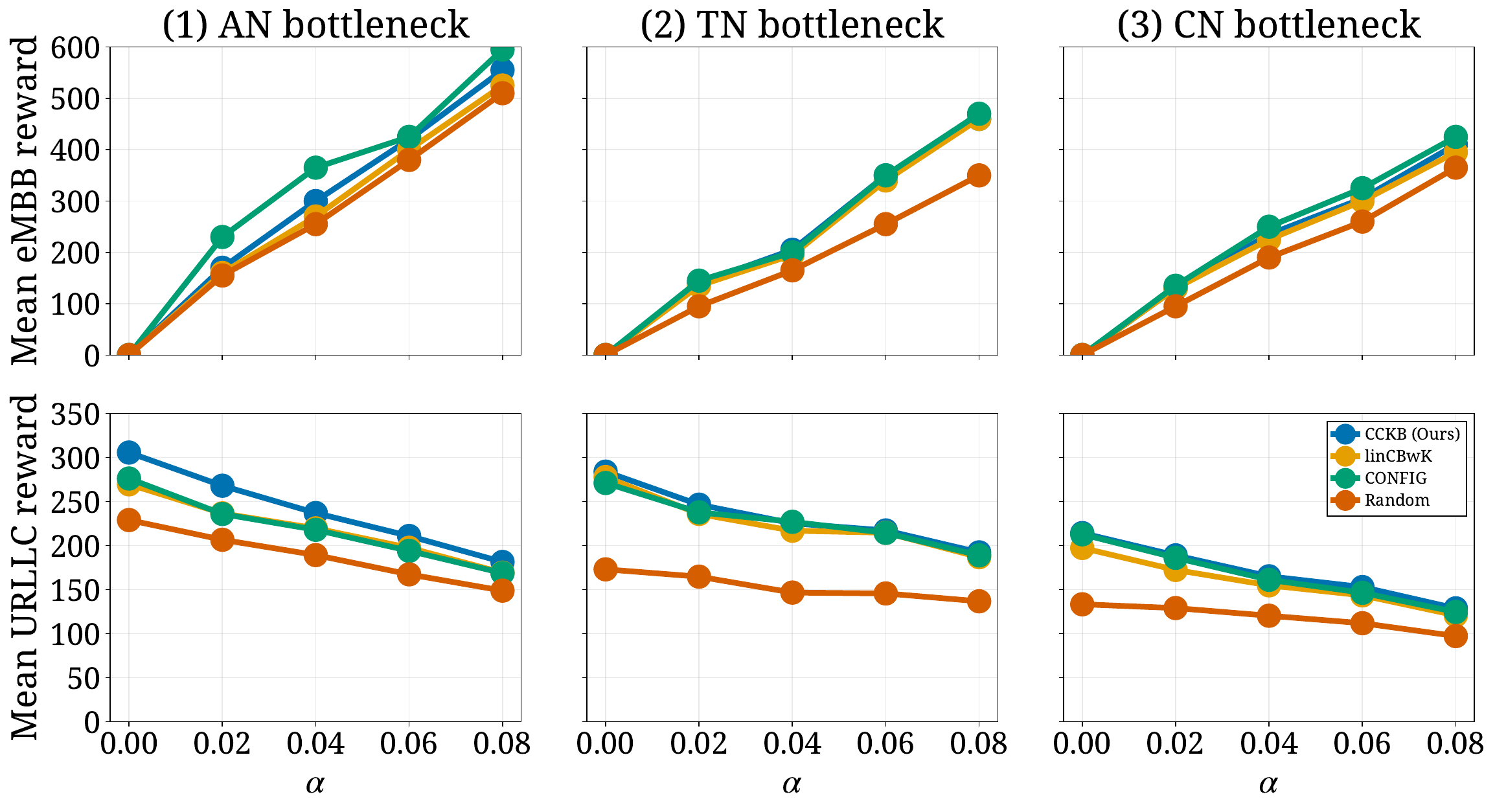}}
\hfill
\subfloat[Ring topology.\label{fig:mixed_traffic_reward_class_ring}]{%
\includegraphics[width=0.49\textwidth]{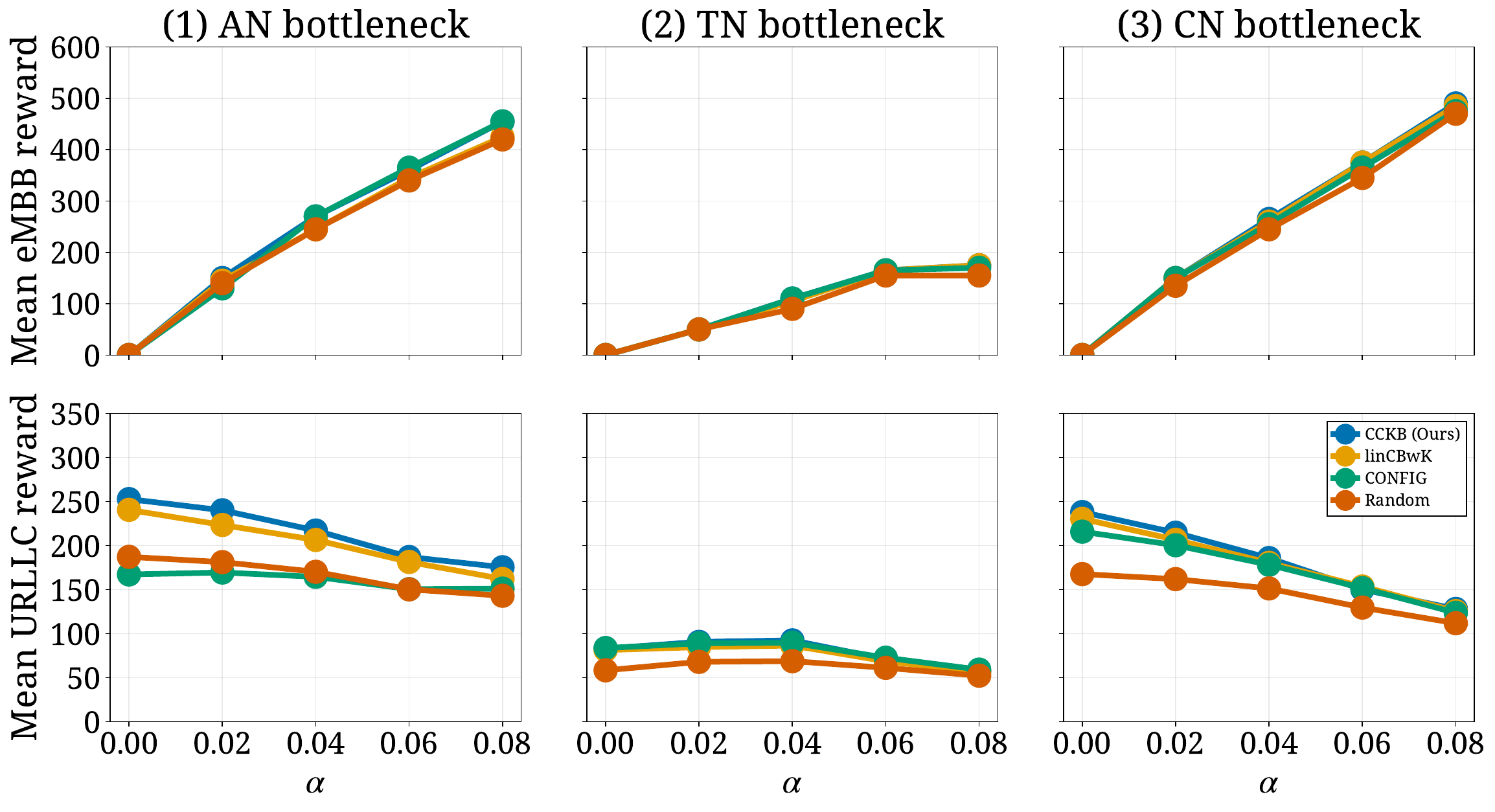}}
\caption{Mean class-specific reward contributions over 10 runs under mixed
traffic.  In each subfigure, the upper panels show the \ac{embb} contribution
\(50N_r^{\mathrm{B}}\), and the lower panels show the \ac{urllc} contribution
\(N_r^{\mathrm{U}}\), under AN, TN, and CN bottlenecks.  The sum of the two
contributions is Total Reward.}
\label{fig:mixed_traffic_reward_class}
\end{figure*}

Across the 24 mixed-traffic points, \ac{cckb} attains higher Total Reward than
\ac{lincbwk} at 23 points, with a mean difference of +19, and than Random at
all 24 points, with a mean difference of +77.  Against \ac{config}, \ac{cckb}
is higher at 15 points, with a mean difference of +6.  Averaged over the mixing
ratios, \ac{cckb} admits more \ac{urllc} slices in each topology--bottleneck
condition; eight of the nine points where \ac{config} attains higher Total
Reward are driven by the \ac{embb} component rather than by \ac{urllc}
admissions.

\end{document}

  \documentclass[src/main.tex]{subfiles}

\begin{document}

\section{Hyperparameter Robustness}
\label{appendix:hyperparameter_robustness}

\subsection{Setup}
\label{appendix:hyperparameter_robustness_setup}
The setup follows Sec.~\ref{subsubsec:ablation_relaxation_setup}.

\textbf{Sweep Design:}
In the \(\rho\) sweep, we fix \(c_{\beta}=0.1\) and vary
\(\rho\in\{0.1,0.4,0.7,1.0\}\).
In the \(c_{\beta}\) sweep, we fix \(\rho=1.0\) and vary
\(c_{\beta}\in\{0.1,0.4,0.7,1.0\}\).

\subsection{Result}
\label{appendix:hyperparameter_robustness_result}

The two sweeps exhibit different patterns.

\begin{figure}[t]
\centering
\includegraphics[width=\linewidth]{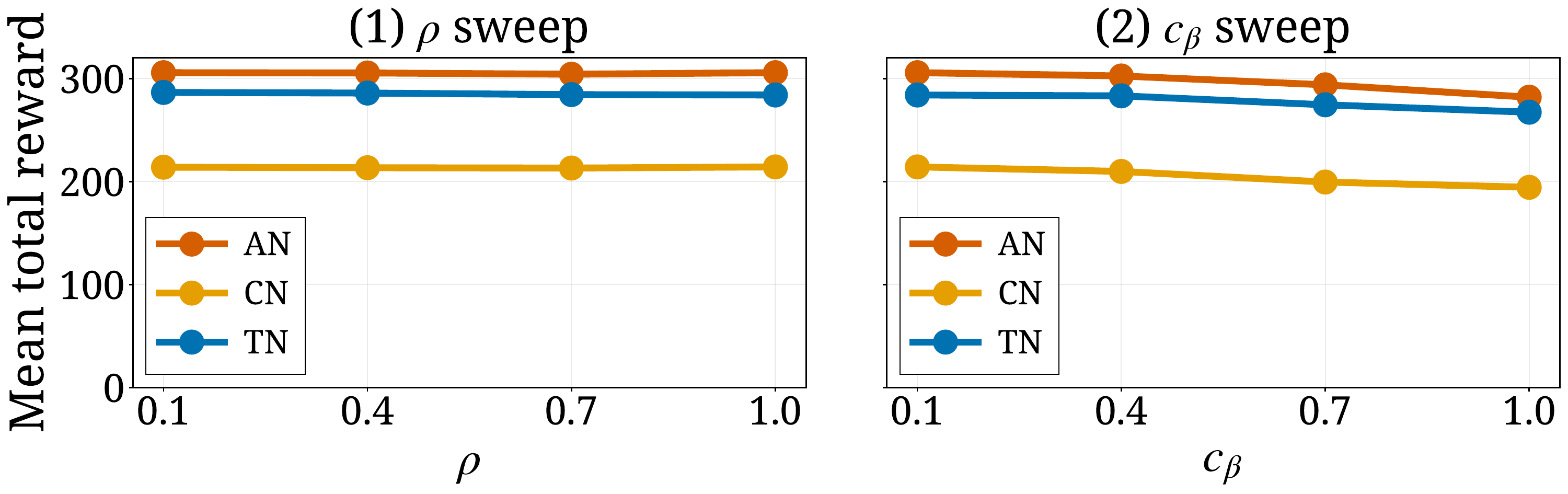}
\caption{Hyperparameter sensitivity of the proposed algorithm. (1) Effect of the dual-cap parameter \(\rho\) with \(c_{\beta}=0.1\). (2) Effect of the exploration-scale parameter \(c_{\beta}\) with \(\rho=1.0\). Each curve reports mean Total Reward over 10 runs for each bottleneck condition.}
\label{fig:hyperparameter_total_reward_comparison}
\end{figure}

The \(\rho\) sweep shows little visible change in Total Reward across the tested range
(Fig.~\ref{fig:hyperparameter_total_reward_comparison}).
The reward curves remain nearly flat for all three bottleneck conditions, suggesting that the proposed method is relatively insensitive to the precise value of the dual-cap parameter in this setting.

For the \(c_{\beta}\) sweep, larger values consistently degrade performance across all three bottleneck settings.
This trend suggests that, in the present problem setting, emphasizing exploration too strongly is not beneficial.

\end{document}

}

\end{document}